%% file: main.tex
\documentclass[a4paper,UKenglish,cleveref, autoref, thm-restate, numberwithinsect]{lipics-v2021}

\pdfoutput=1 
\hideLIPIcs  

\title{Tractable Hyperproperties for MDPs}

\author{Lina Gerlach}{RWTH Aachen University, Aachen, Germany}{gerlach@cs.rwth-aachen.de}{https://orcid.org/0009-0002-5506-6181}{}
\author{Tobias Winkler}{RWTH Aachen University, Aachen, Germany}{}{https://orcid.org/0000-0003-1084-6408}{}
\author{Erika {\'A}brah{\'a}m}{RWTH Aachen University, Aachen, Germany}{}{https://orcid.org/0000-0002-5647-6134}{}
\author{Borzoo Bonakdarpour}{Michigan State University, East Lansing, MI, USA}{}{https://orcid.org/0000-0003-1800-5419}{}
\author{Sebastian Junges}{Radboud University, Nijmegen, the Netherlands}{}{https://orcid.org/0000-0003-0978-8466}{}

\authorrunning{L.\ Gerlach, T.\ Winkler, E.\ {\'A}brah{\'a}m, B.\ Bonakdarpour, S.\ Junges} 

\relatedversion{} 

\acknowledgements{
We thank Oyendrila Dobe for fruitful discussions on earlier versions of this work.
We are grateful to Tim Quatmann for his continuous and reliable \storm support.
Lina Gerlach and Tobias Winkler are supported by the DFG RTG 2236/2 \textit{UnRAVeL}. 
Sebastian Junges is supported by the NWO VENI Grant ProMiSe (222.147).
This work is partially sponsored by the United States National Science Foundation (NSF) Award SaTC 2245114.
}

\nolinenumbers 

\input{packages}

\input{cmds}

\begin{document}

\maketitle

\input{abstract}

\input{intro}
\input{motivating-ex/motivating-ex}

\input{prelim}
\input{problem-statement}

\input{reach}
\input{reach_algo}
\input{reach_complexity}

\input{buechi}
\input{buechi_algo}
\input{buechi_complexity}

\input{conjunction}
\input{conjunction_algo}
\input{moa}
\input{conjunction_complexity}

\input{eval/eval}

\input{related-work}

\input{conclusion}

\bibliography{bibliography}

\appendix
\input{appendix/reach_full-algo}
\input{appendix/reach_algo_proofs}
\input{appendix/reach_compl_proofs}
\input{appendix/buechi_algo_proofs}

\input{appendix/buechi_compl_proofs}
\input{appendix/conjunction_algo_proofs}
\input{appendix/moa_proofs}
\input{appendix/conjunction_compl_proofs}

{\color{cavcolor}
\input{appendix/case-studies}
}

\end{document}

%% file: packages.tex
\usepackage{adjustbox}

\usepackage{graphicx} 
\usepackage[dvipsnames]{xcolor}

\usepackage{wrapfig}

\usepackage{float}

\usepackage{booktabs}
\usepackage{tabulary}
\usepackage{multirow}

\usepackage{multicol}
\usepackage{makecell}

\usepackage{enumitem}

\usepackage{mdframed}

\usepackage[linesnumbered,ruled,vlined]{algorithm2e}
\SetKwInput{Input}{Input}       
\SetKwInput{Output}{Output} 

\SetCommentSty{commentstyle}

\usepackage{tikz}
\usetikzlibrary{arrows,calc,shapes,shadows,decorations.pathmorphing,decorations.pathreplacing,automata,shapes.multipart,positioning}
\usetikzlibrary{fit,backgrounds} 

\tikzstyle{dist}=[circle,fill,inner sep=1.5pt] 
\tikzstyle{state}=[circle, draw, minimum size=2em] 
\tikzstyle{staterectangle}=[draw, rectangle, rounded corners]

\usepackage{mathtools}

\usepackage{amsthm}
\usepackage{thmtools}
\numberwithin{equation}{section}

\usepackage{amsmath,latexsym,amssymb,amsfonts,amscd,stmaryrd,textcomp} 
\usepackage{pifont}
\usepackage[T1]{fontenc}

\usepackage{epsdice} 
\usepackage{nicefrac}

\usepackage{ltl}

\usepackage{csquotes}

\usepackage{hyperref}
\hypersetup{%
	colorlinks=true,           
	allcolors=blue!70!black,   
	pdfstartview=Fit,          
	breaklinks=true}

\renewcommand{\cref}{\Cref} 
\usepackage{textcase} 

\usepackage{xspace} 
\usepackage{xargs} 
\usepackage{microtype}

%% file: cmds.tex
\newtheorem{assumption}{Assumption}

\newmdtheoremenv[backgroundcolor=white,linewidth=0.75pt,innertopmargin = 0mm]{framedproblem}{Problem}

\newcommand{\orcid}[1]{\href{https://orcid.org/#1}{\includegraphics[width=3mm]{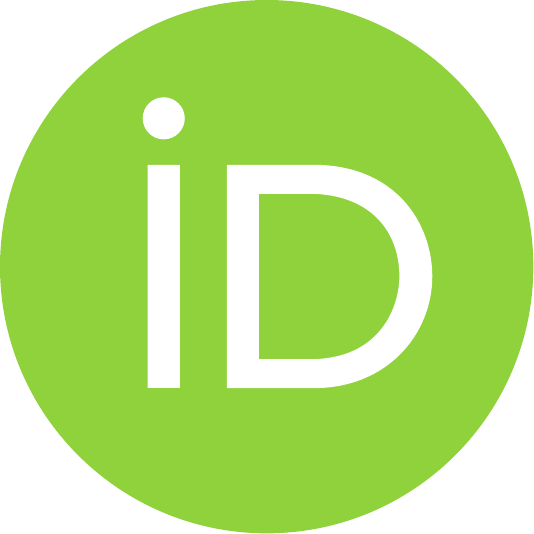}}}

\makeatletter
\newcommand\newtag[2]{#1\def\@currentlabel{#1}\label{#2}}
\makeatother

\newcommand\R{\mathbb R}

\newcommand\N{\mathbb N}

\newcommand\Q{\mathbb Q}
\newcommand\Qnn{\mathbb Q _{\geq 0}}

\newcommand{\Iverson}[1]{[#1]}
\newcommand{\Iff}{\mathbin{\textit{ iff }}}
\renewcommand{\iff}{\Leftrightarrow}
\renewcommand\implies{\Rightarrow}
\newcommand{\Implies}{\textit{ implies }}

\renewcommand{\phi}{\varphi}

\DeclareMathOperator*{\argmax}{arg\,max}

\DeclareMathOperator*{\opt}{opt}

\newcommand{\Expected}{\mathbb{E}}
\newcommand{\Prob}{\mathbb{P}}
\newcommand{\Distr}{\textit{Distr}}
\renewcommand{\Pr}{{\normalfont\text{Pr}}} 

\newcommand{\Limit}{\textit{Limit}}

\newcommand{\comp}{\bowtie}
\newcommand{\diseq}{\trianglelefteq}

\newcommand\mdp{\mathcal{M}}
\newcommand{\states}{S}
\newcommand\Act{\textit{Act}}
\newcommand\Trans{\mathbf{P}}

\newcommand{\quoti}[1][\mdp]{\widehat{#1}_{T_i}}
\newcommand{\quotT}[1][\mdp]{\widehat{#1}}

\newcommand{\processed}[1][\mdp]{\overline{#1}}
\newcommandx{\flip}[2][1=\mdp, 2=\state_{0}]{{#1}_{\circlearrowright#2}} 
\newcommand{\sflip}{\state_{\circlearrowright}}

\newcommand\indexc [1] [\mdp] {#1_c}

\newcommand\indexT [1] [\mdp] {#1_{\mathcal{T}}}

\newcommand\mdptup[1] [\refl]{(#1[S], \allowbreak #1[\Act], \allowbreak #1[\Trans])}

\newcommand\Paths[1][\mdp]{\textit{Paths}^{#1}}
\newcommandx\Pathsfrom[2][1=\mdp,2=\state]{\textit{Paths}^{#1}(#2)}
\newcommand\finPaths[1][\mdp]{\textit{Paths}_{\textit{fin}}^{#1}}
\newcommandx\finPathsfrom[2][1=\mdp, 2=\state]{\textit{Paths}_{\textit{fin}}^{#1}(#2)}
\newcommand{\last}{\textit{last}}
\newcommand{\lastact}{\textit{lastact}}

\newcommand{\state}{s}
\newcommand{\action}{\alpha}
\newcommand{\actiontwo}{\beta}

\newcommand{\rew}{\textit{R}}

\newcommand{\EC}{\textit{EC}}
\newcommand{\MEC}{\textit{MEC}}

\newcommand{\ec}{C}

\newcommand{\mec}{C}

\newcommand\dtmc{\mathcal{D}}
\newcommand\dtmctup[1] [\refl]{(#1[S], \allowbreak #1[\Trans])}

\newcommand\sched{\sigma}
\newcommand{\modes}{Q}
\newcommand{\mode}{q}
\newcommand{\modef}{\textit{mode}}
\newcommand{\start}{\textit{init}}
\newcommand{\act}{\textit{act}}

\newcommand{\Scheds}[1][\mdp]{\Sigma^{#1}}
\newcommand{\SchedsMD}[1][\mdp]{\Sigma^{#1}_{\textit{MD}}}
\newcommand{\SchedsMR}[1][\mdp]{\Sigma^{#1}_{\textit{MR}}}

\newcommand{\AP}{\textit{AP}}
\newcommand{\ap}{\textit{a}}

\newcommand{\IL}{\textit{IL}}
\newcommand{\VarSched}{V_{\textit{Sched}}}

\newcommand{\RelReach}{\textsf{\small RelReach}\xspace} 
\newcommand{\RelReachMD}{\textsf{\small RelReach$_{\textsf{MD}}$}\xspace} 
\newcommand{\RelBuechi}{\textsf{\small RelBuechi}\xspace} 
\newcommand{\RelBuechiMD}{\textsf{\small RelBuechi$_{\textsf{MD}}$}\xspace} 
\newcommand{\ConjRelReach}{\textsf{\small MORelReach}\xspace}

\newcommand{\ConjRelReachMD}{\textsf{\small MORelReach$_{\textsf{MD}}$}\xspace} 
\newcommand{\MORelReach}{\ConjRelReach} 
\newcommand{\MORelReachMD}{\ConjRelReachMD} 

\newcommand{\RelReachNotEq}{\textsf{\small RelReach}$_{\geq,>,\not\approx}$\xspace}
\newcommand{\RelBuechiNotEq}{\textsf{\small RelBuechi}$_{\geq,>,\not\approx}$\xspace}

\newcommand{\HyperPCTL}{\textsf{\small HyperPCTL}\xspace}

\newcommand{\PCTL}{\textsf{\small PCTL}\xspace}
\newcommand{\PHL}{\textsf{\small PHL}\xspace}

\newcommand\varstate{\hat{s}}
\newcommand\varsched{\hat{\sched}}

\newcommand\Globally{\LTLglobally}
\newcommand\Finally{\LTLfinally}

\newcommand{\PTIME}{\textsf{\small PTIME}\xspace}
\newcommand{\NP}{\textsf{\small NP}\xspace}
\newcommand{\coNP}{\textsf{\small coNP}\xspace}
\newcommand{\PSPACE}{\textsf{\small PSPACE}\xspace}
\newcommand{\EXPTIME}{\textsf{\small EXPTIME}\xspace}

\newcommand{\numconj}{l}
\newcommand{\numsum}{m}
\newcommand{\numsched}{n}

\newcommand{\init}{\textit{init}}

\newcommand{\statei}{\state_0}
\newcommand{\statea}{\state_1}
\newcommand{\stateb}{\state_2}

\newcommand{\attractor}{A}
\newcommand{\lb}[1]{\underline{#1}}
\newcommand{\ub}[1]{\overline{#1}}

\newcommand{\comb}{\mathscr{C}}
\newcommand{\relInd}{\mathit{ind}}

\newcommand{\conjpart}{\textit{ConjPart}} 
\newcommand{\indexconjpartelt}[1][\mdp]{#1_{X}}

\newcommand{\conjpartelt}{X}

\newcommand{\combined}[1][\mdp]{#1_{\comb}}
\newcommand{\processedcombined}[1][\mdp]{\processed[{\combined[#1]}]}

\newcommand{\class}[1]{\llbracket #1 \rrbracket}

\newcommand{\collapsedStates}{\states_{\textit{collapsed}}}
\newcommand{\mapPaths}{f}
\newcommand{\mapStates}{f}

\newcommand{\sdead}{\states_\bot}
\newcommand{\nnstates}{\states^{\textit{nn}}} 
\newcommand{\vis}{\mathit{vis}}

\newcommand{\InfS}[1][{\processedsched}]{\mathit{InfS}^{#1}}
\newcommand{\InfA}[1][{\processedsched}]{\mathit{InfA}^{#1}}
\newcommand{\FinS}[1][{\processedsched}]{\mathit{FinS}^{#1}}
\newcommand{\FinA}[1][{\processedsched}]{\mathit{FinA}^{#1}}
\newcommand{\speciald}{d^*}
\newcommand{\specialact}{\action^*}
\newcommand{\specialstate}{\state^*}

\newcommand{\TA}{\textbf{TA}\xspace}
\newcommand{\TAone}{\textbf{TA(1)}\xspace}
\newcommand{\TAtwo}{\textbf{TA(2)}\xspace}

\newcommand{\TS}{\textbf{TS}\xspace}
\newcommand{\PW}{\textbf{PW}\xspace}
\newcommand{\PWone}{\textbf{PW(1)}\xspace}
\newcommand{\PWtwo}{\textbf{PW(2)}\xspace}

\newcommand{\SD}{\textbf{SD}\xspace}
\newcommand{\VN}{\textbf{VN}\xspace}
\newcommand{\RT}{\textbf{RT}\xspace}

\newcommand{\IsraeliJalfon}{\textsc{IJ}\xspace}

\newcommand{\Init}{\textit{Init}}

\definecolor{mGreen}{rgb}{0,0.6,0}
\definecolor{mGray}{rgb}{0.5,0.5,0.5}
\definecolor{mPurple}{rgb}{0.58,0,0.82}
\definecolor{backgroundColour}{rgb}{0.95,0.95,0.92}
\lstdefinestyle{CStyle}{
	backgroundcolor=\color{backgroundColour},   
	commentstyle=\color{mGreen},
	keywordstyle=\color{magenta},
	numberstyle=\tiny\color{mGray},
	stringstyle=\color{mPurple},
	basicstyle=\footnotesize,
	breakatwhitespace=false,         
	breaklines=true,                 
	captionpos=b,                    
	keepspaces=true,                 
	numbers=left,                    
	numbersep=2pt,                  
	showspaces=false,                
	showstringspaces=false,
	showtabs=false,                  
	tabsize=1,
	language=C
}

\newcommand{\TO}{\textbf{TO}}
\newcommand{\OOM}{\textbf{OOM}}

\newcommand{\HyperProb}{\textsc{HyperProb}\xspace} 
\newcommand{\HyperPaynt}{\textsc{HyperPAYNT}\xspace} 
\newcommand{\storm}{\textsc{Storm}\xspace}
\newcommand{\uppaalsmc}{\textsc{\small UPPAAL-SMC}\xspace}
\newcommand{\PRISM}{\textsc{PRISM}\xspace}

\definecolor{DarkGreen}{rgb}{0.0, 0.2, 0.13}
\definecolor{Blue}{rgb}{0.0, 0.0, 1.0}
\definecolor{LightSeaGreen}{rgb}{0.13, 0.7, 0.67}

\colorlet{biasupper}{blue}
\definecolor{biaslower}{RGB}{ 87 171  39}

\colorlet{cavcolor}{black}

%% file: abstract.tex
\begin{abstract}

\emph{Probabilistic hyperproperties} describe probabilistic relations between multiple sets of executions in a stochastic system.
Prominent examples include information-theoretic characterizations of security and privacy policies. 
However, model checking for existing probabilistic hyperlogics, such as \HyperPCTL and \PHL, is undecidable in Markov decision processes (MDPs).
In this paper, we study an underexplored problem: the verification of fragments of probabilistic hyperproperties that relate the probabilities of different events to each other, possibly across independent executions of an MDP. Representative verification questions include: \emph{Can two different target states be reached from the same initial state with the same probability?} (different events), \emph{Can a given target state be reached from two different initial states with the same probability?} (same event, independent executions), and natural combinations of these forms.
Besides reachability, our \emph{relational} probabilistic properties cover safety, B\"uchi, and coB\"uchi objectives. They can also be combined conjunctively, thereby generalizing standard \emph{multi-objective} MDP properties.
We provide efficient algorithms for relevant classes of relational properties, while proving computational hardness and completeness results for others.
An implementation of our approach outperforms solvers for more general probabilistic hyperlogics by orders of magnitude on the subset of their benchmarks that lies within our fragment.

\end{abstract}

%% file: intro.tex
\section{Introduction}
\label{sec:intro}

\emph{Markov decision processes} (MDPs) are a standard modeling formalism for systems that combine probabilistic branching with nondeterministic choices.
At each state of an MDP, a nondeterministic choice is made among available actions, each of which induces a probability distribution over successor states.
A \emph{scheduler} resolves this nondeterminism; schedulers may randomize over actions and may use memory, i.e., depend on the execution history.

Classic verification questions for MDPs, such as
\begin{quote}
\emph{``Can some scheduler reach a target state with at least a given probability?''}
\end{quote}
focus on optimizing the probability of a \emph{single} event.
In contrast, in this paper, we consider \emph{relational properties}.
A natural example question is:
\begin{quote}
\emph{``Do all schedulers reach state $s$ with \underline{the same} probability as state $t$?''}
\end{quote}
In particular, relational properties can \emph{compare} the probabilities of two different events.
This goes beyond standard verification problems, including existing multi-objective frameworks, but is captured by \emph{probabilistic hyperlogics}~\cite{abrahamHyperPCTLTemporal2018,dimitrovaProbabilisticHyperproperties2020,abrahamProbabilisticHyperproperties2020}; see~\cref{sec:related-work} for a detailed discussion.
We define relational properties formally in \Cref{sec:problem-statement} on \cpageref{sec:problem-statement}.%

We show that relational properties form a practical, tractable subclass of probabilistic hyperproperties:
Relational properties can be verified algorithmically---often in polynomial time---whereas model checking probabilistic hyperproperties is undecidable in general~\cite{dobeModelChecking2022,dimitrovaProbabilisticHyperproperties2020}.
A key reason for this tractability is that, once a scheduler is fixed, all probabilities appearing in a relational property (e.g., reachability or Büchi probabilities) can be evaluated independently and directly on the original MDP.
This is not the case for general probabilistic hyperproperties.
For example, evaluating hyperproperties like
\begin{quote}
	\emph{``Does it hold that, with probability 1, all pairs of schedulers either both visit a given state $s$, or neither visits $s$, at every time step?''}
\end{quote}
requires additional expensive constructions such as \emph{self-composition}~(see, e.g.,~\cite{abrahamHyperPCTLTemporal2018,dimitrovaProbabilisticHyperproperties2020}), even when schedulers are fixed.
Nevertheless, as illustrated in~\cref{sec:examples}, many practically relevant probabilistic hyperproperties fall into one of our tractable classes of relational properties.

\paragraph*{Results}
In this paper, we consider several variations of relational properties involving 
reachability, safety, B\"uchi or coB\"uchi objectives.
Examples include {\color{cavcolor} asking for a single scheduler ensuring that a set of states is 
reached with higher probability than a different set of states (inequality properties), or that these 
probabilities are approximately or exactly the same (equality properties).
These probabilities can be evaluated from the same or from different initial states, and also under 
the same or different schedulers; the latter allows us to express, e.g., that any pair of schedulers 
induce roughly the same reachability probability.
}%
More precisely, we consider the scheduler synthesis problem for conjunctions of comparisons (e.g., $\leq$, $=$, $\approx_\epsilon$, etc.) between (weighted) sums of probabilities.
We formally introduce our class of relational properties in \cref{sec:problem-statement} and subsequently focus on three fragments:
\begin{enumerate}
    \item The conjunction-free fragment involving only reachability and safety\footnote{In our formalism, we can readily reduce safety to reachability and coBüchi to Büchi; see \Cref{sec:problem-statement}.} objectives~(\cref{sec:reach}). This fragment was already covered in~\cite{gerlachEfficientProbabilistic2025}.
    \item The conjunction-free fragment over B\"uchi and coB\"uchi objectives~(\cref{sec:buechi}).
    \item The B\"uchi-free fragment~(\cref{sec:conjunction}), i.e., conjunctions over comparisons between weighted sums of reachability objectives.
\end{enumerate}
We describe motivating examples covering all three fragments in~\cref{sec:examples}. 
For each fragment we propose a model-checking algorithm and investigate the complexity of the model-checking problem, as outlined in the following.

\subparagraph*{Verifying relational properties.}
Given an MDP and a relational property from one of the three 
fragments, we provide an algorithm to decide whether or not the MDP satisfies the property.
{\color{cavcolor}
For relational reachability properties, the key insight is that (possibly randomized memoryful) witness schedulers can be constructed by translating the given property to expected reward computations in a series of mildly transformed MDPs~(\cref{sec:reach_algo}). 
For inequality properties, it suffices to optimize the total expected reward. 
For (approximate) equality, the main idea is to construct a randomized memoryful scheduler from the schedulers witnessing the corresponding two inequality properties. 
}

Relational B\"uchi properties can be reduced back to relational reachability properties on a variation 
of the standard maximal end 
component (MEC) quotient of the MDP~(\cref{sec:buechi_algo}).
Note that, while there is a standard construction for transferring a scheduler \emph{optimizing} some B\"uchi objective to a scheduler optimizing some reachability probability on the MEC quotient~\cite{baierPrinciplesModel2008}, our setting here is more challenging because we want to transfer \emph{arbitrary} schedulers while preserving the exact probabilities; we use insights from multi-objective model checking~\cite{baierCertificatesWitnesses2024}.

For verifying conjunctive relational reachability properties, we extend the approach from conjunction-free reachability properties, which results in a multi-objective achievability query~(\cref{sec:conj_algo}).
Since the constructed multi-objective queries are, to the best of our knowledge, not covered by existing work, we give an explicit verification procedure in~\cref{sec:moa}.

\subparagraph*{Computational complexity.} 
The algorithms for relational reachability or B\"uchi properties outlined above are exponential only in the number of different target sets that occur in the property, i.e., the algorithms are \emph{fixed-parameter tractable}~\cite{groheDescriptiveParameterized1999}. 
The problem for reachability objectives is in general \PSPACE-hard and for B\"uchi objectives it is 
\NP-complete, but in both cases it can be solved in \PTIME in the size of the input under one of various (mild) assumptions. 
For conjunctive relational reachability properties, the complexity results differ depending on whether we include the comparison operator $\not\approx_\epsilon$, which requires two values to be at least $\epsilon$ apart and thus corresponds to a \emph{disjunction} inside the conjunction.
If we omit $\not\approx_\epsilon$, the algorithm is again only exponential in the number of different target sets.
If we include $\not\approx_\epsilon$, however, the problem is still \NP-complete if all probability operators share the same, absorbing target set.
\cref{tab:complexity_overview_selected} gives an overview of the model-checking complexities of selected fragments over general schedulers.

For all three fragments, {\color{cavcolor} when restricting the schedulers to be \emph{memoryless and deterministic} (MD), several types of equality properties are strongly \NP-hard.
We also list various fragments where we can compute MD schedulers in polynomial time.
}

We study the computational complexity of the model-checking problem for the three fragments in greater detail in \cref{sec:reach_complexity,sec:buechi_complexity,sec:conj_complexity}, respectively.
We compare and contrast our complexity results in \cref{tab:complexity_overview_full}.

\begin{table}[t]
    \caption{Complexity of selected classes of relational properties over general (memoryful, randomized) schedulers, where $\epsilon>0$.
    The upper half of the table presents single-objective properties, the lower half multi-objective ones.
    }
    \label{tab:complexity_overview_selected}
    \setlength\tabcolsep{0pt}
    \centering
    \begin{tabular*}{\linewidth}{@{\extracolsep{\fill}} l  l }
        \toprule
         \bf Property class & \bf  Complexity 
         \\
        \midrule
        $\exists \sched .\ \sum_{i=1}^{m} \Pr^{\sched}_{\state}(\Finally T_i) \geq 1$
        & \PSPACE-hard, in \EXPTIME [Th.~\ref{th:general_EXPTIME}]
        \\
        $\exists \sched .\ \sum_{i=1}^{m} \Pr^{\sched}_{\state}(\Globally\Finally T_i) \geq 1$
        & \NP-hard [Th.~\ref{th:buechi_NP-complete}]
        \\
        $\exists \sched .\ \sum_{i=1}^{m} \Pr^{\sched}_{\state}(\Finally T_i) \geq m$
        & \PSPACE-hard~\cite{randourPercentileQueries2017}, in \EXPTIME [Th.~\ref{th:general_EXPTIME}]
        \\
        $\exists \sched .\ \sum_{i=1}^{m} \Pr^{\sched}_{\state}(\Globally\Finally T_i) \geq m$
        & \PTIME~\cite{randourPercentileQueries2017} [Rem.~\ref{remark:buechi_sim-a-s}]
        \\
        \midrule 
        $\exists \sched .\ \Pr^{\sched}_{\state}(\Finally T_1) = \ldots = \Pr^{\sched}_{\state}(\Finally T_m)$
        & \PSPACE-hard, in \EXPTIME [Th.~\ref{th:conj_exptime}]
        \\
        $\exists \sched .\ \Pr^{\sched}_{\state_1}(\Finally T) = \ldots = \Pr^{\sched}_{\state_m}(\Finally T)$
        & strongly \NP-complete [Th.~\ref{th:conj_special}\ref{item:conj_fixed-param}]
        \\
        $\exists \sched .\ \Pr^{\sched}_{\state_1}(\Finally T) \not\approx_\epsilon \ldots \not\approx_\epsilon \Pr^{\sched}_{\state_m}(\Finally T)$
        & \PTIME [Th.~\ref{th:conj_special}\ref{item:conj_fixed-param}]
        \\
        $\exists \sched_1, \ldots, \sched_m .\ \Pr^{\sched_1}_{\state}(\Finally T_1) = \ldots = \Pr^{\sched_m}_{\state}(\Finally T_m)$
        & \PTIME [Th.~\ref{th:conj_special}\ref{item:conj_independent}]
        \\
        $\exists \sched_1, \ldots, \sched_m .\ \Pr^{\sched_1}_{\state}(\Finally T_1) \not\approx_{\epsilon} \ldots \not\approx_{\epsilon} \Pr^{\sched_m}_{\state}(\Finally T_m)$
        & \PTIME [Th.~\ref{th:conj_special}\ref{item:conj_independent}]
        \\
        \bottomrule
    \end{tabular*}
\end{table}

\paragraph*{Contributions and Structure} 

This paper significantly extends~\cite{gerlachEfficientProbabilistic2025}, which is restricted to 
relational \emph{reachability} properties (\cref{sec:reach}), while this work additionally considers relational 
\emph{B\"uchi} properties (\cref{sec:buechi}) and \emph{multi-objective} relational reachability properties (\cref{sec:conjunction}).
Additionally, we extend existing work on multi-objective achievability queries to the setting of 
(1) expected reward objectives that collect \emph{positive or negative} finite reward on any path 
combined with 
(2) comparison operators $\{>,\geq,\approx_\epsilon,\not\approx_\epsilon \mid \epsilon \in \Qnn\}$.

{\color{cavcolor}
In summary, in this paper we formally introduce the class of relational properties and present model-checking algorithms for these properties, which go beyond standard probabilistic and multi-objective properties while remaining tractable 
(see~\Cref{sec:related-work} for a discussion of related work). 
The tractability is in sharp contrast with the more general probabilistic hyperlogics.
}%
Our key contributions are efficient model checking algorithms (\cref{sec:reach_algo,sec:buechi_algo,sec:conj_algo}) and a study of the complexity landscape for relational properties (\cref{sec:reach_complexity,sec:buechi_complexity,sec:conj_complexity}). 
In \cref{sec:eval}, we provide a prototypical implementation on top of \storm~\cite{henselProbabilisticModel2022} 
and experimentally demonstrate that our approach {\color{cavcolor} can solve some standard benchmarks for probabilistic hyperlogics orders of magnitudes faster than the state of the art~\cite{andriushchenkoDeductiveController2023,dobeHyperProbModel2021}}.

%% file: motivating-ex/motivating-ex.tex
\section{Motivating Examples}
\label{sec:examples}

To motivate the application of relational probabilistic properties,  we present 
three examples, covering all three fragments of relational properties we consider.
We illustrate the main steps of our algorithms and foreshadow some of our complexity results.


\input{motivating-ex/von-neumann}

\subsection{Simulating a Coin by Random Biased Bits}
\label{ex:motivating_ex}
\label{ex:reach_motivating-ex}

{\color{cavcolor}
We wish to simulate an unbiased, perfectly random coin flip using an infinite stream of possibly \emph{biased} random bits,
each of which is $0$ with an unknown but fixed probability $0<p<1$ and otherwise $1$.
The following simple solution is due to von Neumann~\cite{vonneumannVariousTechniques1951}:
Extract the first two bits from the stream; if they are different, return the value of the first; otherwise try again.
Now, consider a variation of the problem where the stream comprises random bits with 
\emph{different}, unknown biases $p_0, p_1, {\ldots}$ which are, however, all known to lie in an 
interval $[\underline{p}, \overline{p}] \subset (0,1)$~\cite{gerlachEfficientProbabilistic2025}.
Is von Neumann's solution still applicable in this new situation?
To address this question for a concrete interval $[\underline{p}, \overline{p}]$, say $[0.59,0.61]$, 
we may model the situation as shown in~\Cref{fig:illustrativeExample} and formalize the 
corresponding relational reachability property as:
\begin{align}
    \forall \sched .\
    \Pr_{s_0}^{\sched}(\Finally \{01\}) \,\approx_{\epsilon}\, \Pr_{s_0}^{\sched}(\Finally \{10\})
    \tag{$\dagger$}\label{eq:vonNeumannExampleProperty}
\end{align}
where $\sched$ is a universally quantified scheduler, 
$s_0$ is the initial state, and $\approx_{\epsilon}$ means approximate equality up to an absolute error of $\epsilon \geq 0$.
Note that the universal quantification in \eqref{eq:vonNeumannExampleProperty} is over \emph{general} schedulers that may use both unbounded memory and randomization.
This is essential to model the problem properly:
Without randomization, all biases would be either $\underline{p}$ or $\overline{p}$ and a bounded-memory scheduler would induce an ultimately periodic stream of biases.

Using the techniques presented in this paper, we can establish automatically that, as expected, von Neumann's trick continues to work \enquote{approximately} in the new setting.
More precisely, in~\cref{sec:eval_new} we use our algorithm to automatically verify that \eqref{eq:vonNeumannExampleProperty} is false for $\epsilon = 0$ (exact equality), but holds if we relax the constraint to $\epsilon = 0.1$. 
}%
Our algorithm proceeds as follows~(see \cref{sec:reach_algo}):
First, we unfold the MDP w.r.t.\ the target states $01$ and $10$, as depicted in~\cref{fig:illustrativeExampleUnfolded}.
Then, we define reward structures $\rew_{01}, \rew_{10}$ on the unfolded MDP $\mdp'$, collecting reward 1 in states $01$ and $10$, respectively.
In order to check whether the desired property holds, we can then check whether $\forall \sched \in \Scheds[\mdp'] .\ \Expected^{\mdp',\sched}_{\state_0}(\rew_{01} - \rew_{10}) \approx_\epsilon 0$, which is equivalent to checking that $\min_{\sched \in \Scheds[\mdp']} \Expected^{\mdp',\sched}_{\state_0}(\rew_{01} - \rew_{10}) \geq - \epsilon \land \max_{\sched \in \Scheds[\mdp']} \Expected^{\mdp',\sched}_{\state_0}(\rew_{01} - \rew_{10}) \leq \epsilon$.
Note that, in this example, all probability operators share the same scheduler variable and initial state and all target states are absorbing.
Our algorithm runs in time polynomial in the size of the input since both target states are absorbing~(\cref{th:general_PTIME}\ref{item:absorbing}).
For $\epsilon=0$, our algorithm actually finds a \emph{memoryless deterministic} counterexample in polynomial time since both probability operators have the same scheduler variable and initial state and both target states are absorbing~(\cref{th:MD_PTIME}\ref{item:absorbing}).
In fact, checking the property over memoryless deterministic schedulers is equivalent to checking it over general schedulers~(\cref{th:MD_suffice}), independent of $\epsilon$.

\input{motivating-ex/israeli-jalfon}

\subsection{Israeli \& Jalfon's Self-Stabilizing Protocol}
\label{ex:buechi_motivating-ex}

Israeli \& Jalfon's self-stabilizing protocol operates on a ring of $N$ processes that pass on 
tokens between them~\cite{israeliTokenManagement1990}.
At each step in time, exactly one of the processes currently possessing a token can be scheduled.
If a process with a token is scheduled, it passes this token on to either its left or right neighbor with equal probability.
If a process possesses several tokens, they are merged into a single token.
We can model this as an MDP $\mdp_o$ as shown in \cref{fig:israeli-jalfon} for $N=3$.
The classical question for this setting is whether from all initial states (i.e., all possible distributions of tokens among the processes) and under all schedulers, the protocol almost-surely reaches a \emph{stable} state, i.e., a state where exactly one of the processes has a token. 
Indeed, the protocol satisfies this property. 
However, this property is also satisfied in an asymmetric variant of the protocol where process $N$ malfunctions and does not pass on the token, which can be modeled as an MDP $\mdp_a$ as shown in \cref{fig:israeli-jalfon_asym} for $N=3$.
In this setting, we still reach a stable configuration with probability 1 from any initial state under any scheduler, but with probability 1 we end up in the state where only process $N$ has a token, hence only process $N$ has a token infinitely often.
This exemplifies that asking whether a stable configuration is reached does not check whether the token is passed around between the processes \emph{in a fair manner}.
Let us thus check the following, more precise, property, where $\statei \in \states$ ranges over all states and $Q_i$ is the set of states where process $i$ has a token: 
\begin{align*}
    \bigwedge_{\statei \in \states} 
    \bigwedge_{i=1}^{N-1} 
    \forall \sched .\ 
    \Pr^{\sched}_{\statei}(\Globally \Finally Q_i) = \Pr^{\sched}_{\statei}(\Globally \Finally Q_{i+1})~,
\end{align*}
which is a conjunction of relational B\"uchi properties.

Let us illustrate how our algorithm proceeds to check the conjunct with $\state_0 = 111$ and $i=2$ for both variations~(cf.~\cref{sec:buechi_algo}):
We first build a variation of the standard MEC quotient that takes into account the target sets. 
Note that, in this example, both MDPs have only a single MEC which does not contain other end components; see~\cref{fig:buechi_running-ex} for an MDP with a more complex structure of the end components.
For the original protocol we construct $\quotT[\mdp_o]$ by collapsing the single MEC $\{001, 010, 100\}$ and adding a transition from the collapsed state to a sink state that reflects that, when staying in this MEC we must see both $Q_2$ and $Q_3$ infinitely often. Observe that it is not possible to see only $Q_2$, only $Q_3$, or neither $Q_2$ nor $Q_3$ infinitely often.
We define the \emph{success sets} for $Q_2$ and $Q_3$ as the set of all sink states reflecting that $Q_2$, resp.~$Q_3$ is visited infinitely often, i.e., $U_{Q_2} = U_{Q_3} = \{\bot^{Q_2,Q_3}\}$.
The model for the asymmetric variant only has a single-state MEC $\{001\}$; we construct the MEC quotient $\quotT[\mdp_a]$ by removing the self-loop and instead adding a transition to a sink state that reflects that while staying in this MEC we must see $Q_3$ infinitely often. Then, the success set of $Q_2$ is $U_{Q_2} = \emptyset$ since there is no sink state that reflects visiting $Q_2$ infinitely often, and the success set for $Q_3$ is defined as $U_{Q_3} = \{\bot^{Q_3}\}$.
The desired relational B\"uchi property for $\mdp_o$ and $\mdp_a$ is thus equivalent to the following relational reachability properties: 
\begin{align*}
    &
    \forall \sched \in \Scheds[{\quotT[\mdp_o]}] .\ 
    \Pr^{\quotT[\mdp_o],\sched}_{111}(\Finally \{\bot^{Q_2,Q_3}\}) = \Pr^{\quotT[\mdp_o],\sched}_{111}(\Finally \{\bot^{Q_2,Q_3}\}) 
    ~ 
    \textit{, and} 
    \\ &
    \forall \sched \in \Scheds[{\quotT[\mdp_a]}] .\ 
    \Pr^{\quotT[\mdp_a],\sched}_{111}(\Finally \emptyset) = \Pr^{\quotT[\mdp_a],\sched}_{111}(\Finally \{\bot^{Q_3}\}) 
    ~\textit{, respectively.}
\end{align*}
Thus, the property holds for $\mdp_o$ (the standard protocol) but not for $\mdp_a$ (the asymmetric variant), as also confirmed by the experimental evaluation of our algorithm in ~\cref{sec:eval_new}.

Our algorithm runs in time exponential in the number of probability operators~(\cref{th:buechi_correctness}).
Since both probability operators have the same scheduler variable and initial state~(\cref{th:buechi_MD}\ref{item:buechi_MD_n-comb}), it induces a \emph{memoryless deterministic} counterexample scheduler for the asymmetric variant (which can be falsified), 
and checking the property over general schedulers is equivalent to checking it over memoryless deterministic schedulers~(\cref{cor:buechi_MD_suffice}), for both the original protocol and the asymmetric variant.

\subsection{Simulating a Die by Biased Coins}
\label{ex:conj_motivating-ex}

\input{motivating-ex/dice-simulation}

We would like to investigate whether we can extend Von Neumann's trick (see \cref{ex:motivating_ex}) towards simulating a 6-sided die by repeatedly tossing biased coins: Instead of generating a distribution with equal probability for \emph{two} outcomes we now want to generate a distribution with equal probability for \emph{six} outcomes.
There exist different protocols for simulating a die via \emph{unbiased} coin throws, e.g., the Knuth-Yao algorithm~\cite{knuthComplexityNonuniform1976} and the fast dice roller~\cite{lumbrosoOptimalDiscrete2013}. 
For fair coins, both can be modeled as DTMCs where every state either has a uniform successor distribution over two states or is absorbing, as shown in \cref{fig:knuth-yao-fair,fig:fast-dice-roller-fair}. 
As before, let us now assume that we do not have a biased coin but instead can choose any bias from a fixed interval $[\underline{p},\overline{p}]$ every time we flip a coin. 
We extend the DTMCs modeling the protocols for unbiased coins to MDPs modeling the protocols for biased coins by replacing every action in the DTMC with two actions, one for the lower and one for the upper bias bound, as shown in \cref{fig:knuth-yao-biased,fig:fast-dice-roller-biased}.
Note that these MDPs can also be viewed as interval DTMCs~\cite{jonssonSpecificationRefinement1991}; inspired by this observation we briefly address relational properties for interval DTMCs in~\cref{rem:interval} on p.~\pageref{rem:interval}.

For both protocols, we now wish to check whether for every coin throw we can choose a bias such that the probability of each outcome $\epsdice{1}, \epsdice{2}, \epsdice{3}, \epsdice{4}, \epsdice{5}, \epsdice{6}$ is approximately the same.
Formally, we can state this as the following multi-objective relational reachability property
\begin{align*}
    \exists \sched .\ 
    \Pr^{\sched}_{\state_0}(\Finally \{\epsdice{1}\}) \approx_\epsilon \Pr^{\sched}_{\state_0}(\Finally \{\epsdice{2}\}) 
    ~\land~ \ldots ~\land~
    \Pr^{\sched}_{\state_0}(\Finally \{\epsdice{6}\}) \approx_\epsilon \Pr^{\sched}_{\state_0}(\Finally \{\epsdice{1}\})
    ~.
\end{align*}
If the interval $[\underline{p},\overline{p}]$ contains $0.5$, then the property is satisfied trivially because of the correctness of the original protocol with fair coins.
If we again consider the interval $[0.59,0.61]$, as for Von Neumann's trick (\cref{ex:reach_motivating-ex}), the picture is less clear.
Let us illustrate how our algorithm proceeds to check the property for each MDP~(cf.~\cref{sec:conj_algo}):
Similarly to the approach for relational reachability properties, we first unfold the MDP w.r.t.\ all target states $\epsdice{1}, \ldots, \epsdice{6}$ and define reward structures $\rew_{\epsdice{1}}, \ldots, \rew_{\epsdice{6}}$ that collect reward 1 on visiting the respective target state.
Then, we define reward structures capturing each conjunct by letting $\rew^{1} = \rew_{\epsdice{1}} - \rew_{\epsdice{2}}$, \ldots, $\rew^6 = \rew_{\epsdice{6}} - \rew_{\epsdice{1}}$.
Thus, the desired property is equivalent to the following multi-objective achievability query on the unfolded MDP $\mdp'$: $\exists \sched \in \Scheds[\mdp'] .\ \bigwedge_{j=1}^{6} \Expected^{\mdp',\sched}_{\state_0}(\rew^{j}) \approx_\epsilon 0$.
Note that, in this example, all probability operators share the same scheduler variable and initial states.
Checking this multi-objective query for various $\epsilon$ yields that, for the fast dice roller, the desired property holds for $\epsilon\geq0.13$ and for Knuth-Yao only for $\epsilon\geq0.15$, so the latter model is a coarser simulation in this sense.

We can decide whether the desired property holds in polynomial time in the size of the input since all target states are absorbing~(\cref{th:conj_special}\ref{item:conj_absorb}).

%% file: motivating-ex/von-neumann.tex
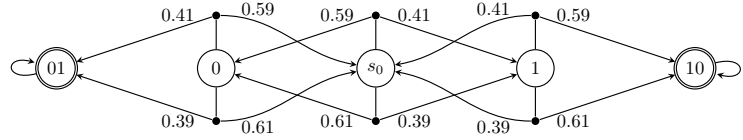
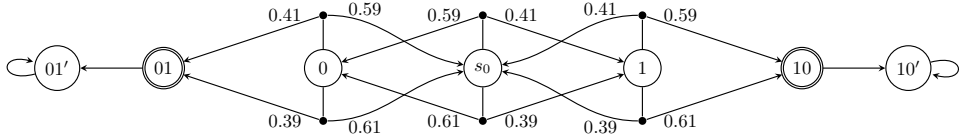
\begin{figure}[t]
    \centering
    \begin{subfigure}[t]{\textwidth}
    \centering
        \scalebox{.7}{
    \begin{tikzpicture}[on grid,node distance=10mm and 30mm,semithick,>=stealth]
        \node[state] (start) {$s_0$};
        \node[dist,above=of start] (startl) {};
        \node[dist,below=of start] (startu) {};
        \node[state,left=of start] (0) {$0$};
        \node[state,right=of start] (1) {$1$};
        \node[dist,above=of 0] (0l) {};
        \node[dist,below=of 0] (0u) {};
        \node[dist,above=of 1] (1l) {};
        \node[dist,below=of 1] (1u) {};
        \node[state,accepting,left=of 0] (01) {$01$};
        \node[state,accepting,right=of 1] (10) {$10$};
        
        \draw[-] (start) -- (startl);
        \draw[-] (start) -- (startu);
        \draw[->] (startl) -- node[near start,above] {$0.59$} (0);
        \draw[->] (startl) -- node[near start,above] {$0.41$} (1);
        \draw[->] (startu) -- node[near start,below] {$0.61$} (0);
        \draw[->] (startu) -- node[near start,below] {$0.39$} (1);
        \draw[-] (0) -- (0l);
        \draw[-] (0) -- (0u);
        \draw[-] (1) -- (1l);
        \draw[-] (1) -- (1u);
        \draw[->] (0l) edge[out=0,in=170] node[near start,above] {$0.59$} (start);
        \draw[->] (0l) -- node[near start,above] {$0.41$} (01);
        \draw[->] (0u) edge[out=0,in=190] node[near start,below] {$0.61$} (start);
        \draw[->] (0u) -- node[near start,below] {$0.39$} (01);
        \draw[->] (1l) edge[out=180,in=10] node[near start,above] {$0.41$} (start);
        \draw[->] (1l) -- node[near start,above] {$0.59$} (10);
        \draw[->] (1u) edge[out=180,in=-10] node[near start,below] {$0.39$} (start);
        \draw[->] (1u) -- node[near start,below] {$0.61$} (10);

        \path[-latex', draw]
            (01) edge[loop left] (01)
            (10) edge[loop right] (10);
    \end{tikzpicture}
    }
    \caption{Modeling Von Neumann's trick.}
    \label{fig:illustrativeExample}
    \end{subfigure}
    \begin{subfigure}[t]{\textwidth}
    \centering
        \scalebox{.7}{
    \begin{tikzpicture}[on grid,node distance=10mm and 30mm,semithick,>=stealth]
        \node[state] (start) {$s_0$};
        \node[dist,above=of start] (startl) {};
        \node[dist,below=of start] (startu) {};
        \node[state,left=of start] (0) {$0$};
        \node[state,right=of start] (1) {$1$};
        \node[dist,above=of 0] (0l) {};
        \node[dist,below=of 0] (0u) {};
        \node[dist,above=of 1] (1l) {};
        \node[dist,below=of 1] (1u) {};
        \node[state,accepting,left=of 0] (01) {$01$};
        \node[state,accepting,right=of 1] (10) {$10$};
        \node[state,left=20mm of 01] (01p) {$01'$};
        \node[state,right=20mm of 10] (10p) {$10'$};
        
        \draw[-] (start) -- (startl);
        \draw[-] (start) -- (startu);
        \draw[->] (startl) -- node[near start,above] {$0.59$} (0);
        \draw[->] (startl) -- node[near start,above] {$0.41$} (1);
        \draw[->] (startu) -- node[near start,below] {$0.61$} (0);
        \draw[->] (startu) -- node[near start,below] {$0.39$} (1);
        \draw[-] (0) -- (0l);
        \draw[-] (0) -- (0u);
        \draw[-] (1) -- (1l);
        \draw[-] (1) -- (1u);
        \draw[->] (0l) edge[out=0,in=170] node[near start,above] {$0.59$} (start);
        \draw[->] (0l) -- node[near start,above] {$0.41$} (01);
        \draw[->] (0u) edge[out=0,in=190] node[near start,below] {$0.61$} (start);
        \draw[->] (0u) -- node[near start,below] {$0.39$} (01);
        \draw[->] (1l) edge[out=180,in=10] node[near start,above] {$0.41$} (start);
        \draw[->] (1l) -- node[near start,above] {$0.59$} (10);
        \draw[->] (1u) edge[out=180,in=-10] node[near start,below] {$0.39$} (start);
        \draw[->] (1u) -- node[near start,below] {$0.61$} (10);

        \draw[->] (10) edge (10p);
        \draw[->] (01) edge (01p);

        \path[-latex', draw]
            (01p) edge[loop left] (01p)
            (10p) edge[loop right] (10p);
    \end{tikzpicture}
    }
    \caption{Unfolding the MDP w.r.t.\ states $01$ and $10$.}
    \label{fig:illustrativeExampleUnfolded}
    \end{subfigure}
    \caption{\color{cavcolor}Simulating unbiased coins by random biased bits. }
\end{figure}

%% file: motivating-ex/israeli-jalfon.tex
\colorlet{pone}{biaslower}

\begin{figure}[t]
    \centering
    \begin{subfigure}[b]{0.49\textwidth}
    \centering
    \scalebox{.7}{
        \begin{tikzpicture}[
            node distance=20mm and 25mm,
            on grid,semithick,>=stealth, every node/.style={scale=0.8}, 
            p1/.style={color=pone}, p2/.style={color=Blue}, p3/.style={color=LightSeaGreen}]
            \node[state] (111) {$111$};
            \node[state, below left=of 111, yshift=5mm] (011) {$011$};
            \node[state, below =of 111, yshift=5mm] (101) {$101$};
            \node[state, below right=of 111, yshift=5mm] (110) {$110$};
            \node[state, below=of 011, yshift=-10mm] (001) {$001$};
            \node[state, below=of 101, yshift=-10mm] (010) {$010$};
            \node[state, below=of 110, yshift=-10mm] (100) {$100$};

            \node[dist, above=of 011, yshift=-12mm, p1, xshift=10mm] (111p1) {};
            \node[dist, above=of 101, yshift=-14mm, p2] (111p2) {};
            \node[dist, above=of 110, yshift=-12mm, p3, xshift=-10mm] (111p3) {};

            \node[dist, below left=of 011, yshift=7mm, xshift=22mm, p2] (011p2) {};
            \node[dist, below right=of 011, yshift=7mm, xshift=-22mm, p3] (011p3) {};
            \node[dist, below left=of 101, yshift=7mm, xshift=22mm, p1] (101p1) {};
            \node[dist, below right=of 101, yshift=7mm, xshift=-22mm, p3] (101p3) {};
            \node[dist, below left=of 110, yshift=7mm, xshift=22mm, p1] (110p1) {};
            \node[dist, below right=of 110, yshift=7mm, xshift=-22mm, p2] (110p2) {};

            \node[dist, below =of 001, yshift=10mm, p3, xshift=5mm] (001p3) {};
            \node[dist, below =of 010, yshift=10mm, p2] (010p2) {};
            \node[dist, below =of 100, yshift=10mm, p1, xshift=-5mm] (100p1) {};

            \path[-]
            (111) edge[p1, bend right=10] (111p1)
            (111) edge[p2] (111p2)
            (111) edge[p3, bend left=10] (111p3);

            \path[->]
            (111p1) edge[p1] (011)
            (111p2) edge[p2] (101)
            (111p3) edge[p3] (110);

            \path[-]
            (011) edge[p2, bend right] (011p2)
            (011) edge[p3, bend right=10] (011p3)
            (101.250) edge[p1, bend left] (101p1)
            (101.280) edge[p3, bend right] (101p3)
            (110) edge[p1, bend left=10] (110p1)
            (110) edge[p2, bend left] (110p2);

            \path[->]
            (011p2) edge[p2] (001)
            (011p2) edge[p2] (101)
            (011p3) edge[p3] (010)
            (011p3) edge[p3] (110)
            (101p1) edge[p1] (001)
            (101p1) edge[p1] (011)
            (101p3) edge[p3] (100)
            (101p3) edge[p3] (110)
            (110p1) edge[p1] (010)
            (110p1) edge[p1] (011)
            (110p2) edge[p2] (100)
            (110p2) edge[p2] (101);

            \path[-]
            (001) edge[p3, bend right] (001p3)
            (010) edge[p2] (010p2)
            (100) edge[p1, bend left] (100p1);

            \path[->]
            (001p3) edge[p3] (010)
            (001p3) edge[p3, bend right=10] (100)
            (010p2) edge[p2, bend left] (001)
            (010p2) edge[p2, bend right] (100)
            (100p1) edge[p1, bend left=10] (001)
            (100p1) edge[p1] (010);
        \end{tikzpicture}
        }
    \caption{Original (\href{https://www.prismmodelchecker.org/casestudies/self-stabilisation.php\#ij}{PRISM model}).}
    \label{fig:israeli-jalfon}
    \end{subfigure}
    \begin{subfigure}[b]{0.49\textwidth}
    \centering
    \scalebox{.7}{
        \begin{tikzpicture}[
            initial text=,initial where=right,node distance=20mm and 25mm,on grid,semithick,>=stealth, every node/.style={scale=0.8}, 
            p1/.style={color=pone}, p2/.style={color=Blue}, p3/.style={color=LightSeaGreen}]
            \node[state] (111) {$111$};
            \node[state, below left=of 111, yshift=5mm] (011) {$011$};
            \node[state, below =of 111, yshift=5mm] (101) {$101$};
            \node[state, below right=of 111, yshift=5mm] (110) {$110$};
            \node[state, below=of 011, yshift=-10mm] (001) {$001$};
            \node[state, below=of 101, yshift=-10mm] (010) {$010$};
            \node[state, below=of 110, yshift=-10mm] (100) {$100$};

            \node[dist, above=of 011, yshift=-12mm, p1, xshift=10mm] (111p1) {};
            \node[dist, above=of 101, yshift=-14mm, p2] (111p2) {};

            \node[dist, below left=of 011, yshift=7mm, xshift=22mm, p2] (011p2) {};
            \node[dist, below left=of 101, yshift=7mm, xshift=22mm, p1] (101p1) {};
            \node[dist, below left=of 110, yshift=7mm, xshift=22mm, p1] (110p1) {};
            \node[dist, below right=of 110, yshift=7mm, xshift=-22mm, p2] (110p2) {};

            \node[dist, below =of 010, yshift=10mm, p2] (010p2) {};
            \node[dist, below =of 100, yshift=10mm, p1, xshift=-5mm] (100p1) {};

            \path[-]
            (111) edge[p1, bend right=10] (111p1)
            (111) edge[p2] (111p2)
            ;

            \path[->]
            (111p1) edge[p1] (011)
            (111p2) edge[p2] (101)
            ;

            \path[-]
            (011) edge[p2, bend right] (011p2)
            (101.250) edge[p1, bend left] (101p1)
            (110) edge[p1, bend left=10] (110p1)
            (110) edge[p2, bend left] (110p2);

            \path[->]
            (011p2) edge[p2] (001)
            (011p2) edge[p2] (101)
            (101p1) edge[p1] (001)
            (101p1) edge[p1] (011)
            (110p1) edge[p1] (010)
            (110p1) edge[p1] (011)
            (110p2) edge[p2] (100)
            (110p2) edge[p2] (101);

            \path[-]
            (010) edge[p2] (010p2)
            (100) edge[p1, bend left] (100p1);

            \path[->]
            (010p2) edge[p2, bend left] (001)
            (010p2) edge[p2, bend right] (100)
            (100p1) edge[p1, bend left=10] (001)
            (100p1) edge[p1] (010);

            \path[->]
            (001) edge[loop below] (001);
        \end{tikzpicture}
        }
    \caption{Asymmetric variant.}
    \label{fig:israeli-jalfon_asym}
    \end{subfigure}
    \begin{subfigure}[b]{0.49\textwidth}
    \centering
    \scalebox{.7}{
        \begin{tikzpicture}[
            node distance=20mm and 25mm,
            on grid,semithick,>=stealth, every node/.style={scale=0.8}, 
            p1/.style={color=pone}, p2/.style={color=Blue}, p3/.style={color=LightSeaGreen}]
            \node[state] (111) {$111$};
            \node[state, below left=of 111, yshift=5mm] (011) {$011$};
            \node[state, below =of 111, yshift=5mm] (101) {$101$};
            \node[state, below right=of 111, yshift=5mm] (110) {$110$};
            
            \node[staterectangle, below=of 101, yshift=-10mm] (s) {$001,010,100$};
            \node[staterectangle, below=of 110, yshift=-10mm] (b23) {$\bot^{Q_2,Q_3}$};

            \node[dist, above=of 011, yshift=-12mm, p1, xshift=10mm] (111p1) {};
            \node[dist, above=of 101, yshift=-14mm, p2] (111p2) {};
            \node[dist, above=of 110, yshift=-12mm, p3, xshift=-10mm] (111p3) {};

            \node[dist, below left=of 011, yshift=7mm, xshift=22mm, p2] (011p2) {};
            \node[dist, below right=of 011, yshift=7mm, xshift=-22mm, p3] (011p3) {};
            \node[dist, below left=of 101, yshift=7mm, xshift=22mm, p1] (101p1) {};
            \node[dist, below right=of 101, yshift=7mm, xshift=-22mm, p3] (101p3) {};
            \node[dist, below left=of 110, yshift=7mm, xshift=22mm, p1] (110p1) {};
            \node[dist, below right=of 110, yshift=7mm, xshift=-22mm, p2] (110p2) {};

            \path[-]
            (111) edge[p1, bend right=10] (111p1)
            (111) edge[p2] (111p2)
            (111) edge[p3, bend left=10] (111p3);

            \path[->]
            (111p1) edge[p1] (011)
            (111p2) edge[p2] (101)
            (111p3) edge[p3] (110);

            \path[-]
            (011) edge[p2, bend right] (011p2)
            (011) edge[p3, bend right=10] (011p3)
            (101.250) edge[p1, bend left] (101p1)
            (101.280) edge[p3, bend right] (101p3)
            (110) edge[p1, bend left=10] (110p1)
            (110) edge[p2, bend left] (110p2);

            \path[->]
            (011p2) edge[p2] (s)
            (011p2) edge[p2] (101)
            (011p3) edge[p3] (s)
            (011p3) edge[p3] (110)
            (101p1) edge[p1] (s)
            (101p1) edge[p1] (011)
            (101p3) edge[p3] (s)
            (101p3) edge[p3] (110)
            (110p1) edge[p1] (s)
            (110p1) edge[p1] (011)
            (110p2) edge[p2] (s)
            (110p2) edge[p2] (101);

            \path[->]
            (s) edge (b23)
            (b23) edge[loop above] (b23);
        \end{tikzpicture}
        }
    \caption{MEC quotient of original model.}
    \label{fig:israeli-jalfon-unfolded}
    \end{subfigure}
    \begin{subfigure}[b]{0.49\textwidth}
    \centering
    \scalebox{.7}{
        \begin{tikzpicture}[
            initial text=,initial where=right,node distance=20mm and 25mm,on grid,semithick,>=stealth, every node/.style={scale=0.8}, 
            p1/.style={color=pone}, p2/.style={color=Blue}, p3/.style={color=LightSeaGreen}]
            \node[state] (111) {$111$};
            \node[state, below left=of 111, yshift=5mm] (011) {$011$};
            \node[state, below =of 111, yshift=5mm] (101) {$101$};
            \node[state, below right=of 111, yshift=5mm] (110) {$110$};
            
            \node[staterectangle, below=of 011, yshift=-10mm, xshift=10mm] (001) {$001$};
            \node[staterectangle, below=of 011, yshift=-10mm, xshift=-10mm] (b3) {$\bot^{Q_3}$};
            
            \node[state, below=of 101, yshift=-10mm] (010) {$010$};
            \node[state, below=of 110, yshift=-10mm] (100) {$100$};

            \node[dist, above=of 011, yshift=-12mm, p1, xshift=10mm] (111p1) {};
            \node[dist, above=of 101, yshift=-14mm, p2] (111p2) {};

            \node[dist, below left=of 011, yshift=7mm, xshift=22mm, p2] (011p2) {};
            \node[dist, below left=of 101, yshift=7mm, xshift=22mm, p1] (101p1) {};
            \node[dist, below left=of 110, yshift=7mm, xshift=22mm, p1] (110p1) {};
            \node[dist, below right=of 110, yshift=7mm, xshift=-22mm, p2] (110p2) {};

            \node[dist, below =of 010, yshift=10mm, p2] (010p2) {};
            \node[dist, below =of 100, yshift=10mm, p1, xshift=-5mm] (100p1) {};

            \path[-]
            (111) edge[p1, bend right=10] (111p1)
            (111) edge[p2] (111p2)
            ;

            \path[->]
            (111p1) edge[p1] (011)
            (111p2) edge[p2] (101)
            ;

            \path[-]
            (011) edge[p2, bend right] (011p2)
            (101.250) edge[p1, bend left] (101p1)
            (110) edge[p1, bend left=10] (110p1)
            (110) edge[p2, bend left] (110p2);

            \path[->]
            (011p2) edge[p2] (001)
            (011p2) edge[p2] (101)
            (101p1) edge[p1] (001)
            (101p1) edge[p1] (011)
            (110p1) edge[p1] (010)
            (110p1) edge[p1] (011)
            (110p2) edge[p2] (100)
            (110p2) edge[p2] (101);

            \path[-]
            (010) edge[p2] (010p2)
            (100) edge[p1, bend left] (100p1);

            \path[->]
            (010p2) edge[p2, bend left] (001)
            (010p2) edge[p2, bend right] (100)
            (100p1) edge[p1, bend left=10] (001)
            (100p1) edge[p1] (010);

            \path[->]
            (001) edge (b3)
            (b3) edge[loop below] (b3);
        \end{tikzpicture}
        }
    \caption{MEC quotient for asymmetric variant.}
    \label{fig:israeli-jalfon_asym-unfolded}
    \end{subfigure}
    \caption{Modeling Israeli \& Jalfon's self-stabilising protocol for 3 processes as an MDP (left), and its asymmetric variant (right).
    Note that states $110$, $010$, and $100$ are not reachable in the asymmetric variant.
    Taking a \textcolor{pone}{green} action corresponds to scheduling process 1, \textcolor{Blue}{blue} to process 2, and \textcolor{LightSeaGreen}{cyan} to process 3.
    }
\end{figure}
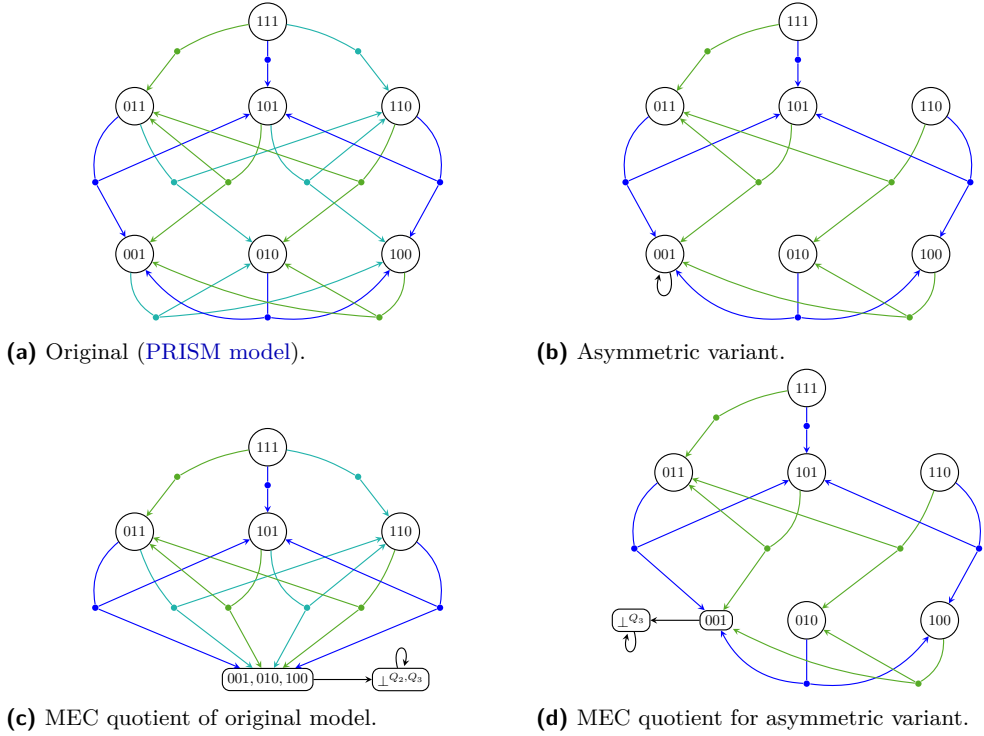

%% file: motivating-ex/dice-simulation.tex
\begin{figure}[t]
    \centering
    \begin{subfigure}[c]{\textwidth}
    \centering
    \begin{subfigure}[b]{0.49\textwidth}
        \centering
        \scalebox{.7}{
        \begin{tikzpicture}[
            node distance=5mm and 25mm,on grid,semithick,>=stealth, every node/.style={scale=0.8}, 
            pl/.style={color=biaslower}, pu/.style={color=biasupper}]
            \node[state] (s0) {$s_0$};
            \node[state, above right=of s0, yshift=12.5mm] (s1) {$s_1$};
            \node[state, below right=of s0, yshift=-12.5mm] (s2) {$s_2$};
            \node[state, above right=of s1] (s3) {$s_3$};
            \node[state, below right=of s1] (s4) {$s_4$};
            \node[state, above right=of s2] (s5) {$s_5$};
            \node[state, below right=of s2] (s6) {$s_6$};

            \node[above right=of s3] (d1) {$\epsdice{1}$};
            \node[above right=of s4] (d2) {$\epsdice{2}$};
            \node[below right=of s4] (d3) {$\epsdice{3}$};
            \node[above right=of s5] (d4) {$\epsdice{4}$};
            \node[below right=of s5] (d5) {$\epsdice{5}$};
            \node[below right=of s6] (d6) {$\epsdice{6}$};

            \node[dist, right=of s0, xshift=-20mm] (dist0) {};
            
            \node[dist, right=of s1, xshift=-20mm] (dist1) {};
            \node[dist, right=of s2, xshift=-20mm] (dist2) {};

            \node[dist, right=of s4, xshift=-20mm] (dist4) {};
            \node[dist, right=of s5, xshift=-20mm] (dist5) {};

            \node[dist, above=of s3, yshift=2.5mm] (dist3) {};
            \node[dist, below=of s6, yshift=-2.5mm] (dist6) {};

            \path[-]
            (s0) edge (dist0)
            (s1) edge (dist1)
            (s2) edge (dist2)
            (s4) edge (dist4)
            (s5) edge (dist5);

            \path[-]
            (s3) edge (dist3)
            (s6) edge (dist6);

            \path[->]
            (dist0) edge[bend left] node[left, pos=0.3] {$0.5$} (s1)
            (dist0) edge[bend right] node[left] {$0.5$} (s2)
            (dist1) edge[bend left] (s3)
            (dist1) edge[bend right] (s4)
            (dist2) edge[bend left] (s5)
            (dist2) edge[bend right] (s6)
            (dist4) edge[bend left] (d2)
            (dist4) edge[bend right] (d3)
            (dist5) edge[bend left] (d4)
            (dist5) edge[bend right] (d5)
            ;

            \path[->]
            (dist3) edge[bend left] (d1)
            (dist3) edge[bend right] (s1)
            (dist6) edge[bend right] (d6)
            (dist6) edge[bend left] (s2)
            ;

            \path[->]
            (d1) edge[loop right] (d1)
            (d2) edge[loop right] (d2)
            (d3) edge[loop right] (d3)
            (d4) edge[loop right] (d4)
            (d5) edge[loop right] (d5)
            (d6) edge[loop right] (d6);
        \end{tikzpicture}
        }
        \subcaption{Knuth-Yao \cite{knuthComplexityNonuniform1976}.}
        \label{fig:knuth-yao-fair}
    \end{subfigure}
    \begin{subfigure}[b]{0.49\textwidth}
        \centering
        \scalebox{.7}{
        \begin{tikzpicture}[
            initial text=,initial where=right,node distance=5mm and 25mm,on grid,semithick,>=stealth, every node/.style={scale=0.8}, 
            pl/.style={color=biaslower}, pu/.style={color=biasupper}]
            \node[state] (s0) {$s_0$};
            \node[state, above right=of s0, yshift=12.5mm] (s1) {$s_1$};
            \node[state, below right=of s0, yshift=-12.5mm] (s2) {$s_2$};
            \node[state, above right=of s1] (s3) {$s_3$};
            \node[state, below right=of s1] (s4) {$s_4$};
            \node[state, above right=of s2] (s5) {$s_5$};
            \node[state, below right=of s2] (s6) {$s_6$};

            \node[above right=of s3] (d1) {$\epsdice{1}$};
            \node[above right=of s4] (d2) {$\epsdice{2}$};
            \node[below right=of s4] (d3) {$\epsdice{3}$};
            \node[above right=of s5] (d4) {$\epsdice{4}$};
            \node[below right=of s5] (d5) {$\epsdice{5}$};
            \node[below right=of s6] (d6) {$\epsdice{6}$};

            \node[dist, right=of s0, xshift=-20mm] (dist0) {};
            
            \node[dist, right=of s1, xshift=-20mm] (dist1) {};
            \node[dist, right=of s2, xshift=-20mm] (dist2) {};

            \node[dist, right=of s4, xshift=-20mm] (dist4) {};
            \node[dist, above=of s5, yshift=2.5mm] (dist5) {};

            \node[dist, right=of s3, xshift=-20mm] (dist3) {};
            \node[dist, right=of s6, xshift=-20mm] (dist6) {};

            \path[-]
            (s0) edge (dist0)
            (s1) edge (dist1)
            (s2) edge (dist2)
            (s4) edge (dist4)
            (s5) edge (dist5);

            \path[-]
            (s3) edge (dist3)
            (s6) edge (dist6);

            \path[->]
            (dist0) edge[bend left] node[left, pos=0.3] {$0.5$} (s1)
            (dist0) edge[bend right] node[left] {$0.5$} (s2)
            (dist1) edge[bend left] (s3)
            (dist1) edge[bend right] (s4)
            (dist2) edge[bend left] (s5)
            (dist2) edge[bend right] (s6)
            (dist3) edge[bend left] (d1)
            (dist3) edge[bend right] (d2)
            (dist4) edge[bend left] (d3)
            (dist4) edge[bend right] (d4)
            (dist5) edge[bend left=5] (s1)
            (dist5) edge[bend right] (s2)
            (dist6) edge[bend left] (d5)
            (dist6) edge[bend right] (d6)
            ;

            \path[->]
            (d1) edge[loop right] (d1)
            (d2) edge[loop right] (d2)
            (d3) edge[loop right] (d3)
            (d4) edge[loop right] (d4)
            (d5) edge[loop right] (d5)
            (d6) edge[loop right] (d6);
        \end{tikzpicture}   
        }
        \subcaption{Fast dice roller \cite{lumbrosoOptimalDiscrete2013} (reduced model from~\cite{dBLP:journals/jar/MertensKQW25}).}
        \label{fig:fast-dice-roller-fair}
    \end{subfigure}
    \end{subfigure}
    \begin{subfigure}[c]{\textwidth} 
    \centering
    \begin{subfigure}[b]{0.49\textwidth}
        \centering
        \scalebox{.7}{
        \begin{tikzpicture}[
            node distance=5mm and 25mm,on grid,semithick,>=stealth, every node/.style={scale=0.8}, 
            pl/.style={color=biaslower}, pu/.style={color=biasupper}]
            \node[state] (s0) {$s_0$};
            \node[state, above right=of s0, yshift=12.5mm] (s1) {$s_1$};
            \node[state, below right=of s0, yshift=-12.5mm] (s2) {$s_2$};
            \node[state, above right=of s1] (s3) {$s_3$};
            \node[state, below right=of s1] (s4) {$s_4$};
            \node[state, above right=of s2] (s5) {$s_5$};
            \node[state, below right=of s2] (s6) {$s_6$};

            \node[above right=of s3] (d1) {$\epsdice{1}$};
            \node[above right=of s4] (d2) {$\epsdice{2}$};
            \node[below right=of s4] (d3) {$\epsdice{3}$};
            \node[above right=of s5] (d4) {$\epsdice{4}$};
            \node[below right=of s5] (d5) {$\epsdice{5}$};
            \node[below right=of s6] (d6) {$\epsdice{6}$};

            \node[dist, above right=of s0, pu, xshift=-15mm, yshift=5mm] (d0u) {};
            \node[dist, below right=of s0, pl, xshift=-15mm, yshift=-5mm] (d0l) {};
            
            \node[dist, above right=of s1, pu, xshift=-20mm] (d1u) {};
            \node[dist, below right=of s1, pl, xshift=-20mm] (d1l) {};
            \node[dist, above right=of s2, pu, xshift=-20mm] (d2u) {};
            \node[dist, below right=of s2, pl, xshift=-20mm] (d2l) {};

            \node[dist, above right=of s4, pu, xshift=-20mm] (d4u) {};
            \node[dist, below right=of s4, pl, xshift=-20mm] (d4l) {};
            \node[dist, above right=of s5, pu, xshift=-20mm] (d5u) {};
            \node[dist, below right=of s5, pl, xshift=-20mm] (d5l) {};

            \node[dist, above left=of s3, pu, yshift=2.5mm, xshift=25mm] (d3u) {};
            \node[dist, above right=of s3, pl, yshift=2.5mm, xshift=-25mm] (d3l) {};
            \node[dist, below left=of s6, pl, yshift=-2.5mm, xshift=25mm] (d6l) {};
            \node[dist, below right=of s6, pu, yshift=-2.5mm, xshift=-25mm] (d6u) {};

            \path[-]
            (s0) edge[pl] (d0l)
            (s0) edge[pu] (d0u)
            (s1) edge[pl] (d1l)
            (s1) edge[pu] (d1u)
            (s2) edge[pl] (d2l)
            (s2) edge[pu] (d2u)
            (s4) edge[pl] (d4l)
            (s4) edge[pu] (d4u)
            (s5) edge[pl] (d5l)
            (s5) edge[pu] (d5u);

            \path[-]
            (s3) edge[pl] (d3l)
            (s3) edge[pu] (d3u)
            (s6) edge[pl] (d6l)
            (s6) edge[pu] (d6u);

            \path[->]
            (d0l) edge[pl, bend left] node[left, pos=0.3] {$\underline{p}$} (s1)
            (d0l) edge[pl, bend right] node[left] {$1-\underline{p}$} (s2)
            (d0u) edge[pu, bend left] node[above] {$\overline{p}$} (s1)
            (d0u) edge[pu, bend right] node[right] {$1-\overline{p}$} (s2)
            (d1l) edge[pl, bend left] (s3)
            (d1l) edge[pl, bend right] (s4)
            (d1u) edge[pu, bend left] (s3)
            (d1u) edge[pu, bend right] (s4)
            (d2l) edge[pl, bend left] (s5)
            (d2l) edge[pl, bend right] (s6)
            (d2u) edge[pu, bend left] (s5)
            (d2u) edge[pu, bend right] (s6)
            (d4l) edge[pl, bend left] (d2)
            (d4l) edge[pl, bend right] (d3)
            (d4u) edge[pu, bend left] (d2)
            (d4u) edge[pu, bend right] (d3)
            (d5l) edge[pl, bend left] (d4)
            (d5l) edge[pl, bend right] (d5)
            (d5u) edge[pu, bend left] (d4)
            (d5u) edge[pu, bend right] (d5)
            ;

            \path[->]
            (d3l) edge[pl, bend left] node[below, pos=0.2] {$\underline{p}$} (d1)
            (d3l) edge[pl, bend right=40] (s1)
            (d3u) edge[pu, bend left] node[below, pos=0.1] {$\overline{p}$} (d1)
            (d3u) edge[pu, bend right] (s1)
            (d6l) edge[pl, bend right] node[above, pos=0.1] {$\underline{p}$} (d6)
            (d6l) edge[pl, bend left] (s2)
            (d6u) edge[pu, bend right] node[above, pos=0.2] {$\overline{p}$} (d6)
            (d6u) edge[pu, bend left=40] (s2)
            ;

            \path[->]
            (d1) edge[loop right] (d1)
            (d2) edge[loop right] (d2)
            (d3) edge[loop right] (d3)
            (d4) edge[loop right] (d4)
            (d5) edge[loop right] (d5)
            (d6) edge[loop right] (d6);
        \end{tikzpicture}
        }
        \subcaption{Knuth-Yao with biased coins. }
        \label{fig:knuth-yao-biased}
    \end{subfigure}
    \begin{subfigure}[b]{0.49\textwidth}
        \centering
        \scalebox{.7}{
        \begin{tikzpicture}[
            initial text=,initial where=right,node distance=5mm and 25mm,on grid,semithick,>=stealth, every node/.style={scale=0.8}, 
            pl/.style={color=biaslower}, pu/.style={color=biasupper}]
            \node[state] (s0) {$s_0$};
            \node[state, above right=of s0, yshift=12.5mm] (s1) {$s_1$};
            \node[state, below right=of s0, yshift=-12.5mm] (s2) {$s_2$};
            \node[state, above right=of s1] (s3) {$s_3$};
            \node[state, below right=of s1] (s4) {$s_4$};
            \node[state, above right=of s2] (s5) {$s_5$};
            \node[state, below right=of s2] (s6) {$s_6$};

            \node[above right=of s3] (d1) {$\epsdice{1}$};
            \node[above right=of s4] (d2) {$\epsdice{2}$};
            \node[below right=of s4] (d3) {$\epsdice{3}$};
            \node[above right=of s5] (d4) {$\epsdice{4}$};
            \node[below right=of s5] (d5) {$\epsdice{5}$};
            \node[below right=of s6] (d6) {$\epsdice{6}$};

            \node[dist, above right=of s0, pu, xshift=-15mm, yshift=5mm] (d0u) {};
            \node[dist, below right=of s0, pl, xshift=-15mm, yshift=-5mm] (d0l) {};
            
            \node[dist, above right=of s1, pu, xshift=-20mm] (d1u) {};
            \node[dist, below right=of s1, pl, xshift=-20mm] (d1l) {};
            \node[dist, above right=of s2, pu, xshift=-20mm] (d2u) {};
            \node[dist, below right=of s2, pl, xshift=-20mm] (d2l) {};

            \node[dist, above right=of s3, pu, xshift=-20mm, yshift=-2.5mm] (d3u) {};
            \node[dist, below right=of s3, pl, xshift=-20mm, yshift=2.5mm] (d3l) {};
            \node[dist, above right=of s4, pu, xshift=-20mm, yshift=-7.5mm] (d4u) {};
            \node[dist, below right=of s4, pl, xshift=-20mm, yshift=-2.5mm] (d4l) {};
            
            \node[dist, above left=of s5, pu, yshift=2.5mm, xshift=25mm] (d5u) {};
            \node[dist, above right=of s5, pl, yshift=2.5mm, xshift=-25mm] (d5l) {};
            
            \node[dist, above right=of s6, pu, xshift=-20mm, yshift=-2.5mm] (d6u) {};
            \node[dist, below right=of s6, pl, xshift=-20mm, yshift=2.5mm] (d6l) {};

            \path[-]
            (s0) edge[pl] (d0l)
            (s0) edge[pu] (d0u)
            (s1) edge[pl] (d1l)
            (s1) edge[pu] (d1u)
            (s2) edge[pl] (d2l)
            (s2) edge[pu] (d2u)
            (s4) edge[pl] (d4l)
            (s4) edge[pu] (d4u)
            (s5) edge[pl] (d5l)
            (s5) edge[pu] (d5u);

            \path[-]
            (s3) edge[pl] (d3l)
            (s3) edge[pu] (d3u)
            (s6) edge[pl] (d6l)
            (s6) edge[pu] (d6u);

            \path[->]
            (d0l) edge[pl, bend left] node[left, pos=0.3] {$\underline{p}$} (s1)
            (d0l) edge[pl, bend right] node[left] {$1-\underline{p}$} (s2)
            (d0u) edge[pu, bend left] node[above] {$\overline{p}$} (s1)
            (d0u) edge[pu, bend right] node[right] {$1-\overline{p}$} (s2)
            (d1l) edge[pl, bend left] (s3)
            (d1l) edge[pl, bend right] (s4)
            (d1u) edge[pu, bend left] (s3)
            (d1u) edge[pu, bend right] (s4)
            (d2l) edge[pl, bend left] (s5)
            (d2l) edge[pl, bend right] (s6)
            (d2u) edge[pu, bend left] (s5)
            (d2u) edge[pu, bend right] (s6)
            (d3l) edge[pl, bend left] (d1)
            (d3l) edge[pl, bend right] (d2)
            (d3u) edge[pu, bend left] (d1)
            (d3u) edge[pu, bend right] (d2)
            (d4l) edge[pl, bend left] (d3)
            (d4l) edge[pl, bend right] (d4)
            (d4u) edge[pu, bend left] (d3)
            (d4u) edge[pu, bend right] (d4)
            (d5l) edge[pl, bend left=5] node[above, pos=0.1] {$\underline{p}$} (s1) %
            (d5l) edge[pl, bend right=20] (s2) 
            (d5u) edge[pu, bend left=5] node[below] {$\overline{p}$} (s1) %
            (d5u) edge[pu, bend right] (s2) 
            (d6l) edge[pl, bend left] (d5)
            (d6l) edge[pl, bend right] (d6)
            (d6u) edge[pu, bend left] (d5)
            (d6u) edge[pu, bend right] (d6)
            ;

            \path[->]
            (d1) edge[loop right] (d1)
            (d2) edge[loop right] (d2)
            (d3) edge[loop right] (d3)
            (d4) edge[loop right] (d4)
            (d5) edge[loop right] (d5)
            (d6) edge[loop right] (d6);
        \end{tikzpicture}      
        }
        \subcaption{Fast dice roller with biased coins.}
        \label{fig:fast-dice-roller-biased}
    \end{subfigure}
    \end{subfigure}
    \caption{Models for different approaches for simulating a die via repeated coin tosses.
    \ref{fig:knuth-yao-fair},\ref{fig:fast-dice-roller-fair}: using fair coins, all distributions are uniform.
    \ref{fig:knuth-yao-biased},\ref{fig:fast-dice-roller-biased}: using biased coins with probability of outcome 0 in the interval $[\underline{p},\overline{p}]$, the \textcolor{biaslower}{green} action always corresponds to the coin with bias $\underline{p}$ and the \textcolor{biasupper}{blue} action to the coin with bias $\overline{p}$. 
    The probability distribution is depicted fully for $s_0$ and otherwise omitted for readability; we adhere to the following notational convention:
    The `upper' transition always corresponds to outcome 0 (i.e., probability $\underline{p}$ / $\overline{p}$), the `lower' one to outcome 1; if there is no `upper' and `lower' transition the transition probabilities for outcome 0 are indicated.
    }
    \label{fig:dice-simulation}
\end{figure}
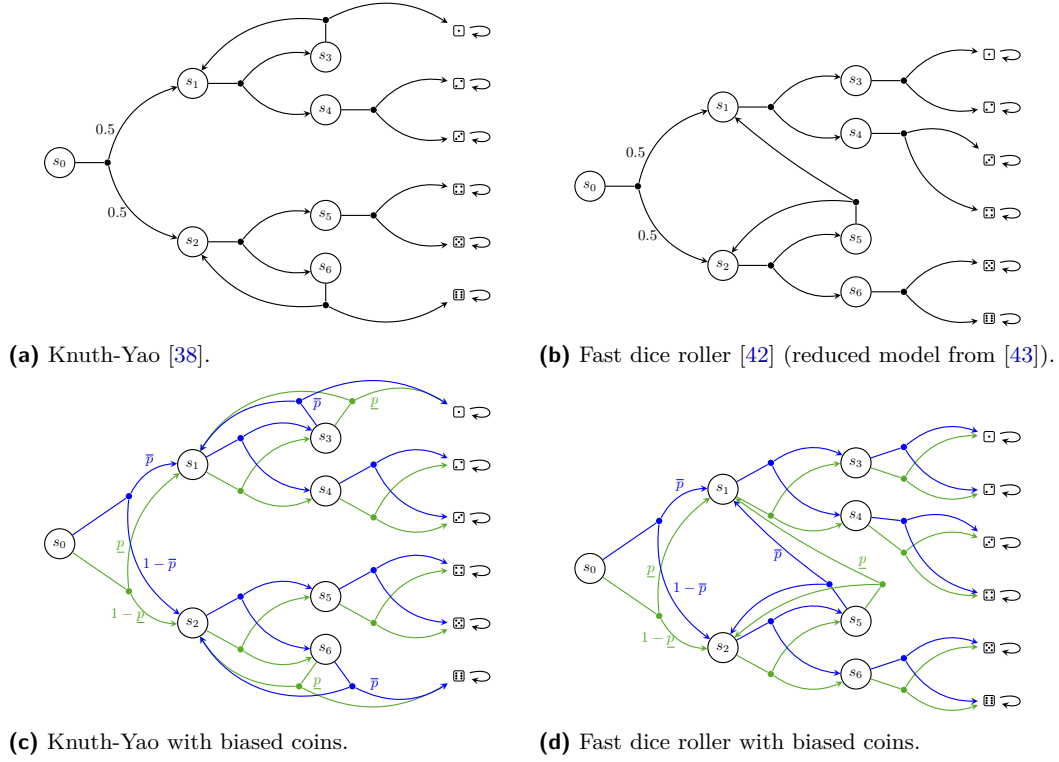

%% file: prelim.tex
\section{Preliminaries: MDPs, End Components, Schedulers}
\label{sec:prelim}

We fix some general notation first:
For a logical statement $P$, we use \emph{Iverson brackets} $\Iverson{P}$ to denote the function evaluating to 1 if $P$ holds and to 0 otherwise.
{\color{cavcolor}
We use $\N$, $\Q$, $\Qnn$, and $\R$ to denote the sets of non-negative integers, rationals, non-negative rationals, and real numbers, respectively. 
For $r,r', \epsilon \in \R$ with $\epsilon \geq 0$, we write $r \approx_\epsilon r'$ iff $|r-r'| \leq \epsilon$, and $r \not\approx_\epsilon r'$ iff $|r - r'| > \epsilon$.
The set $\Distr(V)$ of \emph{probability distributions} over a finite set $V$ contains all functions $\mu \colon V \rightarrow [0,1]$ s.t.\ $ \sum_{v \in V}\mu(v) = 1$.
}%

The remaining definitions in this section closely follow~\cite[Ch.~10]{baierPrinciplesModel2008}.
We are primarily concerned with \emph{Markov decision processes}, which we define formally as follows:

{\color{cavcolor}
\begin{definition}
    \label{def:MDP}
    A \emph{Markov decision process (MDP)} is a triple $\mdp = \mdptup$, where
    $\states$ is a non-empty finite set of \emph{states},
    $\Act$ is a non-empty finite set of \emph{actions}, and
    $\Trans \colon \states \times \Act \times \states \rightarrow [0,1]$ is a \emph{transition 
    probability function} s.t.\ for all $\state \in \states$ the set of its \emph{enabled actions}
	$
		\Act(\state)= \left\{\action\in\Act \mid \sum_{\state' \in \states} \Trans(\state, 
		\action, \state') =1\right\}
	$
	is non-empty and 
	$\sum_{\state' \in \states} \Trans(\state, \action, \state') = 0$
 for all $\action \in \Act \setminus \Act(s)$.
\end{definition}
We use $|\mdp|$ to denote the size of (an explicit encoding of) an MDP $\mdp = \mdptup$.
A state $\state \in \states$ is \emph{absorbing} if $\Trans(\state, \action, \state) = 1$ for all $\action \in \Act(\state)$.
A set of states $S' \subseteq \states$ is absorbing if all $\state \in S'$ are absorbing.
}%
The set of \emph{successor states} of $\state \in \states$ under $\action \in \Act$ is defined as $\textit{Succ}(\state,\action) = \{ \state' \in \states \mid \Trans(\state,\action, \state') >0\}$.

{\color{cavcolor}
An \emph{(infinite) path} of an MDP $\mdp$ is an infinite sequence of states and actions 
$\pi = s_0 \action_0 s_1 \action_1 \ldots$ such that for all $i \geq 0$ we have 
$\Trans(s_i, \action_i, s_{i+1})>0$.
A \emph{finite path} is a finite prefix of an infinite path ending in a state.
We use $\Paths$ (respectively, $\finPaths$) to denote the set of all infinite (respectively, finite) paths of $\mdp$, and for some state $\state \in \states$ we use $\Pathsfrom$ (respectively, $\finPathsfrom$) to denote the set of all infinite (respectively, finite) paths of $\mdp$ starting at $\state$.
For a finite or infinite path $\pi$, we use $\pi(i) := s_i$ to denote the $i^{th}$ state of $\pi$.
For a finite path $\pi = s_0 \action_0 \ldots \action_{n-1} s_n$, we define $\last(\pi) := s_n$,} 
$\lastact(\pi) := \action_{n-1}$, {\color{cavcolor}and $|\pi| := n$.} 
We say that an infinite path $\pi$ \emph{stays in $(S',A) \subseteq \states \times \Act$} iff for all $i$ we have $s_i \in S'$ and $\action_i \in A$, and analogously for finite paths.

\begin{definition}[End component]%
    \label{def:EC}%
    An \emph{end component (EC)} of an MDP $\mdp = \mdptup$ is a tuple $\ec = (S',A) \subseteq \states \times \Act$ such that
    \begin{enumerate}
        \item $S' \neq \emptyset$, $A \neq \emptyset$ and $A \subseteq \bigcup_{s' \in S'} \Act(s')$,
        \item $C$ is closed, i.e., for each $s' \in S'$, $\action \in A$ we have $\textit{Succ}(s',\action) \subseteq S'$, and
        \item $C$ is strongly connected, i.e., for each pair of states $s',s'' \in S'$ there exists a path from $s'$ to $s''$ that stays in $C$.
    \end{enumerate}
    A \emph{maximal end component} (\emph{MEC}) is an EC that is not contained in any other EC.
\end{definition}

We use $\EC(\mdp)$ and $\MEC(\mdp)$ to denote the set of all ECs and MECs, respectively, of an MDP $\mdp$.
Furthermore, we let $\states_{\MEC} = \{ \state \in \states \mid \exists C=(S',A) \in \MEC(\mdp) .\ s \in S' \}$ be the set of states that are contained in some MEC.

{\color{cavcolor}
A \emph{scheduler} resolves the nondeterminism in an MDP.
\begin{definition}[Scheduler]%
    \label{def:MDPsched}%
    A \emph{scheduler} for an MDP $\mdp = \mdptup$ is a function $\sched \colon \finPaths \allowbreak\to\Distr(\Act)$ with $\sched(\pi)(\action) = 0$ for all $\pi \in \finPaths$ and $\action \in \Act \setminus \Act(\last(\pi))$.
\end{definition}

A scheduler is \emph{memoryless}} if for all pairs of finite paths $\pi, \pi' \in \finPaths$ with $\last(\pi) = \last(\pi')$ we have $\sched(\pi) = \sched(\pi')$, and \emph{memoryful} (oder history-dependent) otherwise.
We often identify a memoryless scheduler with a function $\sched \colon \states \to \Distr(\Act)$.
{\color{cavcolor}
A scheduler is \emph{deterministic} if $\sched(\pi)(\action) \in \{0,1\}$ for all $\pi \in \finPaths$ and all $\alpha\in\Act$, and \emph{randomized} otherwise.
The set of all \emph{general} (i.e., history-dependent randomized) schedulers for an MDP $\mdp$ is denoted by $\Scheds$, the set of all memoryless randomized (MR) schedulers by $\SchedsMR$ and the set of all memoryless deterministic (MD) schedulers by $\SchedsMD$.
}
%

{\color{cavcolor}
Applying a scheduler to an MDP induces a 
\emph{discrete-time Markov chain (DTMC)}, which is an MDP where the set of actions is a singleton. 
We usually omit the actions and define a DTMC as a tuple $\dtmc = \dtmctup$.
}%

For an MDP $\mdp$, a scheduler $\sched$, and a state $\state$, we write $\Pr^{\mdp, \sched}_{\state}$ for the probability measure over infinite paths of the DTMC induced by $\sched$ on $\mdp$, assuming initial state $\state$ (see~\cite[Ch.~10]{baierPrinciplesModel2008} for details).
When the MDP is clear from the context, we simply write $\Pr^\sched_\state$.
{\color{cavcolor}
For example, for a target set $T \subseteq \states$, we use $\Pr^\sched_\state(\Finally T)$ to denote the \emph{reachability probability} of $T$ from $\state$ under $\sched$,
}%
$\Pr^\sched_\state(\Globally T)$ to denote the probability of staying in $T$ forever (\emph{safety}),
$\Pr^\sched_\state(\Globally \Finally T)$ to denote the probability of visiting $T$ infinitely often (\emph{B\"uchi}), and
$\Pr^\sched_\state(\Finally \Globally T)$ to staying in $T$ forever from some point on (\emph{coB\"uchi}).

%% file: problem-statement.tex
\section{Problem Statement: Checking MDPs Against Relational Properties}
\label{sec:problem-statement}

In this section, we formally introduce \emph{relational properties}, which constitute a subclass of more general probabilistic hyperproperties~\cite{abrahamProbabilisticHyperproperties2020,dimitrovaProbabilisticHyperproperties2020} (see \cref{sec:related-work} for a detailed comparison).
We begin by defining the syntax and semantics of relational properties, and then formally state the model-checking problem addressed in this paper.

\subparagraph*{Syntax.}
Let $\AP$ be a set of \emph{atomic propositions}, $\IL$ a set of \emph{initial state labels}, and $\VarSched$ a set of \emph{scheduler variables}.
A \emph{relational formula} $\phi$ adheres to the following grammar:
\begin{align*}
    &\textit{Propositional formula:} & \Psi & ~\Coloneqq~ \ap ~\mid~ \neg \Psi ~\mid~ \Psi \wedge \Psi ~\mid~ \Psi \vee \Psi \\
    &\textit{Path formula:} & \psi & ~\Coloneqq~ \Finally \Psi ~\mid~ \Globally \Psi ~\mid~ \Globally\Finally \Psi ~\mid~ \Finally\Globally \Psi \\
    &\textit{Expression:} & E & ~\Coloneqq~ \Prob^{\varsched}_{l}(\psi) ~\mid~ E + E ~\mid~ q \cdot E ~\mid~ q \\ 
    &\textit{Unquantified relational formula:} & \Phi & ~\Coloneqq~ E \comp E ~\mid~ \Phi \wedge \Phi  \\
    &\textit{Relational formula:} & \phi & ~\Coloneqq~ \exists \varsched .\ \phi ~\mid~ \Phi 
\end{align*}
where 
$\ap \in \AP$,
$l \in \IL$, 
$\varsched \in \VarSched$, 
$q \in \Q$, and
$\comp \,\in \{ >, \geq, \approx_\epsilon, \neq,\not\approx_\epsilon, \leq, < \mid \epsilon \in \Qnn\}$.
A relational formula is \emph{closed} if every occurrence of a scheduler variable $\varsched$ is in the scope of a quantifier $\exists \varsched$.

\begin{remark}[On quantifiers in relational formulas]
    In relational formulas, we allow purely existential quantifier prefixes only; in particular, we do not allow quantifier alternations.
    However, formulas with a universal quantifier prefix (without alternations; e.g., \eqref{eq:vonNeumannExampleProperty} on page~\pageref{eq:vonNeumannExampleProperty}) are readily reducible to the existential form via negation.
    Importantly, even though universally quantified formulas appear in several examples throughout the paper, we reserve the term \emph{relational formula} for the existential variant as defined above.
    This distinction is relevant for our complexity results (e.g., an $\NP$-complete problem for existential formulas becomes $\coNP$-complete for universally quantified formulas).
\end{remark}

\subparagraph*{Semantics.}
A relational formula is evaluated over an MDP together with a function labeling each state with a set of atomic propositions as well as a mapping from initial state labels to states.
Let $\mdp$ be an MDP, $L \colon \states \to 2^{\AP}$ a state labeling function, and $i \colon \IL \to \states$ a mapping of initial state labels.
$\mdp$ with $L$ and $i$ satisfies a closed relational formula $\phi = \exists \varsched_1 \ldots \exists \varsched_\numsched .\ \Phi$, written $(\mdp, L, i) \models \phi$, iff there exists some assignment of the scheduler variables $\mathcal{I} = (\varsched_i \mapsto \sched_i)_{i=1,\ldots,\numsched}$ such that $\phi[\mathcal{I}]$ holds, where $\phi[\mathcal{I}]$ corresponds to $\phi$ with the scheduler variables instantiated according to $\mathcal{I}$ and the probability, temporal, Boolean, arithmetic and comparison operators are interpreted in the usual way, see, e.g.,~\cite[Ch.~5.1.2 and Ch.~10.2]{baierPrinciplesModel2008}.

\subparagraph*{Model-checking problem.}
In this paper, we address the following problem.
\begin{framedproblem}
    \label{prob:relational}
    Given an MDP $\mdp = \mdptup$ with a state labeling function $L \colon \states \to 2^{\AP}$ and a mapping of initial state labels $i \colon \IL \to \states$,
    as well as a closed relational formula $\phi$, 
    decide whether $(\mdp, \textit{L}, \textit{i}) \models \phi$.
\end{framedproblem}

In the remainder, we assume relational formulas to be given in some `normal form', and consider their `instantiation' w.r.t.\ a given MDP, i.e., we quantify directly over schedulers for a given MDP and instantiate the initial state labels and atomic propositions with states and sets of states, respectively. 

\begin{lemma}[Normal form for relational properties]
    \label{le:normal-form}
    Given an MDP $\mdp = \mdptup$ with a state labeling function $L \colon \states \to 2^{\AP}$ and a mapping of initial state labels $i \colon \IL \to \states$, as well as a closed relational formula $\phi$, there exist
    \begin{itemize}
        \item natural numbers $\numconj,\numsum,\numsched$,
        \item rational coefficients $q_{1,1}, \ldots, q_{\numsum,\numconj}$,
        \item rational bounds $q_1, \ldots, q_\numconj$,
        \item (not necessarily distinct) initial states $\state_{1,1}, \ldots, \state_{\numsum,\numconj}$,
        \item a set of indices $\{k_{1,1}, k_{1,2} \ldots, k_{\numsum-1,\numconj}, k_{\numsum,\numconj}\} = \{1, \ldots, \numsched\}$,
        \item temporal operators $\heartsuit_{1,1}, \ldots, \heartsuit_{\numsum, \numconj} \in \{ \Finally, \Globally\Finally\}$,
        \item (not necessarily distinct) target sets $T_{1,1}, \ldots, T_{\numsum,\numconj}$, and 
        \item comparison operators $\comp_1, \ldots, \comp_\numconj \;\in\! \{>, \geq, \approx_{\epsilon}, \not\approx_{\epsilon} \mid \epsilon \in \Q_{\geq 0} \}$,
    \end{itemize}
    such that $
    \displaystyle \quad 
    (\mdp, \textit{L}, \textit{i}) \models \phi
    \; \mathrel{\text{iff}} \;
    \exists \sched_1, \ldots, \sched_\numsched \in \Scheds .\
        \bigwedge_{j=1}^{\numconj} \sum_{i=1}^{\numsum} q_{i,j} 
        \Pr^{\sched_{k_{i,j}}}_{\state_{i,j}}(\mathrel{\heartsuit_{i,j}} T_{i,j}) \comp_j q_j
        ~.
    $
\end{lemma}

\begin{proof}
    Safety objectives can be reduced to reachability objectives, and coB\"uchi-objectives to B\"uchi objectives.    
    Properties with $\comp\,\in\!\!\{<,\leq\}$ can be transformed to the above form by multiplying coefficients with $-1$.
    Properties like $\exists \sched .\ \Pr^\sched_s(\Finally T) = \Pr^\sched_s(\Finally T)$ can be brought into the desired form by subtracting the right-hand-side on both sides of the equality, since we allow positive and negative coefficients.
\end{proof}

In this paper, we focus on the following three increasingly complex fragments of relational properties.
We start by analyzing \emph{relational reachability} formulas~(\cref{sec:reach}), 
i.e., properties of the form
\begin{align*} 
    \exists \sched_1, \ldots, \sched_\numsched \in \Scheds .\
    \sum_{i=1}^{\numsum} q_{i} 
    \Pr^{\sched_{k_{i}}}_{\state_{i}}(\Finally T_{i}) \comp q
    ~,
    \tag{\RelReach}
\end{align*}
which corresponds to the normal form of \cref{le:normal-form} with $\numconj = 1$ and $\heartsuit_{i,1} = \Finally$ for all $i=1,\ldots,\numsum$.
Afterwards, we extend this towards \emph{relational B\"uchi} formulas~(\cref{sec:buechi}), i.e.,
\[
    \exists \sched_1, \ldots, \sched_\numsched \in \Scheds .\
        \sum_{i=1}^{\numsum} q_{i} 
        \Pr^{\sched_{k_{i}}}_{\state_{i}}(\Globally \Finally T_{i}) \comp q
    ~,
    \tag{\RelBuechi}
\]
and finally to \emph{multi-objective relational reachability} formulas~(\cref{sec:conjunction}), i.e.,
\[
    \exists \sched_1, \ldots, \sched_\numsched \in \Scheds .\
        \bigwedge_{j=1}^{\numconj} \sum_{i=1}^{\numsum} q_{i,j} 
        \Pr^{\sched_{k_{i,j}}}_{\state_{i,j}}(\Finally T_{i,j}) \comp_j q_j
    ~.
    \tag{\ConjRelReach}
\]

Note that, by \Cref{le:normal-form}, \ConjRelReach already covers the full class of relational formulas restricted to reachability objectives.
Together with our results on \RelBuechi, a further extension of \ConjRelReach to mixtures of reachability and Büchi---and thus to general relational formulas as defined at the beginning of this section---is rather straightforward, but somewhat technical; we thus omit this case.

\input{interval}

%% file: interval.tex
\begin{remark}[Excursus: Interval DTMCs]
\label{sec:interval}
\label{rem:interval}

The MDPs from \cref{ex:conj_motivating-ex}, depicted in~\cref{fig:dice-simulation}~(p.~\pageref{fig:dice-simulation}), can be viewed as \emph{interval DTMCs}~\cite{jonssonSpecificationRefinement1991}, leading us to the question whether we can lift our results for relational properties on MDPs to relational properties on interval DTMCs.
Given some interval DTMC $\mathcal{I}$, we can create a corresponding MDP $\mdp_{\mathcal{I}}$ with the same state set as follows (for details on this reduction see, e.g.,~\cite{senModelCheckingMarkov2006}).
For every state $\state$ of $\mathcal{I}$, find the `basic feasible solutions' $\mu^{\state}_1, \ldots, \mu^{\state}_{n_\state} \in \Distr(\states)$ of the set of successor functions from $\state$, i.e., functions $\mu^{\state}_1, \ldots, \mu^{\state}_{n_\state}$ such that every possible successor function for $\state$ can be represented as a convex combination of $\mu^{\state}_1, \ldots, \mu^{\state}_{n_\state}$.
Then, we associate $\state$ in $\mdp_{\mathcal{I}}$ with actions $\action_1, \ldots, \action_{n_{\state}}$ and let $\Trans(\state, \action_i, \state') = \mu^{\state}_i(\state')$.
In general, this reduction introduces an exponential blow-up: For each state $\state$ in $\mdp_{\mathcal{I}}$, the number of actions $n_{\state}$ is bounded exponentially in its number of successors in $\mathcal{I}$.

Intuitively, \emph{dynamic semantics} for the interval DTMC correspond to memoryful schedulers for the constructed MDP, and \emph{static semantics} to memoryless ones (see, e.g.,~\cite{chenComplexityModel2013} for details on the different semantics).
For example, for any interval DTMC $\mathcal{I}$ and the corresponding MDP $\mdp_{\mathcal{I}}$, 
following the notation from~\cite{chenComplexityModel2013}, we have:
\begin{itemize}
    \item Under static semantics, an interval DTMC $\mathcal{I}$ is interpreted as an infinite set of DTMCs $[\mathcal{I}]$. 
    Then,
    there exists some instantiation $\dtmc \in [\mathcal{I}]$ s.t.\ 
    $\sum_i q_i \Pr^{\dtmc}_{\state_i}(\Finally T_i) \comp q$
    iff 
    there exists some \emph{memoryless} scheduler $\sched \in \SchedsMR[{\mdp_{\mathcal{I}}}]$ s.t.\ $\sum_i q_i \Pr^{\mdp_{\mathcal{I}}, \sched}_{\state_i}(\Finally T_i) \comp q$.

    \item Under dynamic semantics, an interval DTMC $\mathcal{I}$ is interpreted as an \emph{interval MDP} $\lceil \mathcal{I} \rceil$, which is defined like an MDP, but states might have infinitely many successors. 
    Then,
    there exists some scheduler $\tau \in \Scheds[\lceil \mathcal{I} \rceil]$ s.t.\ 
    $\sum_i q_i \Pr^{\lceil \mathcal{I} \rceil, \tau}_{\state_i}(\Finally T_i) \comp q$
    iff
    there exists some (memoryful) scheduler $\sched \in \Scheds[{\mdp_{\mathcal{I}}}]$ s.t.\ $\sum_i q_i \Pr^{\mdp_{\mathcal{I}}, \sched}_{\state_i}(\Finally T_i) \comp q$.
\end{itemize}

Hence, we can solve relational properties on interval DTMCs under dynamic semantics by applying our techniques for relational properties on MDPs.
Note that we can \emph{not} transfer our complexity results since the transformation from an interval DTMC to an MDP outlined above introduces an exponential blowup.

\end{remark}

%% file: reach.tex
\section{Relational Reachability}
\label{sec:reach}

An earlier version of this section appeared in \cite{gerlachEfficientProbabilistic2025}.
We consider the following problem:

{\color{cavcolor}
\begin{framedproblem}[\RelReach]
    \label{prob:relreach}
    Given an MDP $\mdp = \mdptup$, decide whether
    \[
        \exists \sched_1, \ldots, \sched_\numsched \in \Scheds .~
        \sum_{i=1}^{\numsum} q_i \cdot \Pr^{\sched_{k_i}}_{\state_{i}}(\Finally T_i) ~\comp~ q
        ~,
        \qquad\text{where}
    \]
    
    \begin{itemize}
        \item $\numsum,\numsched$ are natural numbers,
        \item $q_1, \ldots, q_{\numsum}$ are rational coefficients,
        \item $q$ is a rational bound,
        \item $\state_1, \ldots, \state_\numsum \in \states$ are (not necessarily distinct) initial states,
        \item $\{k_1, \ldots, k_\numsum\} = \{1, \ldots, \numsched\}$ is a set of indices,
        \item $T_1, \ldots, T_\numsum \subseteq \states$ are (not necessarily distinct) target sets, and 
        \item $\comp~ \in \{ >, \geq, \approx_\epsilon, \not\approx_\epsilon 
        \mid \epsilon \in \Qnn\}$ is a comparison operator.
    \end{itemize}
\end{framedproblem}
}

Recall that the motivating example from \cref{ex:reach_motivating-ex} is (the negation of) a relational reachability property.
{\color{cavcolor}
Below, we provide some further examples properties:
\begin{itemize}
    \item \emph{The probability of reaching $T$ from $s$ is \enquote{approximately scheduler-independent}:}
    \[\forall \sigma_1 \forall \sigma_2.~\Pr_s^{\sigma_1}(\Finally T) \approx_\epsilon \Pr_s^{\sigma_2}(\Finally T) ~.\]
    \item \emph{The probability to reach $T$ from $s_1$ is at least twice the probability of reaching $T$ from $s_2$, no matter the scheduler:}
        \[
            \forall \sigma .~ \Pr_{s_1}^{\sigma}(\Finally T) \geq 2\cdot \Pr_{s_2}^{\sigma}(\Finally T)~.
        \]
    \item \emph{There exists a scheduler reaching $T$ from $s_1$ with probability at least 10\% higher than reaching $T$ from $s_2$:}
        \[
            \exists \sigma .~ \Pr_{s_1}^{\sigma}(\Finally T) \geq \Pr_{s_2}^{\sigma}(\Finally T) + 0.1~.
        \]
    \item \emph{There is a scheduler that, \emph{in expectation}, visits more (different) targets from $\{T_1,\ldots,T_k\}$ than from $\{U_1,\ldots,U_\ell\}$:}
    \[\exists \sigma.~\Pr_s^\sigma(\Finally T_1) + \ldots + \Pr_s^\sigma(\Finally T_k) > \Pr_s^\sigma(\Finally U_1) + \ldots + \Pr_s^\sigma(\Finally U_\ell)~.\]
\end{itemize}
}

\subparagraph*{Outline of this section.}
In \cref{sec:reach_algo}, we show how to efficiently solve relational reachability properties by reducing them to the computation of expected rewards on the \emph{goal unfolding} of the MDP.
In~\cref{sec:reach_complexity}, we study the complexity of \RelReach under both general and MD schedulers.
Under general schedulers, the problem is in \EXPTIME but fixed-parameter tractable; under MD schedulers it is strongly \NP-hard. 
We identify further fragments where our algorithm runs in polynomial time and/or returns MD schedulers.
\cref{tab:complexity_overview_m=2} gives an overview on the complexity of selected fragments that compare two reachability probabilities, illustrating the border between \PTIME and strong \NP-hardness for MD schedulers.

\begin{table}[t]
    \caption{{\color{cavcolor}
        Complexity of selected classes of simple relational reachability properties over MD schedulers, where $\epsilon > 0$ and $\diseq~ \in \{ \geq, >, \not\approx_{\epsilon'} \mid \epsilon' \in \Qnn \}$.
        Over general schedulers, all variants considered here are in \PTIME (\Cref{th:general_PTIME}).
        }
    }
    \label{tab:complexity_overview_m=2}
    \setlength\tabcolsep{0pt}
    \centering
    \begin{tabular*}{\linewidth}{@{\extracolsep{\fill}} l  l }
        \toprule
         \bf Property class & \bf  Complexity over MD schedulers 
         \\
        \midrule
         $\exists \sched .\ \Pr^{\sched}_{\state}(\Finally T_1) =
       \Pr^{\sched}_{\state}(\Finally T_2)$
        & strongly \NP-complete [Th.~\ref{th:MD_NP-complete}\ref{item:1sched1state}] 
        \\
        $\exists \sched .\ \Pr^{\sched}_{\state}(\Finally T_1) \approx_\epsilon 
       \Pr^{\sched}_{\state}(\Finally T_2)$
        & \NP-complete [Th.~\ref{th:MD_NP-complete}\ref{item:1sched1state}] 
        \\
        $\exists \sched .\ \Pr^{\sched}_{\state}(\Finally T_1) \diseq \Pr^{\sched}_{\state}(\Finally T_2)$
        & in \NP [Th.~\ref{th:MD_NP}]; \PTIME if $T_1, T_2$ absorb. [Th.~\ref{th:MD_PTIME}\ref{item:MD_n-comb}]
        \\
        \midrule 
        $\exists \sched .\ \Pr^{\sched}_{\state_1}(\Finally T_1) = \Pr^{\sched}_{\state_2}(\Finally T_2)$
        & strongly \NP-complete [Th.~\ref{th:MD_NP-complete}\ref{item:1sched2state}]
        \\
        $\exists \sched .\ \Pr^{\sched}_{\state_1}(\Finally T_1) \approx_\epsilon \Pr^{\sched}_{\state_2}(\Finally T_2)$
        & \NP-complete [Th.~\ref{th:MD_NP-complete}\ref{item:1sched2state}]
        \\
        $\exists \sched .\ \Pr^{\sched}_{\state_1}(\Finally T_1) \diseq \Pr^{\sched}_{\state_2}(\Finally T_2)$
        & in \NP [Th.~\ref{th:MD_NP}]
        \\
        \midrule 
        $\exists \sched_1, \sched_2 .\ \Pr^{\sched_1}_{\state_1}(\Finally T_1) = \Pr^{\sched_2}_{\state_2}(\Finally T_2)$
        & strongly \NP-complete [Th.~\ref{th:MD_NP-complete}\ref{item:2sched2state}]
        \\
        $\exists \sched_1, \sched_2 .\ \Pr^{\sched_1}_{\state_1}(\Finally T_1) \approx_\epsilon \Pr^{\sched_2}_{\state_2}(\Finally T_2)$
        & \NP-complete [Th.~\ref{th:MD_NP-complete}\ref{item:2sched2state}]
        \\
        $\exists \sched_1, \sched_2 .\ \Pr^{\sched_1}_{\state_1}(\Finally T_1) \diseq \Pr^{\sched_1}_{\state_2}(\Finally T_2)$
        & \PTIME [Th.~\ref{th:MD_PTIME}\ref{item:MD_independent}]
        \\
        \bottomrule
    \end{tabular*}
\end{table}

%% file: reach_algo.tex
\subsection{Verifying Relational Reachability Properties}
\label{sec:reach_algo}

{\color{cavcolor}
Assume an arbitrary MDP $\mdp = \mdptup$ and a \RelReach property 
\begin{align}
    \exists \sched_1, \ldots, \sched_n \in \Scheds .~ \sum_{i=1}^{m} q_i \cdot \Pr^{\sched_{k_i}}_{\state_{i}}(\Finally T_i) ~\comp~ q 
    \label{eq:genRelReachProp}
\end{align}
as in~\cref{prob:relreach}.
In the following, we outline a four-step procedure that checks whether property (\ref{eq:genRelReachProp}) holds and, if yes, constructs (possibly memoryful randomized) witness schedulers.
The procedure is summarized in \cref{alg:linear_general_simplified} for the comparison relation $\approx_\epsilon$. 

\begin{example}
    \label{ex:runningExIntro}
    The MDP in \Cref{fig:runningEx} (left) together with the property
    \[
        \exists \sched .~~ \Pr^{\sched}_{s_1}(\Finally T) -\nicefrac{1}{2} \cdot \Pr^{\sched}_{s_1}(\Finally T') - \nicefrac{1}{2} \cdot \Pr^{\sched}_{s_2}(\Finally T') ~\approx_\epsilon~ 0
        ~,
    \]
    (\emph{Does there exist a scheduler such that the probability of reaching $T$ from $s_1$ is approximately equal to the mean of the probabilities of reaching $T'$ from $s_1$ and $T'$ from $s_2$?}) will serve as a running example throughout the section.
\end{example}

\subparagraph*{Step 1: Collect combinations of initial states and schedulers.}
\label{step1}
We start by analyzing the relationship of schedulers and states in the property. 

\begin{definition}[State-scheduler combinations]
    \label{def:comb}
    We define 
    $\comb = \{ (s_i, \sched_{k_i}) \mid i=1, \ldots, m\}$, the set of all different combinations of initial states and schedulers that occur in the property~\eqref{eq:genRelReachProp}.
    Furthermore, for every $c = (s, \sched) \in \comb$, we define $\relInd(c)$ as the set of indices $i$ such that $(s_i, \sched_{k_i}) = c$ and $\indexc[\mathcal{T}] = \{T_i \mid i \in \relInd(c)\}$.
\end{definition}

\begin{example}
    In the property from \Cref{ex:runningExIntro} we have $\comb = \{c_1=(s_1,\sched), c_2=(s_2,\sched)\}$.
    Furthermore, $\relInd(c_1) = \{1,2\}$, $\relInd(c_2) = \{3\}$, $\mathcal{T}_{c_1}=\{T,T'\}$, and $\mathcal{T}_{c_2}=\{T'\}$. 
\end{example}

Notice that $n \leq |\comb| \leq m$.
State-scheduler combinations allow introducing fresh scheduler variables, one per combination:
\begin{restatable}{lemma}{whyComb}
    \label{thm:whyComb}
    Let $\comb = \{c_1,\ldots,c_{k}\}$.
    Then property \eqref{eq:genRelReachProp} is equivalent to
    \[
        \exists \sched_{c_1},\ldots,\sched_{c_k}.~ \sum\nolimits_{i=1}^{k}  \Big[ \sum\nolimits_{j \in \relInd(c_i)} q_j \cdot \Pr_{s_j}^{\sched_{c_i}}(\Finally T_j) \Big]
        ~\comp~ q
        ~.
    \]
\end{restatable}
\begin{proof}[Proof (Sketch)]
    Quantifying over each state-scheduler combination individually is justified because schedulers may use memory and thus remember the initial state, see \cref{app:whyComb} for details. 
\end{proof}

\begin{example}
    Applying \Cref{thm:whyComb} to the property from \Cref{ex:runningExIntro} yields the following equivalent property (over general, memoryful schedulers):
    \[
        \exists \sched_{c_1}, \sched_{c_2} .~
        \underbrace{\left[1 \cdot \Pr^{\sched_{c_1}}_{s_1}(\Finally T) -\tfrac{1}{2} \cdot \Pr^{\sched_{c_1}}_{s_1}(\Finally T')\right]}_{\text{combination } c_1 \,=\, (s_1,\textcolor{gray}{\sched})}
        +
        \underbrace{\left[-\tfrac{1}{2} \cdot \Pr^{\sched_{c_2}}_{s_2}(\Finally T')\right]}_{\text{combination } c_2 \,=\, (s_2,\textcolor{gray}{\sched})}
        ~\approx_\epsilon~
        0
    ~.
    \]
\end{example}
}

\begin{figure}[t]
    \centering
    \begin{minipage}[t]{0.34\textwidth}
    \scalebox{.7}{
        \begin{tikzpicture}[on grid,node distance=15mm and 20mm,semithick,>=stealth]
            \node[state] (s1) {$s_1$};
            \node[state,below=of s1] (s2) {$s_2$};
            \node[state,accepting,right=of s1,label=90:{$T'\,{=}\,\{t_2\}$}] (t2) {$t_2$};
            \node[dist,above =of s1] (a) {};
            \node[state,accepting,left=of s1,label=90:{$T\,{=}\,\{t_1\}$}] (t1) {$t_1$};
            \draw[->] (s1) edge node[left] {$\beta$} (s2);
            \draw[-] (s1) edge[bend left=15] node[left] {$\alpha$} (a);
            \draw[->] (s2) edge node[below] {$\alpha$} (t2);
            \draw[->] (s2) edge[loop below] node[left] {$\beta$} (s2);
            \draw[->] (t2) edge[loop below] (t2);
            \draw[->] (t1) edge (s1);
            \draw[->] (a) edge[bend left=15] node[right] {$\tfrac 1 3$} (s1);
            \draw[->] (a) edge node[above] {$\tfrac 1 3$}(t1);
            \draw[->] (a) edge node[above] {$\tfrac 1 3$} (t2);
        \end{tikzpicture}
        }
    \end{minipage}
    \hfill
    \begin{minipage}[t]{0.65\textwidth}
    \scalebox{.7}{
        \begin{tikzpicture}[on grid,node distance=15mm and 20mm,semithick,>=stealth]
            \node[staterectangle] (s1) {$s_1,\emptyset$};
            \node[staterectangle,below=of s1] (s2) {$s_2,\emptyset$};
            \node[staterectangle,accepting,right=of s1,label={90:\textcolor{red}{$-\frac 1 2$}}] (t2) {$t_2,\emptyset$};
            \node[dist,above =of s1] (a) {};
            \node[staterectangle,accepting,left=of s1,label={-90:\textcolor{red}{$+1$}}] (t1) {$t_1,\emptyset$};
            \node[staterectangle,below=of t2] (t2new) {$t_2,\{T'\}$};
            
            \node[staterectangle,right=65mm of s1] (s1') {$s_1,\{T\}$};
            \node[staterectangle,below=of s1'] (s2') {$s_2,\{T\}$};
            \node[staterectangle,accepting,right=25mm of s1',label={90:\textcolor{red}{$-\frac 1 2$}}] (t2') {$t_2,\{T\}$};
            \node[dist,above =of s1'] (a') {};
            \node[staterectangle,left=25mm of s1'] (t1') {$t_1,\{T\}$};
            \node[staterectangle,below=of t2'] (t2new') {$t_2,\{T,T'\}$};
            
            \draw[->] (s1) edge node[left] {$\beta$} (s2);
            \draw[-] (s1) edge[bend left=15] node[left] {$\alpha$} (a);
            \draw[->] (s2) edge node[below] {$\alpha$} (t2);
            \draw[->] (s2) edge[loop below] node[left] {$\beta$} (s2);
            \draw[->] (t2) edge (t2new);
            \draw[->] (t2new) edge[loop below] (t2new);
            \draw[->] (t1) edge[out=75,in=155] (s1');
            \draw[->] (a) edge[bend left=15] node[right] {$\tfrac 1 3$} (s1);
            \draw[->] (a) edge node[above] {$\tfrac 1 3$}(t1);
            \draw[->] (a) edge node[above] {$\tfrac 1 3$} (t2);
            
            \draw[->] (s1') edge node[left] {$\beta$} (s2');
            \draw[-] (s1') edge[bend left=15] node[left] {$\alpha$} (a');
            \draw[->] (s2') edge node[below] {$\alpha$} (t2');
            \draw[->] (s2') edge[loop below] node[left] {$\beta$} (s2');
            \draw[->] (t2') edge (t2new');
            \draw[->] (t2new') edge[loop below] (t2new');
            \draw[->] (t1') edge (s1');
            \draw[->] (a') edge[bend left=15] node[right] {$\tfrac 1 3$} (s1');
            \draw[->] (a') edge node[above] {$\tfrac 1 3$}(t1');
            \draw[->] (a') edge node[above] {$\tfrac 1 3$} (t2');
        \end{tikzpicture}
        }
    \end{minipage}
    \caption{{\color{cavcolor}
    An MDP (left) and its goal unfolding with \textcolor{red}{rewards} (right).
    }}
    \label{fig:runningEx}
\end{figure}
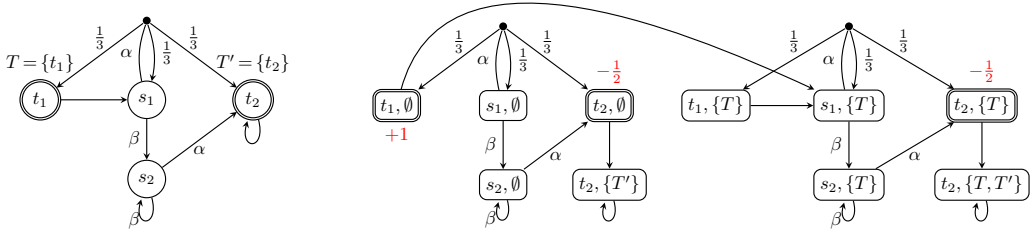

{\color{cavcolor}
\subparagraph*{Step 2: Unfold targets and set up reward structures.}
\label{step2}
Next we process each combination $c \in \comb$ individually.
We rely on two established techniques from the literature:
Including reachability information in the state space \cite{quatmannVerificationMultiobjective2023,forejtQuantitativeMultiobjective2011} 
and encoding reachability probabilities as \emph{expected rewards} (e.g.,~\cite[pp.~51~ff.]{quatmannVerificationMultiobjective2023}).
For the sake of completeness, we detail these steps nonetheless:
}

\begin{definition}[{Goal unfolding w.r.t.\ $\mathcal{T}$}]
    \label{def:goalUnfolding}
    Let $\mathcal{T}\subseteq 2^{\states}$ be a set of sets of states of $\mdp$ with $\mathcal{T}\neq\emptyset$.
    The \emph{goal unfolding} $\indexT$ of $\mdp$ w.r.t.\ $\mathcal{T}$ 
    is the
    MDP $\indexT = (\indexT[\states], \Act, \indexT[\Trans])$ where $\indexT[\states] = \states \times 2^{\mathcal{T}}$, and
    $\indexT[\Trans]$ is defined as follows:
    For $\state, \state' \in \states$, $\mathcal{T}',\mathcal{T''} \subseteq \mathcal{T}$, and $\action \in \Act$,
    \[
        \indexT[\Trans]\big((\state, \mathcal{T}'), \action, (\state', \mathcal{T}'')\big)
        ~=~ 
        \begin{cases}
            \Trans(\state, \action, \state') & \text{if } 
            \mathcal{T}'' = \mathcal{T}' \cup \mathcal \{ T \in \mathcal{T} \mid s \in T \} \\
        0 & \text{otherwise}~.
        \end{cases}
    \]
\end{definition}

For $c\in \comb$, we use $\indexc$ to denote the goal unfolding of $\mdp$ w.r.t.\ $\indexc[\mathcal{T}] = \{T_i \mid i \in \relInd(c)\}$.
{\color{cavcolor}
For a combination $c = (s,\textcolor{gray}{\sched})$ we use $\state_{c}$ to denote the state $(\state, \emptyset)$ in $\mdp_c$.

\begin{example}
    \label{ex:unfolding}
    The goal unfolding of the MDP in \Cref{fig:runningEx} (left) w.r.t.\ combination $c_1 = (s_1,\textcolor{gray}{\sched})$, for which we have $\mathcal T_{c_1} = \{T,T'\}$, is depicted in \Cref{fig:runningEx} (right).
\end{example}

\begin{definition}[Reward structure for state-scheduler combination]
    \label{def:rewForComb}
    Let $c \in \comb$.
    We define the reward structure $\indexc[\rew] \colon \states_c \to \Q$ on the goal unfolding $\indexc$ by $\indexc[\rew] = \sum_{T \in \mathcal T_c} q_T \cdot \rew_{T}$, 
    where $q_T = \sum_{i \in \{\relInd(c)\mid T = T_i\}} q_i$ and
    \[
        \rew_T \colon \indexc[\states] \to \Q,~
        (\state, \mathcal T) \mapsto
        \begin{cases}
            1 & \text{if } \state \in T \wedge T \not\in \mathcal T \\
            0 & \text{otherwise}~.
        \end{cases}
    \] 
\end{definition}
Intuitively, we collect reward equal to the sum of the coefficients occurring together with a target $T \in \mathcal T_c$ when we visit $T$ \emph{for the first time}.
For any $\sched \in \Scheds[\indexc]$, the reward function $\indexc[\rew]$ can be naturally lifted to $\indexc^\sched$ and further to infinite paths of $\indexc^\sched$ by letting $\indexc[\rew](\pi) = \sum_{i=0}^{\infty} \indexc[\rew](\pi(i))$ for $\pi \in \Paths[\indexc]$; this is well-defined because we collect reward only finitely often on any path.
The \emph{expected reward} of $\indexc[\rew]$ on $\indexc$ from $s_c$ under some $\sched \in \Scheds[\indexc]$ is then defined as the expectation of the function $\indexc[\rew](\pi) = \sum_{i=0}^{\infty} \indexc[\rew](\pi(i))$.
Then, we can reduce our query to a number of expected reward queries (see \cref{app:weightedReachProbsAsExpectedRew} for the proof):

\begin{restatable}{lemma}{weightedReachProbsAsExpectedRew}
    \label{thm:weightedReachProbsAsExpectedRew}
    For every combination $c = (s, \textcolor{gray}{\sched}) \in \comb$ and $\opt 
    \in \{\min,\max\}$:
    \[
        \opt_{\sched \in \Scheds[\mdp]} \sum_{j \in \relInd(c)} q_j \cdot \Pr_{s}^{\mdp,\sched}(\Finally T_j)
        ~=~
        \opt_{\sched \in \Scheds[\indexc]} \Expected^{\indexc, \sched}_{s_c}(\indexc[\rew])
        ~.
    \]
\end{restatable}

\begin{example}
    Following \Cref{ex:unfolding}, the (non-zero) rewards $\rew_c$ for $c = (s_1,\textcolor{gray}{\sched})$ are given in \textcolor{red}{red} next to the states in \Cref{fig:runningEx} (right).
\end{example}

\subparagraph*{Step 3: Compute expected rewards.}
\label{step3}
The next step is to compute, for each individual scheduler-state combination $c$, the maximal and minimal rewards occurring in \Cref{thm:weightedReachProbsAsExpectedRew}.
Again, we rely on existing techniques from the literature~\cite{putermanMarkovDecision1994,karmarkarNewPolynomialtime1984,quatmannVerificationMultiobjective2023}.\footnote{\color{cavcolor}Note that we (must) rely on results supporting positive and negative reward, since we allow positive and negative coefficients.} 
We refer to \cref{app:exRewPtimeMD} for the proof.

\begin{restatable}{lemma}{exRewPtimeMD}
    \label{thm:exRewPtimeMD}
    Let $c \in \comb$ and $\opt \in \{\max,\min\}$.
    The optimal expected reward of $\indexc[\rew]$ from $\indexc[\state]$, $\opt_{\sched \in \Scheds[\indexc]} \Expected^{\indexc, \sched}_{\indexc[\state]}(\indexc[\rew]) \in \Q$, is computable in time polynomial in the size of $\mdp_c$.
    Moreover, the optimum is attained by an MD scheduler $\sigma \in \SchedsMD[\mdp_c]$.
\end{restatable}

\begin{example}
    \label{ex:expectedRewards}
    Reconsider the MDP in \Cref{fig:runningEx} (right).
    The maximal expected reward from initial state $(s_1,\emptyset)$ is $\tfrac 1 4$ and is attained by the MD strategy that always chooses $\alpha$ in $(s_1,\emptyset)$ and thus eventually reaches either $(t_2,\emptyset)$ or $(t_1,\emptyset)$ with probability $\tfrac 1 2$ each.
    In the latter case, the strategy then selects $\beta$ in $(s_1,\{T\})$ to reach $(s_2,\{T\})$ and remain there forever, not collecting any further reward.
    Overall, this strategy collects a total expected reward of $\tfrac 1 2 \cdot 1 + \tfrac 1 2 \cdot (-\tfrac 1 2) = \tfrac 1 4$.
    The minimal expected reward is easily seen to be $-\tfrac 1 2$.
\end{example}

\begin{remark}[Approximate vs exact]  
    \label{remark:approx_exact}
    In practice, \emph{exact} computation of the optimal expected reward via LP as suggested by \Cref{thm:exRewPtimeMD} (and its proof) may be significantly slower than approximation~\cite{hartmannsPractitionersGuide2023}.
    Fortunately, it is possible to amend our algorithm to \emph{approximate} expected reward computation.
    To retain soundness, it is crucial to employ a procedure such as \emph{Sound Value Iteration}~\cite{quatmannSoundValue2018} that yields guaranteed \emph{under-} and \emph{over-approximations} of the true result.
    Appropriate handling of such approximations is detailed in \cref{app:reach_full-algo}.
    Note that approximation inherently leads to incompleteness, i.e., the algorithm may return \enquote{inconclusive} in some cases.
    Further, for each $c \in \comb$ with $|\relInd(c)|>1$ we can view $\opt_{\sched \in \Scheds[\mdp]} \sum_{j \in \relInd(c)} q_j \cdot \Pr_{s}^{\mdp,\sched}(\Finally T_j)$ as a \emph{weighted-sum optimization problem} and employ multi-objective model-checking techniques~\cite{quatmannVerificationMultiobjective2023}. 
    (For $|\relInd(c)|=1$ this is a single-objective model-checking query.)
\end{remark}

\subparagraph*{Step 4: Aggregate results and check relational property.}
\label{step4}
We now combine the optimal expected rewards $\opt_{\sched \in \Scheds[\indexc]} \Expected^{\indexc, \sched}_{s_c}(\indexc[\rew])$ for each state-scheduler combination $c \in \comb$ obtained via the previous two steps. 
We exemplify this for the comparison relation $\approx_\epsilon$, the other relations are handled similarly (see \cref{app:reach_full-algo}).

\begin{restatable}{lemma}{cruxApproxEqual}
    \label{thm:cruxApproxEqual}
    For each $c \in \comb$ and $\opt \in \{\max,\min\}$ let
    \[
        v_c^{\opt}
        ~=~
        \opt_{\sched \in \Scheds[\indexc]} \Expected^{\indexc, \sched}_{s_c}(\indexc[\rew])
        ~.
    \]
    Furthermore, let $v^{\opt} = \sum_{c \in \comb} v_c^{\opt}$.
    Then, assuming that the comparison operator $\comp$ in the property \eqref{eq:genRelReachProp} is $\approx_\epsilon$ for some $\epsilon \geq 0$:
    \[
        q \in [v^{\min} - \epsilon, v^{\max} + \epsilon] ~\Iff~ \text{property \eqref{eq:genRelReachProp} holds} ~.
    \]
\end{restatable}
\Cref{thm:cruxApproxEqual} relies on the fact that any value in the interval of achievable probabilities
$[v_c^{\min}, v_c^{\max}]$ can be achieved by constructing the \emph{convex combination} (e.g.,~\cite[p.\ 71]{quatmannVerificationMultiobjective2023}) of the minimizing and the maximizing scheduler.

\begin{proof}[Proof of \NoCaseChange{\cref{thm:cruxApproxEqual}} (Sketch)]
    The interesting direction is ``$\Rightarrow$'':
    By \Cref{thm:whyComb,thm:weightedReachProbsAsExpectedRew} there exist schedulers $\sched_1^{\geq},\ldots,\sched_n^{\geq}$ and $\sched_1^{\leq},\ldots,\sched_n^{\leq}$ for $\mdp$ such that   
    \begin{align*}
        \ub{v} := \sum_{i=1}^{m} q_i \cdot \Pr^{\sched^{\geq}_{k_i}}_{\state_{i}}(\Finally T_i)
        ~\geq~
        q - \epsilon~, \quad
        \lb{v} := \sum_{i=1}^{m} q_i \cdot \Pr^{\sched^{\leq}_{k_i}}_{\state_{i}}(\Finally T_i)
        ~\leq~
        q + \epsilon ~ .
    \end{align*}
    If $\ub{v} \leq q + \epsilon$ \emph{or} $\lb{v} \geq q - \epsilon$, then $\sched_1^{\geq},\ldots,\sched_n^{\geq}$ or $\sched_1^{\leq},\ldots,\sched_n^{\leq}$ are already witnessing schedulers and there is nothing else to show.
    Otherwise, $\ub{v} > q + \epsilon$ \emph{and} $\lb{v} < q - \epsilon$.
    Thus there exists $\lambda \in (0,1)$ such that $q = \lambda \lb v + (1-\lambda) \ub v$.
    The schedulers $\sched_i^{\lambda} = [\sched^{\leq}_i \oplus_\lambda \sched^{\geq}_i]$, $i = 1,\ldots,n$ \cite[p.~71]{quatmannVerificationMultiobjective2023} then witness satisfaction of property \eqref{eq:genRelReachProp} with \emph{exact} equality ($\approx_0$).
    See \cref{app:cruxApproxEqual} for more details.
\end{proof}

\begin{example}
    We wrap up our running example by proving that the property from \Cref{ex:runningExIntro} indeed holds in the MDP in \Cref{fig:runningEx} (left).
    For combination $c_1 = (s_1,\textcolor{gray}{\sched})$  we have already established in \Cref{ex:expectedRewards} that $v_{c_1}^{\max} = \tfrac 1 4$ and $v_{c_1}^{\min} = -\tfrac 1 2$.
    For the other combination $c_2$ one easily finds $v_{c_2}^{\max} = 0$ and $v_{c_2}^{\min} = -\tfrac 1 2$. Summing up these values yields $v^{\min} = -1$ and $v^{\max} = \tfrac 1 4$.
    Since $0 \in [v^{\min}, v^{\max}]$, the property is satisfiable, even with exact equality $\approx_0$.
\end{example}

We remark that for $\comp~ \in \{\geq, >\}$, it actually suffices to compute \emph{only} $v^{\max}$ rather than both $v^{\max}$ and $v^{\min}$ as in \Cref{thm:cruxApproxEqual}.
}
{\color{cavcolor}
\subparagraph*{Overall algorithm.}
}

\begin{algorithm}[t]
    \caption{{\color{cavcolor}Solving \RelReach with $\approx_\epsilon$}} 
    \label{alg:linear_general_simplified}
    \label{alg:reach}
    \Input{
        MDP $\mdp = \mdptup$ and a \RelReach property
        \\ \medskip 
        \phantom{Input: } 
        $\displaystyle \exists \sched_1, \ldots, \sched_n \in \Scheds .~ \sum_{i=1}^{m} q_i \cdot \Pr^{\sched_{k_i}}_{\state_{i}}(\Finally T_i) \approx_\epsilon q$
        \DontPrintSemicolon\tcp*{See \cref{prob:relreach}}
        \medskip
    }
    \Output{Whether the property is true in $\mdp$}
    \tcp{Step 1: Loop over all state-scheduler combinations:}
    \For{$c = (s,\sched) \in \comb = \{ (s_i, \sched_{k_i}) \mid i=1, \ldots, m\}$}{
        \tcp{Step 2: Unfold and define reward structures:}
        $\indexc \gets$ unfold $\mdp$ w.r.t.\ target sets for $c$ \tcp*{See \Cref{def:goalUnfolding}}
        $\indexc[\rew] \gets$ reward structure on $\indexc$ for $c$ \tcp*{See \Cref{def:rewForComb}}
        \tcp{Step 3: Compute (or approximate) expected rewards:}
        $\indexc[{v}]^{\max} \gets  {\displaystyle\max_{\sched \in \Scheds[\indexc]}} \Expected^{\indexc, \sched}_{s_c}(\indexc[\rew]) ~;~
        \indexc[{v}]^{\min} \gets  {\displaystyle\min_{\sched \in \Scheds[\indexc]}} \Expected^{\indexc, \sched}_{s_c}(\indexc[\rew])$
        \label{line:linear_general_simplified_step3}
    }
    \tcp{Step 4: Aggregate results from state-scheduler combinations and check:}
    $v^{\max} \gets \sum_{c \in \comb} \indexc[{v}]^{\max} ~;~ {v}^{\min}  \gets \sum_{c \in \comb} \indexc[{v}]^{\min}$ \;
    \label{line:linear_general_simplified_step4}
    \Return{$q \in [v^{\min} - \epsilon, v^{\max} + \epsilon]$}
    \label{line:linear_general_simplified_return}
\end{algorithm}

{\color{cavcolor}
The algorithm resulting from Steps 1--4 is stated explicitly as \cref{alg:linear_general_simplified} for $\approx_\epsilon$ and in full generality in \cref{app:reach_full-algo}.

\begin{theorem}[Correctness and time complexity]
    \label{thm:correctness}
    \cref{alg:linear_general_simplified}
    adheres to its input-output specification.
    It can be implemented with worst-case running time of $\mathcal{O}(m \cdot \textit{poly}(2^m \cdot |\mdp|))$, where $m$ is the number of probability operators in the property.
\end{theorem}
\begin{proof}
    \Cref{thm:cruxApproxEqual} establishes correctness.
    Regarding time complexity, notice that 
    for each $c \in \comb$, the size of the goal unfolding $\indexc$ is bounded by $2^m \cdot |\mdp|$ (Step 2).
    Exact computation of expected rewards is possible in time polynomial in the size of the MDP $\indexc$ (Step 3).
    Steps 2 and 3 have to be executed for at most $|\comb|\leq m$ state-scheduler combinations.
\end{proof}
}

%% file: reach_complexity.tex
\subsection{Complexity of Relational Reachability}
\label{sec:reach_complexity}

{\color{cavcolor}
In this section, we analyze the computational complexity of the \RelReach problem over general and over memoryless deterministic schedulers, respectively.
We also identify restricted variants of \RelReach that are decidable in polynomial time.

\subsubsection{General Schedulers}

The runtime analysis from \cref{thm:correctness} yields an \EXPTIME upper bound for the complexity of the \RelReach problem.

\begin{theorem}
    \label{th:general_EXPTIME}
    Problem \RelReach is \PSPACE-hard and decidable in \EXPTIME.
\end{theorem}
\begin{proof}
    For \PSPACE-hardness observe that \emph{simultaneous almost-sure reachability} of $m$ target sets $T_1,\ldots,T_m$ (which is known to be \PSPACE-complete~\cite[Th.\ 2]{randourPercentileQueries2017}) is expressible as the \RelReach property \enquote{$\exists \sigma .\ \Pr_{s}^\sigma(\Finally T_1) + \ldots \Pr_{s}^\sigma(\Finally T_m) \geq m$}.

    Membership in \EXPTIME follows from \cref{thm:correctness}.
\end{proof}
}

We leave tighter complexity bounds for \RelReach as an open problem.
In particular, it remains unclear whether \RelReach belongs to \PSPACE:
While the exponentially large goal unfolding need not be stored entirely in memory---e.g., by processing the individual MDP copies sequentially in a bottom-up manner (cf.~\cite[Proof of Th.\ 2]{randourPercentileQueries2017})---it is unclear whether the intermediate expected rewards needed to determine the final answer can be represented using only polynomially many bits.

{\color{cavcolor}
However, we observe that \RelReach can be solved in \PTIME if the number of probability operators $\numsum$ is fixed. 
Hence, \RelReach is \emph{fixed-parameter tractable} \cite{groheDescriptiveParameterized1999} with parameter $\numsum$. 
The following theorem generalizes this observation and further states that the exponential blow-up of the goal unfolding can be avoided if all target states are absorbing, or if each probability operator is evaluated under a different scheduler (i.e., if $n=m$).
We refer to \cref{app:general_PTIME} for the proof.

\begin{restatable}{theorem}{generalPTIME}
    \label{th:general_PTIME}
    The following special cases of \RelReach are in \PTIME:
    \begin{enumerate}[label=(\alph*)]
        \item \label{item:fixed-param}
        The number of \emph{different} target sets $|\{ T_1,\ldots,T_m\}|$ is \emph{at most a constant}.

        \item \label{item:absorbing}
        The target sets $T_1,\ldots,T_m$ are all \emph{absorbing}.

        \item \label{item:independent}
        $n=m$, i.e., each probability operator in the property has \emph{its own quantifier}.
    \end{enumerate}
\end{restatable}

\subsubsection{Memoryless Deterministic Schedulers}

We now consider \RelReachMD, the \RelReach problem over MD schedulers. \RelReachMD is in \NP because we can non-deterministically guess schedulers and verify whether they are witnesses in polynomial time by computing the (exact) reachability probabilities in the induced DTMC~\cite[Ch.~10]{baierPrinciplesModel2008}.
Further, \RelReach is \emph{strongly} \NP-hard\footnote{\color{cavcolor}A problem is strongly \textsf{NP}-hard if it is \textsf{NP}-hard even if all numerical quantities (here: rational transition probabilities) in a given input instance are encoded in unary.
}~\cite{gareyStrongNPCompleteness1978}
over MD schedulers already for simple variants with equality.

\begin{restatable}{theorem}{MDNPComplete}
    \label{th:MD_NP-complete}
    \label{th:MD_NP}
    \RelReachMD is strongly \NP-complete.
    Strong \NP-hardness already holds for the following special cases:
    For a given MDP $\mdp$, 
    initial states $\state_1,\state_2 \in \states$,
    target sets $T_1, T_2 \subseteq \states$,
    decide if
    \begin{enumerate}[label=\textnormal{(\alph*)},itemsep=2pt]
        \item 
        \label{item:1sched1state}
        \label{item:1sched2state}
        $
            \exists \sched \in \SchedsMD .\ 
            \Pr^{\sched}_{\state_1}(\Finally T_1) - \Pr^{\sched}_{\state_2}(\Finally T_2) = 0 \ .
        $
        
        \item 
        \label{item:2sched2state}
        $
            \exists \sched_1, \sched_2 \in \SchedsMD .\ 
            \Pr^{\sched_1}_{\state_1}(\Finally T_1) - \Pr^{\sched_2}_{\state_2}(\Finally T_2) = 0 \ .
        $
    \end{enumerate}
    Strong \NP-hardness of \ref{item:1sched1state} and \ref{item:2sched2state} holds irrespective of whether 
    $\state_1 = \state_2$
    and whether $T_1$ and/or $T_2$ are absorbing.
    Moreover, \ref{item:1sched1state} and \ref{item:2sched2state} with relation $\approx_\epsilon$, $\epsilon > 0$, are \NP-hard.
\end{restatable}

\begin{proof}[Proof (Sketch)]
    We show strong \NP-hardness by giving a pseudo-polynomial transformation from the \emph{Hamiltonian path} problem, which is known to be strongly \NP-hard \cite{gareyStrongNPCompleteness1978}, inspired by \cite{footePolynomialTime}.
    \NP-hardness of the cases for approximate equality follows by an analogous transformation, but for $\epsilon$ the transformation is only polynomial, not pseudo-polynomial, hence establishing \NP-hardness but not strong \NP-hardness.
    We refer to \cref{app:MD_NP-complete} for the full proof.
\end{proof}

Note that the hardness of the problem does not rely on whether all probability operators are evaluated under the same scheduler and from the same initial state. 

We observe that Steps 1--4 detailed in \cref{sec:reach_algo} construct memoryful randomized witness schedulers in case of approximate equality, if they exist.
Memory and/or randomization are necessary for constructing a scheduler that exactly achieves some specified reachability probability, in general.

\begin{restatable}{theorem}{generalNecessary}
    \label{th:general_necessary}
    Memory and randomization are necessary for \RelReach with (approximate or exact) equality.
\end{restatable}

\begin{proof}[Proof (Sketch)]
    Memory: Consider the MDP in \cref{fig:memory-necessary}.
    Over general schedulers, the property
    \enquote{$\exists \sched_1 .\ \Pr^{\sched}_{s}(\Finally \{t_1\}) = \Pr^{\sched}_{s}(\Finally \{t_2\})$} holds, 
    but over memoryless schedulers it does not.
    Randomization: Consider the MDP in \cref{fig:randomization-necessary}.
    Over general schedulers, the property 
    \enquote{$\exists \sched .\ \Pr^\sched_s(\Finally \{t\}) \approx_{0.1} 0.5$} holds, but over deterministic schedulers it does not.
    See \cref{app:general_necessary} for details.
\end{proof}
}

\begin{figure}[t]
    \centering
    \begin{subfigure}[b]{0.49\textwidth}
        \centering
        \scalebox{.7}{
    \begin{tikzpicture}
        \node[state] (s) {$s$};
        
        \node[state, below left= of s] (t1) {$t_1$}; 
        \node[state, below right= of s] (t2) {$t_2$}; 

       \path[-latex', draw]
            (s) edge[bend left] node[right] {$\alpha$} (t1)
            (s) edge node[right] {$\beta$} (t2);

        \path[-latex', draw]
            (t1) edge[bend left] node[right] {$\gamma$} (s)
            (t2) edge[loop right] (t2);
    \end{tikzpicture}
    }
    \caption{$\exists \sched .\ \Pr^{\sched}_{s}(\Finally \{t_1\}) = \Pr^{\sched}_{s}(\Finally \{t_2\})$}
    \label{fig:memory-necessary}
    \end{subfigure}%
    \begin{subfigure}[b]{0.49\textwidth}
        \centering
        \scalebox{.7}{
    \begin{tikzpicture}
        \node[state] (sinita) {$s$};  
        
        \node[state, below left= of sinita] (t) {\phantom{$s_{2}$}}; 
        \node at (t) (th) {$t$};
        \node[state, below right= of sinita] (sx) {$s_\bot$};

        \path[-latex', draw]
            (sinita) edge node[left] {$\alpha$} (t)
            (sinita) edge node[right] {$\beta$} (sx);

        \path[-latex', draw]
            (t) edge[loop left] (t)
            (sx) edge[loop right] (sx);
    \end{tikzpicture}
    }
    \caption{$\exists \sched .\ \Pr^\sched_s(\Finally \{t\}) \approx_{0.1} 0.5$ 
    }
    \label{fig:randomization-necessary}
    \end{subfigure}%
    \caption{
    {\color{cavcolor}
        MDPs where memory and/or randomization are necessary for relational reachability properties with equality.
        }
    }
\end{figure}
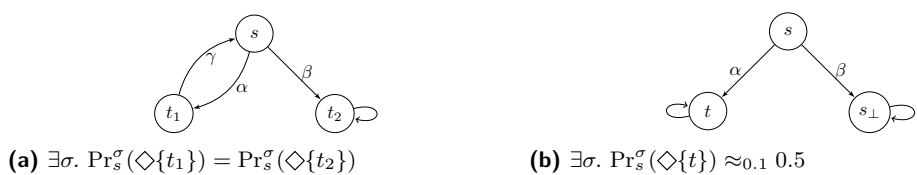

{\color{cavcolor}
Note that \cref{th:MD_NP-complete,th:general_necessary} only make statements about properties with (approximate) equality.
Let us now consider \RelReach with inequality or disequality ($\comp~ \in \{\geq, >, \not\approx_\epsilon \mid \epsilon \geq 0 \}$).
Recall that, for inequality, we only have to check the maximizing (or, for $\not\approx_\epsilon$, possibly also the minimizing) schedulers for the goal unfoldings and transform them back to schedulers for the original MDP.
This transformation introduces memory in general:
If several probability operators are associated with the same scheduler but different initial states, then, intuitively, it might be necessary to switch behavior depending on the initial state.
Further, if several probability operators are evaluated under the same scheduler but for different target sets, then it might be necessary to switch behavior depending on which target sets were already visited. 

The next example illustrates that memory may be necessary for relational reachability properties with disequality even for just a single, absorbing target.

\begin{example}
    \label{ex:oneschedtwostate-exists-greater-memory}
    Consider the MDP depicted in \cref{fig:oneschedtwostate-exists-greater-memory} 
    and the property
    \[
        \exists \sched \in \Scheds .\ \Pr^{\sched}_{s_2}(\Finally \{ t \})  < \Pr^{\sched}_{s_1}(\Finally \{ t \}) \ .
    \]
    Here, both probability operators are evaluated under the \emph{same scheduler} but \emph{different initial states}, and have the same, absorbing target set.
    Over MD schedulers this property cannot be satisfied: For all MD schedulers it holds that $\Pr^{\sched}_{s_1}(\Finally \{ t \}) = \Pr^{\sched}_{s_2}(\Finally \{ t \})$.
    In contrast, there does exist a memoryful scheduler such that the probability of reaching $t$ from $\statea$ exceeds the probability of reaching $t$ from $\stateb$, namely the scheduler that chooses $\alpha$ at $\statea$ if the execution was started at $\statea$, and $\beta$ otherwise.
    This also implies that 
    \[
        \exists \sched \in \Scheds .\ \Pr^{\sched}_{s_1}(\Finally \{ t \}) \neq \Pr^{\sched}_{s_2}(\Finally \{ t \})  \ 
    \]
    can be satisfied over general schedulers but not over MD schedulers.
\end{example}
}

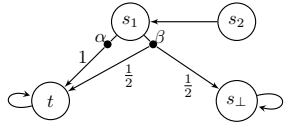
\begin{figure}
    \centering
    \scalebox{.7}{
    \begin{tikzpicture}[on grid,node distance=15mm and 20mm,semithick,>=stealth]
        \node[state] (sinita) {$\statea$};   
        \node[state, right= of sinita] (sinitb) {$\stateb$};   
        
        \node[state, below left= of sinita, xshift=0.5cm] (t) {\phantom{$s_{2}$}}; 
        \node at (t) (th) {$t$};
        \node[state, below right= of sinita] (sx) {$s_\bot$}; 

        \node[dist, below left of=sinita, node distance=4ex] (alpha) {};
        \node[dist, below right of=sinita, node distance=4ex] (beta) {};
        \path[-, draw]
            (sinita) edge node[left] {$\alpha$} (alpha)
            (sinita) edge node[right] {$\beta$} (beta);

        \path[-latex', draw]
            (alpha) edge node[pos=0.25, left] {$1$} (t)
            (beta) edge node[pos=0.25, below] {$\frac{1}{2}$} (t)
            (beta) edge node[pos=0.5, below] {$\frac{1}{2}$} (sx);

        \path[-latex', draw]
            (sinitb) edge (sinita)
            (t) edge[loop left] (t)
            (sx) edge[loop right] (sx);
    \end{tikzpicture}
    }
    \caption{{\color{cavcolor}MDP where memory is necessary for a relational reachability property with in- or disequality.}}
    \label{fig:oneschedtwostate-exists-greater-memory}
\end{figure}

{\color{cavcolor}
\noindent However, under some mild restrictions on the target sets and/or the structure of the state-scheduler combinations, the procedure detailed in \cref{sec:reach_algo} returns MD schedulers in \PTIME, if they exist. 
Correctness (\cref{thm:correctness}) then directly implies that MD schedulers suffice in these cases.
It remains future work to determine whether there are variants of \RelReachMD with disequality that are \NP-hard. 
}%
Let \RelReachNotEq denote the fragment of \RelReach restricted to comparison operators $\{\geq,>,\not\approx_{\epsilon} \mid \epsilon\geq 0\}$.
{\color{cavcolor}

\begin{restatable}{theorem}{MDPTIME}
    \label{th:MD_PTIME}
    For the following special cases of \RelReachNotEq, the procedure in \cref{sec:reach_algo} runs in polynomial time and induces MD schedulers, if it returns true:
    \begin{enumerate}[label=(\alph*)]
        \item \label{item:MD_independent}
        $n=m$, i.e., each probability operator in the property has \emph{its own quantifier}. 

        \item \label{item:MD_n-comb}
        Probability operators with the same scheduler variable have the same initial state 
        (formally, $\forall i,i' .\ \sched_{k_i} {=} \sched_{k_{i'}} \implies \state_i {=} \state_{i'}$)
        and all target sets are absorbing.

        \item \label{item:MD_sign}
        Probability operators with the same scheduler variable have equally signed coefficients and the same target sets
        (formally, $\forall i,i' .\ \sched_{k_i} {=} \sched_{k_{i'}} \implies ((q_i \geq 0 \iff q_{i'} \geq 0) \wedge T_i {=} T_{i'})$).
    \end{enumerate}
\end{restatable}

\begin{corollary}
    \label{th:MD_suffice}
    For the \RelReach variants from \cref{th:MD_PTIME}, MD schedulers suffice.
\end{corollary}

The proof of \cref{th:MD_PTIME} can be found in \cref{app:MD_PTIME}. 
We focus here on the intuition.
Firstly, %
if $n=m$ (\ref{item:MD_independent}),  
then all probability operators are independent in the sense that we only need to solve $n=m$ independent single-objective queries, for which there exist optimal MD schedulers~\cite{putermanMarkovDecision1994}.\footnote{\color{cavcolor}Note that this reasoning does not work for $\approx_\epsilon$, because there we may need to combine the optimal MD schedulers into memoryful randomized schedulers to obtain a witness. }

For \ref{item:MD_n-comb}, we consider the statement for $n=1$. 
If all probability operators are associated with the same scheduler and initial state and all target states are sinks, then
MD schedulers suffice as there is nothing to remember: We know which state we started from (since there is only a single initial state), and we know which target sets have already been visited (since all targets are sinks).   

Lastly, consider \ref{item:MD_sign} for $n=1$ and non-negative coefficients.
In this case,
since all target sets are the same, there must exist an MD scheduler maximizing all reachability probabilities $\Pr^\sched_{\state_i}(\Finally T)$ at the same time~\cite{putermanMarkovDecision1994}.
Since all coefficients are non-negative, this scheduler also maximizes the weighted sum over the probabilities. We can analogously reason about minimizing the weighted sum.

}

\paragraph*{Summary of \cref{sec:reach}.} 
The key insight is that solving \RelReach can be reduced to expected reward computations on the goal unfoldings with respect to the relevant target sets for each state-scheduler combination.
\RelReach can be solved in time exponential only in the number of different target sets that occur in the property.
\RelReach is in general \PSPACE-hard, but can be solved in polynomial time in the size of the input if the number of target sets is at most a constant, the target sets are absorbing or each probability operator has its own scheduler quantifier.
\RelReachMD is strongly \NP-hard, but there are several fragments where we can compute MD schedulers in polynomial time.

%% file: buechi.tex
\section{Relational B\"uchi}
\label{sec:buechi}

We extend our approach from reachability to B\"uchi objectives: 

\begin{framedproblem}[\RelBuechi]
    \label{prob:relbuechi}
    Given an MDP $\mdp = \mdptup$, decide whether
    \[
        \exists \sched_1, \ldots, \sched_\numsched \in \Scheds .\
        \sum_{i=1}^{\numsum} 
        q_i \Pr^{\sched_{k_i}}_{\state_i}\left(\Globally \Finally T_i\right) \comp q
        ~,
        \qquad\text{where}
    \]
    \begin{itemize}
        \item $\numsum, \numsched$ are natural numbers, 
        \item $q_1, \ldots, q_\numsum$ are rational coefficients,
        \item $q$ is a rational bound,
        \item $\state_1, \ldots, \state_\numsum \in \states$ are (not necessarily distinct) initial states,
        \item $\{k_1, \ldots, k_\numsum\} = \{1, \ldots, n\}$ is a set of indices,
        \item $T_1, \ldots, T_\numsum \subseteq \states$ are (not necessarily distinct) target sets, and 
        \item $\comp~ \in \{ >, \geq, \approx_\epsilon, \not\approx_\epsilon 
        \mid \epsilon \in \Qnn\}$ is a comparison relation.
    \end{itemize}    
\end{framedproblem}

The motivating example from~\cref{ex:buechi_motivating-ex} is (the negation of) a relational B\"uchi property.

\subparagraph*{Outline of this section.}
In \cref{sec:buechi_algo}, we show how to solve relational B\"uchi queries by reducing them to relational reachability queries on a variation of the standard \emph{MEC quotient} inspired by \cite{baierCertificatesWitnesses2024,baierFoundationsProbabilityraising2024}.
In \cref{sec:buechi_complexity}, we prove that the problem \RelBuechi is strongly NP-complete over both general and MD schedulers and present special cases that can be solved in \PTIME or where MD scheduler suffice. 
\cref{tab:buechi_complexity_overview_m=2} gives an overview of the border between \PTIME and strong \NP-hardness for MD schedulers, for analogous cases as given in~\cref{tab:complexity_overview_m=2}, differences to the complexity results for \RelReach are indicated in \textbf{bold}.

\begin{table}[t]
    \caption{Complexity of selected classes of simple relational B\"uchi properties over MD schedulers, where $\epsilon > 0$ and $\diseq~ \in \{ \geq, >, \not\approx_{\epsilon'} \mid \epsilon' \in \Qnn \}$, 
    for the cases presented in \cref{tab:complexity_overview_m=2} for \RelReach. Differences to the results for \RelReach are indicated in \textbf{bold}.
    Over general schedulers, all variants considered here are in \PTIME (\Cref{th:buechi_general_PTIME}\ref{item:buechi_fixed-param}). 
    }
    \label{tab:buechi_complexity_overview_m=2}
    \setlength\tabcolsep{0pt}
    \centering
    \begin{tabular*}{\linewidth}{@{\extracolsep{\fill}} l  l }
        \toprule
         \bf Property class & \bf  Complexity over MD schedulers 
         \\
        \midrule
         $\exists \sched .\ \Pr^{\sched}_{\state}(\Globally \Finally T_1) =
       \Pr^{\sched}_{\state}(\Globally \Finally T_2)$
        & strongly \NP-complete [Th.~\ref{th:buechi_MD_NP-complete}\ref{item:buechi_1sched1state}] 
        \\
        $\exists \sched .\ \Pr^{\sched}_{\state}(\Globally \Finally T_1) \approx_\epsilon 
       \Pr^{\sched}_{\state}(\Globally \Finally T_2)$
        & \NP-complete [Th.~\ref{th:buechi_MD_NP-complete}\ref{item:buechi_1sched1state}] 
        \\
        $\exists \sched .\ \Pr^{\sched}_{\state}(\Globally \Finally T_1) \diseq \Pr^{\sched}_{\state}(\Globally \Finally T_2)$
        & in \NP [Th.~\ref{th:buechi_MD_NP}]; \PTIME if $T_1, T_2$ absorb. [Cor.~\ref{cor:buechi_MD_absorb}]
        \\
        \midrule 
        $\exists \sched .\ \Pr^{\sched}_{\state_1}(\Globally \Finally T_1) = \Pr^{\sched}_{\state_2}(\Globally \Finally T_2)$
        & strongly \NP-complete [Th.~\ref{th:buechi_MD_NP-complete}\ref{item:buechi_1sched2state}]
        \\
        $\exists \sched .\ \Pr^{\sched}_{\state_1}(\Globally \Finally T_1) \approx_\epsilon \Pr^{\sched}_{\state_2}(\Globally \Finally T_2)$
        & \NP-complete [Th.~\ref{th:buechi_MD_NP-complete}\ref{item:buechi_1sched2state}]
        \\
        $\exists \sched .\ \Pr^{\sched}_{\state_1}(\Globally \Finally T_1) \diseq \Pr^{\sched}_{\state_2}(\Globally \Finally T_2)$
        & in \NP [Th.~\ref{th:buechi_MD_NP}]
        \\
        \midrule 
        $\exists \sched_1, \sched_2 .\ \Pr^{\sched_1}_{\state_1}(\Globally \Finally T_1) = \Pr^{\sched_2}_{\state_2}(\Globally \Finally T_2)$
        & strongly \NP-complete [Th.~\ref{th:buechi_MD_NP-complete}\ref{item:buechi_2sched2state}]
        \\
        $\exists \sched_1, \sched_2 .\ \Pr^{\sched_1}_{\state_1}(\Globally \Finally T_1) \approx_\epsilon \Pr^{\sched_2}_{\state_2}(\Globally \Finally T_2)$
        & \NP-complete [Th.~\ref{th:buechi_MD_NP-complete}\ref{item:buechi_2sched2state}]
        \\
        $\exists \sched_1, \sched_2 .\ \Pr^{\sched_1}_{\state_1}(\Globally \Finally T_1) \diseq \Pr^{\sched_2}_{\state_2}(\Globally \Finally T_2)$
        & \textbf{in \NP [Th.~\ref{th:buechi_MD_NP}]; \PTIME if $T_1, T_2$ absorb. [Cor.~\ref{cor:buechi_MD_absorb}]}
        \\
        \bottomrule
    \end{tabular*}
\end{table}

%% file: buechi_algo.tex
\subsection{Verifying Relational B\"uchi Properties}
\label{sec:buechi_algo}

Assume an MDP $\mdp = \mdptup$ and 
a \RelBuechi property
\begin{align}
    \exists \sched_1, \ldots, \sched_\numsched \in \Scheds .\
        \sum_{i=1}^{\numsum} 
        q_i \Pr^{\sched_{k_i}}_{\state_i}\left(\Globally \Finally T_i\right) \comp q
    \label{eq:genRelBuechiProp}
\end{align}
as in~\cref{prob:relbuechi}.
\begin{assumption}
\label{ass:distinct-act}
  To simplify notation, in the following we assume w.l.o.g.~that all states in $\mdp$ have disjoint sets of actions.
\end{assumption}
We show how to reduce property~\eqref{eq:genRelBuechiProp} to a relational reachability property on a variation of the standard \emph{MEC quotient} of $\mdp$. 
\cref{alg:buechi} gives an overview of our approach.

\begin{example}
    \label{ex:buechi_running}
    As a running example, we use the MDP depicted in \cref{fig:buechi_running-ex_mdp} along with the \RelBuechi property
    \[
        \exists \sched_1,\sched_2 .\ 
        \Pr^{\sched_1}_{s}(\Globally\Finally T_1) + \Pr^{\sched_1}_{s}(\Globally\Finally T_2) - \Pr^{\sched_2}_{s}(\Globally\Finally T_1) = 0~,
    \]
    which intuitively asks whether there exist two schedulers $\sched_1$, $\sched_2$ such that the sum of the probabilities of visiting $T_1$ and $T_2$, respectively, infinitely often under $\sched_1$ is the same as the probability of visiting $T_1$ infinitely often under $\sched_2$.  
\end{example}

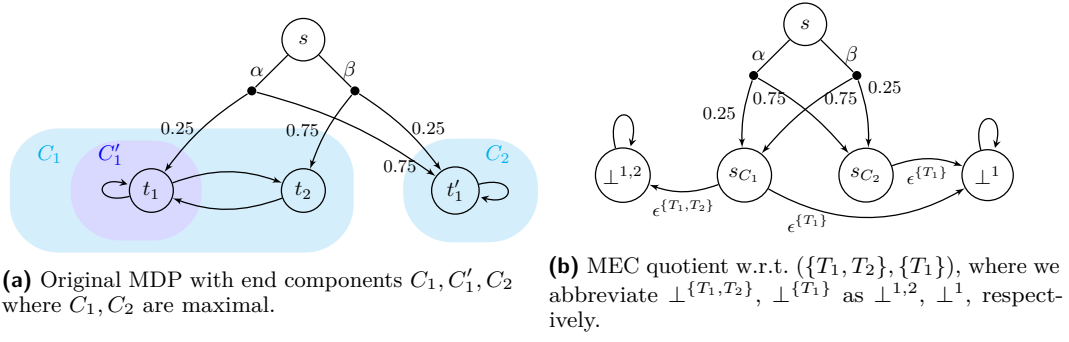
\begin{figure}
    \centering
    \begin{subfigure}[c]{0.48\textwidth}
        \centering
        \scalebox{.8}{
    \begin{tikzpicture}[
            node distance=25mm and 25mm,
            on grid,semithick,>=stealth,
            ]
        \node[state] (s) {$s$};
        \node[dist, below left=12mm of s] (alpha) {};
        \node[dist, below right=12mm of s] (beta) {};
        
        \node[state, below left= of s] (t1) {$t_1$}; 
        \node[state, below= of s] (t2) {$t_2$}; 
        
        \node[state, below right= of s] (t1p) {$t'_1$}; 

        \path[-, draw]
            (s) edge node[left] {$\alpha$} (alpha)
            (s) edge node[right] {$\beta$} (beta);
            
        \path[-latex', draw]
            (alpha) edge[bend right=10] node[font=\footnotesize, left] {$0.25$} (t1)
            (alpha) edge[bend left=10] node[font=\footnotesize, below, pos=0.8] {$0.75$} (t1p)
            (beta) edge[bend right=10] node[font=\footnotesize, left] {$0.75$} (t2)
            (beta) edge[bend left=10] node[font=\footnotesize, right] {$0.25$} (t1p)
            (t1) edge[bend left=20] (t2)
            (t2) edge[bend left=20] (t1);

        \path[-latex', draw]
            (t1) edge[loop left] (t1)
            (t1p) edge[loop right] (t1p);

        \begin{scope}[on background layer]
        \node (box12) [inner sep=13pt, fill=cyan!15, line width=.8pt,rounded corners=0.8cm,fit = {([xshift=-15mm] t1.west) ([yshift=2mm] t2.north) ([yshift=-2mm] t2.south) (t2.east)}, label={[yshift=6mm, xshift=10mm]left:{\textcolor{cyan}{$\mec_1$}}}] {};
        \node (box1) [inner sep=13pt, fill=blue!15, line width=.8pt,rounded corners=0.8cm,fit = {([xshift=-5mm] t1.west) (t1.east) (t1.north) (t1.south)}, label={[yshift=6mm, xshift=10mm]left:{\textcolor{blue}{$\ec'_1$}}}] {};
        \node (box1p) [inner sep=13pt, fill=cyan!15, line width=.8pt,rounded corners=0.8cm,fit = {(t1p.west) ([xshift=5mm] t1p.east) (t1p.north) (t1p.south)}, label={[yshift=6mm, xshift=-10mm]right:{\textcolor{cyan}{$\mec_2$}}}] {};
        \end{scope}[on background layer]
    \end{tikzpicture}
    }
    \caption{Original MDP with end components $C_1, C'_1, C_2$ where $C_1, C_2$ are maximal.}
    \label{fig:buechi_running-ex_mdp}
    \end{subfigure}
    \quad
    \begin{subfigure}[c]{0.48\textwidth}
        \centering
        \scalebox{.8}{
        \begin{tikzpicture}[
            node distance=25mm and 20mm,
            on grid,semithick,>=stealth,
            ]
        \node[state] (s) {$s$};
        \node[dist, below left=12mm of s] (alpha) {};
        \node[dist, below right=12mm of s] (beta) {};
        
        \node[state, below = of s, xshift=-10mm] (sC1) {$s_{\mec_1}$}; 
        \node[state, below= of s, xshift=10mm] (sC2) {$s_{\mec_2}$}; 

        \node[state, left= of sC1, inner sep=2pt] (bot12) {$\bot^{1,2}$}; 
        \node[state, right= of sC2] (bot1) {$\bot^{1}$}; 
    
        \path[-, draw]
            (s) edge node[left] {$\alpha$} (alpha)
            (s) edge node[right] {$\beta$} (beta);
            
        \path[-latex', draw]
            (alpha) edge[bend right=10] node[font=\footnotesize, left] {$0.25$} (sC1)
            (alpha) edge[bend left=10] node[font=\footnotesize, below, pos=0.1] {$0.75$} (sC2)
            (beta) edge[bend right=10] node[font=\footnotesize, below, pos=0.1] {$0.75$} (sC1)
            (beta) edge[bend left=10] node[font=\footnotesize, right, pos=0.1] {$0.25$} (sC2)
            (sC1) edge[bend left=20] node[font=\footnotesize, below] {$\epsilon^{\{T_1, T_2\}}$} (bot12)
            (sC1) edge[bend right=30] node[font=\footnotesize, below, pos=0.2] {$\epsilon^{\{T_1\}}$} (bot1)
            (sC2) edge[bend left=20] node[font=\footnotesize, below] {$\epsilon^{\{T_1\}}$} (bot1);

        \path[-latex', draw]
            (bot1) edge[loop above] (bot1)
            (bot12) edge[loop above] (bot12);
    \end{tikzpicture}
    }
    \caption{MEC quotient w.r.t.\ $(\{T_1, T_2\}, \{T_1\})$, where we abbreviate $\bot^{\{T_1, T_2\}}$, $\bot^{\{T_1\}}$ as $\bot^{1,2}$, $\bot^1$, respectively.}
    \label{fig:buechi_running-ex_quotient}
    \end{subfigure}
    \caption{
    MDP from the running example introduced in~\cref{ex:buechi_running}, 
    where $T_1 = \{t_1, t'_1\}$, $T_2 = \{t_2\}$.
    Edges without label correspond to transitions with probability 1.}
    \label{fig:buechi_running-ex}
\end{figure}

\subparagraph*{Step 1: Collect combinations of initial states and schedulers. }
We define the set $\comb$ of combinations of initial states and schedulers, and $\relInd(c)$ and $\indexc[\mathcal{T}]$ for $c \in \comb$ as before (see \cref{def:comb}).

\begin{example}
    Continuing \cref{ex:buechi_running}, we have $\comb = \{(\state,\sched_1), (\state,\sched_2)\}$, with $\mathcal{T}_{(\state,\sched_1)} = \{ T_1, T_2\}$ and $\mathcal{T}_{(\state,\sched_2)} = \{T_1\}$, and $\relInd(\state,\sched_1) = \{1,2\}$ and $\relInd(\state, \sched_2) = \{3\}$.
\end{example}

\subparagraph*{Step 2: Build MEC quotient and define success sets.}
For each $c \in \comb$, we analyze for each MEC and each subset $\mathcal{T}$ of $\indexc[\mathcal{T}]$ whether we can enforce to see the target sets from $\mathcal{T}$ infinitely often and all other target sets from $\indexc[\mathcal{T}]$ only finitely often while staying inside this MEC. 
Let us stress that our goal is to visit \emph{at least one} state from \emph{each} $T \in \mathcal{T}$ infinitely often, which is a stronger condition than requiring to visit at least one state from $\bigcup_{T \in \mathcal{T}} T$ infinitely often.
Since each path of an MDP stays inside some MEC in the limit almost-surely, we analyze in which MECs we can ensure to visit only the target sets from a given $\mathcal{T} \subseteq \indexc[\mathcal{T}]$ infinitely often.

\begin{definition}
    Let $c \in \comb$.
    For
    $\mathcal{T} \subseteq \indexc[\mathcal{T}]$
    we define 
    \begin{align*}
        \EC_{\mathcal{T}} &= \left\{ \ec'=(S',A') \in \EC(\mdp) \mid \left(\forall T \in \mathcal{T} .\ S' \cap T \neq \emptyset \right) \land 
        \left(\forall T \in \indexc[\mathcal{T}] \setminus \mathcal{T} .\ S' \cap T = \emptyset \right) \right\},
        \\
        \MEC_{\mathcal{T}} &= \left\{ \mec \in \MEC(\mdp) \mid 
        \exists \ec' \in \EC(\mdp) .\ \ec' \subseteq \mec \land {}
        \ec' \in \EC_{\mathcal{T}} 
        \right\} 
        ~.
    \end{align*}
\end{definition}

\begin{example}
    For the MDP from \cref{fig:buechi_running-ex_mdp}, we have $\EC_{\{T_1\}} = \{C'_1, C_2\}$, $\EC_{\{T_2\}} = \emptyset$, $\EC_{\{T_1, T_2\}} = \{C_1\}$ and 
    $\MEC_{\{T_1\}} = \{C_1, C_2\}$, $\MEC_{\{T_2\}} = \emptyset$, $\MEC_{\{T_1, T_2\}} = \{C_1\}$.
\end{example}

\begin{remark}
    Since the number of ECs of an MDP might be exponential in the number of states, it might be inefficient to construct $\MEC_{\mathcal{T}}$ for a given $\mathcal{T} \subseteq 2^\states$ by computing the set of all sub-ECs for every MEC $\mec$. 
    Instead, in our prototypical implementation~(see~\cref{sec:implement}), we check for each MEC $\mec$ whether it is contained in $\MEC_{\mathcal{T}}$ by removing the states from $\bigcup_{T \not \in \mathcal{T}} T$ from $\mec$ and then checking whether the remaining subsystem has an MEC containing at least one state from each $T \in \mathcal{T}$.
\end{remark}

We build the MEC quotient with a separate sink state $\bot^{\mathcal{T}}$ for every subset $\mathcal{T}$ of some $\indexc[\mathcal{T}]$ for some $c \in \comb$ with $\MEC_{\mathcal{T}} \neq \emptyset$.
For each MEC $\mec \in \MEC_{\mathcal{T}}$ we add a transition from the collapsed state $\state_\mec$ to $\bot^{\mathcal{T}}$.
Intuitively, moving to a sink $\bot^{\mathcal{T}}$ represents staying in some MEC forever and seeing exactly the target sets contained in $\mathcal{T}$ infinitely often. 
This definition differs from the standard definition of an MEC quotient~(see, e.g., \cite{alfaroFormalVerification1997,baierFoundationsProbabilityraising2024}) only in the introduction of several different sink states.
Before giving the formal definition of the MEC quotient, we establish a correspondence between the states of the original MDP and the states of the MEC quotient.

\begin{definition}[Collapsed states]
\label{def:collapsedStates}
    Let $\collapsedStates := (\states \setminus \states_{\MEC}) \cup \{\state_\mec \mid \mec \in \MEC(\mdp)\}$.
    We define $\mapStates \colon \states \to \collapsedStates$ as follows:
    For $\state \in \states$, let $\mapStates(\state) = \state_C$ if there exists some $\mec \in \MEC(\mdp)$ with $s \in \mec$ and $\mapStates(\state) = \state$ otherwise.
\end{definition}

\begin{definition}[MEC quotient $\quotT$]
    \label{def:MEC-quotient_set}
    Let $\collapsedStates$ and $\mapStates \colon \states \to \collapsedStates$ as defined in~\cref{def:collapsedStates}.
    The \emph{MEC quotient of $\mdp$ w.r.t.\ $(\mathcal{T}_c)_{c \in \comb}$} is the MDP
    $\quotT = \mdptup[\quotT]$ with 
    \begin{itemize}
        \item $\quotT[\states] = \collapsedStates \cup \overbrace{\bigcup_{c \in \comb} \{\bot^{\mathcal{T}} \mid {\mathcal{T}} \subseteq \mathcal{T}_c, \MEC_{\mathcal{T}} \neq \emptyset \}}^{\states_\bot}$.

        \item $\quotT[\Act] = \Act \cup \bigcup_{c\in \comb} \{ \epsilon^{\mathcal{T}} \mid {\mathcal{T}} \subseteq \mathcal{T}_c \}$
    
        \item
        We define the transition function $\quotT[\Trans] \colon \quotT[\states] \times \quotT[\Act] \times \quotT[\states] \to [0,1]$ as follows, distinguishing between collapsed states, sink states and non-MEC states: 
        \begin{itemize}
            \item For $\state \in \states \setminus \states_{\MEC}$, $\action \in \Act$ and $\quotT[\state] \in \quotT[\states] \setminus \states_{\bot}$ we let $\quotT[\Trans](\state, \action, \quotT[\state]) = \sum_{s' \in \mapStates^{-1}(\quotT[\state])} \Trans(\state, \action, s')$.
        
            \item For $\mec=(S',A) \in MEC(\mdp)$, $\action \in \Act$ and $\quotT[\state] \in \quotT[\states] \setminus \states_{\bot}$:
            \[
                \quotT[\Trans](\state_\mec, \action, \quotT[\state]) = 
                \begin{cases}
                    \sum_{s' \in \mapStates^{-1}(\quotT[\state])} \Trans(\state, \action, s') & \text{if } \action \not\in A \land \exists \state' \in S' .\ \action \in \Act(\state') \\
                    0 & \text{otherwise}~.
                \end{cases}
            \]
    
            \item For $\mec=(S',A) \in MEC(\mdp)$, and $\mathcal{T} \subseteq \mathcal{T}_c$ for some $c \in \comb$ with $\mec \in \MEC_{\mathcal{T}}$ we let $\quotT[\Trans](s_\mec, \epsilon^{\mathcal{T}}, \bot^{\mathcal{T}}) = 1$ and $\quotT[\Trans](\bot^{\mathcal{T}}, \epsilon^{\mathcal{T}}, \bot^{\mathcal{T}}) = 1$.
        \end{itemize} 

        \item All other transition probabilities are 0.
    \end{itemize}
\end{definition}

\begin{remark}
    Note that the well-definedness of the transition function of $\quotT$ relies on our assumption~(p.~\pageref{ass:distinct-act}) that each state has a distinct set of actions.
    We could alternatively relax this assumption and set $\quotT[\Act] = \{ (\state,\action) \mid \action \in \Act(\state), \state \in \states\}$; we chose the current form to improve readability.
\end{remark}

Observe that $|\quotT[\states]|$ is exponential in the number of objectives $m$, in fact already $|\states_\bot|$ is exponential in the number of different objectives $m$.

\begin{example}
    \label{ex:MEC-quotient}
    The MEC quotient of the MDP from \cref{fig:buechi_running-ex_mdp} w.r.t.\ $(\{T_1, T_2\}, \{T_1\})$ is depicted in \cref{fig:buechi_running-ex_quotient}.
\end{example}

\begin{remark}
    Dividing the set of target sets by state-scheduler combination, i.e., defining $\sdead$ as in~\cref{def:MEC-quotient_set} instead of letting $\sdead = \{ \bot^{\mathcal{T}} \mid \mathcal{T} \subseteq \{T_1, \ldots, T_m\}, \MEC_{\mathcal{T}} \neq \emptyset \}$, is not necessary for the correctness of our reduction from \RelBuechi to \RelReach, but important for efficiency.
    If there is only a single target (or a constant number of targets) per state-scheduler combination, then $|\sdead|$ is linear $m$ for our definition, while it would be exponential in $m$ for the alternative definition of $\sdead$.
    In general, however, dividing the sink states by state-scheduler combination does not improve the upper bound on the size of $\sdead$ below $2^m$.
\end{remark}

\begin{definition}[Success Set]
    \label{def:success-set}
    Let $\quotT[\mdp]$ be as defined in \cref{def:MEC-quotient_set}.
    For $i=1,\ldots,m$, the \emph{success set of $T_i$ in $\quotT$} is defined as
    \begin{align*} 
        U_{T_i} 
        &= \{ \bot^{\mathcal{T}} \mid \exists c \in \comb .\ T_i \in \mathcal{T} \subseteq \mathcal{T}_c \wedge \MEC_{\mathcal{T}} \neq \emptyset \} 
        ~.
    \end{align*}%
\end{definition}

\begin{example}
    \label{ex:buechi_success-sets}
   Continuing our example from \cref{ex:MEC-quotient}, we have $U_{T_1} = \{\bot^{\{T_1, T_2\}}, \bot^{\{T_1\}}\}$ and $U_{T_2} = \{\bot^{\{T_1, T_2\}}\}$.
\end{example}

\subparagraph*{Step 3: Solve relational reachability query.}
We can now reduce the relational B\"uchi query on the original MDP to a relational reachability query for the success sets on the MEC quotient~(see~\cref{def:success-set} and \cref{def:MEC-quotient_set}, respectively).
We use ideas from~\cite{baierCertificatesWitnesses2024,baierFoundationsProbabilityraising2024} for transferring witnessing schedulers between the original MDP and the MEC quotient.
Recall that $\mapStates$ maps states of $\mdp$ to states of $\quotT$.

\begin{restatable}{theorem}{relBuechiToRelReach}
    \label{th:relBuechi-to-relReach}
    Let $\quotT[\mdp]$ and $U_{T_1}, \ldots, U_{T_m}$ as defined in \cref{def:MEC-quotient_set} and \cref{def:success-set}, respectively. Then,
    \begin{align*}
        & \exists \sched_1, \ldots, \sched_\numsched \in \Scheds .\
        \sum_{i=1}^{\numsum} 
        q_i \Pr^{\sched_{k_i}}_{\state_i} \left( \Globally \Finally T_i \right)
        \comp q
        \\ \Iff \quad &
        \exists \quotT[\sched_1], \ldots, \quotT[\sched_\numsched]  \in \Scheds[\quotT] .\ 
        \sum_{i=1}^{\numsum} 
        q_i \Pr^{\quotT,\quotT[\sched_{k_i}]}_{\mapStates(\state_i)}(\Finally U_{T_i}) \comp q
        ~.
    \end{align*}
\end{restatable}

\begin{example}
    Our running example from~\cref{ex:buechi_running} is equivalent to the following \RelReach property on the MEC quotient $\quotT$ w.r.t.\ $(\{T_1, T_2\}, \{T_1\})$ (depicted in~\cref{fig:buechi_running-ex_quotient}):
    \[
        \exists \sched_1, \sched_2 \in \Scheds[\quotT] .\ 
        \Pr^{\quotT,\sched_1}_{s}(\Finally U_{T_1}) + \Pr^{\quotT,\sched_1}_{s}(\Finally U_{T_2}) - \Pr^{\quotT,\sched_2}_{s}(\Finally U_{T_1}) = 0
        ~,
    \]
    with $U_{T_1}, U_{T_2}$ as defined in \cref{ex:buechi_success-sets}.
    This property can be satisfied by letting both $\sched_1$ and $\sched_2$ choose $\alpha$ and $\beta$ with equal probability in $s$, and choose $\epsilon^{\{T_1\}}$ in both $s_{C_1}$ and $s_{C_2}$.
\end{example}

\begin{proof}[Proof of \NoCaseChange\cref{th:relBuechi-to-relReach}]
    Recall from \cref{thm:whyComb} that we can split a relational reachability query into a separate query for each state-scheduler combination. 
    We can do the same for relational B\"uchi properties with analogous reasoning.
    Let $\comb = \{c_1, \ldots, c_l\}$ be the set of state-scheduler combinations, then we have
    \begin{align*}
        &\exists \sched_1, \ldots, \sched_\numsched \in \Scheds .\
        \sum_{i=1}^{\numsum} 
        q_i \Pr^{\sched_{k_i}}_{\state_i}\left(\Globally \Finally T_i\right) \comp q
        \\ \Iff\quad &
        \exists \sched_{c_1}, \ldots, \sched_{c_l} \in \Scheds .\
        \sum_{i=1}^{l} 
        \sum_{j \in \relInd(c_i)}
        q_j \Pr^{\sched_{c_i}}_{\state_j}\left(\Globally \Finally T_j\right) \comp q
        ~.
    \end{align*}

    In order to show the claim, it suffices to show that for any $c \in \comb$, we can transfer any scheduler $\sched_c \in \Scheds$ to a scheduler $\quotT[\sched_c] \in \Scheds[\quotT]$ such that the weighted sum of B\"uchi objectives under $\sched_c$ matches the weighted sum of corresponding reachability objectives under $\quotT[\sched_c]$, and vice versa.
    More precisely, we show that for $c \in \comb$ and $q_c \in \Q$, it holds that
    \begin{align*}
        &\exists \sched \in \Scheds .\
        \sum_{j \in \relInd(c)}
        q_j \Pr^{\mdp,\sched}_{\state_j}\left(\Globally \Finally T_j \right) = q_c
        \\ \Iff\quad & 
        \exists \quotT[\sched] \in \Scheds[\quotT] .\ 
        \sum_{j \in \relInd(c)}
        q_j \Pr^{\quotT, \quotT[\sched]}_{\mapStates(\state_j)}\left(\Finally U_{T_j} \right) = q_c
    \end{align*}

    We show both directions of the claim separately in the following lemmas, \cref{th:relBuechi-to-relReach_Rightarrow} and \cref{th:relBuechi-to-relReach_Leftarrow}, using ideas from \cite{baierCertificatesWitnesses2024,baierFoundationsProbabilityraising2024}. 
\end{proof}

\begin{restatable}[``$\Rightarrow$'' of \cref{th:relBuechi-to-relReach}]{lemma}{relBuechiToTelReachRightarrow}
    \label{th:relBuechi-to-relReach_Rightarrow}
    Let $c \in \comb$ and $q_c \in \Q$, then
    \begin{align*}
        &\exists \sched \in \Scheds .\
        \sum_{j \in \relInd(c)}
        q_j \Pr^{\mdp,\sched}_{\state_j}\left(\Globally \Finally T_j \right) = q_c
        \\ \Implies\quad & 
        \exists \quotT[\sched] \in \Scheds[\quotT] .\ 
        \sum_{j \in \relInd(c)}
        q_j \Pr^{\quotT, \quotT[\sched]}_{\mapStates(\state_j)}\left(\Finally U_{T_j} \right) = q_c
        ~.
    \end{align*}
\end{restatable}

\begin{proof}[Proof (Sketch)]
    We follow \cite[Lem.\ 2.4]{baierFoundationsProbabilityraising2024}.
    Given $\sched \in \Scheds$, we construct $\quotT[\sched] \in \Scheds[\quotT]$ as follows:
    On $\states \setminus \states_{\MEC}$, $\quotT[\sched]$ follows $\sched$.
    On states $\state_\mec$ for $\mec=(S',A) \in \MEC(\mdp)$, $\quotT[\sched]$ mimics whether $\sched$ leaves $\mec$ via some $\action \not\in A$, or stays in $\mec$ forever.
    More specifically, $\quotT[\sched]$ takes $\action\not\in A$ with the probability with which $\sched$ leaves $\mec$ via $\action$, and takes $\epsilon^{\mathcal{T}}$ (and thus transitions to sink $\bot^{\mathcal{T}}$) with the probability with which $\sched$ stays in $\mec$ and visits exactly the target sets from $\mathcal{T}$ infinitely often for $\mathcal{T} \subseteq \indexc[\mathcal{T}]$.
    %
    The full proof can be found in \cref{app:relBuechi-to-relReach_Rightarrow_proof}.
\end{proof}

For the other direction, we first show the claim for the case that the given scheduler for the quotient $\quotT$ is memoryless, see~\cref{app:relBuechi-to-relReach_Leftarrow_proof} for the proof.
\begin{restatable}{lemma}{relBuechiToTelReachLeftarrowMR}
    \label{th:relBuechi-to-relReach_Leftarrow_MR}
    Let $c \in \comb$ and $q_c \in \Q$, then
        \begin{align*}
            &\exists \quotT[\sched] \in \SchedsMR[\quotT] .\ 
            \sum_{j \in \relInd(c)}
            q_j \Pr^{\quotT, \quotT[\sched]}_{\mapStates(\state_j)}\left(\Finally U_{T_j} \right) = q_c
            \\ \Implies\quad &
            \exists \sched \in \Scheds .\
            \sum_{j \in \relInd(c)}
            q_j \Pr^{\mdp,\sched}_{\state_j}\left(\Globally \Finally T_j \right) = q_c
            ~.
        \end{align*}
\end{restatable}

We can then prove that \cref{th:relBuechi-to-relReach_Leftarrow_MR} suffices to show ``$\Leftarrow$'' of \cref{th:relBuechi-to-relReach}.

\begin{restatable}[``$\Leftarrow$'' of \cref{th:relBuechi-to-relReach}]{lemma}{relBuechiToTelReachLeftarrow}
    \label{th:relBuechi-to-relReach_Leftarrow}
    Let $c \in \comb$ and $q_c \in \Q$, then
    \begin{align*}
        &\exists \quotT[\sched] \in \Scheds[\quotT] .\ 
        \sum_{j \in \relInd(c)}
        q_j \Pr^{\quotT, \quotT[\sched]}_{\mapStates(\state_j)}\left(\Finally U_{T_j} \right) = q_c
        \tag{1}
        \label{eq:1}
        \\ \Implies\quad & 
        \exists \sched \in \Scheds .\
        \sum_{j \in \relInd(c)}
        q_j \Pr^{\mdp,\sched}_{\state_j}\left(\Globally \Finally T_j \right) = q_c
        ~.
        \tag{2}
        \label{eq:2}
    \end{align*}
\end{restatable}

\begin{proof}[Proof of \NoCaseChange\cref{th:relBuechi-to-relReach_Leftarrow} (Sketch)]
    For $c \in \comb$, let $\indexc[\state]$ be the unique state in $\mdp$ with $\state_j = \indexc[\state]$ for all $j \in \relInd(c)$.
    We define the \emph{flip-extension} $\flip[\mdp][]$ of $\mdp$ as the product of $\mdp$ with the memory structure of schedulers that are the convex combination of two schedulers for $\mdp$, with a fresh initial state $\sflip$. We analogously define $\flip[(\quotT)][]$ with a fresh initial state $\mapStates(\sflip) = \sflip$.
    Then, every memoryless scheduler for $\flip[\mdp][]$ can directly be interpreted as the convex combination of two schedulers for $\mdp$, and vice versa; analogously for $\flip[(\quotT)][]$.
    Further, $\flip[(\quotT)][ ] = \quotT[{\flip[\mdp][]}]$ 
    and hence the following claim
        \begin{align*}
            & 
            \exists \quotT[\sched]' \in \SchedsMR[{\flip[(\quotT)][]}] .\ 
            \sum_{j \in \relInd(c)}
            q_j \Pr^{\flip[(\quotT)][], \quotT[\sched]'}_{\sflip}\left(\Finally U_{T_j} \right) = q_c
            \tag{1'}
            \label{eq:1p}
            \\ \Implies\quad &
        \exists \sched' \in \Scheds[{\flip[\mdp][]}] .\ 
        \sum_{j \in \relInd(c)}
        q_j \Pr^{\flip[\mdp][], \sched'}_{\sflip}\left(\Globally \Finally  T_j \right) = q_c
        \tag{2'}
        \label{eq:2p}
    \end{align*}
    is an instance of \cref{th:relBuechi-to-relReach_Leftarrow_MR} with $\mdp=\flip[\mdp][]$ and $\quotT = \quotT[{\flip[\mdp][]}]$.

    The key observation is now that the sets $U_{T_j}$ are absorbing and hence any value achievable for the weighted sum $\sum_{j \in \relInd(c)} q_j \Pr^{\quotT, \quotT[\sched]}_{\mapStates(\indexc[\state])}\left(\Finally U_{T_j} \right)$ can be achieved by the convex combination of two MD schedulers~(see reasoning for~\cref{thm:cruxApproxEqual}).
    Thus, any value achievable for the weighted sum of reachability probabilities on the quotient can be achieved by a \emph{memoryless} randomized scheduler on the flip-extension of the quotient, and vice versa, i.e., `$\eqref{eq:1} \Iff \eqref{eq:1p}$'. 
    We further have `$\eqref{eq:2} \Iff \eqref{eq:2p}$'.
    Putting everything together yields the desired claim `$\eqref{eq:1} \Implies \eqref{eq:2}$'. 
    The full proof can be found in \cref{app:relBuechi-to-relReach_Leftarrow_proof}.
\end{proof}

We stress that the claim from \cref{th:relBuechi-to-relReach_Leftarrow_MR} is an implication, not an equivalence. 
We further observe that the construction from \cref{th:relBuechi-to-relReach_Leftarrow_MR} introduces memory but not randomization.
Given some \emph{memoryless deterministic} scheduler witnessing the constructed relational reachability property on $\quotT$, however, we can construct a witness for the original relational B\"uchi property that is again memoryless deterministic. 
This observation will be used later for a statement about the complexity of \RelBuechi in~\cref{th:buechi_NP-complete}.

\begin{corollary}
    \label{th:relBuechi-to-relReach_Leftarrow_MD}
    Let $c \in \comb$ and $q_c \in \Q$, then
        \begin{align*}
            & 
            \exists \quotT[\sched] \in \SchedsMD[\quotT] .\ 
            \sum_{j \in \relInd(c)}
            q_j \Pr^{\quotT, \quotT[\sched]}_{\mapStates(\state_j)}\left(\Finally U_{T_j} \right) = q_c
            \\ \Implies\quad &
            \exists \sched \in \SchedsMD .\
            \sum_{j \in \relInd(c)}
            q_j \Pr^{\mdp,\sched}_{\state_j}\left(\Globally \Finally T_j \right) = q_c
        \end{align*}
\end{corollary}

\begin{proof}[Proof (Sketch)]
    Let $\quotT[\sched] \in \SchedsMD[\quotT]$, then for each MEC $\quotT[\sched]$ must deterministically decide to either stay in the MEC forever or leave the MEC via a selected action.
    We construct $\sched \in \SchedsMD$ as follows:
    Outside MECs, $\sched$ behaves like $\quotT[\sched]$.
    Upon entering an MEC $\mec=(S',A)$, if $\quotT[\sched]$ chooses to go to a sink $\bot^{\mathcal{T}}$, then $\sched$ switches to a memoryless deterministic scheduler going to a sub-EC $\ec'$ of $\mec$ with $\ec' \in \EC_{\mathcal{T}}$, i.e., we ensure to visit exactly the target sets from $\mathcal{T}$ infinitely often (such a sub-EC and therefore such a scheduler must exist by construction).   
    Otherwise, if $\quotT[\sched]$ chooses some outgoing action $\action \in \Act \setminus A$ with $\action \in \Act(s')$ for some $s' \in S'$, then $\sched$ switches to a memoryless deterministic scheduler deterministically going to $s'$ and taking $\action$.
\end{proof}

\subparagraph*{Overall algorithm.}
The decision procedure resulting from Steps 1--3 is stated in \cref{alg:buechi}. 

\begin{algorithm}[t]
    \caption{Solving \RelBuechi by reduction to \RelReach} 
    \label{alg:buechi}
    \Input{%
        MDP $\mdp = \mdptup$ and a \RelBuechi property 
        \\ \medskip 
        \phantom{Input: } 
        $\displaystyle 
            \exists \sched_1, \ldots, \sched_\numsched \in \Scheds .\
        \sum_{i=1}^{\numsum} 
        q_i \Pr^{\sched_{k_i}}_{\state_i}\left(\Globally \Finally T_i\right) \comp q
        $
        \DontPrintSemicolon\tcp*{See \cref{prob:relbuechi}}
        \medskip
        }
    \Output{Whether the property is true in $\mdp$}
    \tcp{Step 1: Collect state-scheduler combinations: }
    {$\comb \gets \{ (\state_i, \sched_{k_i}) \mid i=1,\ldots,m\}$ \;}
    \tcp{Step 2: Build the MEC quotient and define success sets:}
    $\quotT \gets$ MEC quotient of $\mdp$ w.r.t.\ $(\mathcal{T}_c)_{c \in \comb}$
    \tcp*{See \Cref{def:MEC-quotient_set}}
    \For{$i = 1, \ldots, \numsum$}{
        \tcp{Collect the sinks in $\quotT$ that represent visiting a set of target state sets $\mathcal{T}$ infinitely often that contains $T_i$}
        $U_{T_i} \gets \{ \bot^{\mathcal{T}} \mid \exists c \in \comb .\ T_i \in \mathcal{T} \subseteq \mathcal{T}_c \wedge \MEC_{\mathcal{T}} \neq \emptyset \}$
        \tcp*{See \Cref{def:success-set}}
    }
    \tcp{Step 3: Solve \RelReach query:}
    \Return{\textsf{Solve(}$\exists \quotT[\sched_1], \ldots, \quotT[\sched_\numsched]  \in \Scheds[\quotT] .\ 
        \sum_{i=1}^{\numsum} 
        q_i \Pr^{\quotT,\quotT[\sched_{k_i}]}_{\mapStates(\state_i)}(\Finally U_{T_i}) \comp q$\textsf{)}} 
        \tcp*{See~\cref{sec:reach}}
\end{algorithm}

\begin{theorem}[Correctness and time complexity]
    \label{th:buechi_correctness}
    \cref{alg:buechi} adheres to its input-output specification and can be implemented with a worst-case running time of $\mathcal{O}(\numsum\cdot \textit{poly}(2^\numsum\cdot |\mdp|))$, where $\numsum$ is the number of probability operators in the property.
\end{theorem}

\begin{proof}
    Since all sets $U_{T_j}$ are absorbing, we can solve the \RelReach query from our reduction in \cref{th:relBuechi-to-relReach} in time $\mathcal{O}(m \cdot \textit{poly}(|\quotT|))$ by the reasoning for \cref{thm:correctness}.
    The size of $\quotT$ is bounded by $|\mdp| + 2^m$.
\end{proof}

%% file: buechi_complexity.tex
\subsection{Complexity of Relational B\"uchi}
\label{sec:buechi_complexity}

Let us now analyze the computational complexity of the \RelBuechi problem, both over general and memoryless deterministic schedulers, and additionally identify fragments of the problem that are decidable in polynomial time or where memoryless deterministic schedulers suffice.

\subsubsection{General Schedulers}

\begin{restatable}{theorem}{buechiNPComplete}
    \label{th:buechi_NP-complete}
    Problem \RelBuechi is strongly \NP-complete.
\end{restatable}

\begin{proof}[Proof (Sketch)]    
    For membership in \NP, we show that it suffices to guess a polynomial number of MD schedulers in order to guess a witness for the query.
    Given some MDP $\mdp$ and \RelBuechi property $\phi$, let $\phi'$ be the constructed \RelReach property.
    The key observation is that all target sets $U_{T_i}$ of $\phi'$ on $\quotT$ are absorbing and therefore for each $c \in \comb$, the goal unfolding of $\quotT$ w.r.t.\ the target sets for $c$, $\indexc[\mathcal{T}]' = \{U_{T_i} \mid i \in \relInd(c)\}$, corresponds to $\quotT$ with an additional copy of each sink in some $U_{T_i}$ for some $i \in \relInd(c)$.
    An MD scheduler for $\indexc[(\quotT)]$ can thus be translated to an MD scheduler for $\quotT$ while preserving the reachability probabilities.
    By \cref{th:relBuechi-to-relReach_Leftarrow_MD}, an MD witness for $\phi'$ on $\quotT$ can be translated to an MD witness for $\phi$ on $\mdp$.
    Details can be found in \cref{app:buechi_NP-complete}.

    Strong \NP-hardness follows by reduction from SAT~\cite{karpReducibilityCombinatorial1972}, see \cref{app:buechi_NP-complete} for details. 
\end{proof}

Note that the reasoning for the membership in \NP from the proof of~\cref{th:buechi_NP-complete} does not imply that MD schedulers suffice for \RelBuechi.

Under some mild restrictions on the target sets, the procedure detailed in \cref{sec:buechi_algo} runs in polynomial time.
The following theorem generalizes \cref{th:general_PTIME}, see~\cref{app:buechi_general_PTIME} for the proof.

\begin{restatable}{theorem}{buechiGeneralPTIME}
    \label{th:buechi_general_PTIME}
    The following special cases of \RelBuechi are in \PTIME:
    \begin{enumerate}[label=(\alph*)]
        \item \label{item:buechi_fixed-param} 
        The number of \emph{different} target sets $|\{ T_1,\ldots,T_m\}|$ is \emph{at most a constant}.        
        \item \label{item:buechi_absorb}
        The target sets $T_1,\ldots,T_m$ are all \emph{absorbing}.

        \item \label{item:buechi_single-target}
        Probability operators with the same scheduler variable and the same initial state have the same target sets 
        (formally, $\forall i,i' .\ (\sched_{k_i} {=} \sched_{k_{i'}} \wedge \state_i = \state_{i'}) \implies T_i {=} T_{i'}$).
    \end{enumerate}
\end{restatable}

We note that \ref{item:buechi_single-target} includes the case that $n=m$ (case~\ref{item:independent} from \cref{th:general_PTIME}), where every probability operator is associated with a different scheduler variable.

\begin{remark}
    \label{remark:buechi_sim-a-s}
    A \RelBuechi query of the form $\exists \sched .\ \sum_{i=1}^{\numsum} \Pr_{\statei}^{\sched}(\Globally \Finally T_i) \geq \numsum$ 
    is a \emph{multi-dimensional percentile query} with a $\limsup$ payoff function in the sense of~\cite{randourPercentileQueries2017}.
    More precisely, it is equivalent to the query $\exists \sched .\ \bigwedge_{i=1}^{\numsum} \Pr_{\statei}^{\sched}(\limsup w_i \geq 1) \geq 1$ where $w_i$ assigns weight 1 to all outgoing actions of $T_i$.
    Such queries can be solved in \PTIME by \cite[Th.\ 8]{randourPercentileQueries2017}, this case is not covered by~\cref{th:buechi_general_PTIME}.
    
    However, 
    \cref{alg:buechi} does not necessarily run in polynomial time for queries of the form $\exists \sched .\ \sum_{i=1}^{\numsum} \Pr_{\statei}^{\sched}(\Globally \Finally T_i) \geq \numsum$, since the number of sink states and thus the size of the MEC quotient may still be exponential in the number of different target sets.
    
    Observe that the analogous special case of \RelReach is \PSPACE-hard, as the formula from the proof of \PSPACE-hardness for \RelReach is of this form (cf.\ \cref{th:general_EXPTIME}).
    \RelReach being `harder' than \RelBuechi can intuitively be explained as follows: In order to reach $\numsum$ targets we need to remember which targets we have already seen, but if we want to visit them all infinitely often this is not necessary.
\end{remark}

\subsubsection{Memoryless Deterministic Schedulers}

Let \RelBuechiMD denote the \RelBuechi problem over MD schedulers.
We can show strong \NP-completeness analogously to \cref{th:MD_NP-complete}, see~\cref{app:buechi_MP_NP_complete} for details. 

\begin{restatable}{theorem}{buechiMDNPComplete}
    \label{th:buechi_MD_NP-complete}
    \label{th:buechi_MD_NP}
    \RelBuechiMD is strongly \NP-complete.
    Strong \NP-hardness already holds for the following special cases:
    For a given MDP $\mdp$, 
    initial states $\state_1,\state_2 \in \states$,
    target sets $T_1, T_2 \subseteq \states$,
    decide if
    \begin{enumerate}[label=\textnormal{(\alph*)},itemsep=2pt]
        \item 
        \label{item:buechi_1sched1state}
        \label{item:buechi_1sched2state}
        $
            \exists \sched \in \SchedsMD .\ 
            \Pr^{\sched}_{\state_1}(\Globally \Finally T_1) - \Pr^{\sched}_{\state_2}(\Globally \Finally T_2) = 0 \ .
        $
        
        \item 
        \label{item:buechi_2sched2state}
        $
            \exists \sched_1, \sched_2 \in \SchedsMD .\ 
            \Pr^{\sched_1}_{\state_1}(\Globally \Finally T_1) - \Pr^{\sched_2}_{\state_2}(\Globally \Finally T_2) = 0 \ .
        $
    \end{enumerate}
    Strong \NP-hardness of \ref{item:buechi_1sched1state} and \ref{item:buechi_2sched2state} holds irrespective of whether 
    $\state_1 = \state_2$
    and whether $T_1$ and/or $T_2$ are absorbing.
    Moreover, \ref{item:buechi_1sched1state} and \ref{item:buechi_2sched2state} with relation $\approx_\epsilon$, $\epsilon > 0$, are \NP-hard.
\end{restatable}

\begin{restatable}{theorem}{buechiGeneralNecessary}
    \label{th:buechi_general_necessary}
    Memory and randomization are necessary for \RelBuechi with (approximate or exact) equality. 
\end{restatable}

\begin{proof}[Proof (Sketch)]
    Consider again the MDP from~\cref{fig:memory-necessary}, depicted again in \cref{fig:buechi_memory-randomization-necessary} and the \RelBuechi property $\exists \sched .\ \Pr^{\sched}_{\state}(\protect\Globally \protect\Finally \{t_1\}) \approx_{\epsilon} \Pr^{\sched}_{\state}(\protect\Globally \protect\Finally \{t_2\})$ for some $\epsilon\geq 0$. This property holds over memoryful randomized schedulers, but does not hold if we restrict to either memoryless or deterministic schedulers.
    We refer to \cref{app:buechi_general_necessary} for details.
\end{proof}

\begin{figure}
    \centering
    \begin{tikzpicture}
        \node[state] (s) {$s$};
        
        \node[state, below left= of s] (t1) {$t_1$}; 
        \node[state, below right= of s] (t2) {$t_2$}; 

       \path[-latex', draw]
            (s) edge[bend left] node[right] {$\alpha$} (t1)
            (s) edge node[right] {$\beta$} (t2);

        \path[-latex', draw]
            (t1) edge[bend left] node[right] {$\gamma$} (s)
            (t2) edge[loop right] (t2);
    \end{tikzpicture}
    \caption{MDP where memory and randomization are necessary for the relational B\"uchi property $\exists \sched .\ \Pr^{\sched}_{\state}(\protect\Globally \protect\Finally \{t_1\}) \approx_{\epsilon} \Pr^{\sched}_{\state}(\protect\Globally \protect\Finally \{t_2\})$ for some $\epsilon\geq 0$.}
    \label{fig:buechi_memory-randomization-necessary}
\end{figure}

Combining \cref{th:MD_PTIME}\ref{item:MD_n-comb} and \cref{th:relBuechi-to-relReach_Leftarrow_MD} yields restricted versions of \RelBuechi for which the procedure detailed in \cref{sec:buechi_algo} induces MD schedulers.
Let \RelBuechiNotEq denote the fragment of \RelBuechi restricted to comparison operators $\{\geq,>,\not\approx_{\epsilon} \mid \epsilon\geq 0\}$.

\begin{theorem}
\label{th:buechi_MD}
    For the following special cases of \RelBuechiNotEq, the procedure detailed in \cref{sec:buechi_algo} induces MD schedulers, if it returns true:
    \begin{enumerate}[label=(\alph*)]
        \item \label{item:buechi_MD_n-comb}
        Probability operators with the same scheduler variable have the same initial state
        (formally, $\forall i,i' .\ \sched_{k_i} = \sched_{k_{i'}} \implies \state_i = \state_{i'}$).

        \item \label{item:buechi_MD_sign}
        Probability operators with the same scheduler variable have equally signed coefficients and the same target sets
        (formally, $\forall i,i' .\ \sched_{k_i} {=} \sched_{k_{i'}} \implies ((q_i \geq 0 \Iff q_{i'} \geq 0) \wedge T_i {=} T_{i'})$).
    \end{enumerate}
    In case \ref{item:buechi_MD_sign}, the aforementioned procedure runs in polynomial time.
\end{theorem}

\begin{proof}
    \ref{item:buechi_MD_n-comb}: 
    If probability operators with the same scheduler variable have the same initial state, then in the constructed relational reachability property also probability operators with the same scheduler variable have the same initial state and additionally all target sets are absorbing by construction.
    By \cref{th:MD_PTIME}\ref{item:MD_n-comb}, the procedure for solving \RelReach thus induces MD witness schedulers on the MEC quotient, if they exist.
    By \cref{th:relBuechi-to-relReach_Leftarrow_MD}, we can transform any MD scheduler $\quotT[\sched]$ for $\quotT$ back to an MD scheduler $\sched$ for $\mdp$.
        
    \ref{item:buechi_MD_sign}:
    If probability operators with the same scheduler variable have equally signed coefficients and the same target sets, then each set $\indexc[\mathcal{T}]$ must be a singleton.
    Hence also in the constructed relational reachability property, we must have $U_{T_j} = U_{T_{j'}}$ for all $j,j' \in \relInd(c)$ and hence probability operators with the same scheduler variable have the same target sets.
    By \cref{th:MD_PTIME}\ref{item:MD_sign}, the procedure for solving \RelReach thus induces MD witness schedulers on the MEC quotient, if they exist, and runs in polynomial time by \cref{th:buechi_general_PTIME}\ref{item:buechi_single-target}.
    MD schedulers on the MEC quotient can be translated back to MD schedulers on the original MDP by \cref{th:relBuechi-to-relReach_Leftarrow_MD}.
\end{proof}

Observe that \ref{item:buechi_MD_n-comb} includes the case that $n=m$, and thus subsumes both \cref{th:MD_PTIME}\ref{item:MD_n-comb} and \ref{item:MD_sign}.
However, the procedure does not run in polynomial time in case \ref{item:buechi_MD_n-comb} since the size of the MEC quotient may still be exponential in $m$. 

\begin{corollary}
    \label{cor:buechi_MD_suffice}
    For the \RelBuechi variants from \cref{th:buechi_MD}, MD schedulers suffice.
\end{corollary}

\begin{corollary}
    \label{cor:buechi_MD_absorb}
    If all target sets are absorbing, \RelBuechiMD is in \PTIME in the cases listed in~\cref{th:MD_PTIME}.
\end{corollary}

\begin{proof}
    A relational B\"uchi property with absorbing target sets can be reformulated as a relational reachability property.
\end{proof}

\paragraph*{Summary of \cref{sec:buechi}}
The key insight is that \RelBuechi can be solved by reduction to \RelReach on (a slight variation of) the MEC quotient.
This observation yields an algorithm that runs in time exponential only in the number of target sets, but in polynomial time for a fixed number of target sets, absorbing targets, or if each state-scheduler combination only has a single relevant target set.
\RelBuechi is strongly \NP-complete over both general and MD schedulers; we identified fragments where our algorithm induces MD witnesses in polynomial time.

%% file: conjunction.tex
\section{Multi-Objective Relational Reachability}
\label{sec:conjunction}

Multi-objective relational reachability properties generalize traditional multi-objective reachability properties~\cite{chatterjeeMarkovDecision2006,chatterjeeMarkovDecision2007,etessamiMultiObjectiveModel2008,quatmannVerificationMultiobjective2023}.
The latter ask for a scheduler satisfying multiple predicates comparing a reachability probability to a bound, such as ``\emph{Does there exist a scheduler such that the probability of reaching $T_1$, $T_2$, $T_3$ from $\state$ exceeds $q_1$, $q_2$, $q_3$, respectively?}''
We lift this setting to \emph{relational reachability predicates} that compare weighted sums of probabilities---from potentially different states---to a bound, i.e., $\sum_{i=1}^{\numsum} q_{i,j} \Pr^{\sched_{k_{i,j}}}_{\state_{i,j}}(\Finally T_{i,j}) \comp_j q_j$.
In the following, `predicate' always refers to a relational reachability predicate unless specified otherwise.
We focus on the multi-objective part here and simplify matters by not considering B\"uchi objectives.
We conjecture that the approach can be lifted to multi-objective relational properties including B\"uchi constraints by combining it with techniques from~\cref{sec:buechi}.
Recall that the motivating example from \cref{ex:conj_motivating-ex} is (the negation of) a multi-objective relational reachability property.

\begin{framedproblem}[\ConjRelReach]
    \label{prob:morelreach}
    Given an MDP $\mdp = \mdptup$, decide whether 
    \[
        \exists \sched_1, \ldots, \sched_n \in \Scheds .\
        \bigwedge_{j=1}^{\numconj} \sum_{i=1}^{\numsum} q_{i,j} 
        \Pr^{\sched_{k_{i,j}}}_{\state_{i,j}}(\Finally T_{i,j}) \comp_j q_j
        ~,
        \qquad\text{where}
    \]
    \begin{itemize}
        \item $\numconj,\numsum,\numsched$ are natural numbers,
        \item $q_{1,1}, \ldots, q_{\numsum,\numconj}$ are rational coefficients,
        \item $q_1, \ldots, q_\numconj$ are rational bounds,
        \item $\state_{1,1}, \ldots, \state_{\numsum,\numconj}$  are (not necessarily distinct) initial states,
        \item $\{k_{1,1}, k_{1,2} \ldots, k_{\numsum-1,\numconj}, k_{\numsum,\numconj}\} = \{1, \ldots, \numsched\}$ is a set of indices,
        \item $T_{1,1}, \ldots, T_{\numsum,\numconj}$ are (not necessarily distinct) target sets, and 
        \item $\comp_1, \ldots, \comp_\numconj \in \{>, \geq, \approx_{\epsilon}, \not\approx_{\epsilon} \mid \epsilon \in \Q_{\geq 0} \}$ are comparison operators.
    \end{itemize}    
\end{framedproblem}

Note that, by allowing the coefficients $q_{i,j}$ to be zero, we implicitly allow a different number of summands for each predicate.

\begin{remark}
    As for traditional multi-objective properties, simply maximizing the `sum of the predicates' $\Pr^\sched_s(\Finally T_1) - \Pr^\sched_s(\Finally T_2) + \Pr^\sched_s(\Finally T_3) - \Pr^\sched_s(\Finally T_4)$ does not reveal whether there exists a scheduler satisfying both $\Pr^\sched_s(\Finally T_1) - \Pr^\sched_s(\Finally T_2) \geq 0$ and $\Pr^\sched_s(\Finally T_3) - \Pr^\sched_s(\Finally T_4) \geq 0$.
\end{remark}

Intuitively, a \ConjRelReach property with the comparison operator $\not\approx_{\epsilon}$ contains a disjunction inside a conjunction, since $x \not\approx_\epsilon y$ is equivalent to $x < y - \epsilon \lor x > y + \epsilon$.
This nested disjunction makes \ConjRelReach with $\not\approx_{\epsilon}$ computationally more complex than \ConjRelReach without $\not\approx_{\epsilon}$.
In fact, we will see that \ConjRelReach can be solved in \PTIME for a fixed number of target sets \emph{only} if we exclude $\not\approx_\epsilon$ for $\epsilon>0$. 
However, the special case of $\neq$ (i.e., $\not\approx_0$) can still be handled efficiently.

\begin{remark}[Extended MOA Queries]
    \label{rem:extended-MOA}
    In this section, we work with \emph{multi-objective achievability} (\emph{MOA}) queries in the sense of~\cite{quatmannMultiobjectiveModel2016,forejtQuantitativeMultiobjective2011}, i.e., queries of the form \[
        \exists \sched \in \Scheds .\ 
        \bigwedge_{j=1}^{\numconj} \Expected_{\state}^{\mdp, \sched}(\rew^j) \comp_j q_j 
    \]
    for 
    an MDP $\mdp$ with state $\state$ and reward structures $\rew^1, \ldots, \rew^\numconj$, 
    comparison operators $\comp_1, \ldots, \comp_\numconj \in \{>, \geq, \approx_{\epsilon}, \not\approx_{\epsilon} \mid \epsilon \in \Q_{\geq 0} \}$,
    and rational bounds $q_1, \ldots, q_\numconj$.
    See~\cref{sec:related_moa} for details. 
    Note that \cite{quatmannMultiobjectiveModel2016,forejtQuantitativeMultiobjective2011} do not allow $\not\approx_\epsilon$
    and that other works such as~\cite{quatmannVerificationMultiobjective2023,etessamiMultiObjectiveModel2008} use a more restricted notion of achievability where only $\geq$ is allowed as a comparison operator. 
    In the following, we explicitly call MOA queries with arbitrary comparison operators ``\emph{extended MOA queries}'' to distinguish in particular from the restricted version with only $\geq$.
\end{remark}

\subparagraph*{Outline of this section.}
We first show that multi-objective relational queries can be reduced to extended MOA queries in \cref{sec:conj_algo}, and outline how to solve the constructed MOA queries in \cref{sec:moa}.
In \cref{sec:conj_complexity}, we prove that \ConjRelReach is \PSPACE-hard and decidable in \EXPTIME.
If all probability operators share the same, absorbing target set, it is \NP-complete, and if we additionally exclude $\not\approx_\epsilon$ for $\epsilon>0$, it is in \PTIME.

%% file: conjunction_algo.tex
\subsection{Verifying Multi-Objective Relational Reachability Properties}
\label{sec:conj_algo}

Throughout this section, we fix an MDP $\mdp = \mdptup$ and a \ConjRelReach property
\begin{align}
    \exists \sched_1 \ldots \sched_n \in \Scheds .\
    \bigwedge_{j=1}^{\numconj} \sum_{i=1}^{\numsum} q_{i,j} 
    \Pr^{\sched_{k_{i,j}}}_{\state_{i,j}}(\Finally T_{i,j}) \comp_j q_j
    \label{eq:genConjRelReachProp}
\end{align}
as in~\cref{prob:morelreach}.
In the following, we assume $q_{i,j} \neq 0$ for all $(i,j) \in \{1,\ldots,\numsum\} \times \{1, \ldots \numconj\}$ for simplicity.
In this section, we outline a procedure for verifying whether \eqref{eq:genConjRelReachProp} holds, as summarized in \cref{alg:conj}.
In \cref{sec:conj_to_MOMC}, we first show how to reduce \eqref{eq:genConjRelReachProp} to a multi-objective achievability query; and then address how these can be solved in \cref{sec:moa}.

\begin{example}
    \label{ex:conj_running}
    Throughout this section, we use the MDP depicted in~\cref{fig:conj_running} together with the \ConjRelReach property
    \[
        \exists \sched .\
        \Pr^{\sched}_{s_1}(\Finally T) - \Pr^{\sched}_{s_2}(\Finally T) < 0
        \land 
        \Pr^{\sched}_{s_1}(\Finally T) - \Pr^{\sched}_{s_1}(\Finally T') > 0
    \]
    as a running example.
\end{example}

\colorlet{rew1}{magenta}
\colorlet{rew2}{blue}
\begin{figure}
    \centering
    \scalebox{.7}{
    \begin{tikzpicture}[
            node distance=25mm and 15mm,
            on grid,semithick,>=stealth,
            ]
        \node[state] (s1) {$s_1$};
        \node[dist, right=10mm of s1] (d1) {};
        \node[state, right=10mm of d1] (s2) {$s_2$};
        \node[dist, below=10mm of s2] (d2) {};
        \node[state, below left=25mm and 10mm of s1] (t) {$t$};
        \node[state, below=of d1] (sp) {$s'$};
        \node[state, below right=25mm and 10mm of s2] (tp) {$t'$};

        \path[-, draw]
        (s1) edge (d1)
        (s2) edge (d2)
        ;
            
        \path[-latex', draw]
        (d1) edge[bend left=15] node[left] {$\nicefrac{1}{3}$} (t)
        (d1) edge node[left] {$\nicefrac{1}{3}$} (sp)
        (d1) edge node[above] {$\nicefrac{1}{3}$} (s2)
        (d2) edge[bend right=15] node[above] {$\nicefrac{2}{3}$} (sp)
        (d2) edge[bend left=15] node[above] {$\nicefrac{1}{3}$} (tp)
        (sp) edge node[above] {$\alpha$} (t)
        (sp) edge[bend left=15] node[above] {$\beta$} (tp)
        (tp) edge[bend right=30] node[right] {$\gamma$} (s2)
        ;

        \path[-latex', draw]
        (t) edge[loop above] (t)
        (tp) edge[loop below] node[right] {$\delta$} (tp);

        \node[staterectangle, right=10cm of s1] (sX) {$s_{c_1,c_2}$};
        \node[dist, below=10mm of sX] (dX) {};

        \node[staterectangle, below=20mm of sX, xshift=-4cm] (s1c1) {$s_1, \emptyset, c_1$};
        \node[dist, right= of s1c1] (d1c1) {};
        \node[staterectangle, right= of d1c1] (s2c1) {$s_2, \emptyset, c_1$};
        \node[dist, below=7mm of s2c1] (d2c1) {};
        \node[staterectangle, below left=of s1c1, xshift=5mm, label={[yshift=-1mm]100:{\textcolor{rew1}{$\mathbf{2}$}, \textcolor{rew2}{$\mathbf{2}$}}}] (tc1) {$t, \emptyset, c_1$};
        \node[staterectangle, below=of d1c1] (spc1) {$s', \emptyset, c_1$};
        \node[staterectangle, below right=of s2c1, xshift=-5mm, label={[yshift=-1mm]110:{\textcolor{rew2}{$\mathbf{-2}$}}}] (tpc1) {$t', \emptyset, c_1$};

        \node[staterectangle, left=30mm of tc1] (tTc1) {$t, \{T\}, c_1$};

        \node[staterectangle, below=35mm of s2c1] (s2Tpc1) {$s_2, \{T'\}, c_1$};
        \node[dist, below=7mm of s2Tpc1] (d2Tpc1) {};
        
        \node[staterectangle, below=30mm of tTc1] (tTTpc1) {$t, \{T, T'\}, c_1$};
        \node[staterectangle, below=30mm of tc1, label={[yshift=-1mm]110:{\textcolor{rew1}{$\mathbf{2}$}, \textcolor{rew2}{$\mathbf{2}$}}}] (tTpc1) {$t, \{T'\}, c_1$};
        \node[staterectangle, below=30mm of spc1] (spTpc1) {$s', \{T'\}, c_1$};
        \node[staterectangle, below=30mm of tpc1] (tpTpc1) {$t', \{T'\}, c_1$};

        \node[below=20mm of sX, xshift=3.5cm] (s1c2) {};
        \node[right= of s1c2] (d1c2) {};
        \node[staterectangle, right=15mm of d1c2] (s2c2) {$s_2, \emptyset, c_2$};
        \node[dist, below=10mm of s2c2] (d2c2) {};
        \node[staterectangle, below left=of s1c2, xshift=5mm, label={[yshift=-1mm]110:{\textcolor{rew1}{$\mathbf{-2}$}}}] (tc2) {$t, \emptyset, c_2$};
        \node[staterectangle, below=of d1c2] (spc2) {$s', \emptyset, c_2$};
        \node[staterectangle, below right=of s2c2, xshift=-5mm] (tpc2) {$t', \emptyset, c_2$};

        \node[staterectangle, below=15mm of tc2] (tTc2) {$t, \{T\}, c_2$};

        \path[-, draw] (sX) edge node[left] {$\epsilon$} (dX);
        
        \path[-latex', draw] 
        (dX) edge[bend right=10] node[above] {$\nicefrac{1}{2}$} (s1c1.north)
        (dX) edge[bend left=10] node[above] {$\nicefrac{1}{2}$} (s2c2.north);
        
        \path[-, draw]
        (s1c1) edge (d1c1)
        (s2c1) edge (d2c1)
        (s2Tpc1) edge (d2Tpc1)
        ;
            
        \path[-latex', draw]
        (d1c1) edge[bend left=15] node[left] {$\nicefrac{1}{3}$} (tc1)
        (d1c1) edge node[left] {$\nicefrac{1}{3}$} (spc1)
        (d1c1) edge node[above] {$\nicefrac{1}{3}$} (s2c1)
        (d2c1) edge[bend right=15] node[above] {$\nicefrac{2}{3}$} (spc1)
        (d2c1) edge[bend left=15] node[above] {$\nicefrac{1}{3}$} (tpc1)
        (spc1) edge[ultra thick] node[above] {$\boldsymbol{\alpha}$} (tc1)
        (spc1) edge[bend left=15] node[above, pos=0.2] {$\beta$} (tpc1)
        ;

        \path[-latex', draw]
        (tc1) edge (tTc1);

        \path[-latex', draw] 
        (tpc1) edge[bend right=30] node[right, pos=0.5, xshift=1mm] {$\gamma$} (s2c1)
        (spTpc1) edge[ultra thick, bend right=15] node[above, pos=0] {$\boldsymbol\alpha$} (tTpc1.north)
        (spTpc1) edge[bend left=15] node[above, pos=0.3,yshift=-1mm] {$\beta$} (tpTpc1.120)
        (tpTpc1) edge[bend right=30] node[right, pos=0.1] {$\gamma$} (s2Tpc1)
        ;

        \path[-latex', draw]
        (tpc1) edge[ultra thick, bend right=30] node[left] {$\boldsymbol\gamma$} (s2Tpc1.north)
        (d2Tpc1) edge[bend right=15] node[above] {$\nicefrac{2}{3}$} (spTpc1)
        (d2Tpc1) edge[bend left=15] node[above] {$\nicefrac{1}{3}$} (tpTpc1)
        (tTpc1) edge (tTTpc1)
        ;

        \path[-latex', draw]
        (tTc1) edge[loop above] (tTc1)
        (tTTpc1) edge[loop above] (tTTpc1)
        (tpc1) edge[bend left=20] node[right] {$\delta$} (tpTpc1.30)
        (tpTpc1) edge[loop right, looseness=10, out=5, in=355, ultra thick] node[right] {$\boldsymbol\delta$} (tpTpc1);

        \path[-, draw]
        (s2c2) edge (d2c2)
        ;
            
        \path[-latex', draw]
        (d2c2) edge[bend right=15] node[above] {$\nicefrac{2}{3}$} (spc2)
        (d2c2) edge[bend left=15] node[above] {$\nicefrac{1}{3}$} (tpc2)
        (spc2) edge node[above, ultra thick] {$\boldsymbol\alpha$} (tc2)
        (spc2) edge[bend left=15] node[above] {$\beta$} (tpc2)
        (tpc2) edge[bend right=30, ultra thick] node[right] {$\boldsymbol\gamma$} (s2c2)
        ;
        \path[-latex', draw]
        (tc2) edge (tTc2);

        \path[-latex', draw]
        (tTc2) edge[loop right, looseness=10, out=5, in=355] (tTc2)
        (tpc2) edge[loop below] node[right] {$\delta$} (tpc2);

        \begin{scope}[on background layer]
        \node[state, right=40mm of s2c1, color=gray, inner sep=0pt] (sink) {$\bot_{\emptyset}$};
        \node[state, right=55mm of tpTpc1, yshift=3mm, color=gray, inner sep=0pt] (sinkTp) {$\bot_{T'}$};  
        \node[state, left=30mm of tTc1, yshift=-5mm, color=gray, inner sep=0pt] (sinkT) {$\bot_{T}$};  
        \node[state, below=15mm of sinkT, color=gray, inner sep=0pt] (sinkTTp) {$\bot_{T,T'}$};  

        \path[-latex', draw, color=gray, dashed]
        (tTc1) edge (sinkT)
        (tTTpc1) edge (sinkTTp)
        (tTc2) edge[bend left=15] (sinkT)
        (s2c1) edge (sink)
        (spc1) edge (sink)
        (tpc1) edge (sink)
        (s2c2) edge (sink)
        (spc2) edge (sink)
        (tpc2) edge (sink)
        (s2Tpc1) edge[bend right=15] (sinkTp)
        (spTpc1) edge[bend left=15] (sinkTp)
        (tpTpc1) edge[bend left=5] (sinkTp)
        (sink) edge[loop below] (sink)
        (sinkT) edge[loop above] (sinkT)
        (sinkTp) edge[loop right] (sinkTp)
        (sinkTTp) edge[loop below] (sinkTTp)
        ;
        \end{scope}
    \end{tikzpicture}
    }
    \label{fig:conj_running-ex_mdp}
    \caption{MDP with $T = \{t\}$ and $T' = \{t'\}$ (left) and the pre-processed combined unfolded MDP $\processed[\mdp_{\comb}]$ (right) for ~\cref{ex:conj_running}. 
    States and transitions added by the pre-processing~(\cref{def:conj_preprocessing}) are indicated in \textcolor{gray}{gray}, the MDP without these states and transitions is the combined unfolded MDP \emph{without} the pre-processing $\mdp_{\comb}$~(\cref{def:MDP_combined}).
    If a state only has a single available action, we may omit action labels.
    Non-zero rewards for \textcolor{rew1}{$\rew^1$}, \textcolor{rew2}{$\rew^2$} are indicated next to the states.
    Actions chosen by the witness scheduler for $\mdp_{\comb}$ from~\cref{ex:conj_rew} are indicated in \textbf{bold}.
    } 
    \label{fig:conj_running}
\end{figure}

\subsubsection{Reducing to Extended Multi-Objective Achievability}
\label{sec:conj_to_MOMC}

\subparagraph*{Step 0: Split query (optional).}
\label{step:conj_split}
We define the \emph{class} of a predicate index $j \in \{1,\ldots,\numconj\}$ as the set of all objectives $j' \in \{1,\ldots,\numconj\}$ sharing a scheduler with objective $j$, thus partitioning the set of objectives $\{1,\ldots,\numconj\}$.
Splitting the query is not strictly necessary for the correctness of the approach but improves its practical efficiency (though it does not improve its \emph{theoretical} worst-case complexity). 

\begin{definition}
    \label{def:conjpart}
    For $j\in \{1,\ldots,\numconj\}$, we define the \emph{class} of $j$ as 
    \[
        \class{j} = \{ j' \in \{1,\ldots,\numconj\} \mid \exists i,i' .\ \sched_{k_{i,j}} = \sched_{k_{i',j'}} \} 
        ~.
    \]
    We let $\conjpart = \{ \class{1}, \ldots, \class{\numconj} \}$.
\end{definition}

For each class of predicate indices we collect the set of all relevant scheduler indices.
\begin{definition}
    For $\conjpartelt \in \conjpart$, we let  
    $\textit{Scheds}(\conjpartelt) = \{ k_{i,j'} \mid i \in \{1,\ldots,\numsum\} \wedge j' \in \conjpartelt\} $. 
    For $n_{\conjpartelt} = |\textit{Scheds}(\conjpartelt)|$, we let $\{h^\conjpartelt_1, \ldots, h^\conjpartelt_{\numsched_\conjpartelt} \} = \textit{Scheds}(\conjpartelt)$.
\end{definition}

\begin{example}
    \label{ex:conj_split}
    Consider the following extension of the property from~\cref{ex:conj_running}
    \begin{align*}
        \exists \sched_1, \sched_2 .\
        &\Pr^{\sched_1}_{s_1}(\Finally T) - \Pr^{\sched_1}_{s_2}(\Finally T) < 0
        \;\;\land\;\;
        \Pr^{\sched_1}_{s_1}(\Finally T) - \Pr^{\sched_1}_{s_1}(\Finally T') > 0
        \;\;\land\;\; {} \\ &
        \Pr^{\sched_2}_{s_1}(\Finally T) + \Pr^{\sched_2}_{s_1}(\Finally T') > 2
        ~.
    \end{align*}
    We have $\conjpart = \{ \class{1} = \class{2}, \class{3}\}$, i.e., the first and second predicate share a scheduler but not with the third predicate.
    More precisely, the first two predicates use $\sched_1$, so $\textit{Scheds}(\class{1}) = \{ 1 \}$, 
    while the third predicate only refers to $\sched_2$, so $\textit{Scheds}(\class{3}) = \{ 2 \}$, 
\end{example}

We can now split the problem into a number of independent smaller problems, one for every partition element $\conjpartelt \in \conjpart$.

\begin{lemma}
\label{le:split-conjunction}
    Let $\conjpart$ as defined in~\cref{def:conjpart}, then
    \begin{align*}
        &\exists \sched_1 \ldots \sched_n \in \Scheds .\ 
        \bigwedge_{j=1}^{\numconj}
        \sum_{i=1}^{\numsum} q_{i,j} 
        \Pr^{\sched_{k_{i,j}}}_{\state_{i,j}}(\Finally T_{i,j}) \comp_j q_j \\
        \Iff\quad \bigwedge_{\conjpartelt \in \conjpart}
        &\exists \sched_{h^\conjpartelt_1} \ldots \sched_{h^\conjpartelt_{\numsched_\conjpartelt}} \in \Scheds .\ \bigwedge_{j \in \conjpartelt} \sum_{i=1}^{\numsum} q_{i,j} 
        \Pr^{\sched_{k_{i,j}}}_{\state_{i,j}}(\Finally T_{i,j}) \comp_j q_j 
        ~.
    \end{align*}
\end{lemma}

\begin{example}
    For the property from \cref{ex:conj_split}, we have 
    \begin{align*}
        & \; \exists \sched_1, \sched_2 .\
        \Pr^{\sched_1}_{s_1}(\Finally T) - \Pr^{\sched_1}_{s_2}(\Finally T) < 0
        \;\;\land\;\;
        \Pr^{\sched_1}_{s_1}(\Finally T) - \Pr^{\sched_1}_{s_1}(\Finally T') > 0
        \;\;\land\;\; {} 
        \\ & \qquad\qquad 
        \Pr^{\sched_2}_{s_1}(\Finally T) + \Pr^{\sched_2}_{s_1}(\Finally T') > 2
        \\ \Iff\quad  & 
        \left( 
        \exists \sched_1 .\
        \Pr^{\sched_1}_{s_1}(\Finally T) - \Pr^{\sched_1}_{s_2}(\Finally T) < 0
        \;\;\land\;\; 
        \Pr^{\sched_1}_{s_1}(\Finally T) - \Pr^{\sched_1}_{s_1}(\Finally T') > 0
        \right)
        \;\;\land\;\; {}
        \\ & 
        \left( 
        \exists \sched_2 .\
        \Pr^{\sched_2}_{s_1}(\Finally T) + \Pr^{\sched_2}_{s_1}(\Finally T') > 2
        \right)
        ~.
    \end{align*}
    Note that the first conjunct corresponds to our running example from~\cref{ex:conj_running}.
\end{example}

We can now process each $\conjpartelt \in \conjpart$ individually.
For the sake of simplicity, we detail the following steps assuming this optional partitioning step was skipped.

\subparagraph*{Step 1: Collect state-scheduler combinations.}
We extend the notion of state-scheduler combinations to multiple objectives, overwriting the definition from~\cref{def:comb}.

\begin{definition}
    \label{def:conj_comb}
    We define $\comb = \{ (s_{i,j}, \sched_{k_{i,j}}) \mid i=1,\ldots,\numsum, j=1, \ldots, \numconj \}$.
    For $c \in \comb$, we let $\relInd(c) = \{ (i,j) \mid (\state_{i,j}, \sched_{k_{i,j}}) = c \}$ and $\indexc[\mathcal{T}] = \{T_{i,j} \mid (i,j) \in \relInd(c)\}$.
    For $j \in \conjpartelt$, we further let $\comb^j = \{ (s_{i,j}, \sched_{k_{i,j}}) \mid i=1,\ldots,\numsum \}$ and $\relInd^j(c) = \{ i \mid (\state_{i,j}, \sched_{k_{i,j}}) = c \}$ for $c \in \comb^j$.
\end{definition}

The number of state-scheduler combinations, $|\comb|$, is bounded by $\numconj \cdot \numsum$.

\begin{example}
    For the running example from in~\cref{ex:conj_running}, we have
    $\comb = \{(s_1, \sched), (s_2, \sched)\}$.
    Let $c_1 = (s_1, \sched)$ and $c_2 = (s_2, \sched)$, then we have
    $\mathcal{T}_{c_1} = \{T, T'\}$ and $\mathcal{T}_{c_2} = \{T\}$, and further
    $\comb^1 = \{c_1, c_2\}$
    and
    $\comb^2 = \{ c_1 \}$
\end{example}

\begin{remark}
    If we split the query as detailed in Step 0~(p.~\pageref{step:conj_split}), then, by definition, different partition elements do not share schedulers.
    Formally, for $\conjpartelt, \conjpartelt' \in \conjpart$ with $\conjpartelt \neq \conjpartelt'$:
    \[\{ (s_{i,j}, \sched_{k_{i,j}}) \mid i=1,\ldots,\numsum, j \in \conjpartelt \} \cap \{ (s_{i,j}, \sched_{k_{i,j}}) \mid i=1,\ldots,\numsum, j \in \conjpartelt' \} = \emptyset~.\]
\end{remark}

\subparagraph*{Step 2: Unfold targets and set up reward structures.}
First, we construct the goal unfolding $\indexc = \mdptup[\indexc]$ w.r.t.\ all relevant target sets for each state-scheduler combination $c \in \comb$.
As in~\cref{sec:reach_algo},  we use $\indexc[\state] \in \indexc[\states]$ to denote the state $(\state, \emptyset)$ for $c=(\state, \cdot) \in \comb$.
We then combine all unfolded MDPs $\indexc$ for $c \in \comb$ into a single MDP $\combined$, which corresponds to the disjoint union of the unfolded MDPs plus a fresh initial state $\combined[\state]$ from which we transition to each unfolded MDP with equal probability.
This enables us to rephrase \eqref{eq:genConjRelReachProp} as an extended multi-objective achievability query (\cref{rem:extended-MOA}) on a single MDP.

\begin{definition}[Combined MDP]
    \label{def:MDP_partelt}
    \label{def:MDP_combined}
    For $c \in \comb$, let $\indexc[\mdp] = (\indexc[\states], \Act, \indexc[\Trans])$ be the goal unfolding of $\mdp$ w.r.t.\ $\indexc[\mathcal{T}] = \{T_{i,j} \mid (i,j) \in \relInd(c)\}$ as in \cref{def:goalUnfolding}.
    The combined MDP is then defined as $\combined = \mdptup[\combined]$
    where
    \begin{itemize}
        \item $\combined[\states] = \{ \combined[\state] \} \cup \bigcup_{c \in \comb} \indexc[\states] \times \{c\} $ for a fresh state $\combined[\state]$,

        \item $\combined[\Act] = \Act \cup \{\epsilon\}$ for a fresh action $\epsilon$,

        \item 
        We define $\combined[\Trans] \colon \combined[\states] \times \combined[\Act] \times \combined[\states] \to [0,1]$ as follows, distinguishing between the fresh initial state and states from the unfolded MDPs: 
        \begin{itemize}
            \item $\combined[\Trans](\combined[\state], \epsilon, (\indexc[\state], c)) = \frac{1}{|\comb|}$ for $c \in \combined$,

            \item $\combined[\Trans]((\state, c), \action, (\state', c)) = \indexc[\Trans](\state, \action, \state')$ for $\action \in \Act$ and $\state, \state' \in \indexc[\states]$ for some $c \in \combined$, and

            \item all other transition probabilities are 0.
        \end{itemize}
    \end{itemize} 
\end{definition}

\begin{example}
    \label{ex:combined}
    The combined MDP $\mdp_{\comb}$ for the running example from~\cref{ex:conj_running} is depicted in~\cref{fig:conj_running}. 
\end{example}

The size of $\combined[\states]$ is (coarsely) bounded by $1 + |\comb| \cdot \max_{c \in \comb} |\indexc| \leq 1 + \numconj \cdot \numsum \cdot |\states| \cdot 2^{\numconj \cdot \numsum}$.
More precisely, it is exponential in the maximal number of different target sets relevant for some state-scheduler combination.

For each $c \in \comb$, we define reward structures $\rew_T \colon \indexc[\states] \to \Q$ for $T \in \indexc[\mathcal{T}]$ on $\indexc$ encoding the reachability probabilities as for \RelReach~(cf.~\cref{def:rewForComb}).
Recall that for \RelReach we defined reward structures $\indexc[\rew]$ for each $c \in \comb$ that aggregate all reward structures $\rew_T$ for $T \in \indexc[\mathcal{T}]$~(\cref{def:rewForComb}).
We now define analogous reward structures $\indexc[\rew]^j$ for each predicate $j$ and combine them into a reward structure $\rew^j$ expressing the weighted sum of reachability probabilities for predicate $j$.
We have to scale the collected reward by $|\comb|$ in order to compensate for reaching each $\indexc$ with probability $\nicefrac{1}{|\comb|}$ from $\combined[\state]$.

\begin{definition}[Reward structure for predicate]
    \label{def:rewForConj}
    For $c \in \combined$ and $T \in \indexc[\mathcal{T}]$, let $\rew_{T}$ be the reward structure on $\indexc$ from~\cref{def:rewForComb}.
    For $j \in \{1, \ldots, \numconj\}$ and $c \in \comb^j$, we define the reward structure $\indexc[\rew]^j \colon \indexc[\states] \to \Q$ on $\indexc$ with 
    \[
        \indexc[\rew]^j(\state,\mathcal{T})
        = 
        \sum_{T \in \indexc[\mathcal{T}] \cap \mathcal{T}^j} \left( \sum_{i \in \{\relInd^j(c) \mid T=T_{i,j}\}} q_{i,j} \right) \cdot \rew_T(\state, \mathcal{T})
    \]
    for $(\state, \mathcal{T}) \in \indexc[\states]$.
    Then, we define the reward structure $\rew^j \colon \combined[\states] \to \Q$ on $\combined$ by letting $\rew^j(\combined[\state]) = 0$ and 
    \[
        \rew^j((\state,\mathcal{T}), c) = 
        |\comb| \cdot \indexc[\rew]^j(\state,\mathcal{T})
        ~.
    \]
\end{definition}

\begin{example}
    The rewards for \textcolor{rew1}{$\rew^1$}, \textcolor{rew2}{$\rew^2$} are depicted next to the states in \cref{fig:conj_running}. 
\end{example}

\subparagraph*{Step 3: Reduce to extended MOA query of expected rewards.}
\label{step:reduce-to-moa}
We can reduce the quantification over schedulers for $\mdp$ to quantifying over a single scheduler for $\combined$ (\cref{def:MDP_combined}), and reduce the computation of a weighted sum of probabilities to computing some expected reward, see~\cref{app:reduce-to-moa_proof} for the proof.

\begin{restatable}{lemma}{reduceToMOA}
    \label{le:reduce-to-moa}
    Let $\combined$ as in~\cref{def:MDP_combined} and $\rew^j$ as  in~\cref{def:rewForConj}, then
    \begin{align*}
        & \exists \sched_{1} \ldots \sched_{\numsched} \in \Scheds .\ 
        \bigwedge_{j=1}^{\numconj} 
        \sum_{i=1}^{\numsum} q_{i,j} \Pr^{\mdp, \sched_{k_{i,j}}}_{\state_{i,j}}(\Finally T_{i,j}) 
        \comp_j q_j  
        \\ \Iff\quad & 
        \exists \sched \in \Scheds[\combined] .\ 
        \bigwedge_{j=1}^{\numconj} 
        \Expected_{\combined[\state]}^{\combined, \sched}(\rew^j) \comp_j q_j  
        ~.
    \end{align*}
\end{restatable}

We observe that
\begin{align}
    \exists \sched \in \Scheds[\combined] .\ 
        \bigwedge_{j=1}^{\numconj} 
        \Expected_{\combined[\state]}^{\combined, \sched}(\rew^j) \comp_j q_j  
    ~
    \label{eq:moa-query}
\end{align}
from \cref{le:reduce-to-moa}
is an (extended) multi-objective achievability query in the sense of \cite{quatmannMultiobjectiveModel2016,forejtQuantitativeMultiobjective2011} (see \cref{rem:extended-MOA} for a details on extended MOA queries).
We will address solving \eqref{eq:moa-query} in \cref{sec:moa}; let us assume for now that we have a black-box solving such queries.
Note, however, that for solving this multi-objective query it does not suffice to only compute the \emph{optimal} expected reward w.r.t.~some reward structures, in contrast to \RelReach.

\begin{example}
    \label{ex:conj_rew}
    Our running example from~\cref{ex:conj_running} is equivalent to the following property on the combined MDP $\mdp_{\comb}$:
    \[
        \exists \sched \in \Scheds[\mdp_{\comb}] .\ 
        \Expected_{\state_{\class{1}}}^{\mdp_{\comb}, \sched}(\rew^1) < 0 
        \land 
        \Expected_{\state_{\class{1}}}^{\mdp_{\comb}, \sched}(\rew^2) > 0 
        ~.
    \]
    Consider the memoryless scheduler $\sched$ for $\combined$ defined as follows, which is indicated in~\cref{fig:conj_running} in \textbf{bold}:
    \begin{align*}
        &\sched(s', \emptyset, c_1) = \alpha && \sched(s', \{T'\}, c_1) = \alpha & &\sched(s', \emptyset, c_2) = \alpha  \\
        &\sched(t', \emptyset, c_1) = \gamma &&\sched(t', \{T'\}, c_1) = \delta && \sched(t', \emptyset, c_2) = \gamma 
    \end{align*}
    which induces a \emph{memoryful} scheduler on $\mdp$. 
    This scheduler is a witness since
    \begin{align*}
        \Expected_{\state_{\class{1}}}^{\mdp_{\comb}, \sched}(\rew^1) 
        & = 2 \cdot \left(\frac{1}{2} \cdot \Expected_{s_1,\emptyset,c_1}^{\mdp_{\comb}, \sched}(\rew^1) + \frac{1}{2} \cdot \Expected_{s_2, \emptyset, c_2}^{\mdp_{\comb}, \sched}(\rew^1) \right)
        \\ & = 
        \left(\frac{1}{3} \cdot 2 + \frac{1}{3} \cdot 2 + \frac{1}{3} \left(\frac{2}{3} \cdot 2 + \frac{1}{3} \left(\frac{2}{3} \cdot 2 + \frac{1}{3} \cdot 0 \right) \right) \right) -2
        = \frac{4}{3} + \frac{16}{3 \cdot 9}  - 2
        < 0
        ~,
        \\
        \Expected_{\state_{\class{1}}}^{\mdp_{\comb}, \sched}(\rew^2)
        & = \Expected_{s_1, \emptyset, c_1}^{\mdp_{\comb}, \sched}(\rew^2) 
        = \frac{1}{3} \cdot 2 + \frac{1}{3} \cdot 2 + \frac{1}{3} \left(\frac{2}{3} \cdot 2 + \frac{1}{3} \left(-2 + \frac{2}{3} \cdot 2 \right) \right) \\
        & = \frac{4}{3} + \frac{1}{3} \left(\frac{4}{3} + \frac{1}{3} \left(\frac{-2}{3} \right) \right)
        =  \frac{4}{3} + \frac{1}{3} \left(\frac{12}{9} + \frac{-2}{9} \right) 
        > 0
        ~.
        \qedhere
    \end{align*}
\end{example}

\subparagraph*{Step 4: Combine partition elements (optional).}
If we initially split up the query into the different partition elements $\conjpartelt \in \conjpart$ in Step 0~(p.~\pageref{step:conj_split}), we need to combine the results for each $\conjpartelt \in \conjpart$ again to answer the original query \eqref{eq:genConjRelReachProp} in the end.

\begin{corollary}
    \label{le:conj_combine}
    Let $\conjpart$ as defined in~\cref{def:conjpart}, $\combined$ as defined in~\cref{def:MDP_combined} and $\rew^j$ as defined in~\cref{def:rewForConj}, then
    \begin{align*}
        &\exists \sched_1 \ldots \sched_n .\
        \bigwedge_{j=1}^{\numconj} \sum_{i=1}^{\numsum} q_{i,j} 
        \Pr^{\sched_{k_{i,j}}}_{\state_{i,j}}(\Finally T_{i,j}) \comp_j q_j 
        \\ \Iff\quad & 
        \bigwedge_{\conjpartelt \in \conjpart} \left( \exists \sched_{\conjpartelt} \in \Scheds[\combined] .\ 
        \bigwedge_{j \in \conjpartelt} \Expected_{\combined[\state]}^{\combined, \sched}(\rew^j) \comp_j q_j  \right)
    \end{align*} 
\end{corollary}

\begin{proof}
    Straightforward from \cref{le:split-conjunction} and \cref{le:reduce-to-moa}.
\end{proof}

\subparagraph*{Overall algorithm.}
The algorithm resulting from Steps 1--3 is summarized in~\cref{alg:conj}, and the practically more efficient version including Steps 0 and 4 in~\cref{alg:conj_practice}.

\begin{theorem}[Correctness and time complexity]
    \label{th:conj_runtime}
    \label{th:conj_correctness}
    \cref{alg:conj,alg:conj_practice} adhere to their input-output specification and can be implemented with a worst-case running time of $\mathcal{O}(2^{\numconj_{\not\approx}} \cdot \textit{poly}(\numconj \cdot \numsum \cdot |\mdp| \cdot 2^{\numconj \cdot \numsum}))$
    where $\numconj$ is the number of predicates, $\numconj_{\not\approx}$ the number of predicates with $\not\approx_\epsilon$ with $\epsilon>0$, and $\numsum$ the number of probability operators per objective in the property.
\end{theorem}

\begin{proof}
    Correctness follows from \cref{le:conj_combine}, runtime complexity from \cref{th:MOA_correctness}.
\end{proof}

\begin{algorithm}[t]
    \caption{Solving \ConjRelReach by reduction to multi-objective query} 
    \label{alg:conj}
    \Input{%
        MDP $\mdp = \mdptup$ and a \ConjRelReach property 
        \\ \medskip 
        \phantom{Input: } 
        $\displaystyle \exists \sched_1, \ldots, \sched_n \in \Scheds .\ \bigwedge_{j=1}^{\numconj} \sum_{i=1}^{\numsum} q_{i,j} 
            \Pr^{\sched_{k_{i,j}}}_{\state_{i,j}}(\Finally T_{i,j}) \comp_j q_j $
        \DontPrintSemicolon\tcp*{See \cref{prob:morelreach}}
        \medskip
    }
    \Output{Whether the property is true in $\mdp$} 
    \tcp{Step 1: Analyze state-scheduler combinations:}
    {$\comb \gets \{ (s_{i,j}, \sched_{k_{i,j}}) \mid i=1,\ldots,\numsum, j=1, \ldots, \numconj \}$ \tcp*{See \cref{def:conj_comb}}}
    \For{$c \in \comb$}{
        \tcp{Step 2: Unfold MDP:}
        $\indexc \gets$ unfold $\mdp$ w.r.t.\ target sets for $c$\hspace{-1ex}
        \tcp*{See \Cref{def:goalUnfolding}}
    }
    \tcp{Step 2, cont'd: Combine unfolded MDPs and define reward structures:}
    $\combined \gets$ combine $\indexc$ for $c \in \comb$ \tcp*{See \Cref{def:MDP_partelt}}
    \For{$j \in \conjpartelt$}{
        $\rew^j \gets$ reward structure on $\combined$
        \tcp*{See \Cref{def:rewForConj}}
    }
    \tcp{Step 3: Solve extended MOA query:} 
    {$\textit{res}_\conjpartelt \gets \textsf{Solve(}\exists \sched \in \Scheds[\combined] .\ \bigwedge_j \Expected_{\combined[\state]}^{\combined, \sched}(\rew^j) \comp_j q_j$\textsf{)}}\tcp*{Using, e.g., \cref{alg:MOA}}
\end{algorithm}

\begin{algorithm}[t]
    \caption{Solving \ConjRelReach by reduction to multi-objective query with improved practical efficiency} 
    \label{alg:conj_practice}
    \Input{%
        MDP $\mdp = \mdptup$ and a \ConjRelReach property 
        \\ \medskip 
        \phantom{Input: } 
        $\displaystyle \exists \sched_1, \ldots, \sched_n \in \Scheds .\ \bigwedge_{j=1}^{\numconj} \sum_{i=1}^{\numsum} q_{i,j} 
            \Pr^{\sched_{k_{i,j}}}_{\state_{i,j}}(\Finally T_{i,j}) \comp_j q_j $
        \DontPrintSemicolon\tcp*{See \cref{prob:morelreach}}
        \medskip
    }
    \Output{Whether the property is true in $\mdp$} 
    \tcp{Step 0: Splitting the query}
    $\conjpart \gets \big\{ \{ j' \in \{1,\ldots,l\} \mid \exists i,i' .\ k_{i,j} = k_{i',j'} \} \;\big|\; j \in \{1, \ldots, l\} \big\}$\hspace*{-0.75cm} \tcp*{\cref{def:conjpart}}
    \For{$\conjpartelt \in \conjpart$}{
        {$\textit{res}_\conjpartelt \gets \textsf{Solve(}\exists \sched_{h^\conjpartelt_1} \ldots \sched_{h^\conjpartelt_{\numsched_\conjpartelt}} \in \Scheds .\ \bigwedge_{j \in \conjpartelt} \sum_{i=1}^{\numsum} q_{i,j} 
        \Pr^{\sched_{k_{i,j}}}_{\state_{i,j}}(\Finally T_{i,j}) \comp_j q_j $\textsf{)}}\tcp*{Using \cref{alg:conj}}
    }
    \tcp{Step 4: Combine partition elements:}
    \Return{ $\bigwedge_{\conjpartelt \in \conjpart} \textit{res}_\conjpartelt$} 
    \label{line:conj_combine}
\end{algorithm}

%% file: moa.tex
\subsubsection{Solving Extended Multi-Objective Achievability Queries}
\label{sec:moa}

In this section, we address solving the extended MOA query \eqref{eq:moa-query} constructed in \cref{sec:conj_to_MOMC}, Step~3~(p.~\pageref{step:reduce-to-moa}), i.e.,
\[
    \exists \sched \in \Scheds[\combined] .\ 
        \bigwedge_{j=1}^{\numconj} \Expected_{\combined[\state]}^{\combined, \sched}(\rew^j) \comp_j q_j  
    ~.
\]
To the best of our knowledge, while does exist work on both MOA queries with \emph{positive and negative} rewards~\cite{quatmannVerificationMultiobjective2023} on the one hand, and extended MOA queries with \emph{strict and non-strict} comparison operators~\cite{etessamiMultiObjectiveModel2008,forejtQuantitativeMultiobjective2011}, on the other hand, these have not been considered together before and, more interestingly, the comparison operator $\not\approx_\epsilon$ (or even just $\neq$) has not been addressed before.
We refer to~\cref{sec:related_moa} for a detailed discussion of related work on multi-objective model checking.
%
We extend \cite{etessamiMultiObjectiveModel2008,forejtQuantitativeMultiobjective2011} to solving \eqref{eq:moa-query} by encoding the problem in quantifier-free \emph{linear real arithmetic}, though the encoding corresponds to an LP for queries without~$\not\approx_\epsilon$.
The approach is summarized in \cref{alg:MOA}.

\begin{algorithm}[t]
    \caption{Solving extended MOA queries for \ConjRelReach } 
    \label{alg:MOA}
    \Input{%
        Extended MOA query constructed by \cref{alg:conj}
        \\ \medskip 
        \phantom{Input: } 
        $\displaystyle \exists \sched \in \Scheds[\combined] .\ \bigwedge_{j=1}^{\numconj} \Expected_{\combined[\state]}^{\combined, \sched}(\rew^j) \comp_j q_j
        $
        \DontPrintSemicolon\tcp*{See~\cref{le:reduce-to-moa}}
    }
    \Output{Whether the query is satisfiable} 
    {$\processed[\combined] \gets $ pre-processing of $\combined$ }\tcp*{See \cref{def:conj_preprocessing}}
    {$L \gets$ encoding from \cref{fig:LP} }\tcp*{Potentially containing $\neq, \not\approx$}
    \Return{\textsf{Solve($L$)}}
    \tcp*{SMT solver or iterative calls to LP solvers (\cref{alg:LP})}
\end{algorithm}

\subparagraph*{Pre-processing.}
Following \cite{forejtQuantitativeMultiobjective2011}, we first process $\combined$ to ensure that any scheduler witnessing the original query corresponds to a witness scheduler for the query on the processed MDP $\processedcombined$ that visits all states from $\combined$ in $\processedcombined$ only finitely often, and vice versa.
Recall that $\combined$ is already a goal unfolding, i.e., for any MEC in $\combined$ there must exist some $\mathcal{T}$ and $c \in \comb$ such that all states from the MEC are of the form $((s, \mathcal{T}), c)$ for some $\state \in \states$.
Concretely, we construct $\processed[\combined]$ from $\combined$ by adding new sink states $\bot_{\mathcal{T}}$ for every combination of target sets $\mathcal{T}$ and adding transitions to $\bot_{\mathcal{T}}$ via a fresh action $\dagger$ from all states $\state \in \states_{\MEC}$ that have already seen exactly the target sets from $\mathcal{T}$.
Intuitively, going to $\bot_{\mathcal{T}}$ represents staying forever in the current MEC (and thus not visiting any more target sets).

\begin{definition}
    \label{def:conj_preprocessing}
    Let $\combined[\mathcal{T}] := \bigcup_{c \in \combined[\comb]} \indexc[\mathcal{T}]$.
    Given $\combined$ as constructed in \cref{def:MDP_partelt}, we construct its \emph{pre-processing} $\processed[\combined] = \mdptup[\processedcombined]$ with
    \begin{itemize}
        \item $\processedcombined[\states] = \combined[\states] \cup \overbrace{\{\bot_{\mathcal{T}} \mid \mathcal{T} \subseteq \combined[\mathcal{T}] \}}^{\sdead}$,
        \item $\processedcombined[\Act] = \combined[\Act] \cup \{ \dagger \}$ for a fresh action $\dagger$,
        \item $\processedcombined[\Trans](\state,\action,\state') = \combined[\Trans](\state,\action,\state')$ for $\state,\state' \in \combined[\states]$, $\action \in \combined[\Act]$, \\
        $\processedcombined[\Trans](\state,\dagger, \bot_{\mathcal{T}}) = 1$ for $\state \in \combined[\states]$ and $\mathcal{T} \subseteq \combined[\mathcal{T}]$ if there exists some $C=(T,A) \in \MEC(\combined)$ such that $\state \in T$ and there exist $\state' \in \states$, $c \in \combined[\comb]$ such that $\state = ((\state', \mathcal{T}), c)$, \\
        $\processedcombined[\Trans](\bot_{\mathcal{T}},\dagger, \bot_{\mathcal{T}}) = 1$ for $\mathcal{T} \subseteq \combined[\mathcal{T}]$,
        and all other transition probabilities are 0.
    \end{itemize}
\end{definition}

\begin{example}
    Continuing our running example from~\cref{ex:combined}, the states and transitions added by the pre-processing of $\mdp_{\class{1}}$ are indicated in \textcolor{gray}{gray} in~\cref{fig:conj_running}.
\end{example}

If some scheduler for $\processed[\combined]$ almost-surely reaches $\sdead$, then the expected number of times we transition into or out of a state is finite for all states $\state \in \combined[\states]$ in $\processedcombined$ under that scheduler.
Hence, the pre-processing ensures that the expected number of times we leave a state $\state$ via an action $\action$ is finite for all states from $\combined$ and thus enables us to encode the problem using variables representing these values.
We show that finding a \emph{memoryless} randomized scheduler for the pre-processed MDP that satisfies the extended MOA query and reaches $\sdead$ with probability 1 is equivalent to solving the original query.
This claim can be generalized to any MDP $\mdp$ with reward functions $\rew^j$ whose expected reward is finite under all schedulers, analogously to~\cite{forejtQuantitativeMultiobjective2011}.
Our constructed MDP $\processedcombined$ satisfies the stronger requirement that we collect finite reward \emph{on all paths} under all schedulers, or, equivalently, that we collect non-zero reward only outside MECs.
We state the claim here for our special case only for the sake of a more stream-lined presentation

\begin{restatable}{theorem}{moaPreprocTransferMR}
    \label{th:moa_preproc_transfer-MR}
    Let $\processedcombined$ as defined in~\cref{def:conj_preprocessing} with sink states $\states_\bot$, then
    \begin{align*}
        &\exists \sched \in \Scheds[\combined] .\ 
        \bigwedge_{j=1}^{\numconj} \Expected_{\combined[\state]}^{\combined, \sched}(\rew^j) \comp_j q_j  
        \\ \Iff\quad &
        \exists \processed[\sched] \in \SchedsMR[{\processed[\combined]}] .\ 
        \bigwedge_{j=1}^{\numconj} \Expected_{\combined[\state]}^{{\processed[\combined]}, \processed[\sched]}(\rew^j) \comp_j q_j \wedge \Pr_{\combined[\state]}^{{\processed[\combined]}, \processed[\sched]}(\Finally \sdead) = 1.
    \end{align*}
\end{restatable}

\begin{proof}[Proof (Sketch)]
    We show this claim in two steps, analogously to \cite{forejtQuantitativeMultiobjective2011}.
    First, we show that any scheduler for $\combined$ induces a scheduler for $\processed[\combined]$ that reaches $\sdead$ with probability 1, and vice versa.
    Then we show that memoryless (randomized) schedulers suffice on $\processed[\combined]$.
    We refer to \cref{app:moa_preproc_transfer-MR_proof} for the proof, which uses ideas from \cite{kretinskyLTLConstrainedSteadyState2021a,forejtQuantitativeMultiobjective2011}.
\end{proof}

\begin{example}
    Consider again the running example introduced in~\cref{ex:conj_running} and the corresponding pre-processed MDP, depicted in~\cref{fig:conj_running}.
    Following~\cref{ex:conj_rew} and \cref{th:moa_preproc_transfer-MR}, we have
    $\states_\bot = \{ \bot_{\emptyset}, \bot_{\{T\}}, \bot_{\{T'\}}, \bot_{\{T,T'\}} \}$
    and
    \begin{align*}
        &\exists \sched \in \Scheds[\mdp_{\class{1}}] .\ 
        \Expected_{\state_{\class{1}}}^{\mdp_{\class{1}}, \sched}(\rew^1) < 0 
        \land 
        \Expected_{\state_{\class{1}}}^{\mdp_{\class{1}}, \sched}(\rew^2) > 0 
        \\ \Iff\quad &
        \exists \sched \in \Scheds[{\processed[\mdp_{\class{1}}]}] .\ 
        \Expected_{\state_{\class{1}}}^{\mdp_{\class{1}}, \sched}(\rew^1) < 0 
        \land 
        \Expected_{\state_{\class{1}}}^{\mdp_{\class{1}}, \sched}(\rew^2) > 0 
        \land 
        \Pr_{\state_{\class{1}}}^{\mdp_{\class{1}}, \sched}(\Finally \states_\bot) = 1
        ~. 
    \end{align*}
    We can transform the witness scheduler to the original query from~\cref{ex:conj_rew}, which is indicated in \textbf{bold} in~\cref{fig:conj_running}, to a witness satisfying the modified query on $\processed[\mdp_{\class{1}}]$ by taking $\dagger$ in absorbing states, i.e., in $(t, \{T\}, c_1)$, $(t, \{T,T'\}, c_1)$, $(t', \{T'\}, c_1)$ and $(t, \{T\}, c_2)$.  
\end{example}

The following example illustrates that MR schedulers in general do not suffice for the combined MDP $\combined$ without pre-processing.

\begin{example}   
    Consider again the MDP from \cref{fig:memory-necessary}, restated in \cref{fig:moa_memory-necessary_original}, and the following property
    \[  
        \exists \sched .\ \Pr_\state^\sched(\Finally \{t_2\}) = 0.5~.
    \]
    The combined MDP $\mdp_{\class{1}}$ is shown in \cref{fig:moa_memory-necessary_combined}, and the equivalent query on $\mdp_{\class{1}}$ is
    \[
        \exists \sched \in \Scheds[{\mdp_{\class{1}}}] .\ 
        \Expected_{s_{\class{1}}}^{{\mdp_{\class{1}}}, \sched}(\rew^1) = 0.5
    \]
    where $\rew^1$ collects reward 1 when visiting $t_2$ in $\mdp_{\class{1}}$ and 0 otherwise.
    Let $\sched'$ be a memoryless randomized scheduler for $\mdp_{\class{1}}$. If $\sched'$ selects $\beta$ in $s$ with some positive probability, we must reach $t_2$ almost-surely from $s$, otherwise (if $\sched'$ does not select $\beta$ in $s$), we cannot reach $t_2$ from $s$.
    Hence, any memoryless randomized scheduler $\sched'$ for $\mdp_{\class{1}}$ 
    can achieve only values 0 or 1 for $\Expected_{s_{\class{1}}}^{\mdp_{\class{1}}, \sched}(\rew^1)$.
\end{example}

\begin{figure}[t]
\centering
    \begin{subfigure}[c]{0.49\textwidth}
    \centering
    \begin{tikzpicture}[semithick,>=stealth]
        \node[state] (s) {$s$};
        
        \node[state, below left= of s] (s1) {$t_1$}; 
        \node[state, below right= of s] (t) {$t_2$}; 
    
       \path[-latex', draw]
            (s) edge[bend left] node[right] {$\alpha$} (s1)
            (s) edge node[left] {$\beta$} (t);

        \path[-latex', draw]
            (s1) edge[bend left] node[right] {$\gamma$} (s)
            (t) edge[loop above] (t);
    \end{tikzpicture}
    \caption{Original MDP $\mdp$.}
    \label{fig:moa_memory-necessary_original}
    \end{subfigure}
    \begin{subfigure}[c]{0.49\textwidth}
    \centering
    \begin{tikzpicture}
        \node[state] (s) {$s$};
        \node[staterectangle, left= of s] (sX) {$s_{\class{1}}$};
        
        \node[state, below left= of s] (s1) {$t_1$}; 
        \node[state, below right= of s] (t) {$t_2$}; 
        \node[state, right= of t] (sink) {$t_2'$};

       \path[-latex', draw]
            (sX) edge node[above] {$\epsilon$} (s)
            (s) edge[bend left] node[right] {$\alpha$} (s1)
            (s) edge node[left] {$\beta$} (t)
            (t) edge (sink);

        \path[-latex', draw]
            (s1) edge[bend left] node[right] {$\gamma$} (s)
            (sink) edge[loop above] (sink);
    \end{tikzpicture}
    \caption{Combined MDP $\mdp_{\class{1}}$.}
    \label{fig:moa_memory-necessary_combined}
    \end{subfigure}
    \caption{MDP from \cref{fig:memory-necessary} and the combined MDP $\mdp_{\class{1}}$ w.r.t.\ $\exists \sched .\ \Pr_\state^\sched(\protect\Finally \{t_2\}) = 0.5$, where some state names are simplified for better readability.} 
    \label{fig:moa_memory-necessary}
\end{figure}

\subparagraph*{Encoding in linear real arithmetic.}
\cref{fig:LP} encodes the query \eqref{eq:moa-query} in quantifier-free linear real arithmetic, analogously to~\cite[Fig.\ 2]{forejtQuantitativeMultiobjective2011}, where we introduce $\not\approx_{\epsilon}$ as syntactic sugar.
The encoding uses variables $y_{\state, \action}$ for $\state \in \combined[\states]$, $\action \in \processed[{\combined[\Act]}](\state)$ encoding the expected number of times we leave $\state$ via $\action$. 
Constraint \eqref{eq:EVTs} intuitively expresses that for $\state \in \combined[\states]$, the expected number of times we leave $\state$ corresponds to the expected number of times we enter $\state$, plus 1 if $\state$ is the initial state $\combined[\state]$.
Constraint \eqref{eq:MOA} expresses the extended MOA query and \eqref{eq:sinks} expresses the added constraint of reaching $\sdead$ with probability 1. 
For a property without $\not\approx_\epsilon$, this corresponds to an LP feasibility problem.
We can use ideas from~\cite{etessamiMultiObjectiveModel2008,forejtQuantitativeMultiobjective2011} to show the correctness of the encoding, see~\cref{app:conj_LP_proof}.

\begin{figure}[t]
    \fbox{\begin{minipage}{\textwidth}
    Find assignments for variables 
    $ y_{\state, \action} \text{ for } \state \in \combined[\states], \action \in \processed[{\combined[\Act]}](\state)$ 
    such that 
    \begin{align}
        \sum_{\action \in \processed[{\combined[\Act]}](\state)} 
        \hspace{-2ex} y_{\state, \action} - 
        \hspace{-1ex} \sum_{\state' \in \combined[\states]} 
        \sum_{\action' \in \combined[\Act](\state')} 
        \hspace{-2ex}\combined[\Trans](\state', \action', \state) \cdot y_{\state', \action'} 
        &= \Iverson{\state {=} \combined[\state]} 
        &&\text{for }  \state \in \combined[\states] \label{eq:EVTs} \\
        \sum_{\state \in \combined[\states]} 
        \hspace{-1ex}\rew^j(\state) \cdot 
        \sum_{\state' \in \combined[\states]} 
        \sum_{\action' \in \combined[\Act](\state')} 
        \hspace{-2ex}\combined[\Trans](\state', \action', \state) \cdot y_{\state', \action'} 
        &\comp_j q_j && \text{for } j=1,\ldots,\numconj \label{eq:MOA} \\
        \sum_{\state \in (\combined[\states])_{\MEC}} y_{\state, \dagger} 
        &= 1 && \label{eq:sinks} \\
        y_{\state, \action} 
        &\geq 0 &&\text{for } \state \in \combined[\states], \action \in \processed[{\combined[\Act]}](\state) \label{eq:nonneg}
    \end{align}
    \end{minipage}
    }
    \caption{Linear real arithmetic encoding for extended MOA queries with $\comp_j \in \{>, \geq, \approx_\epsilon,\not\approx_{\epsilon}\mid \epsilon \geq 0 \}$, where we use $x\not\approx_\epsilon y$ as syntactic sugar for $x < y-\epsilon \lor x >y+\epsilon$. 
    Recall that $\processed[{\combined[\states]}] = \combined[\states] \cup \sdead$ and that $\sdead$ consists of sink states that can only be reached via $\dagger$, i.e., $\processed[{\combined[\Trans]}](\state',\dagger,\state) = 0$ for $\state,\state' \in \combined[\states]$.
    }
    \label{fig:LP}
\end{figure}

\begin{restatable}{lemma}{conjLP}
    \label{le:conj_LP}
    Let $\processedcombined$ as defined in~\cref{def:MDP_combined}, then
    \begin{align*}
        &
            \exists \processed[\sched] \in \SchedsMR[{\processed[\combined]}] .\ 
            \bigwedge_{j=1}^{\numconj} \Expected_{\combined[\state]}^{{\processed[\combined]}, \processed[\sched]}(\rew^j) \comp_j q_j 
            \wedge \Pr_{\combined[\state]}^{{\processed[\combined]}, \processed[\sched]}(\Finally \sdead) = 1
        \\ \Iff\quad & 
        \text{the encoding from \cref{fig:LP} has a feasible solution $y^*$.}
    \end{align*}
\end{restatable}

Due to the presence of $\approx_\epsilon$-, $>$-, and $\not\approx_{\epsilon}$-constraints, the constructed encoding in linear real arithmetic does not necessarily correspond to a linear program, whose standard form only allows $\leq,\geq$-constraints.
While $\approx_\epsilon$- and $>$-constraints can be transformed to $\leq,\geq$-constraints~(see, e.g.,~\cite{dutertreFastLinearArithmetic2006,nalbachExtendingFundamental2021} for transforming $>$ to $\geq$), the situation is more complex for $\not\approx_\epsilon$. 
A linear real arithmetic encoding containing $\not\approx_\epsilon$ can be solved using successive calls to an LP solver, in the worst case exponentially many, while in the special case of $\not\approx_0$, i.e., $\neq$, we only need to solve a \emph{linear} number of calls to an LP solver, see \cref{sec:moa_via_lp} for details.

\subparagraph*{Overall algorithm.}
The algorithm resulting from the approach outlined above is summarized in \cref{alg:MOA}.
It has a worst-case running time exponential in the number of probability operators (even without $\not\approx$).
\label{rem:moa_approx}%
We note that while this approach gives us a tighter upper bound for the complexity of MOA than a Pareto-curve-based approach (which would be exponential also in the size of the MDP), in practice, an approximate Pareto-curve-based approach may significantly outperform our exact approach, see, e.g.,~\cite{forejtParetoCurves2012}.

\begin{restatable}{theorem}{MOACorrectness}
    \label{th:MOA_correctness}
    \cref{alg:MOA} adheres to its input-output specification and can be implemented with worst-case running time of $\mathcal{O}(2^{\numconj_{\not\approx}} \cdot \textit{poly}(\numconj \cdot \numsum \cdot |\mdp| \cdot 2^{\numconj \cdot \numsum}))$ 
    where $\numconj$ is the number of predicates, $\numconj_{\not\approx}$ the number of predicates with $\not\approx_\epsilon$ with $\epsilon>0$, and $\numsum$ the number of probability operators per objective in the original \ConjRelReach property, i.e., the original input to~\cref{alg:conj}.
\end{restatable}

\begin{proof}
    Correctness follows from \cref{le:conj_LP} and \cref{th:moa_preproc_transfer-MR}. 

    If the extended MOA query does not contain $\not\approx_{\epsilon}$ with $\epsilon>0$, then the encoding is a quantifier-free formula with a conjunction over at most $|\combined[\states]| + \numconj + 1 + |\combined[\states]| \cdot |\processed[{\combined[\Act]}]|$ linear real constraints.
    Observe 
    $|\combined[\states]| \cdot |\processed[{\combined[\Act]}]|
    \leq (\numconj \cdot \numsum \cdot |\states| \cdot 2^{\numconj \cdot \numsum} + 1) \cdot (|\Act| + 2)$. 
    The claim then follows since quantifier-free linear real arithmetic with only conjunctions can be solved in time polynomial in the number of constraints~\cite{khachiyanPolynomialAlgorithm1979}.
    In particular, if $\neq$-predicates are present, it suffices to first check the encoding without the $\neq$-constraints---which is an LP---and then checking each $\neq$-constraint one by one~\cite[Th.~1]{lassezIndependenceNegative1989}.
    Otherwise, if the extended MOA query contains $\numconj_{\not\approx} \leq \numconj$ constraints with $\not\approx_{\epsilon}$ with $\epsilon>0$, then the encoding corresponds to a disjunction over $2^{\numconj_{\not\approx}}$ quantifier-free formulas with only conjunctions over linear real constraints without $\not\approx_\epsilon$ with $\epsilon>0$.
    We can solve this disjunction by iterating over the disjuncts and solving each disjunct separately.
    See~\cref{sec:moa_via_lp} for details.
\end{proof}

%% file: conjunction_complexity.tex
\subsection{Complexity of Multi-Objective Relational Reachability}

\label{sec:conj_complexity}
\label{sec:conj_not-approx}

Since already \ConjRelReach with a single predicate is \PSPACE-hard, the problem with an arbitrary number of objectives is \PSPACE-hard as well.

\begin{theorem}
    \label{th:conj_exptime}
    Problem \ConjRelReach is \PSPACE-hard and decidable in \EXPTIME.
\end{theorem}

\begin{proof}
    \PSPACE-hardness follows from the fact that the problem is already \PSPACE-hard for a single predicate~\cref{th:general_EXPTIME}, no matter whether we include $\not\approx_\epsilon$ with $\epsilon>0$ or not.
    Membership in \EXPTIME follows from \cref{th:conj_runtime}.
\end{proof}
The exact complexity of \ConjRelReach remains open for reasons similar to those discussed after the proof of \cref{th:general_EXPTIME}.

Recall the \PTIME-solvable special cases of \RelReach from~\cref{th:general_PTIME}: a constant number of target sets, all target sets absorbing, and properties where each probability operator has its own quantifier. 
Multi-objective relational properties without $\not\approx_\epsilon$ ($\epsilon>0$) remain in \PTIME in these cases.
However, if the property contains at least one $\not\approx_\epsilon$ with $\epsilon>0$, some of these special cases become (strongly) \NP-complete:

\begin{restatable}{theorem}{conjSpecial}
    \label{th:conj_special}
    Consider the following special cases of \ConjRelReach, where $\numconj_{\not\approx}$ denotes the number of predicates with $\not\approx_\epsilon$ with $\epsilon>0$:

    \begin{minipage}[c]{\textwidth}
    \begin{tabular}{ll|c|c}
        & & $\numconj_{\not\approx}>0$ & $\numconj_{\not\approx}=0$ \\
        \hline
        \newtag{(a)}{item:conj_fixed-param} & $|\{T_{1,1}, \ldots, T_{m,l}\}|$ is at most a constant & strongly \NP-complete & \PTIME \\
        \newtag{(b)}{item:conj_absorb} & all target sets are absorbing & strongly \NP-complete & \PTIME \\ 
        \newtag{(c)}{item:conj_independent} & \makecell[tl]{$n = m\cdot l$, i.e., each probability operator \\ has its own quantifier} & \PTIME & \PTIME 
    \end{tabular}
    \end{minipage}
\end{restatable}

\begin{proof}
    $\ref{item:conj_fixed-param}, \ref{item:conj_absorb}$: \ConjRelReach is \NP-hard by reduction from 3SAT to a \ConjRelReach property of the form
    \[
        \exists \sched_1, \ldots, \sched_n .\ \bigwedge_{j=1}^{n} \Pr^{\sched_j}_{\state}(\Finally T) \not\approx_{\epsilon} q_j
        ~,
    \]
    i.e., a property with a single, absorbing target set shared by all probability operators.
    The reduction can be found in \cref{app:conj_approx-diseq_NP}.

    Membership in \NP/\PTIME: If the number of different target sets is at most a constant or all targets are absorbing, then the size of the goal unfolding $\indexc$ is polynomial for each $c \in \comb$ and thus also the size of $\combined$ is polynomial. 
    Hence, the encoding from \cref{fig:LP} has a polynomial number of variables and constraints.
    Thus, we can guess a feasible solution $y^*$, whose size is polynomial in the size of the MDP, and verify whether it is indeed a solution in polynomial time.
    In particular, if the property does not contain $\not\approx_\epsilon$ with $\epsilon>0$, we can even solve the encoding in \PTIME following the reasoning from~\cref{th:MOA_correctness}.

    $\ref{item:conj_independent}$: If $n = m \cdot l$, then no predicates share schedulers, so we can solve each predicate separately. We thus have to solve $l$ independent \RelReach properties where each probability operator has its own scheduler, which can be solved in \PTIME by \cref{th:general_PTIME}\ref{item:independent}. 
\end{proof}

In contrast to previous sections, we do not study cases where the model-checking problem under memoryless deterministic schedulers, \ConjRelReachMD, is in \PTIME, but only make the following two statements that follow directly from \cref{th:MD_NP-complete}\footnote{To be precise, membership in \NP for \MORelReachMD does not follow directly from \cref{th:MD_NP-complete} but with analogous reasoning.} and \cref{th:general_necessary}.
In order to analyze fragments where our algorithm returns MD schedulers in polynomial time, we would need to analyze in which cases MD schedulers suffice for multi-objective achievability queries and can be computed in \PTIME, which we consider out of scope here.
While the complexity of solving MOA queries over MD schedulers has been studied before~\cite{delgrangeSimpleStrategies2020,chatterjeeMarkovDecision2006,wieringComputingOptimal2007,velasquezSteadyStatePolicy2019},
we are not aware of work analyzing cases where MD schedulers suffice for MOA queries.

\begin{corollary}
    \ConjRelReachMD is strongly \NP-complete.
\end{corollary}

\begin{corollary}
    \label{cor:conj_general-necessary}
    Memory and randomization are necessary for \ConjRelReach with (approximate or exact) equality.
\end{corollary}

\paragraph*{Summary of \cref{sec:conjunction}}
The key insight is that \MORelReach can be reduced to multi-objective achievability queries with expected reward objectives.
Even though the constructed reward structures assign both positive and negative reward, we can solve the constructed query in time exponential in the number of probability operators of the property by extending existing approaches~\cite{etessamiMultiObjectiveModel2008,forejtQuantitativeMultiobjective2011}.
\MORelReach is in general \PSPACE-hard.
We investigate the complexity of \MORelReach in the analogous case to those where \RelReach can be solved in \PTIME and observe that \MORelReach without $\not\approx_{\epsilon}$ for $\epsilon>0$ can also be solved in \PTIME in these cases, but is strongly \NP-complete in some of the cases if there are predicates $\not\approx_{\epsilon}$ for $\epsilon>0$.
\MORelReachMD is strongly \NP-complete.

%% file: eval/eval.tex
\section{Implementation and Evaluation}
\label{sec:eval}
\label{sec:implement}

We have implemented model-checking algorithms for \RelReach and \RelBuechi based on the procedures described in \cref{sec:buechi_algo,sec:conj_algo}.
Extending our implementation to \MORelReach is current work; 
we believe that the current results already give a good insight into our approach's level of applicability.
We use our prototypical implementation to investigate the scalability of our approach in terms of model size, and how it performs compared to existing tools that can check relational properties.
We give details on our setup in~\cref{sec:setup}, and report on results for evaluating our prototypical implementation on case studies for relational properties~(\cref{sec:case-studies_new}) as well as benchmarks for probabilistic hyperproperties~(\cref{sec:hyperprob-benchmarks}).
The experiments on relational reachability probabilities are similar to the ones in~\cite{gerlachEfficientProbabilistic2025}, but have been run on the latest version of our software.

\subsection{Setup}
\label{sec:setup}

All experiments were performed on a laptop with a single core of a $1.80$GHz Intel~i7 CPU and 16GB~RAM under Linux Ubuntu 24.04 LTS. 

\subparagraph*{Implementation.}
{{\color{cavcolor}
Our prototype\footnote{Source code and benchmarks: \url{https://github.com/carolinager/RelProp/releases/tag/RelReachBuechi}, artifact with scripts to reproduce experiments: \url{https://doi.org/10.5281/zenodo.19450015}} is implemented on top of (the python bindings of) the probabilistic model checker \storm~\cite{henselProbabilisticModel2022} and supports an expressive fragment of relational properties. It takes as input an MDP (in the form of a \PRISM file~\cite{kwiatkowskaPRISM402011}) with $\numsum$ (not necessarily distinct) initial states and $\numsum$ (not necessarily distinct) target sets,} as well as a \emph{universally quantified} relational property, i.e., the negation of a relational formula as defined in~\cref{sec:problem-statement}.
More specifically, our prototype covers universally quantified relational reachability properties, and universally quantified relational B\"uchi properties. 
{\color{cavcolor}%
Most available case studies for relational properties are universally quantified.
Recall that, while the previous sections addressed existentially quantified relational properties, we can transfer our results to universally quantified relational properties by considering their negation.
}%

Note that 
the tool’s output (`Yes'/`No') depends only on comparing an expected reward to a threshold (\cref{alg:linear_general_simplified},~\cref{line:linear_general_simplified_return}), so the same calculations are performed regardless of the comparison's outcome.
Hence, performance is generally unaffected by an instance's falsifiability in these cases.

We use \storm's internal single- and multi-objective model checking capabilities for the computation of expected rewards (\cref{alg:linear_general_simplified},~\cref{line:linear_general_simplified_step3}), 
{\color{cavcolor}where single-objective model checking employs \emph{Optimistic Value Iteration}~\cite{hartmannsOptimisticValue2020} for approximate expected reward computation with a default tolerance of $10^{-6}$.\footnote{\color{cavcolor}Note that \storm returns a single value $x$ with relative difference at most $10^{-6}$ to the exact result instead of sound lower and upper bounds here. We ensure soundness by computing conservative lower and upper bounds as $\tfrac{x}{1 + 10^{-6}}$ and $\tfrac{x}{1 - 10^{-6}}$.} 
The multi-objective model checking employs \emph{Sound Value Iteration}~\cite{quatmannSoundValue2018} with the same tolerance.\footnote{\color{cavcolor}\storm (currently) returns a single value with relative difference at most $10^{-6}$ to the exact result as both lower and upper bound and we again ensure soundness manually.}
Here, we report results for approximate computation of expected rewards with tolerance $10^{-6}$.
}

While parts of this section already appeared in~\cite{gerlachEfficientProbabilistic2025}, we performed all experiments again with our updated prototype, which is based on a newer version of \storm compared to the artifact provided for~\cite{gerlachEfficientProbabilistic2025}.

\subparagraph*{Baseline.}
{\color{cavcolor}
We compare our tool against two baselines:
\HyperProb\footnote{\color{cavcolor}\url{https://github.com/TART-MSU/HyperProb}}~\cite{dobeHyperProbModel2021} and \HyperPaynt\footnote{\color{cavcolor}\url{https://github.com/probing-lab/HyperPAYNT}, we used this docker container: \url{https://zenodo.org/records/8116528}. Also referred to as \enquote{AR loop} in \cite{andriushchenkoDeductiveController2023}.}~\cite{andriushchenkoDeductiveController2023}. 
These tools handle (fragments of much more expressive) probabilistic hyperproperties that partially overlap with relational properties. 
\HyperProb encodes the \HyperPCTL model checking problem in SMT with an exponential number of variables.
\HyperPaynt uses abstraction refinement to model-check a fragment of \HyperPCTL, potentially exploring an exponential number of schedulers.
To the best of our knowledge, no other tools support a (nontrivial) fragment of relational properties.
\HyperProb and \HyperPaynt support approximate comparison operators via an equivalent conjunction or disjunction of two inequalities.\footnote{\color{cavcolor}While not covered in \cite{andriushchenkoDeductiveController2023}, \HyperPaynt also allows to add a constant to one side of the (in)equality.}
}%
\HyperPaynt restricts to reachability objectives, it does not cover B\"uchi objectives. 
While \HyperPCTL does not allow direct nesting of temporal operators, B\"uchi probabilities are \PCTL-definable via $\Prob(\Globally \Finally T) \equiv \Prob(\Finally \Prob(\Globally (\Prob(\Finally T)=1) =1))$ (see, e.g.,~\cite[Th.~10.47]{baierPrinciplesModel2008}).

{\color{cavcolor}
Both \HyperProb and \HyperPaynt search for policies and restrict themselves to MD schedulers. Solving a property over MD schedulers is not always equivalent to checking it over general schedulers (see, e.g., \cref{th:general_necessary}, \cref{ex:oneschedtwostate-exists-greater-memory}), but it coincides sometimes (\cref{th:MD_suffice}).
Further, \HyperProb supports only properties with a single scheduler quantifier (i.e., all probability operators are evaluated over the same MD scheduler).
\HyperPaynt, on the other hand, supports an arbitrary number of scheduler quantifiers.%
\footnote{\color{cavcolor}For both \HyperProb and \HyperPaynt, one could alternatively consider a single scheduler on a manually created self-composition to simulate multiple schedulers (from the same initial state). Based on the large performance gap with \HyperProb and \HyperPaynt seen below, we did not consider such tweaks.}
To account for the different scheduler classes considered by the tools, we will in the following denote the validity of a property over general as well as over MD schedulers as \emph{HR result} and \emph{MD result}, respectively, when comparing the tools.%
\footnote{\color{cavcolor}Note that \HyperPaynt expects existentially quantified properties, but we report the validity w.r.t.\ the universally quantified property, so the MD result stated here is the opposite of the result reported by \HyperPaynt.}
We will also denote for each benchmark family whether the considered property belongs to a fragment for which checking MD schedulers is equivalent to checking general schedulers.

Our tool's core procedure} for treating relational reachability properties
{\color{cavcolor}
consists of polynomially many calls to well-established, practically efficient subroutines based on value iteration~\cite{quatmannParameterSynthesis2016}.
In contrast, \HyperProb and \HyperPaynt use exponential algorithms to solve the \NP-hard \RelReach problems over MD schedulers and do not optimize for the \PTIME special cases (\cref{th:general_PTIME}).
Further, to optimize a linear combination of expected rewards over different schedulers, we can treat the different state-scheduler combinations independently and then aggregate them in Step~4 of \cref{sec:reach_algo}~(p.~\pageref{step4}).
}

\subsection{Case Studies}
\label{sec:case-studies_new}
\label{sec:eval_new}
Let us first illustrate that relational properties cover interesting problems.
\input{eval/eval_new}

\input{eval/eval_reach_comp}

%% file: eval/eval_new.tex
{\color{cavcolor}

\subparagraph*{Generalization of Von Neumann's trick (\VN).}
We generalize Von Neumann's trick from \cref{ex:reach_motivating-ex} to $2N$ bits.
The idea is as follows: 
Extract the first $2N$ bits from the stream; if the number of zeros equals the number of ones, return the value of the first bit; otherwise try again.
We check whether
\[
    \forall \sched .\
    \Pr_{s_0}^{\sched}(\Finally \{\texttt{return 0}\}) \approx_{\epsilon} \Pr_{s_0}^{\sched}(\Finally \{\texttt{return 1}\})
\]
for $\epsilon=0$ and $\epsilon=0.1$ and varying values for $N$. 
By \cref{th:MD_suffice}\ref{item:absorbing}, checking this property over MD schedulers is equivalent to checking it over general schedulers since all target sets are absorbing.
Therefore, \HyperProb and \HyperPaynt can also check the general problem in this case even though they restrict to MD schedulers.

\input{eval/eval_reach_new_tab}

\subparagraph*{Robot tag (\RT).}
Consider a $N {\times} N$ grid world with a robot and a janitor, where both move in turns, starting with the robot.%
\footnote{\color{cavcolor}Our model is based on a \href{https://www.prismmodelchecker.org/casestudies/robot.php}{\PRISM model} and a \href{https://github.com/sjunges/gridworld-by-storm/blob/master/gridstorm/models/files/evade-two-mdp.nm}{Gridworld-By-Storm model}.}
The robot starts in the lower left corner, the janitor in the upper right corner.
The robot has fixed the following strategy: It moves right until it reaches the lower right corner, then moves up until it reaches the upper right corner, its target.
The janitor can hinder the robot from reaching its target by occupying a cell that the robot wants to move to. 
We now want to check whether this strategy is approximately robust against adversarial behavior by the janitor, 
in the sense that the probability of reaching the target should be approximately independent of the janitor strategy.
Formally,
\[
    \forall \sched_1 , \sched_2 .\ \Pr^{\sched_1}_s(\Finally \{t\}) ~\approx_{10^{-5}}~ \Pr^{\sched_2}_s(\Finally \{t\})
    ~,
\]
where $s$ represents the initial state with robot and janitor in their initial locations, and
$t$ represents that the robot has reached its target.
We check this property for different grid sizes $N$ and different starting positions for the janitor:
If the janitor starts in the upper right corner, they cannot hinder the robot, but if they start in location $(N, N{-}1)$, then they can always stop the robot. 

By \cref{th:MD_suffice}\ref{item:independent}, MD schedulers suffice here. 
However, \HyperProb does not natively support this problem as it only allows a single scheduler quantifier and the property is trivially satisfied if we use a single scheduler quantifier for both probability operators.
}

\subparagraph*{Israeli-Jalfon (\IJ).} 
Recall Israeli \& Jalfon's self-stabilizing protocol and its asymmetric variant introduced in \cref{ex:buechi_motivating-ex},\footnote{We extend a \href{https://www.prismmodelchecker.org/casestudies/self-stabilisation.php\#ij}{PRISM model}.} and that we are interested in the following conjunction of relational B\"uchi properties:
\begin{align*}
    \bigwedge_{\statei \in \states} 
    \bigwedge_{i=1}^{N-1} 
    \forall \sched \in \Scheds .\ 
    \Pr^{\sched}_{\statei}(\Globally \Finally Q_i) = \Pr^{\sched}_{\statei}(\Globally \Finally Q_{i+1})
    ~.
\end{align*}
Here, we benchmark the conjunct where $\statei$ is the state where all processes have a token\footnote{This state does not have any incoming transitions and can reach all other states.} and $i=N{-}1$ for all tools for a varying number of processes $N$.
Observe that memoryless deterministic schedulers suffice here due to \cref{cor:buechi_MD_suffice}.


\subparagraph*{Results.}
The results for both case studies for relational reachability properties, \VN and \RT, are presented in \cref{tab:eval_new}.
{\color{cavcolor}
Firstly, we observe that the quality of the coin simulation decreases with a growing number of bits $N$: While approximate equality with $\epsilon{=}0.1$ holds for $N{=}10$, it does not hold anymore for $N{=}100$.
}%
Further, our tool handles 1 million states in $6s$ for \RT but times out for \VN instances of a quarter the size.
{\color{cavcolor}
Besides the size of the state space, the computation time of expected rewards also depends on structural aspects, e.g. many cycles (as in \VN) cause slower convergence of the approximation algorithm for expected rewards.
}%
We note that our prototype performs worse on \VN than our previous prototype from~\cite{gerlachEfficientProbabilistic2025}, even though the implementation for solving relational reachability properties is the same in both versions. 
We therefore attribute the decreased performance to changes in \storm: For \cite{gerlachEfficientProbabilistic2025} we used \storm 1.9.0, for the current evaluation we use \storm 1.12.0, the most recent version. 
{\color{cavcolor}%
Still, on problem instances that are also supported by \HyperProb or \HyperPaynt, our tool performs drastically better: \HyperProb times out already for \VN with $N{=}10$, and \HyperPaynt for \VN with $N{=}100$ while our tool solves these instances in a matter of seconds.
}

\cref{tab:eval_buechi} shows the results for the case study for relational B\"uchi properties, \IJ.
Recall that \HyperPaynt does not support B\"uchi objectives.
Our tool handles roughly $2\,000\,000$ states in around $35$ minutes, only less than $1$ minute of which consist of the actual computation time excluding model building.
It times out while building the model for over $4\,000\,000$ states, while \HyperProb times out for roughly $500$ states already.

\input{eval/eval_buechi_tab}

%% file: eval/eval_reach_new_tab.tex
\begin{table}[t] 
    \centering
    \caption{
        \color{cavcolor}
        Experimental results for relational reachability case studies.
        \emph{HR res.}/\emph{MD res.}: Does the universally quantified property hold over general (HR)/MD schedulers? 
        \emph{$\equiv?$}: Is checking the property over general scheduler equivalent to checking over MD schedulers?
        \emph{Time} is rounded to the nearest second, with the exception of values ${<}1s$ which are denoted as such.
        For every tool, we give the total time and, in brackets,
        (1) for our tool: the total time excluding model building,  
        (2) for \HyperProb: z3 solving time, 
        (3) for \HyperPaynt: the reported `synthesis time'. 
        \TO: Timeout (1 hour). \OOM: Out of memory.
    }
    \label{tab:eval_new}
    \adjustbox{max width=\textwidth}{  
    \begin{tabular}{l l l r || c | r >{\color{gray}}r || c || c | r >{\color{gray}}r | r >{\color{gray}}r}
        \toprule
         \multicolumn{4}{c||}{Case study} & \multicolumn{3}{c||}{Our tool} & \multirow{2}{*}{$\equiv$?} & \multicolumn{5}{c}{Comparison \textbf{over MD sched.}} \\
         & \multicolumn{2}{c}{Variant} & \makecell{$|\states|$} & HR res. & \multicolumn{2}{c||}{Time} & &
         MD res. & 
         \multicolumn{2}{c|}{\HyperProb} & \multicolumn{2}{c}{\HyperPaynt}
         \\
         \midrule 
         \multirow{10}{*}{\VN} 
         & $N{=}1$ & $\epsilon{=}0$ & 5 
         & No & ${<}1s$ & (${<}1s$)
         & \multirow{10}{*}{\makecell{$\equiv$ \\ Cor.~\ref{th:MD_suffice}}}
         &  No
         & ${<}1s$ & (${<}1s$)
         &  ${<}1s$ & (${<}1s$)
         \\
         & $N{=}1$ & $\epsilon{=}0.1$ & 5 
         & Yes & ${<}1s$ & (${<}1s$)
         &
         & Yes
         & ${<}1s$ & (${<}1s$)
         & ${<}1s$ & (${<}1s$)
         \\
         & $N{=}10$ & $\epsilon{=}0$ & 383 
         & No & ${<}1s$ & (${<}1s$)
         & & No
         & \TO & -
         & ${<}1s$ & (${<}1s$)
         \\
         & $N{=}10$ & $\epsilon{=}0.1$ & 383 
         & No & ${<}1s$ & (${<}1s$)
         & & No
         & \TO & -
         & ${<}1s$ & (${<}1s$)
         \\
         & $N{=}100$ & $\epsilon{=}0$ & $39\,803$ 
         & No & $11s$ & ($10s$)
         & & -
         & \TO & -
         & \TO & -
         \\
         & $N{=}100$ & $\epsilon{=}0.1$ & $39\,803$ 
         & No & $10s$ & ($10s$)
         & & -
         & \TO & -
         & \TO & -
         \\
         & $N{=}200$ & $\epsilon{=}0$ & $159\,603$
         & No & $1\,761s$ & ($1\,761s$)
         & & -
         & \OOM & -
         & \TO  & -
         \\
         & $N{=}200$ & $\epsilon{=}0.1$ & $159\,603$
         & No & $1\,719s$ & ($1\,718s$)
         & & -
         & \OOM & -
         & \TO & -
         \\
         & $N{=}250$ & $\epsilon{=}0$ & $249\,503$
         & - & \TO & -
         & & -
         & \OOM & -
         & \TO  & -
         \\
         & $N{=}250$ & $\epsilon{=}0.1$ & $249\,503$
         & - & \TO & -
         & & -
         & \OOM & -
         & \TO & -
         \\
         \midrule
         \multirow{8}{*}{\RT} 
         & $N{=}10$ & $(N,N)$ & $933$ 
         & Yes & ${<}1s$ & (${<}1s$)
         & \multirow{8}{*}{\makecell{$\equiv$ \\ Cor.~\ref{th:MD_suffice}}}
         & Yes
         & \multicolumn{2}{c|}{\multirow{8}{*}{n/a}}
         & ${<}1s$ & (${<}1s$)
         \\
         & $N{=}10$ & $(N,N{-}1)$ & $1\,021$ 
         & No & ${<}1s$ & (${<}1s$)
         & & No
         & & 
         & ${<}1s$ & (${<}1s$)
         \\
         & $N{=}100$ & $(N,N)$ & $994\,803$ 
         & Yes & $7s$ & ($1s$)
         & & -
         & & 
         & \TO & -
         \\
         & $N{=}100$ & $(N,N{-}1)$ & $1\,004\,701$ 
         & No & $6s$ & (${<}1s$)
         & & -
         & & 
         & \TO & -
         \\
         & $N{=}200$ & $(N,N)$ & $7\,979\,603$
         & Yes & $53s$ & ($12s$)
         & & -
         & &
         & \TO & - 
         \\
         & $N{=}200$ & $(N,N{-}1)$ & $8\,019\,401$
         & No & $48s$ & ($8s$)
         & & -
         & &
         & \TO & -
         \\
         & $N{=}300$ & $(N,N)$ & $26\,954\,403$
         & - & \OOM & -
         & & -
         & &
         & \OOM & -
         \\
         & $N{=}300$ & $(N,N{-}1)$ & $27\,044\,101$
         & - & \OOM & -
         & & -
         & &
         & \OOM & -
         \\
         \bottomrule
    \end{tabular}
    }
\end{table}

%% file: eval/eval_buechi_tab.tex
\begin{table}[t] 
    \centering
    \caption{
        Experimental results for relational B\"uchi case study.
        The meaning of columns and abbreviations is the same as in \cref{tab:eval_new}.
    }
    \label{tab:eval_buechi}
    \adjustbox{max width=\textwidth}{  
    \begin{tabular}{l l l r || c | r >{\color{gray}}r || c || c | r >{\color{gray}}r | r >{\color{gray}}r}
        \toprule
         \multicolumn{4}{c||}{Case study} & \multicolumn{3}{c||}{Our tool} & \multirow{2}{*}{$\equiv$?} & \multicolumn{5}{c}{Comparison \textbf{over MD sched.}} \\
         & \multicolumn{2}{c}{Variant} & \makecell{$|\states|$} & HR res. & \multicolumn{2}{c||}{Time} & &
         MD res. & 
         \multicolumn{2}{c|}{\HyperProb} & \multicolumn{2}{c}{\HyperPaynt}
         \\
         \midrule 
         \multirow{10}{*}{\IsraeliJalfon} 
         & $N{=}3$ & asym. & $4$
         & No & ${<}1s$ & (${<}1s$)
         & \multirow{10}{*}{\makecell{$\equiv$ \\ Cor.~\ref{cor:buechi_MD_suffice}}}
         & No
         & ${<}1s$ & (${<}1s$)
         &  \multicolumn{2}{c}{\multirow{10}{*}{n/a}}
         \\
         & $N{=}3$ & orig. & $7$
         & Yes & ${<}1s$ & (${<}1s$)
         & 
         & Yes
         & ${<}1s$ & (${<}1s$)
         & &
         \\
         & $N{=}10$ & asym. & $512$
         & No & ${<}1s$ & (${<}1s$)
         & 
         & -
         & \TO & -
         &  & 
         \\
         & $N{=}10$ & orig. & $1\,023$ 
         & Yes & ${<}1s$ & (${<}1s$)
         & 
         & -
         & \OOM & -
         &  & 
         \\
         & $N{=}15$ & asym. & $16\,384$
         & No & $2s$ & ($2s$)
         & 
         & -
         & \OOM & -
         &  & 
         \\
         & $N{=}15$ & orig. & $32\,767$
         & Yes & $5s$ & ($4s$)
         & 
         & -
         & \OOM & -
         &  & 
         \\
         & $N{=}20$ & asym. & $524\,288$
         & No & $186s$ & ($180s$)
         & 
         & -
         & \OOM &
         &  & 
         \\ 
        & $N{=}20$ & orig. & $1\,048\,575$
         & Yes & $559s$ & ($545s$)
         & 
         & -
         & \OOM &
         &  & 
         \\
         & $N{=}22$ & asym. & $2\,097\,152$
         & No & $2082s$ & ($2050s$)
         & 
         & -
         & \OOM &
         &  & 
         \\ 
        & $N{=}22$ & orig. & $4\,194\,303$
         & - & \TO & -
         & 
         & -
         & \TO &
         &  & 
         \\
         \bottomrule
    \end{tabular}
    }
 \end{table}

%% file: eval/eval_reach_comp.tex
{\color{cavcolor}
\subsection{Benchmarks for Probabilistic Hyperproperties}
\label{sec:hyperprob-benchmarks}

Next, we investigate the scalability of our tool on benchmarks from the literature on probabilistic hyperproperties that are relational reachability properties. These benchmarks are typically motivated by security use cases.
In particular, three out of four \HyperPCTL case studies for MDPs presented in~\cite{abrahamProbabilisticHyperproperties2020} are covered by our approach: 
    We consider (mild variations of) \TA, \PW, \TS from~\cite{abrahamProbabilisticHyperproperties2020}. 
    \TA and \PW check properties of the form
    \begin{align*}
        \textstyle\bigwedge_{\state_1,\state_2 \in \Init} \bigwedge_{i=0}^{M} \forall \sched_1, \sched_2 .\ \Pr^{\sched_1}_{\state_1}(\Finally T_i) = \Pr^{\sched_2}_{\state_2}(\Finally T_i)
        ~,
    \end{align*}
    where $\Init$ is a set of initial states.
    Here, we benchmark only the first conjunct (which can be falsified) for all tools.
    We consider two variations of \TA and \PW: \TAone and \PWone use only a single scheduler quantifier for both probability operators while \TAtwo and \PWtwo use two scheduler quantifiers, as in the original formulation~\cite{abrahamProbabilisticHyperproperties2020}.\footnote{\color{cavcolor}Note that both variants are equivalent over general schedulers.}
    The property for \TS is analogous, but with only a single scheduler quantifier and we fix a different pair of initial states for each instance. 
    Further, we consider \SD from \cite{andriushchenkoDeductiveController2023}, which checks 
    \[ 
        \forall \sched .\ \Pr^{\sched}_{s_1}(\Finally T) \geq \Pr^{\sched}_{s_2}(\Finally T) ~.
    \]
    Details on all four benchmarks can be found in \cref{app:hyperprob-benchmarks}, including the differences in our models to the models from \cite{abrahamProbabilisticHyperproperties2020,andriushchenkoDeductiveController2023}.

\input{eval/eval_reach_comp_tab}

\subparagraph*{Results.}
\cref{tab:eval} presents our 
experimental 
results on these four case studies and provides a comparison with \HyperProb and \HyperPaynt.
We observe that for all benchmarks in \cref{tab:eval}, the running time of our tool consists almost entirely of building the model.

For the benchmarks with a single scheduler quantifier (\TAone, \PWone, \TS, \SD), checking over general schedulers is, in general, not equivalent to checking over MD schedulers (see \cref{ex:oneschedtwostate-exists-greater-memory}). 
Notably, already instances with $100\,000$ states prove to be challenging over MD schedulers for \HyperProb and \HyperPaynt, while our tool solves these instances over general schedulers in over a minute and can handle instances of \TAone with 1 million states in less than 10 minutes.

For the benchmarks with two scheduler quantifiers (\TAtwo and \PWtwo), 
checking over MD schedulers is equivalent to checking over general schedulers since both probability operators are independent (\cref{th:MD_suffice}).%
\footnote{\color{cavcolor}Checking for two scheduler quantifiers on a manually created self-composition of the MDP with \HyperProb already exceeds memory bounds for \TA with $m{=}8$.}
Notably, our tool solves the instance for $M{=}24$ in over a minute while \HyperPaynt times out.
Our tool can solve instances of \TAtwo with more than 1 million states, but times out for the instance of \PWtwo with almost 1 million states.

In summary, our tool is orders of magnitude faster than existing tools, which restrict to MD schedulers but nevertheless solve an equivalent problem in some cases (\TAtwo, \PWtwo).
For problem instances belonging to fragments whose decision problems are \NP-hard over MD schedulers but in \PTIME over general schedulers, we show that solving the \NP-hard problem via SMT solving (\HyperProb) or an abstraction-refinement approach (\HyperPaynt) is also much harder in practice than solving the \PTIME problem.

}

%% file: eval/eval_reach_comp_tab.tex
\begin{table}[t] 
    \centering
    \caption{
        \color{cavcolor}
        Comparison on benchmarks for probabilistic hyperproperties.
        The meaning of columns and abbreviations is the same as in \cref{tab:eval_new}.
    }
    \label{tab:eval}
    \adjustbox{max width=\textwidth}{  
    \begin{tabular}{l l r || c | r >{\color{gray}}r || c || c | r >{\color{gray}}r | r >{\color{gray}}r}
        \toprule
         \multicolumn{3}{c||}{Case study} & \multicolumn{3}{c||}{Our tool} & \multirow{2}{*}{$\equiv$?} & \multicolumn{5}{c}{Comparison \textbf{over MD sched.}} \\
         & Variant & \makecell{$|\states|$} & HR res. & \multicolumn{2}{c||}{Time} & &
         MD res. & 
         \multicolumn{2}{c|}{\HyperProb} & \multicolumn{2}{c}{\HyperPaynt} \\
         \midrule 
         \multirow{5}{*}{\TAone} 
         & $M{=}8$ & $423$
         & No & ${<}1s$  & (${<}1s$)  
         & \multirow{5}{*}{\makecell{$\not\equiv$ \\ Ex~\ref{ex:oneschedtwostate-exists-greater-memory}}}
         & No 
         & $3\,055s$ & ($2\,858s$) 
         & ${<}1s$  & (${<}1s$)  
         \\
         & $M{=}16$ & $13\,039$ 
         & No & ${<}1s$  & (${<}1s$) 
         & & No 
         & \OOM & - 
         & $21s$ &  ($1s$)
         \\
         & $M{=}24$ & $307\,175$ 
         & No & $25s$ & (${<}1s$)
         & & - 
         & \OOM & - 
         & \TO & -
         \\
         & $M{=}28$ & $1\,425\,379$ 
         & No & $556s$ & (${<}1s$)
         &  & -
         & \OOM & - 
         & \TO & -
         \\
         & $M{=}32$ & $6\,488\,031$
         & -  & \TO & -
         & & -
         & \TO &
         & \TO & -
         \\
         \midrule
         \multirow{5}{*}{\TAtwo} 
         & $M{=}8$ & 423 
         & No & ${<}1s$ & (${<}1s$) 
         &\multirow{5}{*}{\makecell{$\equiv$ \\ Cor.~\ref{th:MD_suffice}}}
         & No 
         & \multicolumn{2}{c|}{\multirow{5}{*}{n/a}}  
         & ${<}1s$ & (${<}1s$)
         \\
         & $M{=}16$ & $13\,039$ 
         & No & ${<}1s$ & (${<}1s$) 
         & & No 
         &  &  
         & $115s$ & ($2s$)
         \\
         & $M{=}24$ & $307\,175$ 
         & No & $23s$ & (${<}1s$) 
         & & - 
         &  &  
         & \TO & - 
         \\
         & $M{=}28$ & $1\,425\,379$ 
         & No & $560s$ & (${<}1s$)
         & & -
         & & 
         & \TO & -
         \\
         & $M{=}32$ &  $6\,488\,031$
         & - & \TO & -
         & & -
         &  &
         & \TO & -
         \\
         \midrule
         \multirow{2}{*}{\PWone} 
         & $M{=}2$ & $2\,307$ 
         & No & ${<}1s$ & (${<}1s$) 
         & \multirow{2}{*}{ \makecell{$\not\equiv$ \\ Ex~\ref{ex:oneschedtwostate-exists-greater-memory}} }
         & No 
         & \OOM & - 
         & ${<}1s$ & (${<}1s$) 
         \\
         & $M{=}4$ & $985\,605$ 
         & - & \TO & - 
         &  & -
         & \TO & - 
         & \TO & -
         \\
         \midrule
         \multirow{2}{*}{\PWtwo} 
         & $M{=}2$ & $2\,307$ 
         & No & ${<}1s$ & (${<}1s$)  
         & \multirow{2}{*}{ \makecell{$\equiv$ \\ Cor.~\ref{th:MD_suffice}} }
         & No 
         & \multicolumn{2}{c|}{\multirow{2}{*}{n/a}} 
         & $5s$ & (${<}1s$)
         \\
         & $M{=}4$ & $985\,605$ 
         & - & \TO & - 
         & & - 
         & &  
         & \TO  & -
         \\
         \midrule
         \multirow{5}{*}{\TS} 
         & $h{=}(10,20)$ & 252 
         & No & ${<}1s$ & (${<}1s$) 
         & \multirow{5}{*}{ \makecell{$\not\equiv$ \\ Ex~\ref{ex:oneschedtwostate-exists-greater-memory}} }
         & No 
         & $112s$ & ($38s$)
         & ${<}1s$ & (${<}1s$)
         \\
         & $h{=}(20,200)$ & $2\,412$ 
         & No & ${<}1s$ & (${<}1s$) 
         &  & No
         & \OOM & - 
         & $2s$ & (${<}1s$)
         \\
         & $h{=}(20,5\,000)$ & $60\,012$ 
         & No & ${<}1s$ & (${<}1s$) 
         &  & No
         & \OOM & -
         & $1\,076s$ & ($20s$)
         \\
         & $h{=}(50,10\,000)$ & $120\,012$ 
         & No & ${<}1s$ & (${<}1s$) 
         &  & -
         & \OOM & -
         & \TO & -
         \\
         & $h{=}(50,20\,000)$ & $240\,012$ 
         & No & ${<}1s$ & (${<}1s$) 
         &  & 
         & \OOM & -
         & \TO & -
         \\
         \midrule 
         \multirow{7}{*}{\SD} 
         & simple & 10 
         & No & ${<}1s$ & (${<}1s$)  
         & \multirow{7}{*}{ \makecell{$\not\equiv$ \\ Ex~\ref{ex:oneschedtwostate-exists-greater-memory}} }
         & Yes 
         & $3s$  & (${<}1s$) 
         & $2s$ & ($2s$)
         \\
         & splash-1 & 16 
         & No & ${<}1s$ & (${<}1s$) 
         &  & No 
         & $1\,184s$  & ($1\,183s$) 
         & ${<}1s$ & (${<}1s$) 
         \\
         & splash-2 & 25 
         & No & ${<}1s$ & (${<}1s$) 
         & & Yes 
         & \TO & - 
         & $2\,487s$ & ($2\,485s$)
         \\
         & larger-1 & 25 
         & No & ${<}1s$ & (${<}1s$) 
         & & No 
         & \TO  & - 
         & $310s$ & ($310s$)
         \\
         & larger-2 & 25 
         & No & ${<}1s$ & (${<}1s$) 
         & & Yes 
         & \TO & - 
         & $925s$ & (${<}1s$)
         \\
         & larger-3 & 25 
         & No & ${<}1s$ & (${<}1s$)
         & & No 
         & \TO & - 
         & ${<}1s$ & (${<}1s$)
         \\
         & train & 48 
         & No & ${<}1s$ & (${<}1s$)
         & & Yes 
         & \TO & - 
         & $14s$ & ($14s$)
         \\
         \bottomrule
    \end{tabular}
    }
\end{table}

%% file: related-work.tex
\section{Related Work}
\label{sec:related-work}

{\color{cavcolor}
We discuss two main areas of related work: \emph{multi-objective MDP model checking} and \emph{probabilistic hyperlogics}.}%

\subparagraph*{Multi-objective model checking.}%
\label{sec:related_moa}%
The techniques and results of this paper are closely related to \emph{multi-objective model checking} (MOMC), a well-established area of research in probabilistic verification for MDPs~\cite{chatterjeeMarkovDecision2006,chatterjeeMarkovDecision2007,etessamiMultiObjectiveModel2008,quatmannVerificationMultiobjective2023}.

A key problem within MOMC is \emph{multi-objective achievability} (MOA; see also \Cref{sec:moa}).
For reachability, the basic MOA question is:
\emph{Is there a scheduler such that target set $A$ is reached with probability \emph{at least} $\lambda_A$ \emph{and} target set $B$ is reached with probability \emph{at least} $\lambda_B$?}
Our relational properties, as defined in \cref{sec:problem-statement}, subsume MOA queries with reachability, safety, Büchi, and coBüchi objectives.
This subsumption is strict:
For instance, no MOA query---nor any other MOMC formalism that we are aware of---can express the prototypical \RelReach property that a scheduler reaches $A$ and $B$ with \emph{equal} probabilities.

Similar to our results for \RelReach, MOA with a fixed number of reachability objectives is \NP-complete over MD schedulers~\cite{quatmannVerificationMultiobjective2023} and in \PTIME over general schedulers~\cite{etessamiMultiObjectiveModel2008}.
Moreover, prominent algorithms for MOA~\cite{forejtParetoCurves2012} also rely on optimizing weighted sums of probabilities, just like our algorithm for \RelReach.
The probabilistic model checking tools \PRISM~\cite{kwiatkowskaPRISM402011} and \storm~\cite{henselProbabilisticModel2022} support various types of MOA queries.
Due to the algorithmic similarities, we were able to reuse some of \storm's MOA capabilities in our implementation for \RelReach.

The relation of our work to MOA goes further:
Our solution for \ConjRelReach (see \Cref{sec:conjunction}) reduces the problem to MOA for total expected reward objectives.
The reduction constructs queries with both \emph{positive and negative} rewards, and with \emph{strict and non-strict} comparison operators, including $\not\approx_\epsilon$ and $\neq$.
However, to the best of our knowledge, this particular type of MOA query is not covered by existing work:
\begin{itemize}
    \item Etessami et al.~\cite{etessamiMultiObjectiveModel2008} address only $\omega$-regular objectives, including reachability as a special case, and comparison operators $<$, $\leq$, $=$, $\geq$, and $>$.
    \item Forejt et al.~\cite{forejtQuantitativeMultiobjective2011} extend~\cite{etessamiMultiObjectiveModel2008} to total expected reward with non-negative action-based payoffs.
    \item Forejt, Kwiatkowska, and Parker~\cite{forejtParetoCurves2012} study MOA for reachability or total expected reward, again with non-negative reward functions only, and non-strict comparison operators;
their method explores the Pareto curve using value iteration.
    \item The previous approach~\cite{forejtParetoCurves2012} was extended later to either entirely positive or entirely negative reward functions, as well as strict and non-strict comparison operators~\cite{quatmannMultiobjectiveModel2016,quatmannMarkovAutomata2017}.
    \item Quatmann~\cite{quatmannVerificationMultiobjective2023} treats MOA queries for general objectives expressed as measurable functions---this includes total expected reward objectives with mixed rewards---but restricts to the comparison operator $\geq$.
\end{itemize}
Since the specific MOA problem arising from our reduction is not covered by the literature, we present a dedicated solution in \cref{sec:moa}.
Our approach relies on the techniques presented in~\cite{etessamiMultiObjectiveModel2008,forejtQuantitativeMultiobjective2011} and is tailored to the specific structure of the MDPs constructed during the reduction.

Beyond MOA, our relational properties overlap with the \emph{multi-dimensional percentile queries} studied in~\cite{randourPercentileQueries2017} (also see \cref{remark:buechi_sim-a-s}), but neither subsumes the other:
We allow weighted sums with real-valued coefficients, whereas percentile queries cover several \emph{payoff functions} such as $\inf$, $\sup$, and discounted sum.

Finally, somewhat tangential to MOMC, \emph{distributional properties} for MDPs have also been studied~\cite{DBLP:conf/lics/AkshayGV18}.
These properties can express, for example, that a scheduler induces a DTMC which, after a \emph{finite} number of steps, reaches a distribution in which two given target sets have equal probability~\cite{DBLP:conf/ijcai/0001CMZ24}.
In contrast, \RelReach only considers reachability in the limit.
A more thorough comparison is beyond the scope of this paper, but it appears to be a promising direction for future work.

{\color{cavcolor}
\subparagraph*{Probabilistic hyperlogics.}
Our paper was motivated by recent emerging interest in algorithms for probabilistic
hyperlogics, in particular \HyperPCTL~\cite{abrahamHyperPCTLTemporal2018,abrahamProbabilisticHyperproperties2020} and
\PHL~\cite{dimitrovaProbabilisticHyperproperties2020}.
}%
The class of relational formulas defined in~\cref{sec:problem-statement} is 
{\color{cavcolor}
a strict fragment of \HyperPCTL.\footnote{Recall that repeated reachability probabilities are PCTL-definable (see, e.g.,~\cite[Th.~10.47]{baierPrinciplesModel2008}).}
While \PHL cannot naturally compare probabilities from different initial states, every MDP $\mdp$ and relational property $\varphi$ can be transformed to an MDP $\mdp'$ and \PHL formula $\varphi'$ s.t.\ $\mdp$ satisfies $\varphi$ iff $\mdp'$ satisfies $\varphi'$ over general schedulers (by making a copy of the MDP for every state-scheduler combination).
\HyperPCTL and \PHL both can express properties like \emph{Does there exist a scheduler such that all paths are trace-equivalent almost-surely?}, which relational formulas does not cover.
Due to their high expressiveness, the corresponding model-checking problems are undecidable.
This paper contributes two main points to the study of probabilistic hyperlogics.
First, searching for randomized and memoryful schedulers may be beneficial complexity-wise, both in theory and in practice.
Second, many of the motivating case studies for probabilistic hyperlogics can also be treated by dedicated and therefore much more efficient routines, which also motivated the use of an AR-loop in~\cite{andriushchenkoDeductiveController2023}.
However, in \cite{andriushchenkoDeductiveController2023} a \emph{search} over a finite amount of (MD) schedulers is suggested, while we study the computational complexity and consider an algorithm for general (and thus uncountably many) schedulers.
While \uppaalsmc is not a tool for probabilistic hyperproperties, it also supports \emph{statistical} model checking for comparison of two cost-bounded reachability probabilities \emph{on DTMCs}~\cite{davidTimeStatistical2011}.
}

%% file: conclusion.tex
\section{Conclusion and Future Work}
\label{sec:conclusion}

The key insights for solving relational properties are as follows:
\begin{enumerate}
    \item \RelReach can be reduced to expected reward computations on the goal unfoldings with respect to the relevant target sets for each state-scheduler combination.

    \item \RelBuechi can be reduced to \RelReach on (a slight variation of) the MEC quotient.

    \item \MORelReach can be reduced to multi-objective achievability queries with total expected reward objectives.
\end{enumerate}

\cref{tab:complexity_overview_full} compares our complexity results for the three fragments over general schedulers; note that \cref{th:buechi_general_PTIME} covers more general fragments of \RelBuechi than those presented here.
We conjecture that we can extend our approach to multi-objective relational B\"uchi properties by combining the techniques for solving multi-objective relational reachability properties with those for solving relational B\"uchi properties.
We are currently working on extending our implementation to include the approach for solving \MORelReach properties, more specifically for universally quantified disjunctive relational reachability properties using \storm's MOA capabilities.

\begin{table}[t]
    \caption{Comparison of complexity results for relational properties over general (memoryful, randomized) schedulers, where $\numconj_{\not\approx}$ is the number of predicates in the property with comparison operator $\not\approx_\epsilon$ with $\epsilon>0$.
    }
    \label{tab:complexity_overview_full}
    \setlength\tabcolsep{0pt}
    \centering
    \begin{tabular*}{\linewidth}{@{\extracolsep{\fill}} l  l l l }
        \toprule
        Fragment &  \RelReach & \RelBuechi & \ConjRelReach
        \\
        \midrule
        no restrictions 
        & \makecell[l]{\PSPACE-hard, \\ in \EXPTIME [Th.~\ref{th:general_EXPTIME}]}
        & \makecell[l]{strongly \NP-complete \\ {[Th.~\ref{th:buechi_NP-complete}]}}
        & \makecell[l]{\PSPACE-hard, \\ in \EXPTIME [Th.~\ref{th:conj_exptime}]}
        \\
        \midrule 
        $|\{T_{1,1}, \ldots, T_{\numsum,\numconj}\}|$ is const.
        & \PTIME [Th.~\ref{th:general_PTIME}\ref{item:fixed-param}]
        & \PTIME [Th.~\ref{th:buechi_general_PTIME}\ref{item:buechi_fixed-param}]
        & \makecell[l]{strongly \NP-complete; \\ \PTIME if $\numconj_{\not\approx} = 0$ \\ {[Th.~\ref{th:conj_special}\ref{item:conj_fixed-param}]} }
        \\
        \midrule
        all targets absorbing
        & \PTIME [Th.~\ref{th:general_PTIME}\ref{item:absorbing}]
        & \PTIME [Th.~\ref{th:buechi_general_PTIME}\ref{item:buechi_absorb}]
        & \makecell[l]{strongly \NP-complete; \\ \PTIME if $\numconj_{\not\approx} = 0$ \\ {[Th.~\ref{th:conj_special}\ref{item:conj_absorb}]} }
        \\
        \midrule
        $n = m \cdot l$
        & \PTIME [Th.~\ref{th:general_PTIME}\ref{item:independent}]
        & \PTIME [Th.~\ref{th:buechi_general_PTIME}\ref{item:buechi_single-target}]
        & \PTIME [Th.~\ref{th:conj_special}\ref{item:conj_independent}]
        \\
        \bottomrule
    \end{tabular*}
\end{table}

\subparagraph*{Future Work.}
We are interested in how much we can extend this subclass of \HyperPCTL while keeping the model-checking (decidable and even) tractable.
Since we can already handle B\"uchi and coB\"uchi conditions, an extension to Rabin conditions, i.e, disjunctions over pairs of a B\"uchi and a coB\"uchi condition, seems natural, thus extending our approach towards relational \emph{$\omega$-regular} properties. 

Another step towards more expressive fragments would be 
relational properties with \emph{scheduler quantifier alternations}. 
Observe that, for example, 
\begin{align*}
    & \exists \sched_1 \forall \sched_2 .\ \Pr^{\sched_1}_{\state_1}(\Finally T_1) = \Pr^{\sched_2}_{\state_2}(\Finally T_2)
    \\ \Iff & 
    \min_{\sched_2} \Pr^{\sched_2}_{\state_2}(\Finally T_2) = \max_{\sched_2} \Pr^{\sched_2}_{\state_2}(\Finally T_2) \in [\min_{\sched_1} \Pr^{\sched_1}_{\state_1}(\Finally T_1), \max_{\sched_1} \Pr^{\sched_1}_{\state_1}(\Finally T_1)]
\end{align*}
and
\begin{align*}
    & \forall \sched_1 \exists \sched_2 .\ \Pr^{\sched_1}_{\state_1}(\Finally T_1) = \Pr^{\sched_2}_{\state_2}(\Finally T_2)
    \\ \Iff & 
    \Big(\min_{\sched_1} \Pr^{\sched_1}_{\state_1}(\Finally T_1) \geq \min_{\sched_2} \Pr^{\sched_2}_{\state_2}(\Finally T_2) \Big)
    \land 
    \Big(\max_{\sched_1} \Pr^{\sched_1}_{\state_1}(\Finally T_1) \leq \max_{\sched_2} \Pr^{\sched_2}_{\state_2}(\Finally T_2) \Big)
    ~.
\end{align*}
Thus, properties of such forms can be solved by only minimizing and maximizing some expected reward, similarly to the case without quantifier alternation.
However, it is not immediately clear whether this generalizes to an arbitrary number of alternations.
In particular, one interesting question is whether the dependencies between the quantifiers are relevant: For a property of the form $\forall \sched_1 \exists \sched_2 .\ \phi$, we may choose a different $\sched_2$ for each $\sched_1$ but, as illustrated above, we can solve such a property without considering the dependency between the two scheduler choices.

Further possible avenues for future work are to
allow nested probability operators,
relate properties that require \emph{trace equivalence},
consider relational \emph{expected reward} or \emph{long run average} properties
or models like \emph{partially observable MDPs} with restricted scheduler classes.

%% file: appendix/reach_full-algo.tex
\section{Full Algorithm for \RelReach (\cref{sec:reach_algo})}
\label{app:reach_full-algo}

The algorithms in this section are taken verbatim from \cite{extendedVersion}.
\textcolor{cavcolor}{
\cref{alg:linear_general} outlines the procedure presented in \cref{sec:reach_algo} for arbitrary comparison operators, including the handling of approximate computations using a black-box for approximate expected reward computations~\cref{alg:black-box}.
}

\begin{algorithm}[h]
    \caption{\color{cavcolor}Approximate expected reward (black box)~\cite{quatmannSoundValue2018},\cite[Alg.~4.6]{quatmannVerificationMultiobjective2023}}
    \Input{%
        MDP $\mdp=\mdptup$,
        reward function $\rew \colon \states \to \Q$,
        initial state $s \in \states$,
        \emph{attracting} set $\attractor \subseteq \states$ (i.e., $\forall \sigma \in \Scheds[\mdp] .~ \Pr_s(\Finally \attractor) = 1$),
        absolute tolerance $\tau \geq 0$ (use $\tau = 0$ for exact computation) 
    }
    \Output{
        $\lb{v}, \ub{v} \in \Q$ such that $\lb{v} \leq \max_{\sched \in \Scheds[\mdp]}\Expected_s^{\sched}(\rew^{\Finally \attractor}) \leq \ub{v}$ and $\ub{v} - \lb{v} \leq \tau$
    }
    \label{alg:black-box}
\end{algorithm}

\begin{algorithm}[p]
    \caption{\color{cavcolor}Efficient solution of \RelReach} 
    \label{alg:linear_general}
    \Input{%
        Tolerance $\tau \geq 0$ ($\tau = 0$ yields exact computation),
        MDP $\mdp = \mdptup$ and a \RelReach property
        \\ \medskip 
        \phantom{Input: } 
        $\displaystyle \quad \exists \sched_1, \ldots, \sched_n \in \Scheds .~ \sum_{i=1}^{m} q_i \cdot \Pr^{\sched_{k_i}}_{\state_{i}}(\Finally T_i) \approx_\epsilon q$
        \DontPrintSemicolon\tcp*{See \cref{prob:relreach}}
        \medskip 
    }
    \Output{Whether the property is true in $\mdp$, or \enquote{inconclusive} (the latter can only happen if $\tau > 0$)}
    \tcp{Step 1: Loop over all state-scheduler combinations:}
    \For{$c = (s,\sched) \in \comb = \{ (s_i, \sched_{k_i}) \mid i=1, \ldots, m\}$}{
        \tcp{Step 2: Unfold and define reward structures:}
        $\indexc \gets$ goal unfolding of $\mdp$ w.r.t.\ target sets for $c$ \tcp*{See \Cref{def:goalUnfolding}}
        $s_{c} \gets $ the state $(s,\emptyset)$ in $\indexc$ \tcp*{Just for readability}
        $\indexc[\rew] \gets$ reward structure on $\indexc$ for $c$ \tcp*{See \Cref{def:rewForComb}}
        \tcp{Step 3: Compute (or approximate) expected rewards:}
        $\mdp_c', \state_\bot \gets$ MEC-quotient of $\mdp_c$ and its absorbing state \;
        $\indexc[\lb{v}]^{\max}, \indexc[\ub{v}]^{\max} \gets \mathsf{AprxExRew}(\indexc[\mdp'],\indexc[\rew],s_{c},\{\state_\bot\}, \tau)$  \tcp*{using the black box} 
        \If{$\comp~\in \{ \approx_\epsilon, \not\approx_\epsilon \}$}{
            \tcp{Compute $\min$-expected reward by flipping signs}
            $\indexc[\ub{v}]^{\min}, \indexc[\lb{v}]^{\min} \gets (-1) \cdot \mathsf{AprxExRew}(\indexc[\mdp'],-\indexc[\rew],s_{c},\{\state_\bot\}, \tau)$ \;
        }
    }
    \tcp{Step 4: Aggregate results from individual state-scheduler combinations and check appropriate conditions (depending on $\comp$):}
    $\lb{v}^{\max} \gets \sum_{c \in \comb} \indexc[\lb{v}]^{\max} \,;~~ \ub{v}^{\max}  \gets \sum_{c \in \comb} \indexc[\ub{v}]^{\max}$ \;
    \If{$\comp~\in \{ >, \geq \}$}{
        \lIf{\hspace{7.2mm}$\lb{v}^{\max} \comp q$}{
            \Return{true} \label{line:linear_general_return_first}
        }
        \lElseIf{$\ub{v}^{\max} \not\comp q$}{
            \Return{false}
        }
        \lElse{\Return{\enquote{inconclusive}}}
    }
    \ElseIf{$\comp~\in \{ \approx_\epsilon, \not\approx_\epsilon \}$}{
        $\lb{v}^{\min} \gets \sum_{c \in \comb} \indexc[\lb{v}]^{\min} \,;~~ \ub{v}^{\min}  \gets \sum_{c \in \comb} \indexc[\ub{v}]^{\min}$ \;
        \If{$\comp$ is $\approx_\epsilon$}{
            \lIf{\hspace{7.2mm}$q \in [\ub{v}^{\min} - \epsilon, \lb{v}^{\max} + \epsilon]$}{
                \Return{true}
            }
            \lElseIf{$q \notin [\lb{v}^{\min} - \epsilon, \ub{v}^{\max} + \epsilon]$}{
                \Return{false}
            }
            \lElse{\Return{\enquote{inconclusive}}}
        }
        \ElseIf{$\comp$ is $\not\approx_\epsilon$}{
            \lIf{\hspace{7.2mm}$q \notin [\lb{v}^{\max} - \epsilon, \ub{v}^{\min} + \epsilon]$}{
                \Return{true}
            }
            \lElseIf{$q \in [\ub{v}^{\max} - \epsilon, \lb{v}^{\min} + \epsilon]$}{
                \Return{false}
                 \label{line:linear_general_return_last}
            }
            \lElse{\Return{\enquote{inconclusive}}}
        }
    }
\end{algorithm}

%% file: appendix/reach_algo_proofs.tex
{\color{cavcolor}
\section{Proofs for \cref{sec:reach_algo}}
\label{sec:reach_algo_proofs}

The proofs in this section are taken verbatim from \cite{extendedVersion}.

\subsection{Proof of \cref{thm:whyComb}}
\label{app:whyComb}

\whyComb*

\begin{proof}
    Quantifying over each state-scheduler combination individually is justified because of the following:
    For every pair of distinct combinations $c = (s,\sched), c' = (s',\sched') \in \comb$ it holds that either the scheduler variables are already distinct (i.e., $\sched\neq\sched'$), or else $\sched = \sched'$ and $s \neq s'$.
    In the latter case, since \emph{memoryful schedulers may depend on the initial state}, we can instead quantify over two different schedulers.
    More formally, the following is true in any MDP $\mdp$ with states $s \neq s'$ (and target sets $T, T'$):
    For every $\sched,\sched' \in \Scheds[\mdp]$ there exists a (memoryful) $\hat{\sched}\in \Scheds[\mdp]$  such that 
    $\Pr^{\sched}_{s}(\Finally T) = \Pr^{\hat{\sched}}_{s}(\Finally T)$ and $\Pr^{\sched'}_{s'}(\Finally T') = \Pr^{\hat{\sched}}_{s'}(\Finally T')$.
\end{proof}

\subsection{Proof of \cref{{thm:weightedReachProbsAsExpectedRew}}}
\label{app:weightedReachProbsAsExpectedRew}

\weightedReachProbsAsExpectedRew*

\begin{proof}
    Let $c = (s, \textcolor{gray}{\sched}) \in \comb$ and $\opt 
    \in \{\min,\max\}$.
    By construction, we have
    \[
        \opt_{\sched \in \Scheds[\mdp]} \sum_{j \in \relInd(c)} q_j \cdot \Pr_{s}^{\mdp,\sched}(\Finally T_j) 
        = 
        \opt_{\sched \in \Scheds[\indexc]} \sum_{j \in \relInd(c)} q_j \cdot \Pr_{s_c}^{\indexc,\sched}(\Finally T_j) 
        ~.
    \]
    Since $\rew_{T_j}$ collects reward 1 only on the first visit to $T_j$, it further follows that 
    \[
        \opt_{\sched \in \Scheds[\indexc]} \sum_{j \in \relInd(c)} q_j \cdot \Pr_{s_c}^{\indexc,\sched}(\Finally T_j)
        = 
        \opt_{\sched \in \Scheds[\indexc]} \sum_{j \in \relInd(c)} q_j \cdot \Expected^{\indexc, \sched}_{s_c}(\rew_{T_j}) 
        ~.
    \]
    By linearity of expectations, it holds that
    \begin{align*}
        \opt_{\sched \in \Scheds[\indexc]} \sum_{j \in \relInd(c)} q_j \cdot \Expected^{\indexc, \sched}_{s_c}(\rew_{T_j}) 
        &= 
        \opt_{\sched \in \Scheds[\indexc]} \Expected^{\indexc, \sched}_{s_c}\left(\sum_{j \in \relInd(c)} q_j \cdot \rew_{T_j}\right) 
    \end{align*}
    and finally
    \begin{align*}
        \opt_{\sched \in \Scheds[\indexc]} \Expected^{\indexc, \sched}_{s_c}\left(\sum_{j \in \relInd(c)} q_j \cdot \rew_{T_j}\right) 
        &= 
        \opt_{\sched \in \Scheds[\indexc]} \Expected^{\indexc, \sched}_{s_c}(\indexc[\rew]) 
        ~,
    \end{align*}
    since 
    $\indexc[\rew] = \sum_{T \in \mathcal T_c} (\sum_{j \in \{\relInd(c)\mid T = T_j\}} q_j) \cdot \rew_{T} = \sum_{j \in \relInd(c)} q_j \cdot \rew_{T_j}$.    
\end{proof}

\subsection{Proof of \cref{thm:exRewPtimeMD}}
\label{app:exRewPtimeMD}

\exRewPtimeMD*

\begin{proof}
    Existence of MD optimal schedulers is a well-known result, see, e.g.~\cite[Thm.~7.1.9]{putermanMarkovDecision1994}.
    To compute the expected rewards in polynomial time, we rely on \emph{linear programming} (LP)~\cite{karmarkarNewPolynomialtime1984}.
    To express the $\max$-expected reward as the optimal solution of an LP we follow \cite[Ch.~4]{quatmannVerificationMultiobjective2023} and first cast the expected total reward objective to an \emph{expected reachability reward} objective~\cite[Def.~2.27]{quatmannVerificationMultiobjective2023} in a certain modified MDP $\mdp_c'$, the so-called \emph{MEC-quotient} of $\mdp_c$.
    Intuitively, $\mdp_c'$ arises from $\mdp_c$ by collapsing its MECs $E$ into individual states $s_E$, and adding a fresh absorbing state $s_\bot$ with an ingoing transition from every $s_E$ to model the option of staying within $E$ forever.
    Since every scheduler eventually remains inside some MEC, the state $s_{\bot}$ is \emph{attracting} in $\mdp_c'$, i.e., $\forall \sched \in \Scheds[\mdp_c'].~ \Pr^\sched(\Finally \{s_\bot\}) = 1$.
    Next, observe that $\rew_c$ (\Cref{def:rewForComb}) assigns zero reward to all states inside the ECs of $\mdp_c$ because reward is only collected upon reaching a target set \emph{for the first time}.
    We can thus naturally translate $\rew_c$ to a reward function $\rew_c'$ in the quotient $\mdp_c'$.
    
    The optimal expected reward in $\mdp_c$ is then equal to the optimal expected \emph{reachability reward} w.r.t.\ $\{\state_\bot\}$ in $\mdp_c'$, i.e., the expected reward collected until visiting $\state_\bot$ for the first time (see~\cite[Sec.~4.1.5]{quatmannVerificationMultiobjective2023}):
    \[
        \opt_{\sched \in \Scheds[\indexc]} \Expected^{\indexc, \sched}_{s_c}(\indexc[\rew])
        ~=~
        \opt_{\sched \in \Scheds[\mdp_c']} \Expected^{\indexc', \sched}_{s_c}(\indexc[\rew]' \Finally \{\state_\bot\})
    \]
    Since $\{\state_\bot\}$ is attracting, the \emph{Bellman-equations} associated with the above reachability reward objective in $\mdp_c'$ have a \emph{unique solution}~\cite[Thm.~4.8]{quatmannVerificationMultiobjective2023} and can thus be readily expressed as an LP of size linear in the size of $\mdp_c'$~\cite[Fig.~4.10]{quatmannVerificationMultiobjective2023}.
    
    The above approach applies to $\min$-expected rewards as well by noticing that in any MDP $\mdp$ with a state $\state$ and a reward function $\rew$ such that the optimal expected rewards are well-defined, $\min_{\sigma \in \Scheds} \Expected_s^\sched(\rew) = -\max_{\sigma \in \Scheds} \Expected_s^\sched(-\rew)$.
\end{proof}

\subsection{Proof of \cref{thm:cruxApproxEqual}}
\label{app:cruxApproxEqual}

\cruxApproxEqual*

\begin{proof}[Proof of \NoCaseChange\cref{thm:cruxApproxEqual}]
    We show \enquote{$\Rightarrow$}.
    It follows from \Cref{thm:whyComb,thm:weightedReachProbsAsExpectedRew} that the properties
    \begin{align*}
        & \exists \sched_1, \ldots, \sched_n \in \Scheds .~ \sum_{i=1}^{m} q_i \cdot \Pr^{\sched_{k_i}}_{\state_{i}}(\Finally T_i)
        ~\boldsymbol{\geq}~
        q \boldsymbol{-} \epsilon \qquad\text{and}\\
        & \exists \sched_1, \ldots, \sched_n \in \Scheds .~ \sum_{i=1}^{m} q_i \cdot \Pr^{\sched_{k_i}}_{\state_{i}}(\Finally T_i)
        ~\boldsymbol{\leq}~
        q \boldsymbol{+} \epsilon
    \end{align*}
    both hold in $\mdp$.    
    Let $\sched_1^{\geq},\ldots,\sched_n^{\geq}$ and $\sched_1^{\leq},\ldots,\sched_n^{\leq}$ be witnessing schedulers for the two properties.
    Further, let $\ub{v} \geq q - \epsilon$ and $\lb{v} \leq q + \epsilon$ be the weighted sums of probabilities induced by these schedulers.
    We distinguish two cases:
    
    First, if we have $\ub{v} \leq q + \epsilon$ \emph{or} $\lb{v} \geq q - \epsilon$, then $\sched_1^{\geq},\ldots,\sched_n^{\geq}$ or $\sched_1^{\leq},\ldots,\sched_n^{\leq}$ are already witnessing schedulers for the property and we are done.
    
    Otherwise, we have $\ub{v} > q + \epsilon$ \emph{and} $\lb{v} < q - \epsilon$.
    This means that $q \in (\lb v,\ub v)$ and thus there exists $\lambda \in (0,1)$ such that $q = \lambda \lb v + (1-\lambda) \ub v$.
    The crux is now to consider the schedulers $\sched_i^{\lambda} = [\sched^{\leq}_i \oplus_\lambda \sched^{\geq}_i]$, $i = 1,\ldots,n$ (for details, see \cite[p.~71]{quatmannVerificationMultiobjective2023}).
    These schedulers satisfy
    \begin{align*}
        \sum_{i=1}^{m} q_i \cdot \Pr^{\sched_{k_i}^{\lambda}}_{\state_{i}}(\Finally T_i) 
        ~=~&\lambda \cdot \sum_{i=1}^{m} q_i \cdot \Pr^{\sched_{k_i}^{\leq}}_{\state_{i}}(\Finally T_i) ~+~ (1-\lambda) \cdot  \sum_{i=1}^{m} q_i \cdot \Pr^{\sched_{k_i}^{\geq}}_{\state_{i}}(\Finally T_i) \\
        ~=~&\lambda \lb v + (1-\lambda) \ub v = q
        ~,
    \end{align*}
    hence they witness satisfaction of property \eqref{eq:genRelReachProp} with \emph{exact} equality ($\approx_0$).
    
    We show \enquote{$\Leftarrow$} by contraposition.
    Assume that $q < v^{\min} - \epsilon$ (the argument for the other case $q > v^{\min} + \epsilon$ is analogous).
    It follows that $q + \epsilon < \min_{\sched_1, \ldots, \sched_n} \sum_{i=1}^{m} q_i \cdot \Pr^{\sched_{k_i}}_{\state_{i}}(\Finally T_i)$, i.e., all schedulers give a value at least $\epsilon$ apart from $q$ and hence property \eqref{eq:genRelReachProp} does not hold.
\end{proof}
}

%% file: appendix/reach_compl_proofs.tex
{\color{cavcolor}
\section{Proofs for \cref{sec:reach_complexity}}
\label{app:reach_compl_proofs}

The proofs in this section are taken verbatim from \cite{extendedVersion}.

\subsection{Proof of \cref{th:general_PTIME}}
\label{app:general_PTIME}

\generalPTIME*

\begin{proof}
    \ref{item:fixed-param}: $m$ is the only parameter occurring exponentially in the runtime complexity of \cref{alg:linear_general} with exact reward computation.
    
    \ref{item:absorbing}: If all target sets are absorbing, then for each $c \in \comb$, the number of reachable states in $\indexc$ is $|(\states \times \emptyset) \cup \bigcup_i T_i \times \{T_i\}| = |\states| + \sum_{i=1}^{m} |T_i| \leq (m+1)\cdot|\states|$. 
    Hence, the size of the unfolding is linear in the size of the original MDP and $m$.
    
    \ref{item:independent}: If each probability operator is evaluated under a different scheduler, then $|\comb| = n= m$ and for each $c \in \comb$ we have $|\relInd(c)| = 1$. Hence, for each $c \in \comb$, we have $|\indexc[\states]| = |\states \times 2^{\relInd(c)}| = |\states| \cdot 2$ and thus the size of the unfolding with respect to $\relInd(c)$ is linear in the size of the original MDP.
\end{proof}

\subsection{Proof of \cref{th:MD_NP-complete}}
\label{app:MD_NP-complete}

\MDNPComplete*

We show \NP-hardness by giving a polynomial transformation from the Hamiltonian path problem, 
which is known to be strongly \NP-hard~\cite{gareyStrongNPCompleteness1978}, inspired by~\cite{footePolynomialTime}.
For exact equality, we establish strong \NP-hardness by showing that our transformation is pseudo-polynomial.

\begin{definition}[Hamiltonian Path Problem]
    Given a directed graph $G = (V, E)$ 
    (where $V$ is a set of vertices and $E \subseteq V \times V$ a set of edges) 
    and some initial vertex $v \in V$,
    decide whether there exists a path from $v$ in $G$ that visits each vertex exactly once.
\end{definition}

}
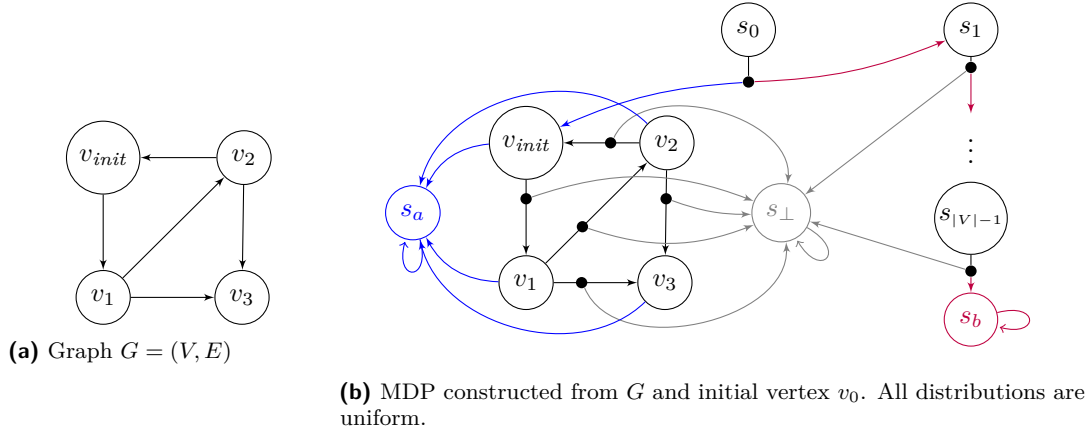
\begin{figure}[t]
    \centering
    \begin{subfigure}[b]{0.3\textwidth}
        \centering
        \begin{tikzpicture}
            \node[state] (v0) {$v_\init$};
            \node[state, below= of v0] (v1) {$v_1$};
            \node[state, right= of v0] (v2) {$v_2$};
            \node[state, right= of v1, xshift=0.1cm] (v3) {$v_3$};

    
            \path[-latex', draw]
                (v0) edge (v1)
                (v1) edge (v2)
                (v1) edge (v3)
                (v2) edge (v0)
                (v2) edge (v3);
        \end{tikzpicture}%
    \caption{Graph $G = (V, E)$}
    \label{fig:ill_ham-path_reduction_graph}
    \end{subfigure}%
    ~
    \begin{subfigure}[t]{0.7\textwidth}
        \centering
    \begin{tikzpicture}
        \node[state] (v0) {$v_\init$};
        \node[state, below= of v0] (v1) {$v_1$};
        \node[state, right= of v0] (v2) {$v_2$};
        \node[state, right= of v1, xshift=0.1cm] (v3) {$v_3$};
        
        \node[dist] at ($(v0)!0.4!(v1)$) (d01) {};
        \node[dist] at ($(v0)!0.6!(v2)$) (d02) {};
        \node[dist] at ($(v1)!0.4!(v2)$) (d12) {};
        \node[dist] at ($(v1)!0.4!(v3)$) (d13) {};
        \node[dist] at ($(v2)!0.4!(v3)$) (d23) {};
        
        \node[] at ($(v2)!0.5!(v3)$) (sinkhelp) {};
        \node[state, right= of sinkhelp, color=gray] (sink) {$s_\bot$};
        
        \node[] at ($(v0)!0.5!(v1)$) (sahelp) {};
        \node[state, left= of sahelp, color=blue] (sa) {$s_a$}; 

        \node[right of = v2, xshift=3cm] (s1help) {};
        \node[state, yshift=1.5cm] at (s1help) (s1) {$s_1$};  
        \node[dist, below of=s1, yshift=0.5cm] (ds1) {};
        \node[below of=ds1] (ds2) {\vdots};
        \node[state, below of=ds2] (svm1h) {\phantom{$s_{|V|}$}};
        \node at (svm1h) (svm1) {$s_{\scriptscriptstyle |V|-1 }$};
        \node[dist, below of=svm1h, yshift=0.3cm] (dsvm1) {};

        \node[state, below= of svm1h, yshift=0.5cm, color=purple] (sb) {$s_b$}; 

        \node[state, yshift=1.5cm] at ($(v0)!0.5!(s1help)$) (sinit) {$\statei$};
        \node[dist, below= of sinit, yshift=0.75cm] (dinit) {};
        
        \path[-latex', draw]
            (v0) edge (v1)
            (v1) edge (v2)
            (v1) edge (v3)
            (v2) edge (v0)
            (v2) edge (v3);

        \path[-, draw]
            (sinit) edge (dinit)
            (s1) edge (ds1) 
            (svm1h) edge (dsvm1)
            ;
        \path[-latex', draw]
            (dinit) edge[color=blue, bend right=10] node[above] {} (v0)
            (dinit) edge[color=purple, bend right=10] node[above] {} (s1)
            (ds1) edge[color=gray] node[above] {} (sink)
            (ds1) edge[color=purple] node[right] {} (ds2)
            (dsvm1) edge[color=gray] node[above] {} (sink)
            (dsvm1) edge[color=purple] node[right] {} (sb); 

        \path[-latex', draw, color=gray]
            (d01) edge[bend left=20] node[pos=0.9, above] {} (sink)
            (d02) edge[bend left=80] node[pos=0.8, right] {} (sink)
            (d12) edge[bend right=20] node[pos=0.8, below] {} (sink)
            (d13) edge[bend right=80] node[pos=0.8, right] {} (sink)
            (d23) edge[bend right=10] node[pos=0.3, below] {} (sink);
            
        \path[-latex', draw, color=blue]
            (v0) edge[bend right] (sa)
            (v1) edge[bend left] (sa)
            (v2) edge[bend right=60] (sa)
            (v3) edge[bend left=60] (sa);

        \path[-latex', draw]
            (sink) edge[in=300, out=330, loop, color=gray] (sink)
            (sa) edge[loop below, color=blue] (sa)
            (sb) edge[loop right, color=purple] (sb);
    \end{tikzpicture}
    \caption{MDP constructed from $G$ and initial vertex $v_0$. All distributions are uniform. }
    \label{fig:ill_ham-path_reduction_mdp}
    \end{subfigure}
    \caption{\color{cavcolor}Illustration of the MDP construction for the reduction from the Hamiltonian path problem. 
    }
    \label{fig:ill_ham-path_reduction}
\end{figure}
{\color{cavcolor}

\begin{proof}[Proof of \NoCaseChange\cref{th:MD_NP-complete}]
    Membership: Given some \RelReach property, MDP $\mdp$ and a memoryless deterministic scheduler $\sched \in \Scheds$, we can verify whether $\sched$ is a witness for the property 
    by computing the (exact) reachability probabilities in the induced DTMC, see e.g.,~\cite[Ch.~10]{baierPrinciplesModel2008}.
    This is possible in time polynomial in the size of the state space~\cite{biancoModelChecking1995}.
    
    \NP-hardness:
    We first give the reduction for \ref{item:1sched1state} with $\state_1 = \state_2$
    and explain how to adjust the construction for the other cases afterwards.
    For a given instance of the Hamiltonian path problem $G = (V, E)$ and $v_{\init} \in V$, 
    we construct the MDP $\mdp = \mdptup$ with
    \begin{itemize}
        \item $\states = V \cup \{ \statei, s_\bot, s_a, s_b \} \cup \{ \state_i \mid i = 1, \ldots, |V|-1\}$
        
        \item $\Act = E \cup \{ \tau \}$
        
        \item 
        $\Trans(\statei, \tau, v_\init) = \frac{1}{2}$ and $\Trans(s_0, \tau, s_1) = \frac{1}{2}$ \\
        $\Trans(s_{i}, \tau, s_{i+1}) = \frac{1}{2}$ and $\Trans(s_{i}, \tau, s_\bot) = \frac{1}{2}$ for $i=1,\ldots, |V|-2$\\
        $\Trans(s_{|V|-1}, \tau, s_b) = \frac{1}{2}$ and $\Trans(s_{|V|-1}, \tau, s_\bot) = \frac{1}{2}$ \\
        $\Trans(v, (v,v'), v') = \frac{1}{2}$ and $\Trans(v, (v,v'), s_\bot) = \frac{1}{2}$ for all $v, v' \in V, (v,v') \in E$ \\
        $\Trans(v, \tau, s_a) = 1$ for all $v \in V$ \\
        $\Trans(s_\bot, \tau, s_\bot) = 1$, $\Trans(s_a, \tau, s_a) = 1$ and $\Trans(s_b, \tau, s_b) = 1$ \\
        $\Trans(s, \alpha, s') = 0$ otherwise
    \end{itemize}
    \cref{fig:ill_ham-path_reduction} illustrates the construction: 
    \cref{fig:ill_ham-path_reduction_mdp} shows the MDP constructed from the graph $G$ with initial vertex $v_0$ depicted in \cref{fig:ill_ham-path_reduction_graph}.

    We observe that for all schedulers $\sched \in \SchedsMD$, it holds that $\Pr^\sched_{\statei}(\Finally \{s_b\}) = \frac{1}{2^{|V|}}$.
    Further, for all schedulers $\sched \in \SchedsMD$, we have
    \begin{align*}
        {\Pr}^{\sched}_{\statei}(\Finally \{s_a\}) 
        &= \frac{1}{2} \cdot {\Pr}^{\sched}_{v_{\init}}(\Finally \{s_a\}) 
    \end{align*}

    \proofsubparagraph{Claim:} \emph{For any $\epsilon < \frac{1}{2^{|V|+1}}$, it holds that there exists a Hamiltonian path from $v_\init$ in $G$ iff there exists some $\sched \in \SchedsMD$ such that $\Pr^\sched_{\statei}(\Finally \{s_a\}) - \Pr^\sched_{\statei}(\Finally \{s_b\}) \approx_\epsilon 0$.}

    ``$\Rightarrow$'': 
    Assume there exists a Hamiltonian path from $v_\init$ in $G$.
    Then, we construct $\sched$ by following this path in $\mdp$ and transitioning to $s_\bot$ from the last vertex. 
    By construction, it holds that ${\Pr}^{\sched}_{\statei}(\Finally \{s_a\}) = \frac{1}{2} \cdot \Pr^{\sched}_{v_{\init}}(\Finally \{s_a\}) = \frac{1}{2^{|V|}}$.
    Since $\Pr^{\sched}_{\statei}(\Finally \{s_b\}) = \frac{1}{2^{|V|}}$ for all schedulers $\sched \in \SchedsMD$, this implies $\Pr^{\sched}_{\statei}(\Finally \{s_a\}) = \Pr^{\sched}_{\statei}(\Finally \{s_b\})$ and hence in particular also $\Pr^\sched_{\statei}(\Finally a) - \Pr^\sched_{\statei}(\Finally b)| \approx_\epsilon 0$.

    ``$\Leftarrow$'': 
    Assume there exists $\sched \in \SchedsMD$ such that $|\Pr^\sched_{\statei}(\Finally \{s_a\}) - \Pr^\sched_{\statei}(\Finally \{s_b\})| \leq \epsilon$.
    We first show that $\Pr^\sched_{\statei}(\Finally \{s_a\}) > 0$: Assume $s_a$ is not reachable from $\statei$. Then $\frac{1}{2^{|V|+1}} > \epsilon \geq |\Pr^\sched_{\statei}(\Finally \{s_a\}) - \Pr^\sched_{\statei}(\Finally \{s_b\})| = |\Pr^\sched_{\statei}(\Finally \{s_b\})| = \frac{1}{2^{|V|}}$, which is a contradiction.
    So $s_a$ is reachable from $\statei$, and hence also from $v_\init$ and the finite path from $v_\init$ to $s_a$ cannot contain a loop since since $\sched$ is memoryless deterministic.
    Let $n$ be the number of $V$-states on the path from $v_\init$ to $s_a$ under $\sched$ (not counting $v_\init$ itself).
    Then $\Pr^{\sched}_{\statei}(\Finally \{s_a\}) = \frac{1}{2^{n+1}}$.
    Since $n \leq |V|-1$, this implies that $\Pr^{\sched}_{\statei}(\Finally \{s_a\}) \geq \frac{1}{2^{|V|}}$.
    Further, 
    \begin{align*}
        |\Pr^{\sched}_{\statei}(\Finally \{s_a\}) - \Pr^{\sched}_{\statei}(\Finally \{s_b\})| \leq \epsilon 
        \Iff
        \Pr^{\sched}_{\statei}(\Finally \{s_a\}) - \frac{1}{2^{|V|}} \leq \epsilon ~ .
    \end{align*}
    Since $\epsilon < \frac{1}{2^{|V|+1}}$, this implies 
    \[ 
        \Pr^{\sched}_{\statei}(\Finally \{s_a\}) < \frac{1}{2^{|V|}} + \frac{1}{2^{|V|+1}} < \frac{1}{2^{|V|-1}} ~ .
    \]
    From $\frac{1}{2^{|V|-1}} > \Pr^{\sched}_{\statei}(\Finally \{s_a\}) = \frac{1}{2^{n+1}} \geq \frac{1}{2^{|V|}}$ we can conclude that $n = |V|-1$.
    Hence, the path from $v_\init$ to $s_a$ visits $|V|-1$ states corresponding to vertices in $G$, which means it must visit each vertex of $G$ exactly once and thus corresponds to a Hamiltonian path from $v_\init$ in $G$.

    \proofsubparagraph{Claim:} \emph{The above construction defines a pseudo-polynomial time transformation to \ref{item:1sched1state}, and a polynomial-time transformation to \ref{item:1sched1state} with approximate equality with $\epsilon>0$.}
    
    The constructed MDP has $2 \cdot |V| + 3$ states and $|E|+1$ actions. Each state has at most two successors.
    All transition probabilities in the MDP are either $0$, $0.5$, or $1$. Hence, the magnitude of the largest number occurring in the constructed MDP is a constant and thus polynomial in the size of the original graph.

    For approximate equality, we must choose some $0 < \epsilon < \frac{1}{2^{|V|+1}}$ for the constructed property. 
    The magnitude of the largest number occurring in the constructed \RelReach instance thus depends exponentially on the size of the original Hamiltonian path problem instance.
    Hence, the transformation is not pseudo-polynomial.
    The transformation is, however, still polynomial since clearly $\epsilon <1$ and hence all numbers are polynomially bounded.

    For exact equality, however, $\epsilon=0$ and hence the magnitude of the largest number occurring in the constructed \RelReach instance is polynomial in the size of the Hamiltonian path problem instance.
    Hence, the transformation is pseudo-polynomial.

    \proofsubparagraph{Handling the other cases.}
    \begin{itemize}
        \item \ref{item:2sched2state}, $s_1 = s_2$:
        We construct the MDP as above, and the following property:
        \enquote{$\exists \sched_1, \sched_2 \in \SchedsMD .\ \Pr^{\sched_1}_{s_0}(\Finally \{s_a\}) - \Pr^{\sched_2}_{s_0}(\Finally \{s_b\}) \approx_\epsilon 0$}.
        The proof works completely analogously since for all schedulers $\sched_1 \in \SchedsMD$, it holds that $\Pr^{\sched_1}_{\statei}(\Finally \{s_b\}) = \frac{1}{2^{|V|}}$
        and for all schedulers $\sched_2 \in \SchedsMD$, we have
       ${\Pr}^{\sched_2}_{\statei}(\Finally \{s_a\}) = \frac{1}{2} \cdot {\Pr}^{\sched_2}_{v_{\init}}(\Finally \{s_a\})$.        
        
        \item \ref{item:1sched2state}, $s_1 \neq s_2$:
        The MDP construction works analogously to the construction above, the only difference being that we do not introduce a fresh initial state. 
        We construct the property
        \enquote{$\exists \sched \in \SchedsMD .\ \Pr^\sched_{v_0}(\Finally \{s_a\}) - \Pr^\sched_{s_1}(\Finally \{s_b\}) \approx_\epsilon 0$}.
        Observe that for all schedulers $\sched \in \SchedsMD$, it holds that $\Pr^{\sched}_{s_1}(\Finally \{s_b\}) = \frac{1}{2^{|V|-1}}$.
        The proof works analogously.

        \item \ref{item:2sched2state}, $s_1 \neq s_2$:
        Again, we construct the MDP as above but without the fresh initial state.
        We construct the property \enquote{$\exists \sched_1, \sched_2 \in \SchedsMD .\ \Pr^{\sched_1}_{v_0}(\Finally \{s_a\}) - \Pr^{\sched_2}_{s_1}(\Finally \{s_b\}) \approx_\epsilon 0$}.
        The proof works analogously.
    \end{itemize}
    
\end{proof}

\subsection{Proof of \cref{th:general_necessary}}
\label{app:general_necessary}

\generalNecessary*

}
\begin{figure}[t]
    \centering
    \begin{subfigure}[b]{0.49\textwidth}
        \centering
    \begin{tikzpicture}
        \node[state] (sinita) {$s$};  
        
        \node[state, below left= of sinita] (t) {\phantom{$s_{2}$}}; 
        \node at (t) (th) {$t$};
        \node[state, below right= of sinita] (sx) {$s_\bot$}; 

        \path[-latex', draw]
            (sinita) edge node[left] {$\alpha$} (t)
            (sinita) edge node[right] {$\beta$} (sx);

        \path[-latex', draw]
            (t) edge[loop left] (t)
            (sx) edge[loop right] (sx);
    \end{tikzpicture}
    \caption{$\exists \sched .\ \Pr^\sched_s(\Finally \{t\}) \approx_\epsilon 0.5$ 
    }
    \label{fig:randomization-necessary_app}
    \end{subfigure}%
    \begin{subfigure}[b]{0.49\textwidth}
        \centering
    \begin{tikzpicture}
        \node[state] (s) {$s$};
        
        \node[state, below left= of s] (t1) {$t_1$}; 
        \node[state, below right= of s] (t2) {$t_2$}; 

       \path[-latex', draw]
            (s) edge[bend left] node[right] {$\alpha$} (t1)
            (s) edge node[right] {$\beta$} (t2);

        \path[-latex', draw]
            (t1) edge[bend left] node[right] {$\gamma$} (s)
            (sx) edge[loop right] (sx);
    \end{tikzpicture}
    \caption{$\exists \sched .\ \Pr^{\sched}_{s}(\Finally \{t_1\}) = \Pr^{\sched}_{s}(\Finally \{t_2\})$}
    \label{fig:memory-necessary_app}
    \end{subfigure}%
    \caption{\color{cavcolor}
        MDPs where memory and/or randomization are necessary for relational reachability properties with equality.
    }
\end{figure}
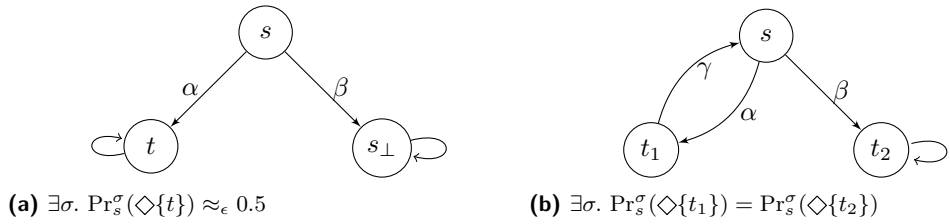
{\color{cavcolor}

\begin{proof}
    We first show that randomization is necessary:
    Consider the MDP in \cref{fig:randomization-necessary_app}.
    Over general schedulers, the property 
    \enquote{$\exists \sched .\ \Pr^\sched_s(\Finally \{t\}) \approx_\epsilon 0.5$} holds for any $\epsilon$, but over (possibly memoryful) deterministic schedulers it holds only for $\epsilon \geq 0.5$.

    We further see that memory is also necessary, i.e., there are instances where memoryless randomized schedulers do not suffice:
    Consider the MDP in \cref{fig:memory-necessary_app}.
    Over general schedulers, the property
    \enquote{$\exists \sched_1 .\ \Pr^{\sched}_{s}(\Finally \{t_1\}) = \Pr^{\sched}_{s}(\Finally \{t_2\})$} holds, 
    but over memoryless schedulers it does not.
    Consider a memoryless randomized scheduler choosing $\alpha$ with probability $p$. If $p=1$ then the property does not hold. If $p<1$ then $\Pr^{\sched}_{s}(\Finally \{t_1\}) = p \neq \Pr^{\sched}_{s}(\Finally \{t_2\}) = \sum_{i=0}^{\infty} p^i (1-p) = \frac{1-p}{1-p} = 1$.
\end{proof}

\subsection{Proof of \cref{th:MD_PTIME}}
\label{app:MD_PTIME}

\MDPTIME*

\begin{proof}   
    For each case, we show that \cref{alg:linear_general} returns an MD scheduler and runs in polynomial time.
    
    \proofsubparagraph{\ref{item:MD_independent}:}
    If $n=m$ and $\comp \in \{\geq, >, \not\approx_\epsilon \mid \epsilon \geq 0 \}$,
    then $|\relInd(c)|=1$ for each $c \in \comb$ and each scheduler $\sched_i \in \{\sched_1, \ldots, \sched_n\}$ corresponds to an MD scheduler for $\indexc$ for the unique $c \in \comb$ with $c = (-, \sched_i)$.
    Further, any MD scheduler on $\indexc$ can be transformed back to an MD scheduler for $\mdp$ since there is only a single target set.
    
    We have already argued in \cref{th:general_PTIME} that \cref{alg:linear_general} runs in polynomial time in this case.

    \proofsubparagraph{\ref{item:MD_n-comb}:}
    We show the claim for the simplest case, namely $n=1$ and $s_1 = \ldots = s_m$ and all target sets are absorbing and $\comp \in \{\geq, >, \not\approx_\epsilon \mid \epsilon \geq 0 \}$.
    Then $|\comb| = 1$, so there is a unique $c \in \comb$.
    We can transform any MD scheduler for $\indexc$ back to an MD scheduler on $\mdp$ since the target sets are absorbing and there is a unique initial state $s_1 = \ldots = s_m$.
    
    Further, since target states are absorbing, the size of the goal unfolding $\indexc$ is linear in $\mdp$ and hence \cref{alg:linear_general} runs in polynomial time.
    
    The reasoning for the general case follows analogously.

    \proofsubparagraph{\ref{item:MD_sign}:}
    We show the claim for the simplest case, namely $n=1$ and $T_1 = \ldots = T_m =: T$ and for all $i=1, \ldots, m$ we have $q_i \geq 0$ and $\comp \in \{\geq, >, \not\approx_\epsilon \mid \epsilon \geq 0 \}$.
    We know that there exists some MD scheduler $\sched^{\max} \in \SchedsMD$ maximizing the reachability probability of $T$ for all $s \in \states$. 
    Since all coefficients $q_1, \ldots, q_m$ are non-negative, it follows that
    \begin{align*}
        \exists \sched \in \Scheds .\ \sum_{i=1}^{m} q_i \cdot \Pr^{\sched}_{\state_i} (\Finally T) \geq q_{m+1}
        \quad\Iff\quad &
        \max_{\sched \in \Scheds} \sum_{i=1}^{m} q_i \cdot \Pr^{\sched}_{\state_i} (\Finally T) \geq q_{m+1}
        \\ \Iff\quad & 
        \sum_{i=1}^{m} q_i \cdot \Pr^{\sched^{\max}}_{\state_i} (\Finally T) \geq q_{m+1}
        \\ \Iff\quad & 
        \max_{\sched \in \SchedsMD} \sum_{i=1}^{m} q_i \cdot \Pr^{\sched}_{\state_i} (\Finally T) \geq q_{m+1}
    \end{align*}
    and analogously for $>$ and $\not\approx_\epsilon$.
    Hence, \cref{alg:linear_general} returns true iff the property is satisfiable over MD schedulers.
    
    Further, since $T_1 = \ldots = T_m$, the size of the goal unfolding is linear in the size of the original MDP and hence \cref{alg:linear_general} runs in polynomial time.
    
    The reasoning for the general case follows analogously.
\end{proof}
}

%% file: appendix/buechi_algo_proofs.tex
\section{Proofs for \cref{sec:buechi_algo}}
\label{app:buechi_algo_proofs}

\subsection{Proof of \cref{th:relBuechi-to-relReach}}

Let us fix some notation needed for both directions of \cref{th:relBuechi-to-relReach}.
For $\pi \in \Paths$, we let $\Limit(\pi)$ be the tuple $(S',A)$ where $S'$ is the set of all states that are visited infinitely often in $\pi$ and $A$ is the set of all actions that are taken infinitely often in $\pi$~(cf., e.g., \cite{baierPrinciplesModel2008}).

    Let $c \in \comb$ and $q_c \in Q$.
    For $\mathcal{T} \subseteq \indexc[\mathcal{T}]$, we use $T \not\in \mathcal{T}$ as shorthand for $T \in \indexc[\mathcal{T}] \setminus \mathcal{T}$.
    Recall that for $j,j' \in \relInd(c)$ we have $\state_j = \state_{j'}$; we use $\statei$ to denote this unique state.
    
For $\sched \in \Scheds$ and $\pi \in \finPaths$ we define the \emph{residual scheduler} $\sched(\pi)$ as the scheduler behaving like $\sched$ after having seen $\pi \in \finPaths$. 

For $s, t \in \states$ we use $\pi \colon s \to t$ to denote $\pi \in \finPaths(s)$ with $\last(\pi) =t$.
    
\begin{definition}[Residual scheduler]
    Let $\mdp$ be an MDP and $\sched \in \Scheds$.
    Let $\pi \in \finPaths$.
    Then we define the \emph{residual scheduler of $\sched$ after $\pi$} as the (partial) scheduler $\sched(\pi) \in \Scheds$ with
    \[
        \sched(\pi)(\state_1 \action_1 \ldots \action_{r-1} \state_r) = \sched( \pi \action_1 \ldots \action_{r-1} \state_r)
    \]
    for $\state_1 \action_1 \ldots \action_{r-1} \state_r \in \finPaths$ with $\state_1 = \textit{last}(\pi)$, and we let $\sched(\pi)$ undefined for $\state_1 \neq \textit{last}(\pi)$.
\end{definition}

We define a correspondence between paths of $\mdp$ and $\quotT[\mdp]$ as in \cite{baierFoundationsProbabilityraising2024}.
    For $\pi \in \Paths[\mdp]$, let $\mapPaths(\pi)$ be the unique path corresponding to $\pi$ in $\quotT[\mdp]$.
    For $\pi \in \finPaths[\mdp]$, let $\mapPaths(\pi)$ be the unique path corresponding to $\pi$ in $\quotT[\mdp]$ that does not contain a state from $\states_\bot$.
    Conversely, for $\quotT[\pi] \in \Paths[\quoti]$, let $\quotT[\mapPaths](\quotT[\pi])$
    be the set of paths in $\mdp$ that correspond to $\quotT[\pi]$, such that the last action does not belong to the same MEC as the last state (i.e., if the last state of the path is in some MEC, then we must have only just entered this MEC).

\subsubsection*{Proof of \cref{th:relBuechi-to-relReach_Rightarrow}}
\label{app:relBuechi-to-relReach_Rightarrow_proof}

\relBuechiToTelReachRightarrow*
\begin{proof}
    We follow \cite[Lem.\ 2.4]{baierFoundationsProbabilityraising2024}.
    Let $\sched \in \Scheds$. 
    The idea for the construction of a witness $\quotT[\sched] \in \Scheds[\quotT]$ is as follows. 
    On $\states \setminus \states_{\MEC}$, $\quotT[\sched]$ copies $\sched$.
    On states $\state_\mec \in \states_{\MEC}$ for $\mec=(S',A) \in \MEC(\mdp)$, $\quotT[\sched]$ mimics whether $\sched$ \emph{(1)} leaves $\mec$ via some $\action \not\in A$, or \emph{(2)} stays in $\mec$ forever and visits a specific set of target sets infinitely often.
    On $\states_\bot$ we choose the unique available action.
    
    Formally, we construct $\quotT[\sched]$ as follows.
    \begin{itemize}
    \item
    For a finite path $\quotT[\pi] \in \finPaths[\quotT](\mapStates(\statei))$\footnote{Technically, we also need to define $\quotT$ for paths not starting at $\mapStates(\statei)$ here; we can just choose some action uniformly at random for these paths.} ending in $\last(\quotT[\pi]) \in \states \setminus \states_{\MEC}$ and for $\action \in \quotT[\Act](\last(\quotT[\pi]))$ we let
    \[
        \quotT[\sched](\quotT[\pi], \alpha) 
        = \frac{
            \sum_{\pi \in \quotT[\mapPaths](\quotT[\pi])} \Pr^{\mdp, \sched}_{\statei}(\pi) \cdot \sched(\pi, \alpha)
        }{
            \sum_{\pi \in \quotT[\mapPaths](\quotT[\pi])} \Pr^{\mdp, \sched}_{\statei}(\pi)
        }
        ~.
    \] 

    \item
    For a finite path $\quotT[\pi] \in \finPaths[\quotT](\mapStates(\statei))$ ending in a state $s_C$ for some MEC $\mec=(S',A)$:
    \begin{itemize}
        \item Choose $\alpha \in \Act \setminus A$ with $\Pr^\sched_{\statei}(\text{`leave } \mec \text{ via } \alpha \text{  after taking } \quotT[\pi]\text{'})$, i.e., 
        \[
            \quotT[\sched](\quotT[\pi], \alpha) 
            = 
            \frac{
                \sum_{\pi \in \quotT[\mapPaths](\quotT[\pi])} \Pr^{\mdp, \sched}_{\statei}(\pi) \cdot 
                \Pr_{\last(\pi)}^{\sched(\pi)}(\text{`leave } \mec \text{ via } \action \text{'})
            }{
                \sum_{\pi \in \quotT[\mapPaths](\quotT[\pi])} \Pr^{\mdp, \sched}_{\statei}(\pi)
            }
        \]
        where \[ 
            \Pr_{\last(\pi)}^{\sched(\pi)}(\text{`leave } \mec \text{ via } \action \text{'}) = \textstyle\sum_{\substack{\pi' \in \finPaths(\last(\pi)), \\\pi' \text{ stays in } C}} \Pr^{\mdp, \sched(\pi)}_{\last(\pi)}(\pi') \cdot \sched(\pi \circ \pi', \alpha)
            ~.
        \]
        
        \item For $\mathcal{T} \subseteq \indexc[\mathcal{T}]$ for $c \in \comb$, choose $\epsilon^{\mathcal{T}}$ with $\Pr^\sched_{\statei}($`stay in $\mec$ and see exactly all $T \in \mathcal{T}$ infinitely often after taking $\quotT[\pi]\text{'})$, i.e.,
        $\quotT[\sched](\quotT[\pi], \epsilon^{\mathcal{T}}) = $
        \[
            \frac{
                \sum_{\pi \in \quotT[\mapPaths](\quotT[\pi])} 
                \Pr^{\mdp, \sched}_{\statei}(\pi) \cdot  
                \Pr^{\mdp, \sched(\pi)}_{\last(\pi)}(\Globally \mec \wedge \bigwedge_{T \in \mathcal{T}} \Globally \Finally T \wedge \bigwedge_{T \not\in \mathcal{T}} \Finally \Globally \overline{T})
            }{
                \sum_{\pi \in \quotT[\mapPaths](\quotT[\pi])} \Pr^{\mdp, \sched}_{\statei}(\pi)
            }
            ~.
        \]
        where $\Pr^{\mdp, \sched(\pi)}_{\last(\pi)}(\Globally C \wedge \bigwedge_{T \in \mathcal{T}} \Globally \Finally T \wedge \bigwedge_{T \not\in \mathcal{T}} \Finally \Globally \overline{T}) = $
        \footnote{Note that here we use $\Globally \mec$ to denote not only that the path only visits states from $S'$ but also that it only uses actions from $A$.}
            \begin{align*} 
                \Pr^{\mdp, \sched(\pi)}_{\last(\pi)}
                \left(\bigcup_{\substack{\ec'=(S'',A') \in \textit{EC}(\mdp) .\ \ec' \subseteq \mec, \\ \forall T \in \mathcal{T} . S'' \cap T \neq \emptyset \wedge \forall T \not\in \mathcal{T} . S'' \cap T = \emptyset}} \{ \pi' \in \Paths(\last(\pi)) \mid \textit{Limit}(\pi') = \mec \}
                \right)
                ~.
            \end{align*}
        \end{itemize}
    \item
    For paths ending in a state $\bot^{\mathcal{T}} \in \states_\bot$ for some $\mathcal{T} \subseteq \indexc[\mathcal{T}]$ for some $c \in \comb$, we choose the unique action $\epsilon^{\mathcal{T}}$ (observe that $\quotT[\Act](\bot^{\mathcal{T}}) = \{ \epsilon^{\mathcal{T}}\}$).
    \end{itemize}
    We can convince ourselves that $\quotT[\sched]$ is well-defined, i.e., $\sum_{\action \in \quotT[\Act]} \quotT[\sched](\quotT[\pi], \action) = 1$ for all $\quotT[\pi] \in \finPaths[\quotT](\mapStates(\statei))$.

    By induction on the length of $\quotT[\pi]$, we can show that for all $\quotT[\pi] \in \finPaths[\quotT](\mapStates(\statei))$ with $\last(\quotT[\pi]) \not \in \states_\bot$ we have
    $\Pr^{\quotT, \quotT[\sched]}_{\mapStates(\statei)}(\quotT[\pi]) = \sum_{\pi \in \quotT[\mapPaths](\quotT[\pi])} \Pr^{\mdp, \sched}_{\statei}(\pi)$. 

    Thus, for each $j \in \relInd(c)$ we have $\Pr^{\mdp, \sched}_{\statei}(\Globally \Finally T_j) = $
    \begin{align*}
        &\sum_{\substack{\mec=(S',A) \in \MEC(\mdp),\, \\S' \cap T_j \neq \emptyset}} \;
        \sum_{\quotT[\pi] \colon \mapStates(\statei) \to \state_\mec}
        \;
        \sum_{\pi \in \quotT[\mapPaths](\quotT[\pi])} \Pr^{\mdp, \sched}_{\statei}(\pi) \cdot 
        \Pr^{\mdp, \sched(\pi)}_{\last(\pi)}(\Globally \mec \wedge \Globally \Finally T_j)
        \\ = &\sum_{\substack{\mec=(S',A) \in \MEC(\mdp),\, \\S' \cap T_j \neq \emptyset}} \;
        \sum_{\quotT[\pi] \colon \mapStates(\statei) \to \state_\mec}
        \;
        \\ & \quad
        \sum_{\pi \in \quotT[\mapPaths](\quotT[\pi])} 
        \Pr^{\mdp, \sched}_{\statei}(\pi) \cdot 
        \sum_{\mathcal{T} \subseteq \mathcal{T}_c, c \in \comb, T_j \in \mathcal{T}}
        \Pr^{\mdp, \sched(\pi)}_{\last(\pi)}(\Globally \mec \wedge \bigwedge_{T \in {\mathcal{T}}} \Globally \Finally T \wedge \bigwedge_{T \not\in \mathcal{T}} \Finally \Globally \overline{T})
        \\ = &\sum_{\substack{\mec=(S',A) \in \MEC(\mdp),\, \\S' \cap T_j \neq \emptyset}} \;
        \sum_{\quotT[\pi] \colon \mapStates(\statei) \to \state_\mec}
        \;
        \\ & \quad
        \sum_{\mathcal{T} \subseteq \mathcal{T}_c, c \in \comb, T_j \in \mathcal{T}}
        \sum_{\pi \in \quotT[\mapPaths](\quotT[\pi])} 
        \Pr^{\mdp, \sched}_{\statei}(\pi) \cdot 
        \Pr^{\mdp, \sched(\pi)}_{\last(\pi)}(\Globally \mec \wedge \bigwedge_{T \in {\mathcal{T}}} \Globally \Finally T \wedge \bigwedge_{T \not\in \mathcal{T}} \Finally \Globally \overline{T})
        ~.
    \end{align*}
    
    On the other hand, we have $\Pr^{\quotT, \quotT[\sched]}_{\mapStates(\statei)}(\Finally U_{T_j}) = $
    \begin{align*}
        & 
        \sum_{\mathcal{T} \subseteq \mathcal{T}_c, c \in \comb, T_j \in \mathcal{T}}
        \Pr^{\quotT, \quotT[\sched]}_{\mapStates(\statei)}(\Finally \bot^{\mathcal{T}})
        \\ = &\sum_{\substack{\mec=(S',A) \in \MEC(\mdp), \\ S' \cap T_j \neq \emptyset}} \;
        \sum_{\quotT[\pi] \colon \mapStates(\statei) \to \state_\mec}
        \Pr^{\quotT, \quotT[\sched]}_{\mapStates(\statei)}(\quotT[\pi])
        \sum_{\mathcal{T} \subseteq \mathcal{T}_c, c \in \comb, T_j \in \mathcal{T}} \quotT[\sched](\quotT[\pi], \epsilon^{\mathcal{T}})
        \\ = &\sum_{\substack{\mec=(S',A) \in \MEC(\mdp),\, \\ S' \cap T_j \neq \emptyset}} \;
        \sum_{\quotT[\pi] \colon \mapStates(\statei) \to \state_\mec}
        \;
        \\ & \quad
        \Pr^{\quotT, \quotT[\sched]}_{\mapStates(\statei)}(\quotT[\pi]) \cdot 
        \hspace{-3ex}\sum_{\mathcal{T} \subseteq \mathcal{T}_c, c \in \comb, T_j \in \mathcal{T}}\hspace{-4ex}
        \frac{
            \sum_{\pi \in \quotT[\mapPaths](\quotT[\pi])} 
            \Pr^{\mdp, \sched}_{\statei}(\pi) \cdot 
            \Pr^{\mdp, \sched(\pi)}_{\last(\pi)}(\Globally \mec \wedge \bigwedge_{T \in \mathcal{T}} \Globally \Finally T \wedge \bigwedge_{T\not\in \mathcal{T}} \Finally \Globally \overline{T})
        }{
            \sum_{\pi \in \quotT[\mapPaths](\quotT[\pi])} 
            \Pr^{\mdp, \sched}_{\statei}(\pi) 
        }
        \\ = &\sum_{\substack{\mec=(S',A) \in \MEC(\mdp),\, \\ S' \cap T_j \neq \emptyset}} \;
        \sum_{\quotT[\pi] \colon \mapStates(\statei) \to \state_\mec}
        \;
        \\ & \quad
        \sum_{\mathcal{T} \subseteq \mathcal{T}_c, c \in \comb, T_j \in \mathcal{T}}
        \sum_{\pi \in \quotT[\mapPaths](\quotT[\pi])} 
        \Pr^{\mdp, \sched}_{\statei}(\pi) \cdot 
        \Pr^{\mdp, \sched(\pi)}_{\last(\pi)}(\Globally \mec \wedge \bigwedge_{T \in \mathcal{T}} \Globally \Finally T \wedge \bigwedge_{T \not\in \mathcal{T}}  \Finally \Globally \overline{T})
    \end{align*}
    which is equivalent to $\Pr^{\mdp, \sched}_{\statei}(\Globally \Finally T_j)$ by our previous reasoning.
    The last equality holds due to the correspondence between the probability of a finite path $\quotT[\pi]$ in $\quotT$ and the probability of the corresponding set of finite paths $\quotT[f](\quotT[\pi])$ in $\mdp$.
\end{proof}

\subsubsection*{Proof of \cref{th:relBuechi-to-relReach_Leftarrow}} 
\label{app:relBuechi-to-relReach_Leftarrow_proof}

\relBuechiToTelReachLeftarrow*

We first show the claim for the case that the given scheduler for the quotient $\quotT$ is memoryless in \cref{th:relBuechi-to-relReach_Leftarrow_MR} and then prove that this suffices to show the desired claim. 

For the scheduler witness transfer from $\quotT$ to $\mdp$, it is convenient to use a scheduler representation with explicit memory modes
instead of the path-based version introduced in \cref{sec:prelim}, we schematically illustrate the new definition in~\cref{fig:sched_mem_illustration}.

\begin{definition}[Explicit-memory scheduler representation]
    \label{def:scheduler_explicit}
    Let $\mdp = \mdptup$ be an MDP.
    An \emph{scheduler $\sched$ for $\mdp$ in explicit-memory representation} is a tuple $(\modes, \start, \modef, \act)$ with
    \begin{itemize}
        \item $\modes$ a countable set of memory modes,

        \item $\start \in \Distr(\modes)$ a stochastic initial mode selection function,

        \item $\modef \colon \modes \times \Act \times \states \to \Distr(\modes)$ a stochastic memory update function, where $\modef(\mode, \action, \state')(\mode')$ gives the probability of updating the memory to $\mode' \in \modes$ if the current mode is $\mode \in \modes$, we chose action $\action$ in the current state and moved to $\state'$,
        and

        \item $\act \colon \modes \times \states \to \Distr(\Act)$ a stochastic action selection function, where $\act(\mode, \state)(\action)$ is the probability of choosing $\action$ in state $\state$ and mode $\mode$. 
    \end{itemize}
\end{definition}

\begin{figure}
    \centering
    \begin{tikzpicture}[on grid]
        \node (s0) {$\state$};
        \node[right=of s0] (alpha) {$\action$};
        \node[right=of alpha] (s1) {$\state'$};
        \node[below=of s0] (m0) {$\mode$};
        \node[below=of s1] (m1) {$\mode'$};
        \node[left=of m0] (init) {};

        \path[->, magenta] (init) edge node[above] {$\start$} (m0);
        \path[->, blue] 
            (m0) edge (alpha)
            (s0) edge node[below, xshift=-1mm] {$\act$} (alpha); 
        \path[->, green]
            (m0) edge (m1)
            (alpha) edge node[below, xshift=-4mm] {$\modef$} (m1)
            (s1) edge (m1);
        \path[->, gray]
            (s0) edge[bend left] node[above, pos=0.9] {$\Trans$} (s1)
            (alpha) edge (s1);
    \end{tikzpicture}
    \caption{Illustration of interaction of \textcolor{magenta}{$\start$}, \textcolor{green}{$\modef$}, \textcolor{blue}{$\act$}, and \textcolor{gray}{$\Trans$} for explicit-memory scheduler representation~(\cref{def:scheduler_explicit})}
    \label{fig:sched_mem_illustration}
\end{figure}
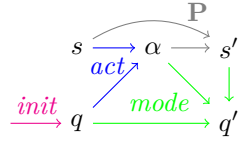

There is a direct correspondence between the explicit-memory and path-based scheduler representations for some given initial state $\statei$.
\begin{lemma}
    Let $\mdp = \mdptup$ be an MDP and $\statei \in \states$.
    For every $\sched = (\modes, \start, \modef, \act)$ there exists $\sched' \colon \finPaths \to \Distr(\Act)$ such that
    \[
        \Pr^{\mdp, \sched}_{\statei}(\pi) = \Pr^{\mdp, \sched'}_{\statei}(\pi)
    \]
    for any finite path $\pi \in \finPaths(\statei)$,
    and vice versa.
\end{lemma}
\begin{proof}
    `$\Rightarrow$': 
    Given $\sched = (\modes, \start, \modef, \act)$, we recursively define
    \begin{itemize}
        \item for $\state \in \states$, $\mode \in \modes$, $\action \in \Act(\state)$, $\pi \beta \state' \in \finPaths$: 
        \[ 
            \modef(\state)(\mode) = \start(\mode)
            \qquad 
            \modef(\pi \beta \state')(\mode) = \sum_{\mode' \in \modes} \modef(\mode', \state', \beta)(\mode) \cdot \modef(\pi)(\mode')
        \]
    
        \item for $\pi \in \finPaths$, $\action \in \Act(\state)$:
        \[
            \act(\pi)(\action) = \sum_{\mode \in \modes} \act(\mode, \last(\pi))(\action) \cdot \modef(\pi)(\mode)
        \]
    \end{itemize}
    We define $\sched' \colon \finPaths \to \Distr(\Act)$ with $\sched'(\pi) = \act(\pi)$ for $\pi \in \finPaths$.
    
    `$\Leftarrow$': 
    Given a scheduler $\sched' \colon \finPaths \to \Distr(\Act)$ we construct $\sched = (\modes, \start, \modef, \act)$ with 
    \begin{itemize}
        \item $\modes = \finPaths(\statei)$
        
        \item $\start(\statei) = 1$ and $\start(\mode) = 0$ for $\mode \in \modes \setminus \{\statei\}$
        
        \item $\modef( \statei \action_0 \ldots \state_n, \action, \state')(\mode) = \Iverson{q = \statei \action_0 \ldots \state_n}$
        for $\statei \action_0 \ldots \state_n \in \finPaths(\statei)$, $\action \in \Act(\state_n)$, $\state' \in \states$ and $\mode \in \modes$

        \item $\act(\statei \action_0 \ldots \state_n, \state_n) = \sched'(\statei \action_0 \ldots \state_n)$ for $\statei \action_0 \ldots \state_n \in \finPaths(\statei)$ \footnote{Technically, we would also need to define $\act(\statei \action_0 \ldots \state_n, \mode)$ for $\mode \neq \state_n$; we can choose any arbitrary distribution.}
        \qedhere
    \end{itemize}
\end{proof}
Note that the transfer from explicit-memory to path-based representation does not depend on a given initial state, but the other direction does.

In the following, for an explicit-memory scheduler $\sched$ we may sometimes use $\sched(\pi)$ to denote $\sched'(\pi)$ for the induced path-based scheduler $\sched'$ from the construction above under abuse of notation.

\relBuechiToTelReachLeftarrowMR*

\begin{proof}
    Let $\quotT[\sched] \in \SchedsMR[\quotT]$.
    We construct a scheduler $\sched \in \Scheds$ analogously to \cite[pp.\ 10--11]{baierCertificatesWitnesses2024}.

    Recall that $\state_j = \state_{j'}$ for $j,j' \in \relInd(c)$.

    Intuitively, when entering an MEC $\mec$ we flip a coin and either switch to some scheduler staying in $\mec$ forever, or to some scheduler leaving $\mec$ via the action chosen by $\quotT[\sched]$.
    More precisely, for $\mathcal{T} \subseteq \indexc[\mathcal{T}]$ for some $c \in \comb$, with probability $\quotT[\sched](\indexc[\state])(\epsilon^{\mathcal{T}})$ we switch to a memoryless scheduler going to a sub-EC $\ec'$ of $\mec$ with $\ec' \in \EC_{\mathcal{T}}$.
    With the remaining probability $1- \sum_{\mathcal{T} \subseteq \indexc[\mathcal{T}], c \in \comb} \quotT[\sched](\indexc[\state])(\epsilon^{\mathcal{T}})$, we switch to a memoryless scheduler that leaves $\mec$ almost-surely, while mimicking the probabilities with which $\quotT[\sched]$ chooses each outgoing action of $\mec$.
    The construction of the latter is technically involved and we refer to~\cite{baierCertificatesWitnesses2024} for details on the formal construction.
    
    We define $\sched$ as a tuple $(\modes, \start, \modef, \act)$ where
    \begin{itemize}
        \item $\modes = \{\mode_0\} \cup \{ \mode_{\mathcal{T}} \mid \mathcal{T} \subseteq \mathcal{T}_c, c \in \comb \}$.

        In the following, we use $\mu_\mec \in \Distr(\modes)$ to denote the distribution defined by $\mu_\mec(\mode_0) = 1 - \sum_{\mathcal{T} \subseteq \mathcal{T}_c, c \in \comb} \quotT[\sched](\state_\mec)(\epsilon^{\mathcal{T}})$ and $\mu_\mec(\mode_{\mathcal{T}}) = \quotT[\sched](\state_\mec)(\epsilon^{\mathcal{T}})$ for $\mathcal{T} \subseteq \mathcal{T}_c, c \in \comb$.

        \item $\start \colon \modes \to [0,1]$ defined as follows.
        If there exists some $\mec_0 \in \MEC(\mdp)$ containing $\statei$ then 
        $\start = \mu_{\mec_0}$.
        Otherwise $\start(\mode) = \Iverson{\mode = \mode_0}$ for $\mode \in \modes$.

        \item $\modef \colon \modes \times \Act \times \states \to \Distr(\modes)$ with
        \[  
            \modef(\mode, \action, \state) = 
            \begin{cases}
                \mu_\mec & \text{if } \exists \mec=(S',A) \in \MEC(\mdp) .\ s \in S' \wedge \action \not\in A  \\
                (\mode \mapsto 1) & \text{otherwise}~.
            \end{cases}
        \]

        Intuitively, we flip a coin when entering an MEC to go to $\mode_0$ or some $\mode_I$.
        Inside MECs and on states outside MECs, we do not update the memory.

        \item $\act \colon \states \times \modes \to \Distr(\Act)$ defined as follows:
        \begin{itemize}
            \item For $\state \in \states \setminus \states_{\MEC}$ and $\mode \in \modes$:
            \[
                \act(\state,\mode) = \quotT[\sched](\state)
            \]
            
            \item For $\state \in \states$ where there exists some $\mec=(S',A) \in \MEC(\mdp)$ with $\state \in S'$:
            \begin{align*}
            \act(\state,\mode_0) &=  \sched_{\textsf{leave},\mec}(\state)  \\
            \act(\state, \mode_{\mathcal{T}}) &=
                \begin{cases}
                    \sched_{\mec,\mathcal{T}}(\state) &\text{if } \mec \in \MEC_{\mathcal{T}} \\
                    0 &\text{otherwise}
                \end{cases}
                \text{ for } \mathcal{T} \subseteq \mathcal{T}_c, c \in \comb
            \end{align*}
        \end{itemize}

        where for $\mec=(S',A) \in \MEC(\mdp)$:
        \begin{itemize}
            \item 
            $\sched_{\textsf{leave},\mec}$ leaves $\mec$ almost surely while mimicking the probabilities with which $\quotT[\sched]$ chooses each outgoing action of $\mec$, i.e., $\action \in \quotT[\Act](\indexc[\state])$. 
            Formally, let $p_\mec = 1 - \sum_{I\subseteq[m]} \quotT[\sched](\indexc[\state])(\epsilon^I)$ for $\mec \in \MEC(\mdp)$, then the scheduler satisfies the following constraint for $\quotT[\pi] = \quotT[\pi']\beta\state_\mec \in \finPaths[\quotT](\mapStates(\statei))$, $\state \in S'$ and $\action \in \Act(\state) \setminus A$:
            \begin{equation}
                \left[
                \sum_{\pi \in \quotT[f](\quotT[\pi])} \hspace{-1ex}
                \Pr_{\statei}^{\sched}(\pi) 
                \hspace{-3ex} \sum_{\pi' \colon \last(\pi) \xrightarrow[]{\mec} \state} \hspace{-4ex}
                \Pr_{\last(\pi)}^{\sched(\pi)}(\pi') 
                \right]
                \sched_{\textsf{leave},\mec}(\state)(\action) p_\mec 
                = 
                \left[
                \sum_{\pi \in \quotT[f](\quotT[\pi])} \hspace{-1ex}
                \Pr_{\statei}^{\sched}(\pi) 
                \right]
                \quotT[\sched](\state_\mec)(\action)
                \label{eq:leave-C-assurance}
            \end{equation}
            where we use $\pi' \colon \last(\pi) \xrightarrow[]{\mec} \state$ as shorthand for $\pi \in \{ \state_0 \action_0 \ldots \state_n \in \finPaths(\last(\pi)) \mid \state_n = \state \wedge \forall i < n .\ \state_i \in S' \wedge \action_i \in A \}$.
            For a formal construction of $\sched_{\textsf{leave},\mec}$ we refer to \cite[p.\ 10--11]{baierCertificatesWitnesses2024} with $p_\mec = 1 - \sum_{I\subseteq[m]} \quotT[\sched](\indexc[\state])(\epsilon^I)$.

            \item $\sched_{\mec,\mathcal{T}}$ ensures that we stay in $\mec$ and see exactly the target sets $T \in \mathcal{T}$ infinitely often.
            Such a scheduler must exist with the following reasoning: 
            By construction, for each $\mathcal{T} \subseteq \mathcal{T}_c$ with $\mec \in \MEC_{\mathcal{T}}$, there must exist some sub-EC $\ec'$ of $\mec$ that (1) contains at least state from each $T \in \mathcal{T}$ and (2) contains no state from any $T \in \indexc[\mathcal{T}] \setminus \mathcal{T}$.
            Hence, there must exist some memoryless deterministic scheduler $\sched_{\mec,\mathcal{T}}$ on $\mdp$ restricted to $\mec$ that almost-surely goes to $\ec'$, and stays there forever. 
        \end{itemize}
    \end{itemize}

    Observe that $\mode_0$ and $\mode_\emptyset$ have different purposes: In $\mode_0$, we aim to leave $\mec$ while in $\mode_\emptyset$ we aim to stay in $\mec$ without seeing any target set infinitely often.

    We can convince ourselves that $\sched$ is well-defined.

    We can show that for all $\quoti[\pi] \in \finPaths[\quoti](\mapStates(\statei))$ with $\last(\quoti[\pi]) \not \in \states_\bot$ it holds that
    $\Pr^{\quoti, \quoti[\sched]}_{\mapStates(\statei)}(\quoti[\pi]) = \sum_{\pi \in \quoti[\mapPaths](\quoti[\pi])} \Pr^{\mdp, \sched}_{\statei}(\pi)$, by induction on the length of $\quoti[\pi]$, using~\eqref{eq:leave-C-assurance}. 
    
    Finally, let us show that $\sched$ satisfies the desired claim.
    Let $j \in \relInd(c)$, then we can rewrite $\Pr_{\statei}^{\mdp,\sched}(\Globally \Finally T_j)$ as follows (see, e.g.,~\cite[Sec.~10.6.3]{baierPrinciplesModel2008}):
    \begin{align*}
        &\Pr_{\statei}^{\mdp,\sched}(\Globally \Finally T_j)
        \\ = & 
        \sum_{\substack{\mec=(S',A) \in \MEC,\\ S' \cap T_j \neq \emptyset}} 
        \Pr_{\statei}^{\mdp,\sched}(\{ \pi \in \Paths(\statei) \mid \Limit(\pi) \subseteq \mec \wedge \Limit(\pi) \cap T_j \neq \emptyset \})
        \\ = & 
        \sum_{\substack{\mathcal{T} \subseteq \mathcal{T}_c, c \in \comb \\ T_j \in \mathcal{T}}} \; 
        \sum_{\mec \in \MEC_{\mathcal{T}}} \\ &\qquad  
        \Pr_{\statei}^{\mdp,\sched}\left( \left\{ \pi \in \Paths(\statei) \;\Bigg|\; \begin{aligned} &\Limit(\pi) \subseteq \mec \wedge \\ &\bigwedge_{T \in \mathcal{T}} \Limit(\pi) \cap T \neq \emptyset \wedge \bigwedge_{T \not\in \mathcal{T}} \Limit(\pi) \cap T = \emptyset \end{aligned} \right\} \right)
        ~.
    \end{align*}
    Recall that $\sched_{\mec,\mathcal{T}}$ ensures that `$\Globally \mec \wedge \bigwedge_{T \in \mathcal{T}} \Globally \Finally T \wedge \bigwedge_{T \not\in \mathcal{T}}  \Finally \Globally \overline{T}$' holds almost-surely. No other $\sched_{\mec,\mathcal{T}'}$ ensures this, and $\sched_{\textsf{leave},\mec}$ leaves $\mec$ almost-surely. Hence, the above property is satisfied if and only if $\sched$ follows $\sched_{\mec,\mathcal{T}}$ after $\pi$, which is the case if and only if $\sched$ moves to mode $\mode_{\mathcal{T}}$ after $\pi$. Hence, we can rewrite the above further to
    \begin{align*}
        & \Pr_{\statei}^{\mdp,\sched}(\Globally \Finally T_j)
        = 
        \sum_{\substack{\mathcal{T} \subseteq \mathcal{T}_c, c \in \comb \\ T_j \in \mathcal{T}}} 
        \sum_{\mec \in \MEC_{\mathcal{T}}}
        \sum_{\quotT[\pi] \colon \mapStates(\statei) \to \state_\mec}
        \sum_{\pi \in \quotT[f](\quotT[\pi])} \Pr_{\statei}^{\mdp,\sched}(\pi) 
        \cdot 
        \modef(\pi)(\mode_{\mathcal{T}})
        ~.
    \end{align*}
    We now observe that $\modef(\pi)(\mode_{\mathcal{T}}) = \quotT[\sched](\state_\mec)(\epsilon^{\mathcal{T}})$ holds for all $\pi \in \quotT[f](\quotT[\pi])$ for $\quotT[\pi] \colon \mapStates(\statei) \to \state_\mec$ for $\mec \in \MEC_{\mathcal{T}}$ for $\mathcal{T} \subseteq \mathcal{T}_c$ for $c \in \comb$ since 
    (1) for $\pi = \statei$ we must have $\statei \in \mec$ and thus $\modef(\statei)(\mode_{\mathcal{T}}) = \start(\mode_{\mathcal{T}}) = \quotT[\sched](\state_\mec)(\epsilon^{\mathcal{T}})$ and 
    (2) for $\pi = \pi' \beta \state'$ we have $s' \in T$ but $\beta \not\in A$ and thus $\modef(\mode', \beta, \state') = \quotT[\sched](\state_\mec)(\epsilon^{\mathcal{T}})$ for all $\mode' \in \modes$ and hence $\modef(\pi' \beta \state')(\mode_{\mathcal{T}}) = $
        \[ 
        \sum_{\mode' \in \modes} \modef(\mode', \state', \beta)(\mode_{\mathcal{T}}) \cdot \modef(\pi')(\mode') = \quotT[\sched](\state_\mec)(\epsilon^{\mathcal{T}}) \cdot \sum_{\mode' \in \modes} \modef(\pi')(\mode') = \quotT[\sched](\state_\mec)(\epsilon^{\mathcal{T}})
        ~.
    \]
    Thus,
    \begin{align*}
        \Pr_{\statei}^{\mdp,\sched}(\Globally \Finally T_j)
        = & 
        \sum_{\substack{\mathcal{T} \subseteq \mathcal{T}_c, c \in \comb \\ T_j \in \mathcal{T}}} \;
        \sum_{\mec \in \MEC_{\mathcal{T}}}
        \sum_{\quotT[\pi] \colon \mapStates(\statei) \to \state_\mec}
        \sum_{\pi \in \quotT[f](\quotT[\pi])} \Pr_{\statei}^{\mdp,\sched}(\pi) 
        \cdot 
        \quotT[\sched](\state_\mec)(\epsilon^{\mathcal{T}})
        \\ = &
        \sum_{\substack{\mathcal{T} \subseteq \mathcal{T}_c, c \in \comb \\ T_j \in \mathcal{T}}} \; 
        \sum_{\mec \in \MEC_{\mathcal{T}}}
        \sum_{\quotT[\pi] \colon \mapStates(\statei) \to \state_\mec}
        \Pr_{\mapStates(\statei)}^{\quotT,\quotT[\sched]}(\pi)
        \cdot 
        \quotT[\sched](\state_\mec)(\epsilon^{\mathcal{T}})
        \\ = &
        \sum_{\substack{\mec=(S',A) \in \MEC,\\ S' \cap T_j \neq \emptyset}} \;
        \sum_{\quotT[\pi] \colon \mapStates(\statei) \to \state_\mec}
        \Pr_{\mapStates(\statei)}^{\quotT,\quotT[\sched]}(\quotT[\pi]) 
        \sum_{\substack{\mathcal{T} \subseteq \mathcal{T}_c, c \in \comb \\ T_j \in \mathcal{T}}} 
        \quotT[\sched](\state_\mec)(\epsilon^{\mathcal{T}})
        \\ = &
        \sum_{\substack{\mathcal{T} \subseteq \mathcal{T}_c, c \in \comb \\ T_j \in \mathcal{T}}} 
        \Pr_{\mapStates(\statei)}^{\quotT,\quotT[\sched]}(\Finally \bot^{\mathcal{T}}) 
        =
        \Pr_{\mapStates(\statei)}^{\quotT, \quotT[\sched]}(\Finally U_{T_j})
        ~.
    \end{align*} 
\end{proof}

In order to show \cref{th:relBuechi-to-relReach_Leftarrow} using \cref{th:relBuechi-to-relReach_Leftarrow_MR}, we define the \emph{flip-extension} of an MDP $\mdp$ with some initial state $\state_{0} \in \states$, which encodes the memory structure of schedulers that are the convex combination of two schedulers for $\mdp$ into the MDP.
    
    \begin{definition}
        Given an MDP $\mdp=\mdptup$ and $\state_{0} \in \states$, the \emph{flip-extension} of $\mdp$ w.r.t.\ $\state_0$ is defined as $\flip[\mdp] = \mdptup[\flip]$ with 
        \begin{itemize}
            \item $\flip[\states] = (\states \times \{1\}) \cup (\{\states\} \times \{2\}) \cup \{\state_{\textit{flip}}\}$,
            \item $\flip[\Act] = \Act \cup \{1, 2\}$, and
            \item $\flip[\Trans] \colon \flip[\states] \times \flip[\Act] \times \flip[\states] \to [0,1]$ is defined by case distinction: \\
            $\flip[\Trans](\state_{\textit{flip}}, i, (\state_0, i))=1$ for $i \in \{1,2\}$, 
            \\
            $\flip[\Trans]((\state, i), \action, (\state',i)) = \Trans(\state, \action, \state')$ for $i \in \{1,2\}$ and $\state, \state' \in \states$ and $\action \in \Act$,
            \\ and all other transition probabilities are 0.
        \end{itemize}
    \end{definition}
    
Hence, every scheduler for $\flip[\mdp]$ can directly be interpreted as the convex combination of two schedulers for $\mdp$, and vice versa.
In particular, memoryless randomized schedulers suffice on $\flip[\mdp]$ to achieve any value for a weighted sum of reachability probabilities that is achievable on $\mdp$:

\label{app:flip-ext_MR_proof}
\begin{restatable}{lemma}{flipExtMR}
        \label{le:flip-ext_MR}
        Let $c \in \comb$ and $q_c \in \Q$, then
        \begin{align*}
            & 
            \exists \quotT[\sched] \in \Scheds[\quotT] .\ 
            \sum_{j \in \relInd(c)}
            q_j \Pr^{\quotT, \quotT[\sched]}_{\mapStates(\indexc[\state])}\left(\Finally U_{T_j} \right) = q_c
            \\ \Iff\quad &
            \exists \quotT[\sched]' \in \SchedsMR[{\flip[(\quotT)][{\mapStates(\indexc[\state])}]}] .\ 
            \sum_{j \in \relInd(c)}
            q_j \Pr^{\flip[(\quotT)][{\mapStates(\indexc[\state])}], \quotT[\sched]'}_{\sflip}\left(\Finally U_{T_j} \right) = q_c
        \end{align*}
    \end{restatable}
\begin{proof}
    First observe that for $i=1,2$, any scheduler for the $\states \times \{i\}$-part of $(\quotT)_{\textit{flip}@\mapStates(\state_j)}$ can be interpreted as a scheduler for $\quotT[\mdp]$, and vice versa.
        
    `$\Leftarrow$': 
    Let $\sched_1, \sched_2 \in \Scheds[\quotT]$ be the schedulers induced by $\quotT[\sched]' \in \SchedsMR[{\flip[(\quotT)][{\mapStates(\state_j)}]}]$ on the $1$- and $2$-copy of $\mdp$.
    We define $\quotT[\sched] \in \Scheds[\quotT]$ as the convex combination of $\sched_1$ and $\sched_2$ w.r.t.\ $\lambda = \quotT[\sched]'(\sflip, 1)$.
    
    `$\Rightarrow$': 
    By construction, the sets $U_{T_j}$ contain only absorbing states, i.e., $\quotT$ is already a goal unfolding w.r.t.\ these sets. 
    Following the construction from \cref{sec:reach_algo}, we can thus reduce the \RelReach query $\exists \sched' \in \Scheds[\quotT] .\ \sum_{j \in \relInd(c)} q_j \Pr^{\quotT, \sched'}_{\mapStates(\statei)}(\Finally U_{T_j}) = x$ to a single-objective expected reward query on $\quotT$ itself such that the expected reward under some scheduler equals exactly the value of the weighted sum of reachability probabilities under that scheduler. 
    There exist memoryless deterministic schedulers minimizing and maximizing the expected reward, respectively (see, e.g., \cite[Th.\ 7.1.9]{putermanMarkovDecision1994}). 
    Since $x$ is achievable (via witness $\quotT[\sched]$), there must exist some $\lambda \in [0,1]$ such that the convex combination of the minimizing and the maximizing scheduler w.r.t.\ $\lambda$ achieves exactly expected reward $x$.
    We construct $\quotT[\sched]' \in \SchedsMR[{\flip[(\quotT)][{\mapStates(\state_j)}]}]$
    by choosing action $1$ in $\sflip$ with probability $\lambda$, behaving like the maximizing scheduler in the $1$-copy of $\quotT$ and like the minimizing scheduler in the $2$-copy of $\quotT$.
\end{proof}

We are now ready to prove \cref{th:relBuechi-to-relReach_Leftarrow}.

\begin{proof}[Proof of \NoCaseChange\cref{th:relBuechi-to-relReach_Leftarrow}]
    We reduce both sides of the implication from \cref{th:relBuechi-to-relReach_Leftarrow_MR} to reasoning about schedulers for the flip-extensions of the respective MDPs, which then allows us to conclude the claim by applying \cref{th:relBuechi-to-relReach_Leftarrow_MR}.
    
    Firstly, observe that the sets $U_{T_j}$ are absorbing and hence any value achievable for the weighted sum $\sum_{j \in \relInd(c)} q_j \Pr^{\quotT, \quotT[\sched]}_{\mapStates(\indexc[\state])}\left(\Finally U_{T_j} \right)$ can be achieved by the convex combination of two MD schedulers (see proof of \cref{thm:cruxApproxEqual}).
    Thus, any value achievable for the weighted sum of reachability probabilities on the quotient can be achieved by a \emph{memoryless} randomized scheduler on the flip-extension of the quotient, and vice versa.
    For $c \in \comb$, let $\indexc[\state]$ be the unique state in $\mdp$ with $\state_j = \indexc[\state]$ for all $j \in \relInd(c)$.

    Let $c \in \comb$ and $q_c \in \Q$, then by \cref{le:flip-ext_MR} it holds that
        \begin{align*}
            & 
            \exists \quotT[\sched] \in \Scheds[\quotT] .\ 
            \sum_{j \in \relInd(c)}
            q_j \Pr^{\quotT, \quotT[\sched]}_{\mapStates(\indexc[\state])}\left(\Finally U_{T_j} \right) = q_c
            \tag{1}
            \\ \Iff\quad &
            \exists \quotT[\sched]' \in \SchedsMR[{\flip[(\quotT)][{\mapStates(\indexc[\state])}]}] .\ 
            \sum_{j \in \relInd(c)}
            q_j \Pr^{\flip[(\quotT)][{\mapStates(\indexc[\state])}], \quotT[\sched]'}_{\sflip}\left(\Finally U_{T_j} \right) = q_c
            ~.
            \tag{1'}
            \label{eq:1p_app}
        \end{align*}
    Observe further that 
    \begin{align*}
        & 
        \exists \sched' \in \Scheds[{\flip[\mdp][{\indexc[\state]}]}] .\ 
        \sum_{j \in \relInd(c)}
        q_j \Pr^{\flip[\mdp][{\indexc[\state]}], \sched'}_{\sflip}\left(\Globally \Finally  T_j \right) = q_c
        \tag{2'}
        \label{eq:2p_app}
        \\ \Iff\quad &
        \exists \sched \in \Scheds[\mdp] .\ 
        \sum_{j \in \relInd(c)}
        q_j \Pr^{\mdp, \sched}_{\indexc[\state]}\left(\Globally \Finally  T_j \right) = q_c
        ~.
        \tag{2}
    \end{align*}

    The claim `$\eqref{eq:1p_app} \Implies \eqref{eq:2p_app}$' is an instance of \cref{th:relBuechi-to-relReach_Leftarrow_MR} with $\mdp=\flip[\mdp][{\indexc[\state]}]$ and $\quotT = \quotT[{(\flip[\mdp][{\indexc[\state]}])}]$, since $\flip[(\quotT)][{\mapStates(\indexc[\state])}] = \quotT[{(\flip[\mdp][{\indexc[\state]}])}]$.
    Hence, we have shown the desired claim `$\eqref{eq:1} \Implies \eqref{eq:2}$'. 
\end{proof}

%% file: appendix/buechi_compl_proofs.tex
\section{Proofs for \cref{sec:buechi_complexity}}
\label{app:buechi_compl_proofs}

\subsection{Proof of \cref{th:buechi_NP-complete}}
\label{app:buechi_NP-complete}

\buechiNPComplete*

\begin{proof}
    We first show membership in \NP and then strong \NP-hardness.
    
    \proofsubparagraph{Membership in \NP:}
    Given some MDP $\mdp$ and \RelBuechi property $\phi$, let $\phi'$ be the constructed \RelReach property.
    We will show that it suffices to guess a polynomial number of MD schedulers in order to guess a witness for the query.
    
    The key observation is that all target sets $U_{T_i}$ of $\phi'$ on $\quotT$ are absorbing and therefore for each $c \in \comb$, the goal unfolding of $\quotT$ w.r.t.\ the target sets for $c$, $\indexc[\mathcal{T}]' = \{U_{T_i} \mid i \in \relInd(c)\}$, corresponds to $\quotT$ with an additional copy of each sink in some $U_{T_i}$ for some $i \in \relInd(c)$.
    An MD scheduler for $\indexc[(\quotT)]$ can thus be translated to an MD scheduler for $\quotT$ while preserving the reachability probabilities.
    By \cref{th:relBuechi-to-relReach_Leftarrow_MD}, an MD witness for $\phi'$ on $\quotT$ can be translated to an MD witness for $\phi$ on $\mdp$.
    
    We distinguish by comparison operator:
    \begin{itemize}
        \item 
        Assume $\comp$ is $\geq$ or $>$. 
        Let us first assume there is only a single state-scheduler combination $\{c\} = \comb$.
        By \cref{th:relBuechi-to-relReach} and \cref{thm:correctness}, solving $\phi$ reduces to maximizing some expected reward on $\indexc[(\quotT)]$, for which memoryless deterministic schedulers suffice.
        By the above reasoning, an MD witness on $\indexc[(\quotT)]$ can be translated to an MD witness on $\mdp$.
        
        For $\geq,>$-queries with $|\comb|>1$, we can thus guess witnesses by guessing an MD witness for each state-scheduler combination. 

        \item 
        Assume $\comp$ is $\approx_\epsilon$. 
        Let us first assume there is only a single state-scheduler combination $\{c\} = \comb$.
        By \cref{th:relBuechi-to-relReach} and \cref{thm:cruxApproxEqual}, solving $\phi$ reduces to both maximizing and minimizing some expected reward on $\indexc[(\quotT)]$, and checking whether the given bound $q$ has distance at least $\epsilon$ to the minimal and maximal achievable values.
        With analogous reasoning to $\geq,>$, we can translate the maximizing and the minimizing schedulers to MD schedulers on $\mdp$.
        We can thus guess a witness for $\phi$ by guessing an MD scheduler maximizing the expected reward and guessing an MD scheduler minimizing the expected reward.

        For $\approx_\epsilon$ with $|\comb|>1$, we can thus guess witnesses by guessing two MD witnesses for each state-scheduler combination.

        \item 
        Assume $\comp$ is $\not\approx_\epsilon$; recall that $\phi$ corresponds to a disjunction of two \RelReach properties with comparison operators $<$ and $>$, respectively.
        We can guess a witness for $\phi$ by guessing which disjunct holds and then guessing an MD witness following the reasoning for $\geq,>$-queries. 
    \end{itemize}
    Note that the above does not imply that MD schedulers suffice for \RelBuechi, since \cref{thm:whyComb} only holds for \emph{memoryful} schedulers: The existence of an MD witness schedulers for each state-scheduler combination does not necessarily imply the existence of MD witnesses for the schedulers quantified in the original property (where several state-scheduler combinations may share a scheduler).

    \proofsubparagraph{Strong \NP-hardness:}
    We show \NP-hardness by reduction from SAT~\cite{karpReducibilityCombinatorial1972}.
    Given a SAT formula $\phi$ over variables $x_1,\ldots, x_N$ in conjunctive normal form with clauses $C_1, \ldots, C_M \subseteq \{x_i, \overline{x_i} \mid i \in \{1, \ldots, N\}\}$.
    We construct the MDP $\mdp_\phi = (\state_\phi, \Act_\phi, \Trans_\phi)$, depicted in \cref{fig:buechi_NP-complete}, with
    \begin{itemize}
        \item $\states_\phi = \{x_i, \overline{x_i} \mid i \in \{1, \ldots, N\} \}$
        \item $\Act_\phi = \{\action, \overline{\action}\}$
        \item 
        $\Trans_\phi(x_i, \action, x_{i+1}) = 1$, and
        $\Trans_\phi(x_i, \overline{\action}, \overline{x_{i+1}}) = 1$ for $i \in \{1, \ldots, N-1\}$, \\
        $\Trans_\phi(x_N, \action, x_{1}) = 1$, and
        $\Trans_\phi(x_N, \overline{\action}, \overline{x_{1}}) = 1$,
        and all other transition probabilities are 0.
    \end{itemize}
    We construct the following \RelBuechi property:
    \begin{align*}
        \exists \sched \in \Scheds[\mdp_\phi] .\
        \sum_{i=1}^{N} \left[ 2 - (\Pr^{\sched}_{x_1}(\Globally\Finally \{x_i\}) + \Pr^{\sched}_{x_1}(\Globally\Finally \{\overline{x_i}\})) \right]
        + 
        \sum_{j=1}^{M} \Pr^{\sched}_{x_1}(\Globally \Finally C_j) 
        \geq N + M
        ~.
        \tag{$\phi'$}
        \label{eq:buechi_from_SAT}
    \end{align*}

    Let us now show that $\phi$ is satisfiable iff \eqref{eq:buechi_from_SAT} holds.
    
    ``$\Rightarrow$'': Assume $\phi$ is satisfiable with model $\mathcal{I}$. Then $\mathcal{I}$ induces a memoryless deterministic scheduler as follows: If $\mathcal{I}$ chooses $x_i$, we choose $\action$ in $x_{i-1}$ (modulo $N$), if $\mathcal{I}$ chooses $\overline{x_i}$ we choose $\overline{\action}$.

    ``$\Leftarrow$'': Assume \eqref{eq:buechi_from_SAT} holds.
    Observe that \emph{deterministic} schedulers suffice for satisfying this property: 
    By \cref{th:relBuechi-to-relReach,thm:correctness} we can reduce this property to computing a scheduler maximizing some expected reward on the MEC quotient of $\mdp_\phi$; for this memoryless deterministic schedulers suffice with analogous reasoning to~\cref{thm:exRewPtimeMD} and by \cref{th:relBuechi-to-relReach_Leftarrow_MD} any MD scheduler on the goal unfolding can be translated back to an MD scheduler on the original MDP.
    
    So let $\sched$ be a deterministic witness to \eqref{eq:buechi_from_SAT}.    
    Observe that for any scheduler $\sched'$ it must hold that $\Pr^{\sched'}_{x_1}(\Globally\Finally \{x_i\}) + \Pr^{\sched'}_{x_1}(\Globally\Finally \{\overline{x_i}\}) \in [1,2]$ and hence $ 2 - (\Pr^{\sched'}_{x_1}(\Globally\Finally \{x_i\}) + \Pr^{\sched'}_{x_1}(\Globally\Finally \{\overline{x_i}\})) \in [0,1]$ for all $i \in \{1,\ldots,N\}$, and further that 
    $\Pr^{\sched'}_{x_1}(\Globally \Finally C_j) \in [0,1]$ for all $j \in \{1,\ldots,M\}$.
    Since $\sched$ satisfies \eqref{eq:buechi_from_SAT}, it must thus hold that 
    $\Pr^{\sched'}_{x_1}(\Globally\Finally \{x_i\}) + \Pr^{\sched'}_{x_1}(\Globally\Finally \{\overline{x_i}\}) = 1$ for all $i \in \{1,\ldots,N\}$, and that 
    $\Pr^{\sched'}_{x_1}(\Globally \Finally C_j) = 1$ for all $j \in \{1,\ldots,M\}$.
    Since $\sched$ is deterministic, we must have $\Pr^{\sched'}_{x_1}(\Globally\Finally \{x_i\}) \in \{0,1\}$ for all $i \in \{1,\ldots,N\}$, and thus $\Pr^{\sched'}_{x_1}(\Globally\Finally \{x_i\}) = 1 \Iff \Pr^{\sched'}_{x_1}(\Globally\Finally \{\overline{x_i}\}) = 0$.
    This induces a model for $\phi$: For each $i \in \{1, \ldots, N\}$, if $\Pr^{\sched'}_{x_1}(\Globally\Finally \{x_i\})=1$, choose $x_i$; otherwise, choose $\overline{x_i}$.

    \textbf{Claim:} \emph{The above construction defines a pseudo-polynomial time transformation.}
    
    The constructed MDP has $2\cdot N$ states. Each state has at exactly two successors. 
    All transition probabilities in the MDP are either 0 or 1. 
    Hence, the magnitude of the largest number occurring in the constructed MDP is a constant.
    The constructed \RelBuechi property sums over $2N+M$ probability operators, and contains coefficients $2$ and $N+M$.
    Hence, the magnitude of the largest number in the constructed \RelBuechi instance is polynomial in the size of the original SAT instance.
\end{proof}

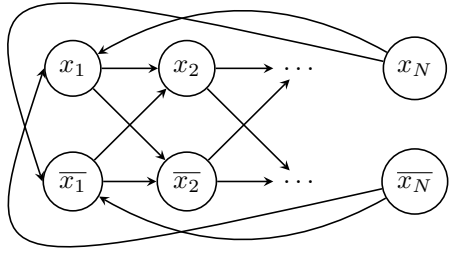
\begin{figure}
    \centering
    \begin{tikzpicture}[on grid,node distance=15mm and 15mm,semithick,>=stealth]
    \node[state] (x1) {$x_1$};
    \node[state, below=of x1] (x1n) {$\overline{x_1}$};  
    \node[state, right=of x1] (x2) {$x_2$};
    \node[state, below=of x2] (x2n) {$\overline{x_2}$};  
    \node[right=of x2] (xi) {\ldots};
    \node[below= of xi] (xin) {\ldots};
    \node[state, right= of xi] (xN) {$x_N$};
    \node[state, below=of xN] (xNn) {$\overline{x_N}$};  

    \node[above left=of x1, xshift=2mm, yshift=-2mm] (x1h) {};
    \node[below left=of x1n, xshift=2mm, yshift=2mm] (x1nh) {};

    \path[->]
    (x1) edge (x2)
    (x1) edge (x2n)
    (x1n) edge (x2)
    (x1n) edge (x2n)
    (x2) edge (xi)
    (x2) edge (xin)
    (x2n) edge (xi)
    (x2n) edge (xin)
    (xN) edge[bend right] (x1)
    (xNn) edge[bend left] (x1n)
    ;
    \draw[->] (xN) .. controls (x1h) .. (x1n.180);
    \draw[->] (xNn) .. controls (x1nh) .. (x1.180);
    \end{tikzpicture}
    \caption{Illustration of the MDP construction for the reduction from SAT to \RelBuechi.}
    \label{fig:buechi_NP-complete}
\end{figure}

\subsection{Proof of \cref{th:buechi_general_PTIME}}
\label{app:buechi_general_PTIME}

\buechiGeneralPTIME*

\begin{proof}
    \ref{item:buechi_fixed-param}: Follows from \cref{th:buechi_correctness}.
    Also, for a fixed $\numsum$, the transformation from \RelBuechi to \RelReach given in \cref{sec:buechi_algo} is a polynomial-time transformation. Hence, for any fixed $\numsum$, \RelBuechi with $m$ target sets is Turing-reducible to \RelReach with $m$ target sets.

    \ref{item:buechi_absorb}: If all targets are absorbing, then reaching a target is equivalent to visiting it infinitely often. Hence, a \RelBuechi query can be restated as a \RelReach query, which is in \PTIME in this case due to \cref{th:general_PTIME}\ref{item:absorbing}.

    \ref{item:buechi_single-target}:
    For every state-scheduler combination, there is a unique target set, so the size of the MEC quotient is linear in the size of the original MDP and hence the procedure detailed above runs in polynomial time.
\end{proof}

\subsection{Proof of \cref{th:buechi_MD_NP-complete}}
\label{app:buechi_MP_NP_complete}

\buechiMDNPComplete*

\begin{proof}
    Membership in \NP:
    Given some \RelBuechi property, MDP $\mdp$ and a memoryless deterministic scheduler $\sched \in \Scheds$, we can verify whether $\sched$ is a witness for the property 
    by computing the (exact) probabilities of the B\"uchi objectives in the induced DTMC.
    This is possible in time polynomial in the size of the state space~\cite{biancoModelChecking1995}.

    \NP-Hardness: We can show \NP-hardness by reduction from the Hamiltonian path problem as in the proof of \cref{th:MD_NP-complete}.
    For a given instance of the Hamiltonian path problem $G =(V,E)$ and $v_{\init} \in V$, we construct the MDP $\mdp$ as before.
    Observe that both target states $\state_a$ and $\state_b$ are absorbing, hence for any scheduler $\sched \in \Scheds$ we have 
    $\Pr^\sched_{\statei}(\Finally \{s_a\}) = \Pr^\sched_{\statei}(\Globally \Finally \{s_a\})$ and 
    $\Pr^\sched_{\statei}(\Finally \{s_b\}) = \Pr^\sched_{\statei}(\Globally \Finally \{s_b\})$.
    Together with our previous reasoning this already implies that for any $\epsilon < \frac{1}{2^{|V|+1}}$, it holds that there exists a Hamiltonian path from $v_\init$ in $G$ iff there exists some $\sched \in \SchedsMD$ such that $\Pr^\sched_{\statei}(\Globally \Finally \{s_a\}) - \Pr^\sched_{\statei}(\Globally \Finally \{s_b\}) \approx_\epsilon 0$.
\end{proof}

\subsection{Proof of \cref{th:buechi_general_necessary}}
\label{app:buechi_general_necessary}

\buechiGeneralNecessary*
\begin{proof}
    Consider again the MDP from \cref{fig:memory-necessary}, depicted again in \cref{fig:buechi_memory-randomization-necessary_app} and consider the property $\exists \sched .\ \Pr^{\sched}_{\state}(\Globally \Finally \{t_1\}) \approx_{\epsilon} \Pr^{\sched}_{\state}(\Globally \Finally \{t_2\})$ for some $\epsilon\geq 0$.
    
    We can construct a memoryful randomized witness for this property as follows for any $\epsilon \geq 0$. 
    Let $\sched_{\alpha}$ and $\sched_{\beta}$ be the memoryless deterministic schedulers that choose $\alpha$ at $\state$ and $\beta$ at $\state$, respectively, and let $\sched := [\sched_{\alpha} \oplus_{0.5} \sched_{\beta}]$ be the convex combination of these two schedulers. 
    Intuitively, $\sched$ initially throws a fair coin to decide whether to always take $\alpha$ or always take $\beta$.
    Hence, this scheduler is randomized with one bit of memory.
    Using \cite[Lem.\ 3.8]{quatmannVerificationMultiobjective2023}, we have $\Pr^{\sched}_{\state}(\Globally \Finally \{t_1\}) = \frac{1}{2} \Pr^{\sched_\alpha}_{\state}(\Globally \Finally \{t_1\}) + \frac{1}{2} \Pr^{\sched_\beta}_{\state}(\Globally \Finally \{t_1\}) = \frac{1}{2}$ and analogously also $\Pr^{\sched}_{\state}(\Globally \Finally \{t_2\}) = \frac{1}{2}$.
    
    If we do not allow memory or randomization, however, the property only holds for $\epsilon \geq 1$.
    Firstly, let us assume $\sched \in \Scheds$ to be memoryless. 
    We distinguish two cases: If $\sched(\state)(\alpha) = 1$, then we almost-surely visit only $t_1$ infinitely often and never visit $t_2$. Otherwise, if $\sched(\state)(\alpha)<1$, the probability of visiting $t_1$ infinitely often is 0 and the probability of finally taking $\beta$ is 1, hence we almost-surely visit $t_2$ infinitely often.

    Now, let us assume that $\sched \in \Scheds$ is memoryful deterministic. Then the DTMC induced by $\sched$ on $\mdp$ has a unique infinite path $\pi$.
    If $\sched$ never takes $\beta$, then $\pi$ never reaches $t_2$, so $\Pr^{\sched}_{\state}(\Globally \Finally \{t_1\}) = 1$ and $\Pr^{\sched}_{\state}(\Globally \Finally \{t_2\}) = 0$. 
    Otherwise, if $\sched$ does take $\beta$ at some point, then $\pi$ cannot visit $t_1$ anymore afterwards, so $\Pr^{\sched}_{\state}(\Globally \Finally \{t_1\}) = 0$ but $\Pr^{\sched}_{\state}(\Globally \Finally \{t_2\}) = 1$.
\end{proof}

\begin{figure}
    \centering
    \begin{tikzpicture}
        \node[state] (s) {$s$};
        
        \node[state, below left= of s] (t1) {$t_1$}; 
        \node[state, below right= of s] (t2) {$t_2$}; 

       \path[-latex', draw]
            (s) edge[bend left] node[right] {$\alpha$} (t1)
            (s) edge node[right] {$\beta$} (t2);

        \path[-latex', draw]
            (t1) edge[bend left] node[right] {$\gamma$} (s)
            (sx) edge[loop right] (sx);
    \end{tikzpicture}
    \caption{MDP where memory and randomization are necessary for the relational B\"uchi property $\exists \sched .\ \Pr^{\sched}_{\state}(\protect\Globally \protect\Finally \{t_1\}) \approx_{\epsilon} \Pr^{\sched}_{\state}(\protect\Globally \protect\Finally \{t_2\})$ for some $\epsilon\geq 0$.}
    \label{fig:buechi_memory-randomization-necessary_app}
\end{figure}
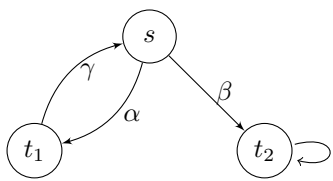

%% file: appendix/conjunction_algo_proofs.tex
\section{Proofs for \cref{sec:conj_algo}}
\label{app:conj_algo_proof}

\subsection{Proof of \cref{le:reduce-to-moa} }
\label{app:reduce-to-moa_proof}

\reduceToMOA*

\begin{proof}
    Let us fix an order on $\comb = \{c_1, \ldots, c_{|\comb|}\}$.
    For $j=1,\ldots,\numconj$ and $T \in \mathcal{T}^j$, let $\comb^j(T) := \{ c \in \comb \mid \exists i \in \relInd^j(c) .\ T_{i,j} = T \}$.

    We lift the reward structure $\indexc[\rew]^j$~(\cref{def:rewForComb}) from $\indexc$ to $\combined$ and denote it again as $\indexc[\rew]^j$, under abuse of notation.

    We show the equivalence step by step.

    \proofsubparagraph{(1) Claim:}\phantom{x}
    
    \nopagebreak\noindent\fbox{\begin{minipage}{\textwidth}
    \begin{alignat*}{2}
        & 
        \exists \sched_{1} \ldots \sched_{\numsched} \in \Scheds . &&
        \bigwedge_{j=1}^{\numconj} 
        \sum_{i=1}^{\numsum} q_{i,j} \Pr^{\mdp, \sched_{k_{i,j}}}_{\state_{i,j}}(\Finally T_{i,j}) 
        \comp_j q_j  
        \\ \Iff\quad & 
        \exists \sched_{c_1} \in \Scheds[\mdp_{c_1}] \ldots \sched_{c_{|\comb|}} \in \Scheds[\mdp_{c_{|\comb|}}] . &&
        \bigwedge_{j=1}^{\numconj} 
        \sum_{c \in \comb^j}
        \sum_{i \in \relInd^j(c)} q_{i,j} \Pr_{\indexc[\state]}^{\indexc, \indexc[\sched]}(\Finally T_{i,j} )
        \comp_j q_j  
    \end{alignat*}
    \end{minipage}}

    \textit{Proof:} We show both directions separately. 

    `$\Rightarrow$': Let $\sched_{1} \ldots \sched_{\numsched} \in \Scheds$. 
    Let $c \in \comb$.
    Let $h$ be the unique index such that $c = (\state, \sched_{h})$ for some $\state \in \states$.
    Then we construct $\indexc[\sched] \in \Scheds[\indexc]$ by letting 
    \[
        \indexc[\sched]((\state_1, \mathcal{T}_1) \action_1 \ldots (\state_r, \mathcal{T}_r)) = \sched_{h}(\state_1 \action_1 \ldots \state_r)
    \]
    for $r \in \N$, $(\state_1, \mathcal{T}_1), \ldots, (\state_r, \mathcal{T}_r) \in \indexc[\states]$, $\action_1, \ldots, \action_{r-1} \in \Act$.
    Then it must hold for any $j=1,\ldots,\numconj$ and $i \in \relInd^j(c)$ that
    $\Pr^{\mdp, \sched_{h}}_{\state_{i,j}}(\Finally T_{i,j}) = \Pr_{\indexc[\state]}^{\indexc, \sched_{c}}(\Finally T_{i,j})$.

    `$\Leftarrow$': Let $\sched_{c_1} \in \Scheds[\mdp_{c_1}] \ldots \sched_{c_{|\comb|}} \in \Scheds[\mdp_{c_{|\comb|}}]$.
    Let $h \in \{1, \ldots, \numsched\}$.
    Then we construct $\sched_{h} \in \Scheds$ as follows:
    Let $\state_1 \action_1 \ldots \state_r \in \finPaths$. 
    \begin{itemize}
        \item 
        If there exists $c \in \comb$ such that $c = (\state_1, \sched_h)$, then
        \[
            \sched_{h}(\state_1 \action_1 \ldots \state_r) = 
                \sched_c((\state_1, \mathcal{T}_1) \action_1 \ldots (\state_r, \mathcal{T}_r))
        \]
        where $\mathcal{T}_1 = \emptyset$ and $\mathcal{T}_t = \mathcal{T}_{t-1} \cup \{ T \in \mathcal{T}_c \mid \state_{t-1} \in T \}$ for $1 < t \leq r$.

        \item Otherwise, choose an action uniformly at random.
    \end{itemize}

    Then for $j=1,\ldots,\numconj$, $c \in \comb^j$, and $i \in \relInd^j(c)$, we have $\Pr_{\indexc[\state]}^{\indexc, \indexc[\sched]}(\Finally T_{i,j} ) = \Pr^{\mdp, \sched_{k_{i,j}}}_{\state_{i,j}}(\Finally T_{i,j})$.

    \proofsubparagraph{(2) Claim:}\phantom{x}
    
    \nopagebreak\noindent\fbox{\begin{minipage}{\textwidth}
    \begin{alignat*}{2}
        & 
        \exists \sched_{c_1} \in \Scheds[\mdp_{c_1}] \ldots \sched_{c_{|\comb|}} \in \Scheds[\mdp_{c_{|\comb|}}] . &&
        \bigwedge_{j=1}^{\numconj} 
        \sum_{c \in \comb^j}
        \sum_{i \in \relInd^j(c)} q_{i,j} \Pr_{\indexc[\state]}^{\indexc, \indexc[\sched]}(\Finally T_{i,j} )
        \comp_j q_j  
        \\ \Iff\quad & 
        \exists \sched_{c_1} \in \Scheds[\mdp_{c_1}] \ldots \sched_{c_{|\comb|}} \in \Scheds[\mdp_{c_{|\comb|}}] . &&
        \bigwedge_{j=1}^{\numconj} 
        \sum_{c \in \comb^j}
        \sum_{i \in \relInd^j(c)} q_{i,j} \Expected_{\indexc[\state]}^{\indexc, \indexc[\sched]}(\rew_{T_{i,j}} )
        \comp_j q_j  
    \end{alignat*}
    \end{minipage}}

    \textit{Proof:} Straightforward.

    \proofsubparagraph{(3) Claim:}\phantom{x}

    \nopagebreak\noindent\fbox{\begin{minipage}{\textwidth}    
    \begin{alignat*}{2}
        & 
        \exists \sched_{c_1} \in \Scheds[\mdp_{c_1}] \ldots \sched_{c_{|\comb|}} \in \Scheds[\mdp_{c_{|\comb|}}] . &&
        \bigwedge_{j=1}^{\numconj} 
        \sum_{c \in \comb^j}
        \sum_{i \in \relInd^j(c)} q_{i,j} \Expected_{\indexc[\state]}^{\indexc, \indexc[\sched]}(\rew_{T_{i,j}} )
        \comp_j q_j  
        \\ \Iff\quad & 
        \exists \sched_{c_1} \in \Scheds[\mdp_{c_1}] \ldots \sched_{c_{|\comb|}} \in \Scheds[\mdp_{c_{|\comb|}}] . && 
        \bigwedge_{j=1}^{\numconj} 
        \sum_{c \in \comb^j} 
        \Expected_{\indexc[\state]}^{\indexc, \indexc[\sched]}\left( \indexc[\rew]^j \right)
        \comp_j q_j  
    \end{alignat*}
    \end{minipage}}

    \textit{Proof:}
    Let $j \in \{1, \ldots, \numconj\}$, $c \in \comb$, and $\indexc[\sched] \in \Scheds[\indexc]$.
    Then,
    \begin{align*}
        \sum_{i \in \relInd^j(c)} q_{i,j} \Expected_{\indexc[\state]}^{\indexc, \indexc[\sched]}(\rew_{T_{i,j}} )
        = &
        \sum_{T \in \indexc[\mathcal{T}] \cap \mathcal{T}^j} \left( \sum_{i \in \{\relInd^j(c) \mid T=T_{i,j}\}} q_{i,j} \right) \cdot \Expected_{\indexc[\state]}^{\indexc, \indexc[\sched]}( \rew_T(\state, \mathcal{T}) )
        \\ = &
        \Expected_{\indexc[\state]}^{\indexc, \indexc[\sched]}\left( \sum_{T \in \indexc[\mathcal{T}] \cap \mathcal{T}^j} \left( \sum_{i \in \{\relInd^j(c) \mid T=T_{i,j}\}} q_{i,j} \right) \cdot \rew_T(\state, \mathcal{T}) \right) 
        \\ = &
        \Expected_{\indexc[\state]}^{\indexc, \indexc[\sched]}\left( \indexc[\rew]^j \right)
        ~.
    \end{align*}
    
    \proofsubparagraph{(4) Claim:}\phantom{x}
    
    \nopagebreak\noindent\fbox{\begin{minipage}{\textwidth}
    \begin{alignat*}{2}
        & \exists \sched_{c_1} \in \Scheds[\mdp_{c_1}] \ldots \sched_{c_{|\comb|}} \in \Scheds[\mdp_{c_{|\comb|}}] . && 
        \bigwedge_{j=1}^{\numconj} 
        \sum_{c \in \comb^j} 
        \Expected_{\indexc[\state]}^{\indexc, \indexc[\sched]}\left( \indexc[\rew]^j \right)
        \comp_j q_j  
        \\ \Iff\quad & 
        \exists \sched \in \Scheds[\combined] .
        && 
        \bigwedge_{j=1}^{\numconj} 
        \Expected_{\combined[\state]}^{\combined, \sched}(\rew^j) 
        \comp_j q_j  
    \end{alignat*}
    \end{minipage}}

     \textit{Proof:} 
     Intuitively, $\combined$ corresponds to initial state $\combined[\state]$ leading with equal probability to one of the goal unfoldings of $\mdp$ with respect to some state-scheduler-combination. Therefore, a number of schedulers for each such goal unfolding can directly be combined into a single scheduler for $\combined$ and vice versa.
     
     Firstly, for $\sched \in \Scheds[\combined]$, let us use $\sched(\combined[\state])$ to denote the residual scheduler of $\sched$ after having visited $\combined[\state]$,
     and observe that for $j \in \{1, \ldots, \numconj\}$ we have
     \begin{align*}
        \Expected_{\combined[\state]}^{\combined, \sched}(\rew^j) 
        &= \sum_{c \in \comb^j} \frac{1}{|\comb|} \cdot \Expected_{(\indexc[\state], c)}^{\combined, \sched(\combined[\state])}(\rew^j) 
        \\ &= \sum_{c \in \comb^j} \frac{1}{|\comb|} \cdot |\comb| \cdot \Expected_{(\indexc[\state], c)}^{\combined, \sched(\combined[\state])}(\indexc[\rew]^j)
        = \sum_{c \in \comb^j} \Expected_{(\indexc[\state], c)}^{\combined, \sched(\combined[\state])}(\indexc[\rew]^j)
        ~.
     \end{align*}

    Let us now show both directions of the claim separately.
    
    `$\Rightarrow$': 
    Let $\sched_{c_1} \in \Scheds[\mdp_{c_1}], \ldots, \sched_{c_{|\comb|}} \in \Scheds[\mdp_{c_{|\comb|}}]$.
    Construct $\sched \in \Scheds[\combined]$ as follows:
    \begin{itemize}
        \item $\sched(\state_\comb)(\epsilon)=1$ (there is no other action available)
        \item All finite paths in $\combined$ from $\state_\comb$ are of the form \[\state_\comb \epsilon (\indexc[\state],c) \alpha_1 (\state_1, c) \ldots \action_r (\state_r, c)\] for some $c \in \comb$, $\state_1, \ldots, \state_r \in \indexc[\states]$, $\action_1, \ldots, \action_r \in \Act$ and $r \in \N$.
        Let $h \in \{1, \ldots, \numsched\}$ be the unique index such that $c = (\state, \sched_h)$ for some $\state \in \states$.
        We let 
        \begin{align*}
        \sched(\combined[\state] \epsilon (\indexc[\state],c) \alpha_1 (\state_1, c) \ldots \action_r (\state_r, c)) 
        &= 
        \sched_{h}(\indexc[\state] \alpha_1 \state_1 \ldots \action_r \state_r)~.
        \end{align*}
        For paths not starting in $\combined[\state]$, we pick an action uniformly at random.
    \end{itemize}

    Then for $j=1,\ldots,\numconj$ we have $\Expected_{\indexc[\state]}^{\indexc, \indexc[\sched]}( \indexc[\rew]^j ) = \Expected_{(\indexc[\state], c)}^{\combined, \sched(\combined[\state])}(\indexc[\rew]^j)$ by definition of the reward structures and our initial observation.

    `$\Leftarrow$': 
    Let $\sched \in \Scheds[\combined]$.
    Let $c \in \comb$. 
    We construct $\sched_c \in \Scheds[\indexc]$ as follows.
    For $\indexc[\state] \action_1 \state_1 \ldots \action_r \state_r \in \finPaths[\indexc](\indexc[\state])$ we let
    \[
        \sched_{c}(\indexc[\state] \alpha_1 \state_1 \ldots \action_r \state_r)
        =
        \sched(\combined[\state] \epsilon (\indexc[\state], c) \alpha_1 (\state_1, c) \ldots \action_r (\state_r, c))
    \]
    and for all other paths we choose some action uniformly at random.
    
    Then for $j=1,\ldots,\numconj$ we have $\Expected_{\indexc[\state]}^{\indexc, \indexc[\sched]}( \indexc[\rew]^j ) = \Expected_{(\indexc[\state], c)}^{\combined, \sched(\combined[\state])}(\indexc[\rew]^j)$ by definition of the reward structures and our initial observation.
\end{proof}

%% file: appendix/moa_proofs.tex
\section{Proofs for \cref{sec:moa}}
\label{app:moa_proofs}

\subsection{Proof of \cref{th:moa_preproc_transfer-MR}}
\label{app:moa_preproc_transfer-MR_proof}

\moaPreprocTransferMR*

We show this claim in two steps, analogously to \cite{forejtQuantitativeMultiobjective2011}.
First, we show that any scheduler for $\combined$ induces a scheduler for $\processed[\combined]$ that reaches $\sdead$ with probability 1, and vice versa in \cref{le:moa_preproc_transfer-memory}.
Then we show that memoryless (randomized) schedulers suffice on $\processed[\combined]$ in \cref{le:moa_preproc_MR-suffice}.

\begin{lemma}
    \begin{align*}
        &\exists \sched \in \Scheds[\combined] .\ 
        \bigwedge_{j=1}^{\numconj} \Expected_{\combined[\state]}^{\combined, \sched}(\rew^j) \comp_j q_j  
        \\ \Iff\quad &
        \exists \processed[\sched] \in \Scheds[{\processed[\combined]}] .\ 
        \bigwedge_{j=1}^{\numconj} \Expected_{\combined[\state]}^{{\processed[\combined]}, \processed[\sched]}(\rew^j) \comp_j q_j \wedge \Pr_{\combined[\state]}^{{\processed[\combined]}, \processed[\sched]}(\Finally \sdead) = 1
    \end{align*}
    \label{le:moa_preproc_transfer-memory}
\end{lemma}

\begin{proof}[Proof (Sketch)]
    ``$\Rightarrow$'': 
    Let $\sched \in \Scheds[\combined]$.
    We proceed similarly to the scheduler witness transfer from, e.g., \cite[Th.\ 1]{kretinskyLTLConstrainedSteadyState2021a}.
    For every MEC $\mec = (S',A)$ and finite path $\pi \in \finPaths[\combined](\combined[\state])$ that has just entered $\mec$ (i.e., $\pi = \pi' \action \state$ and $\state \in S'$ and $\action \not\in A$), we compute the probability $p_{\mec,\pi}$ of staying in $\mec$ forever after having taken $\pi$. 
    We then construct $\processed[\sched] \in \Scheds[{\processed[\combined]}]$ from $\sched \in \Scheds[\combined]$ as follows for $\pi \in \finPaths[{\processed[\combined]}](\combined[\state])$ with $\last(\pi) \in \combined[\states]$:
    \begin{itemize}
        \item If $\pi$ ends in a state that is not contained in an MEC: Behave like $\sched$, i.e. $\processed[\sched](\pi) = \sched(\pi)$ (observe that $\pi$ can be interpreted as a path in $\combined$ since it does not enter $\sdead$ by assumption)
        
        \item If $\pi$ just entered some MEC $\mec$: 
        If $p_{\mec,\pi} = 1$ then we deterministically take $\dagger$, otherwise we copy $\sched$.
    \end{itemize}
    Observe that the set of paths in $\combined^{\sched}$ which from some point on almost-surely stay in an EC has measure 1.
    Hence, the set of paths in ${\processed[\combined]}^{\processed[\sched]}$ which at some point take $\dagger$ to transition to $\sdead$ also has measure 1.

    ``$\Leftarrow$'': 
    Given $\processed[\sched] \in \Scheds[{\processed[\combined]}]$, construct $\sched \in \Scheds[\combined]$ by behaving like $\processed[\sched]$ unless $\processed[\sched]$ chooses $\dagger$, then the current state must be contained in some MEC $\mec$ and there must exists a memoryless deterministic scheduler $\sched_\mec$ staying inside $\mec$ forever almost-surely, so $\sched$ switches to behave like $\sched_\mec$ instead. Formally, for $\pi \in \finPaths[\combined](\combined[\state])$ and $\action \in \Act(\last(\pi))$:
    \[
        \sched(\pi)(\action) = \begin{cases}
            \processed[\sched](\pi)(\action) + \processed[\sched](\pi)(\dagger) \cdot \sched_\mec(\last(\pi))(\action) & \text{if } \exists \mec=(S',A) \in \MEC(\combined) .\ \last(\pi) \in S' \\
            \processed[\sched](\pi)(\action) & \text{otherwise}~.
        \end{cases}
    \]
\end{proof}

Before showing \cref{le:moa_preproc_MR-suffice}, let us first recall the definition of expected visiting times.
\begin{definition}[Expected visiting times]
    Given some MDP $\mdp = \mdptup$ with some initial state $\state_0 \in \states$ and some scheduler $\sched \in \Scheds$, for $\state \in \states$ and $\action \in \Act(\state)$ we let 
    \begin{align*}
        \vis^{\mdp, \sched, \state_0}(\state, \action) &= 
        \sum_{j=0}^{\infty} \left( \sum_{\substack{\pi \in \finPaths(\state_0), |\pi|=j+1 \\ \exists \pi' \in \finPaths(\state_0), \state' \in \states .\ \pi = \pi'\state\action\state'}} \Pr^{\mdp, \sched}_{\state_0}(\pi) \right) \\
        \vis^{\mdp, \sched, \state_0}(\state) &= 
        \sum_{j=0}^{\infty} \left( \sum_{\substack{\pi \in \finPaths, |\pi|=j \\ \last(\pi) = \state}} \Pr^{\mdp, \sched}_{\state_0}(\pi) \right)
    \end{align*}
    and further
    \begin{align*}
        \InfA[\sched] &= \{ (\state, \action) \in \states \times \Act \mid \vis^{\sched}(\state, \action) = \infty \} 
        & \FinA[\sched] &= (\states \times \Act) \setminus \InfA[\sched]
        \\
        \InfS[\sched] &= \{ \state \in \states \mid \vis^{\sched}(\state) = \infty \} 
        & \FinS[\sched] &= \states \setminus \InfS[\sched]
        ~.
    \end{align*}
\end{definition}

We often omit the MDP and/or initial state if clear from context.
Observe that $(\InfS[\sched], \InfA[\sched])$ must be contained in MECs of $\mdp$.
For $\state \in \FinS[\sched]$ we have $\vis^{\sched}(\state) = \sum_{\action \in \Act(\state)} \vis^{\sched}(\action)$.

\begin{restatable}{lemma}{moaPreprocMRSuffice}
    \label{le:moa_preproc_MR-suffice}
    \begin{align*}
        &\exists \processed[\sched] \in \Scheds[{\processed[\combined]}] .\ 
        \bigwedge_{j=1}^{\numconj} \Expected_{\combined[\state]}^{{\processed[\combined]}, \processed[\sched]}(\rew^j) \comp_j q_j \wedge \Pr_{\combined[\state]}^{{\processed[\combined]}, \processed[\sched]}(\Finally \sdead) = 1
        \\ \Iff\quad &
        \exists \processed[\sched]' \in \SchedsMR[{\processed[\combined]}] .\ 
        \bigwedge_{j=1}^{\numconj} \Expected_{\combined[\state]}^{{\processed[\combined]}, \processed[\sched]'}(\rew^j) \comp_j q_j \wedge \Pr_{\combined[\state]}^{{\processed[\combined]}, \processed[\sched]'}(\Finally \sdead) = 1
    \end{align*}
\end{restatable}

\begin{proof}
    The direction ``$\Leftarrow$'' is straightforward.
    For the other direction, first observe that it suffices to show that for all $x_1, \ldots, x_{\numconj} \in \R$, we have
    \begin{align*}
        &\exists \processed[\sched] \in \Scheds[{\processed[\combined]}] .\ 
        \bigwedge_{j=1}^{\numconj} \Expected_{\combined[\state]}^{{\processed[\combined]}, \processed[\sched]}(\rew^j) = x_j 
        \wedge \Pr_{\combined[\state]}^{{\processed[\combined]}, \processed[\sched]}(\Finally \sdead) = 1
        \\ \Implies\quad &
        \exists \processed[\sched]' \in \SchedsMR[{\processed[\combined]}] .\ 
        \bigwedge_{j=1}^{\numconj} \Expected_{\combined[\state]}^{{\processed[\combined]}, \processed[\sched]'}(\rew^j) = x_j
        \wedge \Pr_{\combined[\state]}^{{\processed[\combined]}, \processed[\sched]'}(\Finally \sdead) = 1
        ~.
    \end{align*}
    We show this claim analogously to \cite[Prop.\ 6]{forejtQuantitativeMultiobjective2011}.
    Let $x_1, \ldots, x_{\numconj} \in \R$.
    Let $\processed[\sched] \in \Scheds[{\processed[\combined]}]$ with $\Expected_{\combined[\state]}^{{\processed[\combined]}, \processed[\sched]}(\rew^j) = x_j 
    \wedge \Pr_{\combined[\state]}^{{\processed[\combined]}, \processed[\sched]}(\Finally \sdead) = 1$.
    Since $\Pr_{\combined[\state]}^{{\processed[\combined]}, \processed[\sched]}(\Finally \sdead) = 1$, we have $\InfS \subseteq \sdead$, which further implies $\InfA = \{ (\state,\action) \mid \state \in \InfS, \action \in \processedcombined[\Act](\InfS)\}$.
    
    We construct a memoryless witness for the original claim as follows for $\state \in \processedcombined[\states]$, $\action \in \processedcombined[\Act](\state)$:
    \[
        \processed[\sched]'(\state, \action) = \begin{cases}
            \frac{\vis^{\processed[\sched]}(\state, \action)}{\sum_{\action' \in \processedcombined[\Act](\state)} \vis^{\processed[\sched]}(\state, \action') }
            &\text{if } \state \not\in \InfS 
            \\
            1 &\text{else (for } \state \in \InfS \text{ we have } \{\action\} = \processedcombined[\Act](\state) \text{)}~.
        \end{cases}
    \]

    In order to see that $\processed[\sched]'$ satisfies the desired properties, we first show that the expected visiting times of actions that are visited finitely often under $\processed[\sched]$ match.

    \proofsubparagraph{Claim 1:} \emph{
        For all $(\state, \action) \in \FinA \cap \FinA[{\processed[\sched]'}]$ we have 
        $\vis^{\processed[\sched]}(\state, \action) = \vis^{\processed[\sched]'}(\state, \action)$.}

        For $(\state, \action) \in \FinA \cap \FinA[{\processed[\sched]'}]$ we let $d_{(\state, \action)} = \frac{\vis^{\processed[\sched]}(\state, \action)}{\vis^{\processed[\sched]'}(\state, \action)}$.
        
        Towards contradiction, assume there exists $(\state, \action) \in \FinA \cap \FinA[{\processed[\sched]'}]$ with $d_{(\state, \action)} \neq 1$.
        We assume $\max_{(\state, \action) \in \FinA \cap \FinA[{\processed[\sched]'}]} d_{(\state, \action)} > 1$ and set $\speciald := \max_{(\state, \action) \in \FinA \cap \FinA[{\processed[\sched]'}]} d_{(\state, \action)}$ and $(\specialstate, \specialact) = \argmax_{(\state, \action) \in \FinA \cap \FinA[{\processed[\sched]'}]} d_{(\state, \action)}$.
        If the maximum is smaller than 1, then take the minimal value and argue analogously. 

        Let $\pi \in \finPaths[\processedcombined](\combined[\state])$ with $\pi = \pi'\specialstate\specialact\state'$ for some $\pi' \in \finPaths[\processedcombined](\combined[\state])$, $\state' \in \processedcombined[\states]$.
        Observe that $\specialstate \not\in \InfS$ since $\InfS \subseteq \sdead$ and thus all outgoing actions of states in $\InfS$ must also be taken infinitely often, but $(\specialstate,\specialact) \in \FinA$ by assumption.
        
        By definition,
        $\frac{\vis^{\processed[\sched]}(\specialstate, \specialact)}{\sum_{\action' \in \processedcombined[\Act](\specialstate)} \vis^{\processed[\sched]}(\specialstate, \action')}
        = 
        \processed[\sched]'(\specialstate)(\specialact)$
        and $\vis^{\processed[\sched]}(\specialstate) = \sum_{\action' \in \processedcombined[\Act](\specialstate)} \vis^{\processed[\sched]}(\specialstate, \action')$.
        Since $\processed[\sched]'$ is memoryless we have $\vis^{\processed[\sched]'}(\specialstate, \specialact) = \vis^{\processed[\sched]'}(\specialstate) \cdot \processed[\sched]'(\specialstate, \specialact)$.
        Putting everything together, we have $\vis^{\processed[\sched]}(\specialstate) = \speciald \cdot \vis^{\processed[\sched]'}(\specialstate)$.
        
        If $\specialstate \neq \combined[\state]$, then
        \begin{align*}
            &\sum_{\state' \in \processedcombined[\states]} 
            \sum_{\action' \in \processedcombined[\Act](\state')} \vis^{\processed[\sched]'}(\state', \action') \cdot \Trans(\state', \action', \specialstate)
            \\ = &
            \vis^{\processed[\sched]'}(\specialstate)
            \\ = &
            \frac{1}{\speciald} \cdot\vis^{\processed[\sched]}(\specialstate)
            \\ = & \frac{1}{\speciald} \cdot
            \sum_{\state' \in \processedcombined[\states]} 
            \sum_{\action' \in \processedcombined[\Act](\state')} \vis^{\processed[\sched]}(\state', \action') \cdot \Trans(\state', \action', \specialstate)
            \\ = & \frac{1}{\speciald} \cdot
            \sum_{\state' \in \processedcombined[\states]} 
            \sum_{\action' \in \processedcombined[\Act](\state')} d_{\action'}\cdot \vis^{\processed[\sched]'}(\state', \action') \cdot \Trans(\state', \action', \specialstate)
        \end{align*}
        which implies that $d_{\action'} = \speciald$ for all $\action'$ leading to $\specialstate$ with positive probability.
        In particular, for the second-to-last action $\action_{|\pi|-1}$ in $\pi$ we must also have $d_{\action_{|\pi|-1}} = \speciald$, i.e., this action must also be $\specialact$. By repeatedly applying this reasoning we arrive at $\specialstate = \combined[\state]$, contradicting our assumption.

        So let us assume $\specialstate = \combined[\state]$, then
        \begin{align*}
            &1 + 
            \sum_{\state' \in \processedcombined[\states]} 
            \sum_{\action' \in \processedcombined[\Act](\state')} \vis^{\processed[\sched]'}(\state', \action') \cdot \processedcombined[\Trans](\state', \action', \specialstate)
            \\ = &
            \vis^{\processed[\sched]'}(\specialstate)
            \\ = &
            \frac{1}{\speciald} \vis^{\processed[\sched]}(\specialstate)
            \\ = & \frac{1}{\speciald} \cdot \left(1 + 
            \sum_{\state' \in \processedcombined[\states]} 
            \sum_{\action' \in \processedcombined[\Act](\state')} \vis^{\processed[\sched]}(\state', \action') \cdot \processedcombined[\Trans](\state', \action', \specialstate)
            \right)
            \\ = & \frac{1}{\speciald} \cdot \left(1 + 
            \sum_{\state' \in \processedcombined[\states]} 
            \sum_{\action' \in \processedcombined[\Act](\state')} d_{\action'} \cdot\vis^{\processed[\sched]'}(\state', \action') \cdot \processedcombined[\Trans](\state', \action', \specialstate)
            \right)
        \end{align*}
        but the equality between the first and last term would imply that some $d_{\action'} > \speciald$, which contradicts our assumption that $\speciald$ is maximal.

    \proofsubparagraph{Claim 2:} \emph{$\processed[\sched]'$ satisfies the desired properties.}
    
    Recall that $\InfS$ and $\InfS[{\processed[\sched]'}]$ must be contained in MECs of $\processedcombined$ and recall also that the reward functions $\rew^j$ do not collect reward inside MECs by construction.\footnote{It would suffice here to reason that the expected reward is finite for any reward structure $\rew^j$; thus this reasoning generalizes to any MDP $\mdp$ with reward functions $\rew^j$ whose expected reward is finite under all schedulers.}
    Therefore, we have $\rew^j(\state) = 0$ for $\state \in \InfS \cup \InfS[{\processed[\sched]'}]$.
    Further, for $\state \in \FinS$ we have $\{\state\} \times \processedcombined[\Act](\state) \subseteq \FinA$ and analogously for $\state \in \FinS[{\processed[\sched]'}]$ we have $\{\state\} \times \processedcombined[\Act](\state) \subseteq \FinA[{\processed[\sched]'}]$
    and thus for $\state \in \FinS \cap \FinS[{\processed[\sched]'}]$ we have $\{\state\} \times \processedcombined[\Act](\state) \subseteq \FinA \cap \FinA[{\processed[\sched]'}]$ .
    
    Hence, 
    \begin{align*}
        & \Expected_{\combined[\state]}^{{\processed[\combined]}, \processed[\sched]'}(\rew^j) \\
        {} = &\sum_{\state \in \processedcombined[\states]} \vis^{\processed[\sched]'}(\state) \cdot \rew^j(\state)
        \\ {} = & \sum_{\state \in \FinS \cap \FinS[{\processed[\sched]'}]}
        \sum_{\action \in \processedcombined[\Act](\state)} \vis^{\processed[\sched]'}(\state, \action) \cdot \rew^j(\state)
        \\ \overset{\text{Cl.~2}}{=} \!\! 
        & \sum_{\state \in \FinS \cap \FinS[{\processed[\sched]'}]}
        \sum_{\action \in \processedcombined[\Act](\state)} \vis^{\processed[\sched]}(\state, \action) \cdot \rew^j(\state)
        \\ {} = & \sum_{\state \in \processedcombined[\states]} \vis^{\processed[\sched]}(\state) \cdot \rew^j(\state)
        = \Expected_{\combined[\state]}^{{\processed[\combined]}, \processed[\sched]}(\rew^j) = x_j.
    \end{align*}

    Further, let $\textit{In}(\sdead) = \{ (\state, \action) \in \processedcombined[\states] \times \processedcombined[\Act] \mid \exists t \in \sdead .\ \processedcombined[\Trans](\state, \action, t) = 1 \}$.
    Observe that, by construction actions leading to $\sdead$ can be visited at most once in expectation and thus $\textit{In}(\sdead) \subseteq \FinA \cap \FinA[{\processed[\sched]'}]$.
    Then,
    \begin{align*}
        \Pr_{\combined[\state]}^{{\processed[\combined]}, \processed[\sched]'}(\Finally \sdead)
        &= \sum_{(\state,\action) \in \textit{In}(\sdead)} \vis^{\processed[\sched]'}(\state,\action)
        \overset{\text{Cl.\ 2}}{=} \sum_{\action \in \textit{In}(\sdead)} \vis^{\processed[\sched]}(\action)
        = \Pr_{\combined[\state]}^{{\processed[\combined]}, \processed[\sched]}(\Finally \sdead) = 1~.
        \qedhere
    \end{align*}
\end{proof}

\subsection{Proof of \cref{le:conj_LP}}
\label{app:conj_LP_proof}

\conjLP*

\begin{proof}
    It suffices to show that for $x_1, \ldots, x_{l} \in \R$, it holds that:
    \begin{align*}
        &
            \exists \processed[\sched] \in \SchedsMR[{\processed[\combined]}] .\ 
            \bigwedge_{j=1}^{\numconj} \Expected_{\combined[\state]}^{{\processed[\combined]}, \sched}(\rew^j) = x_j
            \wedge \Pr_{\combined[\state]}^{{\processed[\combined]}, \sched}(\Finally \sdead) = 1
        \\ \Iff\quad & 
        \text{the LRA encoding from \cref{fig:LP} has a feasible solution $y^*$.}
    \end{align*}

    We proceed analogously to \cite[Th.~3.2 (1.)$\Iff$(2.)]{etessamiMultiObjectiveModel2008} \cite[Prop.~4]{forejtQuantitativeMultiobjective2011}.
    We show both directions separately; let us first introduce some notation used in both directions.
    Given some $\processed[\sched] \in \Scheds[{\processed[\combined]}]_{\textit{MR}}$, we let $\combined[\Trans]^{\processed[\sched]}$ (under abuse of notation) denote the transition matrix of $\processedcombined$ restricted to $\combined[\states]$, i.e., the matrix defined by $(\combined[\Trans]^{\processed[\sched]})_{\state, \state'} = \sum_{\action \in \combined[\Act](\state)} \combined[\Trans](\state, \action, \state') \cdot \processed[\sched](\state, \action)$ for $\state,\state' \in \combined[\states]$.
    For $n \in \N$, $((\combined[\Trans]^{\processed[\sched]})^n)_{\state,\state'}$ is the probability of transitioning from $\state$ to $\state'$ in $n$ steps.
    
    \proofsubparagraph{``$\Rightarrow$'':}
    Let $\processed[\sched] \in \Scheds[{\processed[\combined]}]_{\textit{MR}}$ such that 
    \[ 
        \bigwedge_{j=1}^{\numconj} \Expected_{\combined[\state]}^{{\processed[\combined]}, \processed[\sched]}(\rew^j) = x_j
        \wedge \Pr_{\combined[\state]}^{{\processed[\combined]}, \processed[\sched]}(\Finally \sdead) = 1
        ~.
    \]
    
    For $\state \in \combined[\states]$ and $\action \in \processed[{\combined[\Act]}](\state)$, we let 
    \begin{align*}
        y'_{\state, \action} &= \sum_{n=0}^{\infty} ((\combined[\Trans]^{\processed[\sched]})^n)_{\combined[\state], \state} \cdot \processed[\sched](\state, \action) \\
        y'_{\state} &= \sum_{\state' \in \combined[\states]} \sum_{\action' \in \combined[\Act](\state')} \combined[\Trans](\state', \action', \state) \cdot y'_{\state', \action'}
    \end{align*}
    Since $\Pr_{\combined[\state]}^{{\processed[\combined]}, \processed[\sched]}(\Finally \sdead) = 1$, and $\sdead$ is absorbing, the matrix $\combined[\Trans]^{\processed[\sched]}$ must be sub-stochastic. 
    Since $\combined[\Trans]^{\processed[\sched]}$ is sub-stochastic, $y'_{\state,\action}$ must be finite and thus the values $y'_{\state}$ and $y'_{\state,\action}$ are well-defined.
    
    Observe that $y'_{\state,\action}$ is the expected number of times we leave $\state \in \combined[\states]$ via $\action \in \processed[{\combined[\Act]}](\state)$ when starting from $\combined[\state]$ in $\processed[\combined]$ under $\processed[\sched]$.
    Consequently, $\sum_{\action \in \processed[{\combined[\Act]}](\state)} y'_{\state,\action} = \sum_{n=0}^{\infty} ((\combined[\Trans]^{\processed[\sched]})^n)_{\combined[\state], \state}$ is the expected number of times we transition out of $\state \in \combined[\states]$ under $\processed[\sched]$ when starting from $\combined[\state]$.
    Conversely, $y'_{\state}$ is the expected number of times we transition \emph{into} $\state$, i.e., $y'_{\state} = \sum_{n=0}^{\infty} ((\combined[\Trans]^{\processed[\sched]})^n)_{\combined[\state], \state} - \Iverson{\state = \combined[\state]}$, with the following reasoning:
    \begin{align*}
        y'_{\state} = 
        &\sum_{\state' \in \combined[\states]} \sum_{\action' \in \combined[\Act](\state')} \combined[\Trans](\state',\action',\state) \cdot y'_{\state',\action'}
        \\ = &
        \sum_{\state' \in \combined[\states]} 
        \sum_{\action' \in \combined[\Act](\state')} \combined[\Trans](\state',\action',\state) \cdot 
        \sum_{n=0}^{\infty} ((\combined[\Trans]^{\processed[\sched]})^n)_{\combined[\state], \state'} \cdot 
        \processed[\sched](\state', \action')
        \\ = &
        \sum_{\state' \in \combined[\states]} 
        \sum_{n=0}^{\infty} ((\combined[\Trans]^{\processed[\sched]})^n)_{\combined[\state], \state'} \cdot 
        \sum_{\action' \in \combined[\Act](\state')} 
        \combined[\Trans](\state',\action',\state) \cdot \processed[\sched](\state', \action')
        \\ = &
        \sum_{\state' \in \combined[\states]} 
        \sum_{n=0}^{\infty} ((\combined[\Trans]^{\processed[\sched]})^n)_{\combined[\state], \state'} 
        \cdot (\combined[\Trans]^{\processed[\sched]})_{\state', \state} 
        \\ = & \sum_{n=0}^{\infty} ((\combined[\Trans]^{\processed[\sched]})^{n+1})_{\combined[\state], \state}
        = \sum_{n=1}^{\infty} ((\combined[\Trans]^{\processed[\sched]})^{n})_{\combined[\state], \state}
        = \sum_{n=0}^{\infty} ((\combined[\Trans]^{\processed[\sched]})^n)_{\combined[\state], \state} - \Iverson{\state = \combined[\state]}
        ~.
    \end{align*}

    Using this, let us now show that $y'$ is indeed a feasible solution of the encoding.
    \begin{itemize}[itemindent=1em]
        \item[\eqref{eq:EVTs}] For $\state \in \combined[\states]$, by the above reasoning we have 
        \begin{align*}
            & \sum_{\action \in \processed[{\combined[\Act]}](\state)} y'_{\state,\action}
            - 
            \sum_{\state' \in \combined[\states]} \sum_{\action' \in \combined[\Act](\state')} \combined[\Trans](\state',\action',\state) \cdot y'_{\state',\action'}
            \\ =& 
            \sum_{n=0}^{\infty} ((\combined[\Trans]^{\processed[\sched]})^n)_{\combined[\state], \state}
            -
            \left(\sum_{n=0}^{\infty} ((\combined[\Trans]^{\processed[\sched]})^n)_{\combined[\state], \state} - \Iverson{\state = \combined[\state]} \right)
            =\Iverson{\state = \combined[\state]}
            ~.
        \end{align*}

        \item[\eqref{eq:MOA}] \label{it:exp-match}
        Let $j=1,\ldots,\numconj$. We show that $\sum_{\state \in \combined[\states]} \rew^j(\state) y'_{\state} = \Expected_{\combined[\state]}^{{\processed[\combined]}, \sched}(\rew^j)$.

        \begin{align*}
            \Expected_{\combined[\state]}^{{\processed[\combined]}, \sched}(\rew^j)
            &= \int_{\pi \in \Paths[{\processed[\combined]}](\combined[\state])} \sum_{i=0}^{\infty} \rew^j(\pi(i)) \, d \Pr^{\processed[\combined], \processed[\sched]}_{\combined[\state]}
            \\ &= 
            \sum_{i=0}^{\infty} \int_{\pi \in \Paths[{\processed[\combined]}](\combined[\state])} \rew^j(\pi(i)) \, d \Pr^{\processed[\combined], \processed[\sched]}_{\combined[\state]}
            \\ &= 
            \sum_{i=0}^{\infty} 
            \sum_{\pi \in \finPaths[{\processed[\combined]}](\combined[\state]), |\pi|=i}  \Pr^{\processed[\combined], \processed[\sched]}_{\combined[\state]}(\pi) \cdot \rew^j(\pi(i))
        \end{align*}
        Observe that $\rew^j(\combined[\state])=0$ and $\rew^j(\sdead)=0$ and hence we can rewrite the above further to
        \begin{align*}
            \Expected_{\combined[\state]}^{{\processed[\combined]}, \sched}(\rew^j)
            &= 
            \sum_{\state \in \combined[\states]}
            \sum_{i=1}^{\infty} 
            \sum_{\pi \in \finPaths[{\processed[\combined]}](\combined[\state]), |\pi|=i, \last(\pi)=\state} 
            \Pr^{\processed[\combined], \processed[\sched]}_{\combined[\state]}(\pi) \cdot \rew^j(\pi(i))
            \\ &=
            \sum_{\state \in \combined[\states]}
            \rew^j(\pi(i)) \cdot 
            \sum_{i=1}^{\infty} 
            \sum_{\pi \in \finPaths[{\processed[\combined]}](\combined[\state]), |\pi|=i, \last(\pi)=\state} 
            \Pr^{\processed[\combined], \processed[\sched]}_{\combined[\state]}(\pi) 
            \\ &= 
            \sum_{\state \in \combined[\states]}
            \rew^j(\pi(i)) \cdot 
            \sum_{i=1}^{\infty} 
            ((\combined[\Trans]^{\processed[\sched]})^i)_{\combined[\state], \state}
            = 
            \sum_{\state \in \combined[\states]}
            \rew^j(\pi(i)) \cdot y'_{\state} 
        \end{align*}

        \item[\eqref{eq:sinks}]  
        $y'_{\state,\dagger}$ is the expected number of times we leave $\state$ via $\dagger$ and move to $\sdead$.
        Hence, $\sum_{\state \in \combined[\states]} y'_{\state,\dagger}$ is the expected number of times we move from $\combined[\states]$ to $\sdead$. Since $\sdead$ is absorbing, this corresponds to the expected number of times we move to $\sdead$ \emph{for the first time}, i.e., $\Pr_{\combined[\state]}^{{\processed[\combined]}, \processed[\sched]}(\Finally \sdead)$.
        By assumption it holds that $\Pr_{\combined[\state]}^{{\processed[\combined]}, \processed[\sched]}(\Finally \sdead) = 1$ and hence $\sum_{\state \in \combined[\states]} y'_{\state,\dagger} = 1$.

        \item[\eqref{eq:nonneg}] $y'_{\state,\action} \geq 0$ holds by definition for all $\state \in \combined[\states]$, $\action \in \processedcombined[\Act](\state)$.
    \end{itemize}
    Thus, we have shown that the direction ``$\Rightarrow$'' holds.

    \proofsubparagraph{``$\Leftarrow$'':}
    Let $y^*$ be a feasible solution of \cref{fig:LP}.
    Let $\nnstates := \{ \state \in \combined[\states] \mid \sum_{\action \in \processed[{\combined[\Act]}](\state)} y^*_{\state,\action} > 0 \}$.
    
    We construct $\processed[\sched] \in \Scheds[{\processed[\combined]}]_{\textit{MR}}$ as follows. 
    For $\state \in \processed[{\combined[\states]}]$, $\action \in \processed[{\combined[\Act]}](\state)$, let
    \[
        \processed[\sched](\state,\action) = 
        \begin{cases}
            \frac{y^*_{\state,\action}}{\sum_{\action' \in \processed[{\combined[\Act]}]} y^*_{\state',\action}} & \text{if } \state \in \nnstates \\
            \frac{1}{|\processed[{\combined[\Act](\state)}]|} & \text{otherwise}~.
        \end{cases}
    \]
    We define a vector $y'$ from $\processed[\sched]$ as in ``$\Rightarrow$'', i.e., for $\state \in \combined[\states]$ and $\action \in \processed[{\combined[\Act]}](\state)$, we let 
    \begin{align*}
        y'_{\state, \action} &= \sum_{n=0}^{\infty} ((\combined[\Trans]^{\processed[\sched]})^n)_{\combined[\state], \state} \cdot \processed[\sched](\state, \action) 
        \\
        y'_{\state} &= \sum_{\state' \in \combined[\states]} \sum_{\action' \in \combined[\Act](\state')} \combined[\Trans](\state', \action', \state) \cdot y'_{\state', \action'}
        ~.
    \end{align*}
    Note that we do not know yet that these are well-defined.

    In the following, we show that $\nnstates$ contains all states reachable from $\combined[\state]$ under $\processed[\sched]$, which we then use to show that the constructed vector $y'$ equals the feasible solution $y^*$ and is thus well-defined.
    Finally, we show that the constructed scheduler $\processed[\sched]$ satisfies the desired properties.

    Let us first define some useful notation.    
    For $\state \in \combined[\states]$, let 
    \begin{align*}
        y^*_\state &= \sum_{\state' \in \combined[\states]} \sum_{\action' \in \combined[\Act](\state')} \combined[\Trans](\state', \action', \state) \cdot y^*_{\state', \action'} \\
        z^*_{\state} &= \sum_{\action \in \processed[{\combined[\Act]}](\state)} y^*_{\state,\action}
        ~.
    \end{align*}
    We define the set of all states that can `reach' $\sdead$ w.r.t.\ $y^*$ as $W = $ 
    \[
        \left\{ \state \in \combined[\states] \;\middle|\; 
            \exists \pi = \state_0\action_0 \ldots \state_{|\pi|} \in \finPaths[{\processed[\combined]}](\state) .\ 
            \prod_{i=0}^{|\pi|-1} y^*_{\state_i,\action_i} \cdot \processed[{\combined[\Trans]}](\state_i,\action_i,\state_{i+1}) > 0 \wedge \state_{|\pi|} \in \sdead 
        \right\}
        ~.
    \]
    Observe that $\sdead \cap W = \emptyset$ by definition.
    Observe further that $W \subseteq \nnstates$: By contraposition, let $\state \not\in \nnstates$, then $y^*_{\state,\action} = 0$ for all $\action \in \processed[{\combined[\Act]}](\state)$ and hence there cannot exist any path satisfying the requirement for membership in $W$.

    \proofsubparagraph{Claim:} \emph{$\nnstates$ contains all reachable states.}
    
    We show that for all $\state \in \processed[{\combined[\states]}]$ reachable from $\combined[\state]$ under $\processed[\sched]$, it must hold that $\state \in \nnstates$, by induction on the length of the shortest path from $\combined[\state]$ to $\state$.

    \emph{IB:} $s = \combined[\state]$. 
    Since $y^*$ is feasible, by constraint (\ref{eq:EVTs}) we have
    \[
        \sum_{\action \in \processed[{\combined[\Act]}](\combined[\state])} y^*_{\combined[\state], \action} = \sum_{\state' \in \combined[\states]} \sum_{\action' \in \combined[\Act](\state')} \combined[\Trans](\state', \action', \combined[\state]) \cdot y^*_{\state', \action'} + 1
        ~.
    \]
    Since $y^*_{\state, \action} \geq 0$ by constraint (\ref{eq:nonneg}) for all $\state \in \combined[\states]$, $\action \in \processed[{\combined[\Act]}](\state)$, and $\combined[\Trans](\state', \action', \combined[\state]) \geq 0$ by definition, we must have $\sum_{\action \in \processed[{\combined[\Act]}](\combined[\state])} y_{\combined[\state], \action} > 0$, i.e., $\combined[\state] \in \nnstates$.

    \emph{IS:} Assume we can reach $\state \in \processed[{\combined[\states]}]$ under $\processed[\sched]$ in $k>0$ steps from $\combined[\state]$ with positive probability, and we cannot reach $\state$ in less than $k$ steps.
    There must exist some predecessor $\state^p$ of $\state$ that can be reached from $\combined[\state]$ in $k-1$ steps.
    Since $y^*$ is feasible and thus satisfies constraint (\ref{eq:EVTs}), we have
    \begin{align*}
        \sum_{\action \in \processed[{\combined[\Act]}](\state)} y^*_{\state,\action} 
        =
        \sum_{\state' \in \combined[\states]} \sum_{\action' \in \processed[{\combined[\Act]}](\state')} y^*_{\state',\action'} \processed[{\combined[\Trans]}](\state', \action', \state)
        \geq
        \sum_{\action' \in \processed[{\combined[\Act]}](\state^p)} y^*_{\state^p,\action'} \processed[{\combined[\Trans]}](\state^p, \action', \state)
        \overset{\textit{IH}}{>} 0 ~.
    \end{align*}

    \proofsubparagraph{Claim:} \emph{$y^* = y'$ on $W$.}
    
    We show that $y^*_{\state,\action} = y'_{\state,\action}$ for $\state \in W \subseteq \nnstates$ and $\action \in \processed[{\combined[\Act]}](\state)$, step by step.
    \begin{itemize}
        \item We first show that $\Iverson{\state = \combined[\state]} = z^*_{\state} - \sum_{\state' \in \nnstates} (\combined[\Trans]^{\processed[\sched]})_{\state',\state} \cdot z^*_{\state'} $ for all $\state \in \nnstates$.
        Since $y^*$ satisfies constraint (\ref{eq:EVTs}), we have
        \begin{align*}
            \Iverson{\state = \combined[\state]}
            &= z^*_{\state} - \sum_{\state' \in \combined[\states]} \sum_{\action' \in \processed[{\combined[\Act]}]} y^*_{\state',\action'} \combined[\Trans](\state',\action',\state) 
            \\ &= z^*_{\state} - \sum_{\state' \in \nnstates} \sum_{\action' \in \processed[{\combined[\Act]}]} y^*_{\state',\action'} \combined[\Trans](\state',\action',\state) \cdot \frac{\sum_{\action \in \processed[{\combined[\Act]}](\state)} y^*_{\state,\action}}{\sum_{\action \in \processed[{\combined[\Act]}](\state)} y^*_{\state,\action}}
            \\ &= z^*_{\state} - \sum_{\state' \in \nnstates} \sum_{\action' \in \processed[{\combined[\Act]}]} \combined[\Trans](\state',\action',\state) \cdot \processed[\sched](\state',\action') \cdot z^*_{\state'} 
            \\ &= z^*_{\state} - \sum_{\state' \in \nnstates} (\combined[\Trans]^{\processed[\sched]})_{\state',\state} \cdot z^*_{\state'} 
            ~.
        \end{align*}

        \item Let us now show that for $\state \in W$ it holds that $\sum_{\action \in \processed[{\combined[\Act]}]} y^*_{\state,\action} = \sum_{\action \in \processed[{\combined[\Act]}]} y'_{\state,\action}$.

        Observe that $\state \in W$ implies that all predecessors of $\state$ must also be in $W$, and hence, using the above, for $\state \in W$ we have 
        \begin{align*}
            \Iverson{\state = \combined[\state]} = z^*_{\state} - \sum_{\state' \in \nnstates} (\combined[\Trans]^{\processed[\sched]})_{\state',\state} \cdot z^*_{\state'} = z^*_{\state} - \sum_{\state' \in W} (\combined[\Trans]^{\processed[\sched]})_{\state',\state} \cdot z^*_{\state'}
            ~.
        \end{align*}

        Let $\textbf{e}_{\combined[\state]}$ denote the unit vector for the index corresponding to $\combined[\state]$ (according to the order we fixed for constructing $\combined[\Trans]^{\processed[\sched]}$).
        Let $(\combined[\Trans]^{\processed[\sched]})_{W,W}$ denote the restriction of $\combined[\Trans]^{\processed[\sched]}$ to $W$. 
        Observe that $(\combined[\Trans]^{\processed[\sched]})_{W,W}$ must be sub-stochastic since all $\state \in W$ are transient, and hence $(I-(\combined[\Trans]^{\processed[\sched]})_{W,W})^{-1} = \sum_{n=0}^{\infty} ((\combined[\Trans]^{\processed[\sched]})_{W,W})^n$ exists.
        Hence, we can rewrite $\Iverson{\state = \combined[\state]} = z^*_{\state} - \sum_{\state' \in W} (\combined[\Trans]^{\processed[\sched]})_{\state',\state} \cdot z^*_{\state'}$ as 
        \[
            \textbf{e}_{\combined[\state]}|_{W} = (z^*)^T|_{W} \cdot (I-(\combined[\Trans]^{\processed[\sched]})_{W,W})
        \]
        and further as
        \[
            (z^*)^T|_{W} = \textbf{e}_{\combined[\state]}|_{W} \cdot \sum_{n=0}^{\infty} ((\combined[\Trans]^{\processed[\sched]})_{W,W})^n
            ~,
        \]
        i.e., $z^*_{\state} = \sum_{i=0}^{\infty} ((\combined[\Trans]^{\processed[\sched]})^i)_{\combined[\state],\state}$.
        Thus, for $\state \in W$ we have
        \begin{align*}
            \sum_{\action \in \processed[{\combined[\Act]}](\state)} y^*_{\state,\action} 
            \overset{\textit{def.}}{=} z^*_{\state} 
            &= \sum_{n=0}^{\infty} ((\combined[\Trans]^{\processed[\sched]})^n)_{\combined[\state],\state}
            = \sum_{n=0}^{\infty} ((\combined[\Trans]^{\processed[\sched]})^n)_{\combined[\state],\state} 
            \cdot \overbrace{\sum_{\action \in \processed[{\combined[\Act]}](\state)} \processed[\sched](\state, \action)}^{=1} \\
            &= \sum_{\action \in \processed[{\combined[\Act]}](\state)} \sum_{n=0}^{\infty} ((\combined[\Trans]^{\processed[\sched]})^n)_{\combined[\state],\state} 
            \cdot \processed[\sched](\state, \action) 
            \overset{\textit{def.}}{=} \sum_{\action \in \processed[{\combined[\Act]}](\state)} y'_{\state,\action}
            ~.
        \end{align*}

        \item
        Finally, for $\state \in W \subseteq \nnstates$ and $\action \in \processed[{\combined[\Act]}](\state)$ we have
        \begin{align*}
            y'_{\state,\action} 
            &= \left(\sum_{n=0}^{\infty} ((\combined[\Trans]^{\processed[\sched]})^n)_{\combined[\state],\state} \right)
            \cdot \frac{y^*_{\state,\action}}{\sum_{\actiontwo \in \processed[{\combined[\Act]}]} y^*_{\state,\actiontwo}} && \text{// def.\ of $y'$ and $\processed[\sched]$} \\
            &= \left(\sum_{\actiontwo \in \processed[{\combined[\Act]}](\state)} y'_{\state,\actiontwo} \right)
            \cdot \frac{y^*_{\state,\action}}{\sum_{\actiontwo \in \processed[{\combined[\Act]}]} y^*_{\state,\actiontwo}} 
            && \text{// analogously to ``$\Rightarrow$''} \\
            &= \left(\sum_{\actiontwo \in \processed[{\combined[\Act]}](\state)} y'_{\state,\actiontwo} \right)
            \cdot \frac{y^*_{\state,\action}}{\sum_{\actiontwo \in \processed[{\combined[\Act]}]} y'_{\state,\actiontwo}} 
            = y^*_{\state,\action}
            && \text{// by the above}
        \end{align*}
    \end{itemize}

    \proofsubparagraph{Claim:} \emph{$\processed[\sched]$ satisfies the required properties.}
    
    Analogously to ``$\Rightarrow$'' we can show that for the vector $y'$ we constructed from $\processed[\sched]$, it holds that $\sum_{\state \in \combined[\states]} y'_{\state,\dagger}$ is the expected number of times we move from $\combined[\states]$ to $\sdead$ \emph{for the first time} under $\processed[\sched]$, i.e., $\Pr_{\combined[\state]}^{{\processed[\combined]}, \processed[\sched]}(\Finally \sdead)$, so
    \begin{align*}
        \Pr_{\combined[\state]}^{{\processed[\combined]}, \processed[\sched]}(\Finally \sdead) 
        = \sum_{\state \in \combined[\states]} y'_{\state,\dagger}
        &\geq \sum_{\state \in W} y'_{\state,\dagger} &&  \\
        &= \sum_{\state \in W} y^*_{\state,\dagger} && \text{// by the above} \\
        &= \sum_{\state \in \combined[\states]} y^*_{\state,\dagger} && \text{// def.\ of } W \\
        &= 1 && \text{// $y^*$ is feasible sol.}
    \end{align*}

    This also implies that $y'_{\state}=0$ for $\state \in \combined[\states] \setminus W$, since we would not reach $\sdead$ with probability 1 if we visited a state that cannot reach $\sdead$.

    Let $j=1,\ldots,\numconj$.
    Analogously to ``$\Rightarrow$'' we can show that \[
        \Expected_{\combined[\state]}^{{\processed[\combined]}, \processed[\sched]}(\rew^j) = \sum_{\state \in \combined[\states]} \rew^j(\state) y'_{\state} ~.
    \]
    From $y'_{\state}=0$ for $\state \in \combined[\states] \setminus W$ it follows further that
    \[
        \sum_{\state \in \combined[\states]} \rew^j(\state) y'_{\state} = \sum_{\state \in W} \rew^j(\state) y'_{\state} ~.
    \]
    Since $y'_{\state,\action} = y^*_{\state,\action}$ for $\state \in W$, $\action \in \processed[{\combined[\Act]}](\state)$, and all predecessors of $\state \in W$ must also be in $W$, it follows that \[
        y'_{\state} = \sum_{\state' \in W} \sum_{\action' \in \combined[\Act](\state')} \combined[\Trans](\state', \action', \state) \cdot y'_{\state', \action'} =  \sum_{\state' \in W} \sum_{\action' \in \combined[\Act](\state')} \combined[\Trans](\state', \action', \state) \cdot y^*_{\state', \action'} = y^*_{\state}
    \] 
    for $\state \in W$ and thus \[
        \sum_{\state \in W} \rew^j(\state) y'_{\state} = \sum_{\state \in W} \rew^j(\state) y^*_{\state}
        ~.
    \]
    By assumption, $\sum_{\state \in \combined[\states]} \rew^j(\state) y^*_{\state} = x_j$ and hence $\Expected_{\combined[\state]}^{{\processed[\combined]}, \processed[\sched]}(\rew^j) = x_j$ as required.
\end{proof}

\subsection{Proof of \cref{th:MOA_correctness}}
\input{appendix/moa_via_lp}

%% file: appendix/moa_via_lp.tex
\label{sec:moa_via_lp}

\MOACorrectness*

\begin{proof}
    Correctness follows from \cref{le:conj_LP} and \cref{th:moa_preproc_transfer-MR}. 

    For the runtime complexity:
    Observe that the linear real arithmetic encoding is not necessarily a linear program, whose standard form only allows $\leq,\geq$-constraints, due to the presence of $\approx_\epsilon$-, $>$-, and $\not\approx_{\epsilon}$-predicates in the encoding.
    However, we can solve the constructed linear real arithmetic encoding (\cref{fig:LP}) using successive calls to LP solvers, as outlined in \cref{alg:LP}. We show that this approach is correct and yields the desired runtime complexity in \cref{th:lra_via_lp}.
\end{proof}

\begin{algorithm}[t]
    \caption{Solving linear real arithmetic for MOA queries with comparison operators $\{>, \geq, \approx_{\epsilon}, \not\approx_{\epsilon} \mid \epsilon \in \Q_{\geq 0} \}$ via LPs} 
    \label{alg:LP}
    \Input{Linear real arithmetic encoding $L$ from \cref{fig:LP} with comparison operators $\{>, \geq, \approx_{\epsilon}, \not\approx_{\epsilon} \mid \epsilon \in \Q_{\geq 0} \}$}
    \Output{Does $L$ have a feasible optimal solution?
    }
    {$L' \gets$ \textsc{EliminateStrict}(\textsc{EliminateEquality}($L$))} \tcp*{See p.\ \pageref{rem:LP_elim}}
    {$D' \gets $ $\not\approx_{\epsilon}$-constraints of $L'$ with $\epsilon>0$\;}
    \For(\tcp*[f]{Check all possibilities for satisfying all $\not\approx$-constraints}){$J \subseteq D'$}{
        {$J^< \gets$ \textsc{EliminateStrict}($J$ with $x\not\approx_\epsilon y$ replaced by $x< y+\epsilon$) \;}
        {$(D' \setminus J)^> \gets$ \textsc{EliminateStrict}($D' \setminus J$ with $x\not\approx_\epsilon y$ repl.\ by $x>y -\epsilon$) \;}
        {$L'' \gets$ $L'$ with $J$ replaced by $J^< \wedge (D' \setminus J)^>$\;}
        {$C \gets$ $\leq,\geq$-constraints of $L''$ \;}
        {$d_1, \ldots, d_k \gets$ $\neq$-constraints of $L''$ \;}
        \For(\tcp*[f]{Check $\neq$-constraints one by one}){$i=1,\ldots,k$}{
            {$d_i^< \gets$ \textsc{EliminateStrict}($d_i$ with $\neq$ replaced by $<$)\;}
            {$d_i^> \gets$ \textsc{EliminateStrict}($d_i$ with $\neq$ replaced by $>$) \;}
            \tcp{Use standard LP solver ($C$, $d_i^<$, $d_i^>$ contain only $\leq, \geq$-constraints) }
            \If{$C \wedge d_i^<$ does not have feasible opt.\ sol.}{
                \If{$C \wedge d_i^>$ does not have feasible opt.\ sol.}{
                    \Return{False}
                }
            }
        }
        \Return{True}
    }
    \Return{False}
\end{algorithm}

\begin{theorem}
    \label{th:lra_via_lp}
    \cref{alg:LP} adheres to its input-output specification and can be implemented with worst-case running time of $\mathcal{O}(2^{\numconj_{\not\approx}} \cdot \textit{poly}(\numconj \cdot \numsum \cdot |\mdp| \cdot 2^{\numconj \cdot \numsum}))$ 
    where $\numconj$ is the number of predicates, $\numconj_{\not\approx}$ the number of predicates using $\not\approx_\epsilon$ with $\epsilon>0$, and $\numsum$ the number of probability operators per conjunct in the original \ConjRelReach property.
\end{theorem}

\begin{proof}
    Correctness:
    \label{rem:LP_elim}
\begin{itemize}
    \item $\boldsymbol{\approx_{\epsilon}}$: 
    We can eliminate $\approx_{\epsilon}$-constraints by transforming them to a conjunction of a $\leq$-constraint and a $\geq$-constraint ($x \approx_\epsilon y \iff x \geq y-\epsilon \land x \leq y+\epsilon$).

    \item $\boldsymbol{<,>}$: 
    We can eliminate $<,>$-constraints by introducing a new variable $z$, replacing the constraint $x<y$ (or $x>y$) with $x\leq y-z$ (or $x\geq y+z$) and additionally adding the goal of maximizing $z$~(see, e.g., \cite{dutertreFastLinearArithmetic2006,nalbachExtendingFundamental2021}).
    If there exists an optimal solution to the resulting system of inequalities that assigns a positive value to $z$, then there must exist a feasible solution to the original problem, and vice versa.

    Repeatedly eliminating strict constraints yields multiple optimization objectives, i.e., we end up with a multi-objective LP.
    
    \item $\boldsymbol{\neq}$:
    Assume a linear real arithmetic encoding $L$ containing only $\neq, \leq, \geq$-constraints.
    Let $F$ be the set of solutions for the encoding without the $\neq$-constraints, then $F$ must be convex since the encoding without the $\neq$-constraints forms an LP.
    Observe that each $\neq$-constraint only excludes a hyperplane from $F$. 
    Assume the whole equation system $L$ cannot be satisfied but $F$ is non-empty, then $F$ must be contained in the union of the hyperplanes corresponding to the $\neq$-constraints.
    Assume $F$ is contained in more than a single such hyperplane, then the solution space would not be convex, which is a contradiction. 
    Hence, $F$ must be contained in a single hyperplane corresponding to one of the $\neq$-equations. 

    Thus, we can solve $L$ by first solving the equation system without the $\neq$-constraints and then checking the $\neq$-constraints one by one, see also~\cite[Th.~1]{lassezIndependenceNegative1989}. 
    In other words, we can solve a linear real arithmetic encoding containing only $\neq, \leq, \geq$-constraints by solving a \emph{linear} number of LPs.
    
    \item $\boldsymbol{\not\approx_\epsilon}$:
    For $\not\approx_{\epsilon}$-constraints with $\epsilon>0$ the trick from $\neq$ does not work since now a single $\not\approx_{\epsilon}$-constraint may exclude more than just a hyperplane and the union of the areas covered by the $\not\approx_{\epsilon}$-constraints may be convex.

    Instead, each $\not\approx_{\epsilon}$-constraint introduces a disjunction ($x \not\approx_\epsilon y \iff x < y-\epsilon \lor x>y+\epsilon$).
    Let $L$ be the original encoding, let $L^<$ be $L$ with the $\not\approx_{\epsilon}$-constraint replaced by the $<$-disjunct, and let $L^>$ be $L$ with the $\not\approx_{\epsilon}$-constraint replaced by the $>$-disjunct.
    Then $L$ has a solution if either $L^<$ or $L^>$ have a solution.

    Hence, we can solve a linear real arithmetic encoding $L$ containing only $\not\approx_{\epsilon}, \leq, \geq$-constraints by solving an \emph{exponential} number of LPs.
\end{itemize}

    Runtime: 
    If the query contains $\numconj_{\not\approx}$ constraint with $\not\approx_{\epsilon}$ with $\epsilon>0$, then \cref{alg:LP} solves up to $2^{\numconj_{\not\approx}}$ queries with only $=$, $\neq$ and $\leq$ constraints.
    If the query does not contain $\not\approx_{\epsilon}$ with $\epsilon>0$: 
    Let $\numconj_{\neq}$ be the number of $\neq$-constraints in the linear real arithmetic encoding, then \cref{alg:LP} solves up to $\numconj_{\neq}$ LPs with $=$- and $\leq$-constraints.
    For each LP, the number of constraints is $|\indexconjpartelt[\states]| + |\conjpartelt| + |\indexconjpartelt[\states]| + |\indexconjpartelt[\states]| \cdot |\processed[{\indexconjpartelt[\Act]}]| \in \mathcal{O}(\textit{poly}(|\indexconjpartelt[\states]| \cdot |\processed[{\indexconjpartelt[\Act]}]|))$.
    Observe 
    $|\indexconjpartelt[\states]| \cdot |\processed[{\indexconjpartelt[\Act]}]|
    \leq (\numconj \cdot \numsum \cdot |\states| \cdot 2^{\numconj \cdot \numsum} + 1) \cdot (|\Act| + 2)$, hence
    each constructed LP can be solved in time $\mathcal{O}(\textit{poly}(\numconj \cdot \numsum \cdot |\states| \cdot 2^{\numconj \cdot \numsum} \cdot |\Act|))$ since LPs can be solved in time polynomial in the number of constraints~\cite{khachiyanPolynomialAlgorithm1979}.
\end{proof}

%% file: appendix/conjunction_compl_proofs.tex
\section{Proofs for \cref{sec:conj_complexity}}
\label{app:conj_complexity_proofs}

\subsection{Proof of \cref{th:conj_special}}
\label{app:conj_approx-diseq_NP}

\conjSpecial*

We have already argued for membership in \PTIME and \EXPTIME in \cref{sec:conj_complexity}.
It remains to show that \ConjRelReach is strongly \NP-hard in cases $\mathit{(a)}$ and $\mathit{(b)}$. 
We do so by showing the following lemma:

\begin{lemma}
    \ConjRelReach is strongly \NP-hard even if $T_{1,1} = \ldots = T_{m,l}$ and this target set is absorbing.
\end{lemma}

\begin{proof}
    We show strong \NP-hardness by reduction from 3SAT~\cite{karpReducibilityCombinatorial1972}.
    Given a 3SAT formula $\phi = \bigwedge_{j=1}^{M} (l_{j,1} \lor l_{j,2} \lor l_{j,3})$ over variables $x_1,\ldots,x_N$ with literals $l_{j,h} \in \{x_i, \overline{x_i} \mid i \in \{1,\ldots,N\}\}$ for $j \in \{1, \ldots, M\}$, $h \in \{1,2,3\}$.
    We construct the MDP $\mdp_\phi = (\state_\phi, \Act_\phi, \Trans_\phi)$, depicted in \cref{fig:conj_NP-hard}, with 
    \begin{itemize}
        \item $\state_\phi = \{s, t\}$,
        \item $\Act_\phi = \{\alpha, \beta\}$, and 
        \item $\Trans_\phi(s,\alpha,s) = 1$, $\Trans_\phi(s,\beta,t) = 1$, $\Trans_\phi(t,\beta,t) = 1$, and all other transition probabilities are 0.
    \end{itemize}
    We construct the following \ConjRelReach property with a single, absorbing target set $T = \{t\}$:
    \begin{align*}
        \phi' = \exists \sched_1, \ldots, \sched_N \in \Scheds[\mdp_\phi] .\ 
        & \bigwedge_{j=1}^{N}
        \Pr^{\sched_i}_{s}(\Finally T)
        \not\approx_{\nicefrac{1}{4}}
        \nicefrac{1}{2}
        \land 
        \bigwedge_{j=1}^{M} 
        \sum_{h=1}^{3} \alpha_{j,1} \left( \Pr^{\sched_{k_{j,1}}}_{s}(\Finally T) - \nicefrac{1}{2} \right) 
        \geq -1
    \end{align*}
    where 
    $\alpha_{j,h} = \begin{cases}
        2 & \text{if } l_{j,h} \text{ is positive} \\
        -2 & \text{otherwise}
    \end{cases}$
    and
    $k_{j,h}$ is the unique $i \in \{1,\ldots,N\}$ with $l_{j,h} \in \{x_i, \overline{x_i}\}$.
    We now show that $\phi$ is satisfiable iff $\phi'$ holds.

    `$\Rightarrow$':
    Assume $\phi$ is satisfiable, let $\mathcal{I} \colon \{x_1, \ldots, x_N\} \to \{0,1\}$ be a satisfying interpretation.
    For $i \in \{1,\ldots, N\}$, we construct a memoryless deterministic scheduler $\sched_i \in \Scheds[\mdp_\phi]$ by letting $\sched(s)(\beta) = \mathcal{I}(x_i)$ and $\sched(s)(\beta) = 1 - \mathcal{I}(x_i)$.
    Then, $\Pr^{\sched_i}_{s}(\Finally T) = \mathcal{I}(x_i) \not\approx_{\nicefrac{1}{4}} \nicefrac{1}{2}$ for $i \in \{1,\ldots, N\}$ and hence the first part of $\phi'$ holds.
    Further, for $j \in \{1, \ldots, M\}$ and $h \in \{1,2,3\}$, we have 
    \begin{align*}
        \alpha_{j,h} \left( \Pr^{\sched_{k_{j,h}}}_{s}(\Finally T) - \nicefrac{1}{2} \right)
        = \begin{cases}
            2 \cdot \left( 1 - \nicefrac{1}{2} \right) = 1 & \text{if } l_{j,h} = x_{k_{j,h}} \land \mathcal{I}(x_{k_{j,h}}) = 1 \\ 
            2 \cdot \left( 0 - \nicefrac{1}{2} \right) = -1 & \text{if } l_{j,h} = x_{k_{j,h}} \land \mathcal{I}(x_{k_{j,h}}) = 0 \\ 
            -2 \cdot \left( 1 - \nicefrac{1}{2} \right) = -1 & \text{if } l_{j,h} = \overline{x_{k_{j,h}}} \land \mathcal{I}(x_{k_{j,h}}) = 1 \\
            -2 \cdot \left( 0 - \nicefrac{1}{2} \right) = 1 & \text{if } l_{j,h} = \overline{x_{k_{j,h}}} \land \mathcal{I}(x_{k_{j,h}}) = 0
        \end{cases}
    \end{align*}
    Since $\mathcal{I}$ is a satisfying interpretation, for each $j\in \{1, \ldots, M\}$ there must exist at least one $h \in \{1,3\}$ with $l_{j,h}$ being positive, and hence $\sum_{h=1}^{3} \alpha_{j,1} \left( \Pr^{\sched_{k_{j,1}}}_{s}(\Finally T) - \nicefrac{1}{2} \right) \in \{-1, 1, 3\}$ and thus the second part of $\phi'$ holds.

    `$\Leftarrow$': 
    Assume $\phi'$ holds. Let $\sched_1, \ldots, \sched_N$ be witnesses for $\phi'$.
    Then for each $i \in \{1, \ldots, N\}$ it must hold that $\Pr^{\sched_i}_{s}(\Finally T) \not\approx_{\nicefrac{1}{4}} \nicefrac{1}{2}$, i.e., $\Pr^{\sched_i}_{s}(\Finally T) \in [0,\nicefrac{1}{4}) \cup (\nicefrac{3}{4}, 1]$.
    Hence, the following interpretation for $\phi$ is well-defined:
    \[
        \mathcal{I}(x_i) = \begin{cases}
            1 & \text{if } \Pr^{\sched_i}_s(\Finally T) > \nicefrac{3}{4} \\ 
            0 & \text{if } \Pr^{\sched_i}_s(\Finally T) < \nicefrac{1}{4}
        \end{cases}
    \]
    Let us now see that $\mathcal{I}$ is a satisfying interpretation for $\phi$.
    From the above observation on $\Pr^{\sched_i}_{s}(\Finally T)$ it follows that for $j \in \{1, \ldots, M\}$ and $h \in \{1,2,3\}$ we have $\Pr^{\sched_{k_{j,1}}}_{s}(\Finally T) - \nicefrac{1}{2} \in [-\nicefrac{1}{2}, -\nicefrac{1}{4}) \cup (\nicefrac{1}{4}, \nicefrac{1}{2}]$ and thus
    \begin{align*}
        \alpha_{j,h} \left( \Pr^{\sched_{k_{j,h}}}_{s}(\Finally T) - \nicefrac{1}{2} \right) 
        \in 
        \begin{cases}
            (\nicefrac{1}{2}, 1] & \text{if } l_{j,h} = x_{k_{j,h}} \land \mathcal{I}(x_{k_{j,h}}) = 1 \\ 
            [-1,-\nicefrac{1}{2}) & \text{if } l_{j,h} = x_{k_{j,h}} \land \mathcal{I}(x_{k_{j,h}}) = 0 \\ 
            [-1,-\nicefrac{1}{2}) & \text{if } l_{j,h} = \overline{x_{k_{j,h}}} \land \mathcal{I}(x_{k_{j,h}}) = 1 \\
            (\nicefrac{1}{2}, 1] & \text{if } l_{j,h} = \overline{x_{k_{j,h}}} \land \mathcal{I}(x_{k_{j,h}}) = 0
        \end{cases}
    \end{align*}
    By assumption, for $j \in \{1, \ldots, M\}$ it must hold that $\sum_{h=1}^{3} \alpha_{j,h} \left( \Pr^{\sched_{k_{j,h}}}_{s}(\Finally T) - \nicefrac{1}{2} \right) \geq -1$, and hence there must exists at least one $h \in \{1,2,3\}$ with $\alpha_{j,h} \left( \Pr^{\sched_{k_{j,h}}}_{s}(\Finally T) - \nicefrac{1}{2} \right) > \nicefrac{1}{2}$, i.e., with $l_{j,h} = x_{k_{j,h}} \land \mathcal{I}(x_{k_{j,h}}) = 1$ or $l_{j,h} = \overline{x_{k_{j,h}}} \land \mathcal{I}(x_{k_{j,h}}) = 0$.
    Hence, under $\mathcal{I}$, for each clause $j \in \{1,\ldots,M\}$ at least one literal holds and thus $\mathcal{I}$ is a satisfying interpretation for $\phi$.

    It remains to show that the above construction defines a pseudo-polynomial time transformation to \ConjRelReach.
    The constructed MDP has a fixed number of states and transitions, and all transition probabilities are 0 or 1.
    The property quantifies over $N$ schedulers and contains a conjunction of size $N$ as well as a conjunction of size $M$, where $N$ is the number of variables and $M$ the number of clauses in the original 3SAT formula. All coefficients and bounds in the property ($\nicefrac{1}{2}$, $\nicefrac{1}{4}$, $-1$) are fixed.
    Hence, the magnitude of the largest number occurring in the constructed \ConjRelReach instance is polynomial in the size of the 3SAT instance.
\end{proof} 

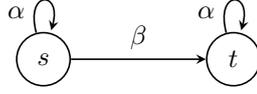
\begin{figure}
    \centering
    \begin{tikzpicture}[on grid,node distance=25mm and 25mm,semithick,>=stealth]
        \node[state] (s) {$s$};
        \node[state, right= of s] (t) {$t$};

        \path[->]
            (s) edge[loop above] node[left, pos=0.2] {$\alpha$} (s)
            (s) edge node[above] {$\beta$} (t)
            (t) edge[loop above] node[left, pos=0.2] {$\alpha$} (t);
    \end{tikzpicture}
    \caption{Illustration of the MDP construction for the reduction from 3SAT to \ConjRelReach.}
    \label{fig:conj_NP-hard}
\end{figure}

%% file: appendix/case-studies.tex
\section{Details on the Benchmarks for Probabilistic Hyperproperties (\cref{sec:hyperprob-benchmarks})}
\label{app:hyperprob-benchmarks}

The contents of this section are taken verbatim from~\cite{extendedVersion}.

\subparagraph*{\TA}
    A standard application of probabilistic hyperproperties
    is to check whether an implementation of modular exponentiation for RSA public-key encryption 
    has a side-channel \emph{timing leak} \cite{abrahamProbabilisticHyperproperties2020}. 
    Concretely, the setting is as follows: One thread computes $a^b \text{ mod } n$ for a given plaintext $a$, encryption key $b$, and modulus $n$ (all non-negative integers), 
    while a second (attacker) thread tries to infer the value of $b$ by keeping a counter $c$ to measure the time taken by the first thread.
    \cref{fig:modexp} shows the implementation of the modular exponentiation algorithm that we want to check for side-channels.

    If $b$ has $k$ bits, we can model the parallel execution of the two threads as an MDP with $2^k$ initial states $\Init$, each representing a different value of $b$.
    In our models, a (memoryful randomized) scheduler for the MDP corresponds to a thread scheduler.
    In contrast, in \cite{abrahamParameterSynthesis2020,andriushchenkoDeductiveController2023} the authors hard-coded fair thread schedulers in their models, and a (memoryless deterministic) scheduler for the MDP corresponds to a secret input. 

    The goal is to now check whether the implementation satisfies \emph{scheduler-specific probabilistic observational determinism} (SSPOD) \cite{ngoConfidentialityProbabilistic2013}, i.e., whether for any scheduling of the two threads that no information about the secret input (here, $b$) can be inferred by the publicly observable information (here, the time taken by the modular exponentiation thread).
    
    In \cite{abrahamProbabilisticHyperproperties2020}, the desired property is formulated in \HyperPCTL as follows:
    \[\forall \varsched_1 \forall \varsched_2 .\ \forall \varstate_1(\varsched_1) \forall \varstate_2(\varsched_2) .\ 
    (\init_{\varstate_1} \wedge \init_{\varstate_2}) 
    \implies 
    \bigwedge_{j=0}^{2k} \Prob(\Finally (c=j)_{\varstate_1}) = \Prob(\Finally (c=j)_{\varstate_2})\]
    where 
    $\init$ marks initial states for different values of the secret input $b$, 
    $b$ has $k$ bits, and 
    $c$ is the counter of the attacking side-channel thread.
    An MDP $\mdp$ satisfies the above \HyperPCTL formula iff it holds that
    \begin{align*}
        &\forall \sched_1, \sched_2 \in \Scheds .\ \bigwedge_{\state_1,\state_2 \in \Init} \bigwedge_{j=0}^{2k} \Pr^{\sched_1}_{\state_1}(\Finally (c=j)) = \Pr^{\sched_2}_{\state_2}(\Finally (c=j))
        ~.
    \end{align*}
    Since universal quantification distributes over conjunction, we can reformulate this to a conjunction of universally quantified relational reachability properties as follows:
    \begin{align*}
        \bigwedge_{\state_1,\state_2 \in \Init} \bigwedge_{j=0}^{2k} \forall \sched_1, \sched_2 .\ \Pr^{\sched_1}_{\state_1}(\Finally (c=j)) = \Pr^{\sched_2}_{\state_2}(\Finally (c=j))
        ~.
    \end{align*}

    In our experiments, we check the first conjunct of this property, i.e., we compare initial values $b=0$ and $b=1$ for $j=0$.
    We check the property for $M \in \{8,16,24,28,32\}$ where $M=2k$ where $k$ is the number of bits for the secret input $b$. 

\begin{figure}[t]
    \centering
    \begin{subfigure}[t]{0.45\textwidth}
    \centering
    \begin{lstlisting}[style=CStyle,tabsize=2,language=ML,basicstyle=\scriptsize,escapechar=/,backgroundcolor=\color{white},]
int mexp(a,b,n){
  d = 0; e = 1; i = k;
  while (i >= 0){
    i = i-1; d = d*2;
    e = (e*e) % n;
    if (b(i) = 1)
      d = d+1;
      e = (e*a) % n;
    }
} \end{lstlisting}
    \end{subfigure}%
    \quad
    \begin{subfigure}[t]{0.45\textwidth}
    \centering
    \begin{lstlisting}[style=CStyle,tabsize=2,language=ML,basicstyle=\scriptsize,escapechar=/,backgroundcolor=\color{white},]
  t = new Thread(mexp(a,b,n)); 
  c = 0; M = 2 * k;
  while (c < M & !t.stop) c++;   \end{lstlisting}
    \end{subfigure}
    \caption{
    \color{cavcolor}
    Modular exponentiation (left) and attacker thread (right) \cite{abrahamParameterSynthesis2020}.
    } 
    \label{fig:modexp}
\end{figure}

\subparagraph*{\PW} 
    Another example for an implementation with possible timing leaks is a careless implementation of string comparison for password verification where a timing leak may reveal information about the password \cite{abrahamParameterSynthesis2020}.
    As for \TA, we want check whether an attacker thread observing the time taken by the main thread (executing a string comparison algorithm) can infer information about the password.

    \cref{fig:string} displays the string comparison algorithm that we want to check for information-leaks.
\begin{figure}[t]
    \centering
    \begin{subfigure}[t]{0.5\textwidth}
    \centering
    \begin{lstlisting}[style=CStyle,tabsize=2,language=ML,basicstyle=\scriptsize, backgroundcolor=\color{white},
    %    ,escapechar=/
      ]
int str_cmp(char * r){
  char * s = 'Bg\$4\0';
  i = 0;
  while (s[i] != '\0'){
    i++;
    if (s[i]!=r[i]) return 0;
    }
    return 1;
}    \end{lstlisting}
    \end{subfigure}
    \caption{
    \color{cavcolor}
    String comparison \cite{abrahamParameterSynthesis2020}.
    }
    \label{fig:string}
\end{figure}

    We model this by mapping strings to integers and allowing $2^{8n}$ different input strings for a string length of $n$.
    As for \TA, our model differs from \cite{abrahamParameterSynthesis2020,andriushchenkoDeductiveController2023}: In our model, (memoryful randomized) MDP schedulers correspond to thread schedulers, while they hard-coded fair thread schedulers in their models, and a (memoryless deterministic) scheduler for the MDP corresponds to a secret input. 

    We check the same property as for \TA:
    In our experiments, we compare string inputs represented by the integers 0 and 1, for $j=0$.
    We check the property for $M \in \{2,4\}$ where $M=2n$ where $n$ is the length of the string.

\subparagraph*{\TS}
    Consider the following classic example of a simple insecure multi-threaded program~\cite{smithProbabilisticNoninterference2003} 
    \[ 
    th \colon \text{\bf{while }} h>0 \text{\bf{ do }} \{h \leftarrow h-1\};\; l \leftarrow 2 
    \quad \parallel \quad
    th' \colon l \leftarrow 1  
    \]
    where $h$ is a secret input and $l$ a public output.
    Intuitively, this program is not secure because the final value of $l$ allows to make a probabilistic inference on the initial value of the secret input $h$.
    For example, under a fair scheduler (that schedules all available threads with the same probability), the probability of outputting $l=2$ is much higher than the probability of outputting $l=1$.
    Formally, the program violates \emph{scheduler-specific probabilistic observational determinism} (SSPOD)~\cite{ngoConfidentialityProbabilistic2013}.
    SSPOD stipulates that for all schedulings of the threads 
    the probability of observing $l=1$ in the end should be the same as the probability of observing $l=2$. 
    \cite{abrahamProbabilisticHyperproperties2020} formulates the desired property in \HyperPCTL as follows:
    \[\forall \varsched .\ \forall \varstate_1(\varsched) \forall \varstate_2(\varsched) .\ 
    (\init_{\varstate_1} \wedge \init_{\varstate_2}) 
    \implies 
    \bigwedge_{j=1}^{2} \Prob(\Finally (l=j)_{\varstate_1}) = \Prob(\Finally (l=j)_{\varstate_2})\]
    where 
    $\init$ marks initial states for different values of the secret input $h$.
    
    This is equivalent to the following conjunction of universally quantified relational reachability properties:
    \begin{align*}
        \bigwedge_{s_1, s_2 \in \Init} \bigwedge_{j=1}^{2} \forall \sched .\ \Pr^{\sched}_{\state_1}(\Finally (l=j)) = \Pr^{\sched}_{\state_2}(\Finally (l=j))
        ~.
    \end{align*}
    where $\Init$ is the set of states representing initial states of the program for different values of $h$ (up to a certain bound).

    We compare the initial values of $h$ indicated in the table, for $j=1$.
    Note that this is actually equivalent to checking the full conjunction since $\Pr^\sched_s(\protect\Finally (l=1)) = 1 - \Pr^\sched_s(\protect\Finally (l=2))$.

\subparagraph*{\SD}    
    \cite{andriushchenkoDeductiveController2023} extend the notion of \emph{stochastic domination} \cite{bartheProbabilisticCouplings2020} to MDP states by defining that a state $s_1$ stochastically dominates a state $s_2$ w.r.t.\ a target set $T$ iff 
    \[ 
        \forall \sched .\ \Pr^{\sched}_{s_1}(\Finally T) \geq \Pr^{\sched}_{s_2}(\Finally T) ~,
    \]
    which is a natural universally quantified relational reachability property.
    In \cite{andriushchenkoDeductiveController2023}, this property is applied to various robot-maze problems with a fixed pair of initial locations each, checking whether the first location can guarantee a better reachability probability than the other location, no matter how the robot behaves.
    We check the same property (over memoryful randomized schedulers) on all mazes provided in \cite{andriushchenkoDeductiveController2023}.